\documentclass[10pt,a4paper,reqno]{amsart}
\usepackage{amsmath,amsthm,amsfonts,amssymb,euscript,bbm,color,slashed, xcolor, stackrel, booktabs}

\usepackage{marvosym,accents, url}
\usepackage{subcaption}
\usepackage[hidelinks]{hyperref}

\newtheorem{theorem}{Theorem}[section]
\newtheorem{proposition}[theorem]{Proposition}
\newtheorem{lemma}[theorem]{Lemma}
\newtheorem{remark}[theorem]{Remark}
\newtheorem{definition}[theorem]{Definition}
\newtheorem{corollary}[theorem]{Corollary}

\usepackage[margin=1in,dvips,marginpar=.6in]{geometry}

\newcommand{\abs}[1]{\left\vert#1\right\vert}
\newcommand{\jb}[1]{\left\langle#1\right\rangle}

\numberwithin{equation}{section}
\numberwithin{theorem}{section}

\usepackage{bbm}
\usepackage{graphicx}
\usepackage{enumerate}

\def\th {\theta}

\def\Hb {\underline{H}}
\def\chib {\underline{\chi}}
\def\chih {\hat{\chi}}
\def\chibh {\hat{\underline{\chi}}}
\def\omegab {\underline{\omega}}
\def\etab {\underline{\eta}}
\def\xib {\underline{\xi}}
\def\betab {\underline{\beta}}
\def\alphab {\underline{\alpha}}

\def\hot{\widehat{\otimes}}

\def\alp {\alpha}

\def\bt {\beta}

\def\nab {\slashed{\nabla}}

\def\ep {\epsilon}

\def\omb {\underline{\omega}}

\def\f {\frac}

\def\rd {\partial}
\def\ls {\lesssim}
\def\de {\delta}
\def\i {\infty}

\newcommand{\ud}{\mathrm{d}}

\def\xib {\underline{\xi}}
\def\R{\mathrm{Ric}}
\def\sR{\slashed{\mathrm{Ric}}}

\makeatletter
\newcommand{\brk}{\@ifstar{\brkb}{\brki}}
\newcommand{\brki}[1]{\langle{#1}\rangle}
\newcommand{\brkb}[1]{\left\langle{#1}\right\rangle}
\makeatother

\newcommand{\pfstep}[1]{\vspace{.5em} {\it \noindent #1.}}

\newcommand{\bea}{\begin{eqnarray}}
\newcommand{\eea}{\end{eqnarray}}
\def\beaa{\begin{eqnarray*}}
\def\eeaa{\end{eqnarray*}}

\renewcommand{\div}{\slashed{\mathrm{div}}}
\newcommand{\curl}{\slashed{\mathrm{curl}}\,}
\newcommand{\trchb}{\slashed{\mathrm{tr}}\chib}
\def\trch{\slashed{\mathrm{tr}}\chi}
\newcommand{\tr}{\slashed{\mathrm{tr}}}

\makeatletter
\newcommand{\nrm}{\@ifstar{\nrmb}{\nrmi}}			
\newcommand{\nrmi}[1]{\Vert{#1}\Vert}
\newcommand{\nrmb}[1]{\left\Vert{#1}\right\Vert}
\newcommand{\Nrm}{\@ifstar{\Nrmb}{\Nrmi}}			
\newcommand{\Nrmi}[1]{\vert\kern-0.25ex\vert\kern-0.25ex\vert{#1}\vert\kern-0.25ex\vert\kern-0.25ex\vert}
\newcommand{\Nrmb}[1]{\left\vert\kern-0.25ex\left\vert\kern-0.25ex\left\vert {#1}\right\vert\kern-0.25ex\right\vert\kern-0.25ex\right\vert}
\makeatother

\newcommand{\aleq}{\lesssim}

\title[Stability of Minkowski]{Stability of the Minkowski spacetime in Newman--Unti gauge}

\author{Jonathan Luk}
\address{Department of Mathematics, 450 Jane Stanford Way, Stanford University, CA 94305, USA}
\email{jluk@stanford.edu}

\author{Sung-Jin Oh}
\address{Department of Mathematics, UC Berkeley, Berkeley, CA 94720, USA and KIAS, Seoul, Korea 02455}
\email{sjoh@math.berkeley.edu}

\author{Claude Warnick}
\address{DAMTP and DPMMS, Centre for Mathematical Sciences, Wilberforce Road, Cambridge CB3 0WA, UK}
\email{c.m.warnick@maths.cam.ac.uk}

\begin{document}

\begin{abstract}
We prove small-data global stability of the Minkowski solution to Einstein's equations in a  centre-normalised outgoing null-geodesic gauge. Our scheme involves first using the $r^p$-estimates of Dafermos--Rodnianski to control certain components of the Weyl tensor which satisfy a decoupled tensorial wave equation. Having established this control, all remaining geometric quantities are controlled by transport equations, taking initial conditions at a regular central axis. This method establishes global stability for initial data which decay only weakly to flat space and can establish additional asymptotic control when the data are assumed to have more structure.
\end{abstract}




\maketitle

\section{Introduction}
In this paper, we investigate the global nonlinear asymptotic stability of the Minkowski spacetime as a solution to the Einstein vacuum equations,
\begin{equation} \label{eq:EVE}
	\mathrm{Ric}(g) = 0.
\end{equation}
This problem possesses a rich history. The first proof of stability was provided by Friedrich \cite{Friedrich1986OnStructure} for asymptotically hyperboloidal data. For the more general class of asymptotically flat data, stability was first established in the monumental work of Christodoulou--Klainerman \cite{CK94}. Subsequent developments have further refined our understanding, including works by Bieri \cite{Bieri2010AnRelativity} and Shen \cite{Shen2023GlobalDecay, Shen2024StabilityRegions} (weakening initial data decay assumptions), Lindblad--Rodnianski \cite{Lindblad2005GlobalCoordinates} (utilizing wave coordinates), Hintz--Vasy \cite{Hintz2017StabilityMetric} (establishing full asymptotic expansions at infinity in wave coordinates), and Ionescu--Pausader \cite{Ionescu2022TheAMS-213} (developing a normal form analysis framework in wave coordinates). For a more comprehensive discussion of these and other prior works, we refer the reader to \S\ref{sec:lit}.

The main result of this paper is a new proof of stability of the Minkowski spacetime in coordinates $(u, r, \theta^{1}, \theta^{2})$, with respect to which the metric takes the following form:
\begin{equation*}
	g = - (\ud u \otimes \ud r+\ud r\otimes \ud u) - f \ud u \otimes \ud u +\gamma_{AB}(\ud\th^A-b^A \ud u)\otimes (\ud\th^B - b^B \ud u).
\end{equation*}
This choice, and the hierarchy of equations we shall use to control the solution, traces back to the work of Newman--Unti \cite{Newman1962BehaviorSpaces} in 1962; for this reason, we refer to it as the \emph{Newman--Unti gauge}. Its defining features are that i)~$u$ is null, and ii) the coordinate vector field $\rd_{r}$ is null-geodesic. This gauge has a long history in the physics literature (particularly, in the study of gravitational radiation), and it has also appeared in several recent mathematical works, including \cite{sAntC2022, xtCsK2026, KS.Kerr}.

What distinguishes our work from the above is a simple normalisation we adopt to fix the remaining coordinate ambiguity, which we term \emph{centre-normalisation}: starting with a fixed timelike geodesic $\Gamma$ (the centre), each level-$u$ hypersurface $H_{u}$ is the null cone generated by all outgoing null geodesics emanating from the point on $\Gamma$ with proper time $u$; for fixed $u$ and $\theta$, the curve $r \mapsto (u, r, \theta)$ represents each such null generator with affine parametrisation; in particular, $\rd_{r}$ is null geodesic. (For the detailed construction, see \S\ref{sec:coords}.) In Minkowski spacetime, if we consider the standard polar coordinates $(t, r, \theta)$, then $\{ r = 0\}$ is the centre geodesic $\Gamma$, $u = t- r$ is the standard retarded time, and the metric takes the form $f = 1$, $b^{A} = 0$, and $\gamma = \mathring{\gamma}$, where $\mathring{\gamma}$ is the standard metric on the round sphere of radius $r$. 

Given the existing literature, one might ask: Why another proof of the stability of the Minkowski spacetime? As already noted in \cite{Newman1962BehaviorSpaces}, the analytic structure of \eqref{eq:EVE} in the Newman--Unti gauge is particularly elegant: the entire dynamics of \eqref{eq:EVE} is represented by the curvature components $\alp_{AB} = R(\rd_{r}, \rd_{\theta^{A}}, \rd_{r}, \rd_{\theta^{B}})$, which, to linear order around Minkowski, solves a decoupled wave equation (Teukolsky equation). Under our centre-normalisation, all metric components are then recovered by integrating the fundamental equations of geometry (i.e., Bianchi and null structure equations) along the generators of the null cones $H_{u}$ starting from the centre $\Gamma$. 

Building upon this simple hierarchical structure, our primary objective is to \emph{develop a streamlined framework for proving stability}. Our proof is concise and encompasses the entire\footnote{The only result in the literature whose decay assumption we do not handle seems to be that of Shen \cite{Shen2024StabilityRegions}, which handles the endpoint case in the setting of exterior stability (i.e., in the domain of dependence of the exterior of a compact set).} spectrum of initial decay assumptions of interest, ranging from the optimal (up to the endpoint) weak decay assumption of Shen \cite{Shen2023GlobalDecay}, to the strongly decaying data originally considered by Christodoulou--Klainerman, to new classes of further structured data that lead to the almost sharp Bondi--Sachs peeling. Furthermore, our framework readily extends to settings such as exterior stability \cite{Klainerman2003TheNicolo.} and the spacelike-characteristic initial value problem \cite{Graf2020GlobalData} (see \S\ref{sec:appendix} for further discussion).

In addition, \emph{the centre-normalised Newman--Unti gauge turns out to be exceptionally well-suited for studying global asymptotics}. This utility is perhaps surprising, given that our gauge is normalised at the centre $\Gamma$ rather than future infinity (see also the discussion of divergent behaviour of metric coefficients below). Nonetheless, in this gauge, all geometric quantities we consider exhibit simple asymptotic behaviour as $u \to +\infty$ (referred to as late-time tails). Our stability proof shows that for solutions with weakly decaying initial data, all curvature components obey the global upper bound $O(\epsilon^{2} r^{-1^{-}} \jb{u}^{-1^{+}})$, where $\epsilon^{2}$ is the initial data size and $\jb{\cdot} := (1+(\cdot)^{2})^{\frac{1}{2}}$. On the other extreme, for solutions exhibiting almost sharp Bondi--Sachs peeling, all curvature components have a decay rate of at least $O(\epsilon^{2} r^{-1^{-}} \jb{u_{+}}^{-4})$, where $(\cdot)_{+} := \max \{(\cdot), 0\}$; this is a new result in this generality. Consistent with the hierarchical structure of the equations, the Ricci (or connection) coefficients and metric components exhibit the same upper bound, but multiplied by $r$ and $r^{2}$, respectively.

Furthermore, combined with the methods of Luk--Oh \cite{LO}, the centre-normalised Newman--Unti gauge allows for the proof of \emph{sharp Bondi--Sachs peeling} for general strongly asymptotically flat initial data (whose precise formulation is an appropriate extension of Definition~\ref{def:AF.data}). One may additionally conjecture that in this gauge, the generic local-in-space decay rate is $u^{-6}$ for the curvature components and $r^{2} u^{-6}$ for the metric components, consistent with the modified Price's law conjecture of Luk--Oh \cite{LO}. This would be markedly different from the case of wave coordinates \cite{hL2017, ycM2026}.

Despite these advantages, the usefulness of the centre-normalised Newman--Unti gauge for proving stability may seem doubtful at first glance. First, there is an apparent loss of derivatives in this gauge \cite{xtCsK2026} so that different connection coefficients appear to have different regularity. More seriously, integrating the equations from the centre reveals that, no matter how fast $\alpha_{AB}$ is assumed to decay in $r$, the metric components diverge(!) from the Minkowskian values as $r \to \infty$:
\begin{equation} \label{eq:metric-diverge}
	f -1 = O(r), \quad \abs{b} = O(r), \quad \abs{\gamma - \mathring{\gamma}} = O(1) \quad \hbox{ on each fixed } H_{u}.
\end{equation}
Our key new insight is that \emph{neither the loss of derivatives nor the divergence of the metric components obstruct sharp wave equation estimates for $\alpha_{AB}$ on this background}. Indeed, \eqref{eq:metric-diverge} notwithstanding, the estimates are accompanied by enough $u$-decay and smallness, which we justify via a suitable bootstrap argument, so that they can still be treated perturbatively. This observation represents a new ramification of the remarkable null structure of the Einstein equations, a classical theme that underlies the proof of Christodoulou--Klainerman \cite{CK94} and all subsequent works. We refer to \S\ref{subsec:ideas} below for a more detailed discussion of the main ideas of the proof.

\begin{remark}\label{rem:BMS}
The Newman--Unti gauge is closely related to the Bondi--Metzner--Sachs (BMS) gauge often used in the physics literature to describe the precise asymptotics of asymptotically flat spacetimes \cite{hB1960, hBmgjvdBawkM1962, rkS1962}. Indeed, the definitions of the two gauges differ only by the choice of the `radial' coordinate $r$: it is an affine parameter on generators of $H_{u}$ in the Newman--Unti gauge, and it is the areal coordinate in the BMS gauge. However, we point out that both gauges used in this context are typically normalised at future infinity, which is different from our centre-normalisation; in particular, the metric behaviour in \eqref{eq:metric-diverge} is precluded. 
\end{remark}

\subsection{First statement of the main theorem}

\begin{definition}
    For $s\in \mathbb Z_{\geqslant 0}$ and $\nu \in \mathbb R$, define $H^{s,\nu} = H^{s,\nu}(\mathbb R^3)$ to be the completion of $C^\infty_c(\mathbb{R}^3;\mathbb{R})$ with respect to the norm
    $$\|f\|_{H^{s,\nu}} = \sum_{|\alp| \leqslant s} \| (1+|x|)^{|\alp|+\nu} \rd^\alp f\|_{L^2}.$$
    We extend this definition to subsets of $\mathbb{R}^3$ and to tensors by working componentwise in Cartesian coordinates in the obvious way.
\end{definition}

\begin{definition}[Asymptotic flat initial data sets in centre-of-mass frame]\label{def:AF.data}
Let $(\Sigma, h, k)$ be a smooth set of Cauchy data for \eqref{eq:EVE}, i.e., $h$ is a Riemannian metric and $k$ is a symmetric $2$-tensor on a $3$-manifold $\Sigma$ that satisfy the \emph{Einstein vacuum constraint equations}:
\begin{equation} \label{eq:EVE-constraint}
    \begin{aligned}
        \mathrm{R}[h] + (\mathrm{tr}_{h} k)^{2} - |k|_{h}^{2} &= 0, \\
        \mathrm{div}_{h} k - \ud \mathrm{tr}_{h} k &= 0.
    \end{aligned}
\end{equation}
Assume that $\Sigma$ is diffeomorphic to $\mathbb{R}^{3}$ and thus equipped with coordinates $x^i$, $i=1, 2, 3$. We say $(\Sigma, h, k)$ is \emph{asymptotically flat} in centre-of-mass frame with decay exponent $\nu \in (0, 3)$ and regularity exponent $s \in \mathbb{Z}_{\geqslant 0}$ if the following holds:
\begin{enumerate}[i)]
    \item If $\nu \in (0, 1)$ and we have
    \begin{equation*}
        (h_{ij}, k_{ij}) \in (\delta_{ij} + H^{s, \nu - \frac{3}{2}}, H^{s-1, (\nu + 1) - \frac{3}{2}}).
    \end{equation*}
    \item If $\nu \in [1, 2)$ and, for some $M \geqslant 0$, we have
    \begin{equation*}
        (h_{ij}, k_{ij})|_{\{|x|>1\}} \in ((1 + \tfrac{2M}{|x|}) \delta_{ij} + H^{s, \nu - \frac{3}{2}}, H^{s-1, (\nu + 1) - \frac{3}{2}}).
    \end{equation*}
    \item If $\nu \in [2, 3)$ and, for some $M \geqslant 0$ and $J \in \mathbb{R}^{3}$, we have
    \begin{equation}\label{eq:AF.data.strong}
        (h_{ij}, k_{ij})|_{\{|x|>1\}} \in ((1 + \tfrac{2M}{|x|} + \tfrac{3 M^{2}}{2 |x|^{2}}) \delta_{ij}
        + H^{s, \nu - \frac{3}{2}},  3 (\in_{i k \ell} J^{\ell} \tfrac{x_{j} x^{k}}{|x|^{5}} + \in_{j k \ell} J^{\ell} \tfrac{x_{i} x^{k}}{|x|^{5}}) + H^{s-1, (\nu + 1) - \frac{3}{2}}).
    \end{equation}
\end{enumerate}
\end{definition}

\begin{remark} \label{rem:af-data}
\begin{enumerate}[i)]
    \item The slowly decaying spatial tails, such as $\frac{2M}{r} \delta_{ij}$, are necessitated for large values of $\nu$ by the Einstein vacuum constraint equations \eqref{eq:EVE-constraint}. Indeed, the parameter $M$ in Definition~\ref{def:AF.data}.ii)--iii) is the \emph{ADM mass}, which is strictly positive for any non-Minkowskian data by the Positive Mass Theorem \cite{Schoen1981ProofII, Witten1981ATheorem}.
    \item While structurally necessary, the precise form of these slowly decaying tails is not unique; our definition is based on physical and mathematical considerations. Specifically, Definition~\ref{def:AF.data}.ii) corresponds to the \emph{strongly asymptotically flat} condition originally implemented by Christodoulou--Klainerman \cite{ChrKl90} and utilised in ensuing works. The terms in Definition~\ref{def:AF.data}.iii) correspond to the asymptotic form of the metric far from a localised gravitating system, as found in Misner--Thorne--Wheeler \cite[(19.13)]{mtw}; alternatively, they match the asymptotics of the Cauchy data for the Kerr spacetime on constant Boyer--Lindquist time slices expressed in suitable spatial coordinates (see the proof of Corollary~\ref{cor:idest}). In this context, the parameter $J$ denotes the \emph{ADM angular momentum}. 
    \item In all cases we consider, the ADM linear momentum $P$ and the ADM centre-of-mass $C$, whenever they are well-defined, are fixed to be zero. This property signifies that the Cauchy data are posed in the \emph{centre-of-mass frame}.
    \item The existence of a large class of Cauchy data in Definition~\ref{def:AF.data}.ii)--iii) follows from the recent works of Fang--Szeftel--Touati \cite{ajFjSaT2025} and Chen--Klainerman \cite{xtCsK2025}. In fact, the Cauchy data in Definition~\ref{def:AF.data} is \emph{general} in the small-data regime in the following sense: for every $0 < \nu < 3$, any small Cauchy data of the aforementioned form can be parameterised -- up to a Poincar\'e transformation to pass to the centre-of-mass frame -- by the linear space of all solutions to the linearised constraint equations around the trivial data $(\delta, 0)$ in the weighted spaces $H^{s, \nu-\frac{3}{2}} \times H^{s-1, \nu-\frac{1}{2}}$. For a precise implementation of this parameterisation, see the upcoming work \cite{MOT2026}. 
\end{enumerate}
\end{remark}

\begin{theorem} [Main theorem, first formulation] \label{thm:main-1}
Suppose $(\mathbb{R}^{3}, h, k)$ is an asymptotically flat initial data set in centre-of-mass frame with any positive decay exponent $\nu > 0$ and regularity exponent\footnote{This choice of $s$ is not optimal: here and below we make no attempt to obtain sharp results with respect to regularity.\label{fn:regularity}} $s\geqslant 30$. There exists $\mathfrak{e}=\mathfrak{e}(s, \delta)>0$ such that if the following small data condition is satisfied
\begin{equation*}
	\mathfrak{n} := \| h_{ij} - \delta_{ij} \|_{H^{s, \delta-\frac{3}{2}}} + \| k_{ij} \|_{H^{s-1, \delta-\frac{1}{2}}} < \mathfrak{e},
\end{equation*}
for some $\delta < \min\{ \nu, \frac{1}{20} \}$, then the following conclusions hold.
    \begin{enumerate}[i)]
        \item The maximal future Cauchy development $(M, g)$ is globally smooth and future complete and may be covered by a single centre-normalised Newman--Unti coordinate chart, which is regular away from the centre geodesic (for the precise definition, see \S\ref{sec:coords}).
        \item $(M, g)$ is an asymptotically flat spacetime, in the following sense. Define the curvature components\footnote{In the Newman--Penrose formalism using the null tetrad with $l = e_{4}$ and $n = e_{3}$, $\alpha$, $\beta$, $(\rho, \sigma)$, $\betab$, and $\alphab$ coincide (up to normalisation) with the complex scalars $\Psi_{0}$, $\Psi_{1}$, $\Psi_{2}$, $\Psi_{3}$, and $\Psi_{4}$, respectively.} (see \cite{CK94} and Definition~\ref{def:weyl})
	\begin{align*}
\alpha_{AB}&= R(e_{A}, e_{4}, e_{B}, e_{4}), &\beta_A&=\frac{1}{2} R(e_{A}, e_{4}, e_{3}, e_{4}), &\rho &= \frac{1}{4} R(e_{4}, e_{3}, e_{4}, e_{3}), \\ \alphab_{AB}&=R(e_{A}, e_{3}, e_{B}, e_{3}),
& \betab_A&=\frac{1}{2}R(e_{A}, e_{3}, e_{3}, e_{4}), &\sigma&=\frac{1}{2} \in^{AB} R(e_{A}, e_{B}, e_{3}, e_{4}),
	\end{align*}
	where $e_{4} = \frac{\rd}{\rd r}$ (which is null and geodesic), $e_{A} = \frac{\rd}{\rd \theta^{A}}$ ($A = 1, 2$; they are orthogonal to $e_{4}$), $e_{3}$ is the unique future-directed null vector characterised by $g(e_{3}, e_{A}) = 0$, $g(e_{3}, e_{4}) = -2$, and $\in_{AB}$ is the volume form on each $S_{u, r}$ (whose indices are raised via $\gamma^{-1}$). We have\footnote{Here, and below, we write $A\ls B$ if $A \leq CB$ for some $C = C(s,\de)>0$.}
	\begin{equation} \label{eq:main-1:curv}
    \begin{aligned}
	|\alpha| + |\beta| + |(\rho, \sigma)| &\aleq \mathfrak{n} (\jb{u}+r)^{-2-\delta} (\tfrac{r}{\jb{u}+r})^{-1+\delta} , \\
	|\betab| &\aleq\mathfrak{n} (\jb{u}+r)^{-2} \jb{u}^{-\delta} (\tfrac{r}{\jb{u}+r})^{-1+\delta}, \\
	|\alphab| &\aleq \mathfrak{n}(\jb{u}+r)^{-1} \jb{u}^{-1-\delta} (\tfrac{r}{\jb{u}+r})^{-1+\delta},
	\end{aligned}
    \end{equation}
	where $| \cdot |$ is defined using the induced metric $\gamma$. In particular, all such curvature components decay to zero along any future-directed causal curve. Moreover, the metric components satisfy the bounds 
        \begin{equation} \label{eq:main-1:metric}
		|f - 1| + |b| \aleq \mathfrak{n} r \jb{u}^{-1-\delta} (\tfrac{r}{\jb{u}+r})^{\delta} , \quad
		|\gamma - \mathring{\gamma}| \aleq \mathfrak{n} \jb{u}^{-\delta} (\tfrac{r}{\jb{u}+r})^{1+\delta} .
	\end{equation}
        \item \label{item:main.thm.3} If, moreover, $\nu>3-\delta$ and there exist $R_0>0, M>0, J\in \mathbb{R}^3$ such that 
     \begin{align*}
    \mathfrak{n}'&:= \left\|h_{ij} - (1 + \tfrac{2M}{|x|} + \tfrac{3 M^{2}}{2 |x|^{2}}) \delta_{ij}\right\|_{H^{s'+1, \nu-\frac{3}{2}}(\{|x|>R_0\})}\\&\qquad + \left\|k_{ij}-3 (\in_{i k \ell} J^{\ell} \tfrac{x_{j} x^{k}}{|x|^{5}} + \in_{j k \ell} J^{\ell} \tfrac{x_{i} x^{k}}{|x|^{5}})\right\|_{H^{s', \nu-\frac{1}{2}}(\{|x|>R_0\})}  <\infty,
    \end{align*}
    then there exists $\mathfrak{e}'=\mathfrak{e}'(s,\de,\nu, R_0)>0$ such that if
    \[
    \mathfrak{n}+\mathfrak{n}'+M+|J|=: \mathfrak{m}<\mathfrak{e}'
    \]
   we have: 
	\begin{equation} \label{eq:main-1:imp}
    \begin{aligned}
	|\alpha| &\aleq \mathfrak{m}(\jb{u}+r)^{-(5-\delta)} (\tfrac{r}{\jb{u}+r})^{-1+\delta}, 
	& |\beta| &\aleq \mathfrak{m}(\jb{u}+r)^{-4} \brk{u_{+}}^{-(1-\delta)} (\tfrac{r}{\jb{u}+r})^{-1+\delta},  \\
	|(\rho, \sigma)| &\aleq\mathfrak{m} (\jb{u}+r)^{-3} \brk{u_{+}}^{-(2-\delta)} (\tfrac{r}{\jb{u}+r})^{-1+\delta}, 
	& |\betab| &\aleq \mathfrak{m} (\jb{u}+r)^{-2} \brk{u_{+}}^{-(3-\delta)} (\tfrac{r}{\jb{u}+r})^{-1+\delta}, \\
	|\alphab|&\aleq\mathfrak{m} (\jb{u}+r)^{-1} \brk{u_{+}}^{-(4-\delta)} (\tfrac{r}{\jb{u}+r})^{-1+\delta}. & &
	\end{aligned}
    \end{equation}
    In particular, for every $u$ as $r \to \infty$ (i.e., towards future null infinity), we have \emph{almost sharp Bondi--Sachs peeling}.
    \end{enumerate}
\end{theorem}

\begin{remark} \label{rem:main-1}
\begin{enumerate}[i)]
\item We say that the spacetime metric satisfies Bondi--Sachs peeling if for every fixed $u \in \mathbb R$, it holds that 
\begin{equation}\label{eq:Bondi.Sachs.def}
    |\alp| \ls r^{-5},\quad |\bt| \ls r^{-4},\quad |(\rho,\sigma)|\ls r^{-3},\quad |\betab|\ls r^{-2},\quad |\alphab| \ls r^{-1}.    
\end{equation}
In \eqref{eq:main-1:imp}, these bounds hold for all curvature components other than $\alp$, for which a slightly weaker estimate holds. We thus call this `almost sharp Bondi--Sachs peeling.'
\item The ideas in this paper are already sufficient for propagating sharp Bondi--Sachs peeling estimates that already hold for spacelike-characteristic initial data; see Remarks~\ref{rem:strong-peeling-rp}, \ref{rem:strong-peeling-geom} and \ref{rem:strong-peeling-spnull}. 
\item In the intermediate regime where $0<\nu<3-\de$ a structurally similar proof will give partial peeling results, where some but not all curvature components exhibit sharp fall-off. For simplicity we only include the endpoints here, but there is no obstruction to repeating the proof for different $\nu$.
\item The exponent $\delta$ above $\tfrac{r}{\jb{u}+r}$ capturing the degeneration at $r=0$ does not need to coincide with the decay rate of the initial data; we have chosen to do so merely for simplicity.
\end{enumerate}

\end{remark}

\subsection{Main ideas of the proof}  \label{subsec:ideas}
We employ the following notation in this subsection.
As was introduced in Theorem~\ref{thm:main-1}, let $e_{4} = \frac{\rd}{\rd r}$, $e_{A} = \frac{\rd}{\rd \theta^{A}}$, and $e_{3} = 2(\frac{\rd}{\rd u} + b^{A} \frac{\rd}{\rd \theta^{A}}) - f \frac{\rd}{\rd r}$ (note that, indeed, $g(e_{3}, e_{3}) = g(e_{3}, e_{A}) = 0$ and $g(e_{4}, e_{3}) = -2$). We say that a tensor field $\phi$ is \emph{$S-$tangent} if it coincides with its projection (defined using $g$-orthogonality) to $S_{u, r}$; important examples of $S-$tangent fields are the curvature components $\alpha$, $\beta$, etc. Let $\nab_3$, $\nab_4$, $\nab_A$ be respectively the projections (defined using orthogonality via $g$) to $S_{u,r}$ of the covariant derivatives $D_3= D_{e_3}$, $D_4 = D_{e_4}$, and $D_A = D_{\rd_{\theta^{A}}}$. In what follows, we shall often omit the subscript $A$ and write $\nab$ for the $S-$tangential gradient. 

\subsubsection{Transport equations along $e_{4}$} \label{subsubsec:transport}
As observed originally by Newman--Unti \cite{Newman1962BehaviorSpaces}, given the curvature components $\alp$ it is possible to recover all coefficients describing the geometry of $(M, g)$ by integrating transport equations along $e_{4}$. For instance, the following chain of transport equations -- arising from the second Bianchi identity and structure equations -- may be used to recover the metric coefficient $f$:
\begin{align*}
\nab_{4} \beta + \frac{4}{r} \beta = \div \alpha + \cdots, \quad
\nab_{4} \rho + \frac{3}{r} \rho = - \div \beta + \cdots, \quad
\nab_{4} \omegab = \rho + \cdots, \quad
\nab_{4} f = 2 \omegab,
 \end{align*}
where we have suppressed the nonlinear terms, which play no important role in this discussion. Moreover, the scalar quantity $\omegab$ is one of the Ricci coefficients, whose precise definition is not needed here (see \S\ref{subsec:ricci} below). Integrating these equations from the trivial initial conditions on the centre geodesic $\{ r = 0 \}$, using the heuristic that every angular derivative (such as $\div$) improves the pointwise decay by $r^{-1}$, and ignoring the nonlinearity, higher order derivatives, and the behaviour near the centre for now (see \S\S\ref{subsubsec:comm}, \ref{subsubsec:centre} for a further discussion of these points), we see that
\begin{align*}
	\alpha = O(\epsilon^{2} (r + \jb{u})^{-2-\delta}) \: &\Rightarrow  \: \beta = O(\epsilon^{2} (r + \jb{u})^{-2-\delta}) \: \Rightarrow
    \rho = O(\epsilon^{2} (r + \jb{u})^{-2-\delta}) \\
    &\Rightarrow  \:  \omegab = O(\epsilon^{2} \jb{u}^{-1-\delta}) \: \Rightarrow  \:  f - 1 = O(\epsilon^{2} r \jb{u}^{-1-\delta}).
\end{align*}
Here, $1$ is the boundary value of $f$ on $\{r = 0\}$. As discussed earlier, $f-1$ exhibits a growth in $r$, but it is accompanied by smallness (i.e., $\epsilon^{2}$) and a decay in $u$ with an integrable rate. 
Similar computations show the rest of the metric component bounds in \eqref{eq:main-1:metric}, as well as the almost sharp Bondi--Sachs peeling bounds (in the latter case, one starts with $\alpha = O(\epsilon^{2} (r + \jb{u})^{-(5-\delta)})$).

The general strategy of combining wave estimates for $\alp$ and transport estimates after choosing a suitable gauge is reminiscent of an essential step in the proof of black hole stability \cite{DHR, DHRT, GKS, KS.Kerr}. Notice, however, in all those works wave estimates were needed for both the curvature components $\alp$ and $\alphab$, and, more importantly, the gauge is normalised to the future at infinity, in constrast to the centre normalisation that we use. 

\subsubsection{Wave estimates for the Teukolsky equation} \label{subsubsec:wave}
As is well-known, the second Bianchi identity and the Einstein vacuum equation imply second-order tensorial wave equations for curvature components. In particular, for the components $\alpha_{AB}$, we have the equation
\begin{equation} \label{eq:ideas:teukolsky}
    \nab_3(r \nab_4(r^2\alpha)) + 2 \nab_4 (r^2 \alpha) + 3 \nab_3 (r^2 \alpha) + \frac{4}{r} (r^2 \alpha) -  r \slashed\Delta (r^2 \alpha) = \cdots,
\end{equation}
which we refer to as the \emph{Teukolsky equation} in reference to black hole perturbation theory. The omitted terms are nonlinear and vanish for linearised gravity around the Minkowski spacetime.

To derive estimates for $\alpha$, the key question we need to address is: \emph{How does one analyse \eqref{eq:ideas:teukolsky} under the metric assumption \eqref{eq:main-1:metric}}, which allows divergence from the Minkowski values? Our basic observation is that (a suitable variant of) the $r^{p}$ multiplier of \cite{Dafermos2009ASpacetimes} is viable, thanks to the smallness and $u$-decay of $f-1$.

To wit, we contract (using the induced metric $\gamma$) the equation against $w (r \nab_{4} (r^{2} \alpha) + 3 r^{2} \alpha)$, where $w = w(u, r)$ is a positive function to be chosen below, and integrate over $S_{u, r}$. Then we obtain
\begin{align*}
    & \frac{\partial}{\partial r} \int_{S_{u, r}}\left[ \frac{1}{2} \abs{r\nab(r^{2} \alpha)}^2 + \abs{r^{2} \alpha}^2 - \frac{f}{2}\abs{r \nab_4 (r^{2} \alpha) + 3 r^{2} \alpha}^2  \right] w\ud A_{\gamma} + \frac{\partial}{\partial u} \int_{S_{u, r}}\abs{r \nab_4 (r^{2} \alpha) + 3 r^{2} \alpha}^2  w\ud A_{\gamma}\nonumber \\ 
    & \qquad + \int_{S_{u,r}} \bigg[ \frac{1}{r}\left(1 + \frac{1}{2}\frac{r \partial_r w}{w} - \frac{r \partial_u w}{w} + \frac{f-1}{2}\frac{r \partial_r w}{w} \right)\abs{r \nab_4 (r^{2} \alpha) + 3 r^{2} \alpha}^2 + 2 r\abs{\nab_4 (r^{2} \alpha)}^2 \\
    &\qquad 
    + \frac{8}{r} (r^{2} \alpha) \cdot(r \nab_4 (r^{2} \alpha))
    + \frac{1}{2 r}\left( 4 -\frac{r \partial_r w}{w}\right) \abs{r\nab(r^{2} \alpha)}^2 + \frac{1}{r}\left(10 -\frac{r \partial_r w}{w} \right) \abs{r^{2} \alpha}^2\bigg] w\ud A_{\gamma} = \cdots,
\end{align*}
where we omitted all nonlinear terms that do not contain metric component deviations $f-1, \gamma - \mathring{\gamma}, b$. While we shall not go through the details of this computation here (the interested reader may refer to \S\ref{TeukEsts}), we note the following main points:
\begin{enumerate}[-]
    \item The metric component $f$ appeared through the coordinate expression $e_{3} = 2 (\rd_{u} + b^{A} \rd_{\theta^{A}}) - f \rd_{r}$ embedded in $\nab_{3}$.
    \item However, \emph{none of the other metric component differences $\gamma - \mathring{\gamma}$ and $b$ appear explicitly in this computation}, thanks to working covariantly on the spheres $S_{u, r}$. Specifically, $\gamma$ is implicitly present in contractions and integrals over the spheres $S_{u, r}$, and $b$ disappears (or more precisely, it is replaced by more favorable Ricci coefficients) by integration by parts on the spheres.
\end{enumerate}
Following \cite{Dafermos2009ASpacetimes} (but with a different normalisation), consider the choice $w = r^{2 p - 3}$ for some $p$ to be fixed below, so that $\frac{r \rd_{r} w}{w} = 2 p - 3$ and $\frac{r \rd_{u} w}{w} = 0$.  Integrating over spacetime regions bounded by the initial spacelike hypersurface $\tilde{\Sigma}_{1}$ and level-$u$ hypersurfaces to its future (denoted by $H_{u}$), we obtain
\begin{align*}
    & \sup_{u} \int_{H_{u}} \abs{r \nab_4 (r^{2} \alpha) + 3 r^{2} \alpha}^2  w \ud A_{\gamma} \ud r 
    + \int \int_{H_{u'}} \bigg[ \left(1 + \frac{2p-3}{2}  \right)\abs{r \nab_4 (r^{2} \alpha) + 3 r^{2} \alpha}^2 + 2 \abs{r \nab_4 (r^{2} \alpha)}^2 \\
    &\qquad + 8 (r^{2} \alpha) \cdot(r \nab_4 (r^{2} \alpha)) 
    + \frac{1}{2}\left( 4 - (2p-3) \right) \abs{r\nab(r^{2} \alpha)}^2 + \left(10 - (2p - 3) \right) \abs{r^{2} \alpha}^2\bigg] \frac{w}{r} \ud A_{\gamma} \ud r \ud u\\
    &\qquad 
    + (2p-3)\int \int_{H_{u'}} \frac{1}{r} \frac{f-1}{2} \abs{r \nab_4 (r^{2} \alpha) + 3 r^{2} \alpha}^2 w \ud A_{\gamma} \ud r \ud u 
    = \hbox{(Data on $\tilde{\Sigma}_{1}$)}
    + \cdots,
\end{align*}
where for simplicity we have ignored the issues of convergence of the integrals and vanishing of a boundary term as $r \to +\infty$ (we remark that, since constant-$(u+r)$ hypersurfaces are not necessarily spacelike under our assumptions, for a proper treatment of these issues, we perform a bootstrap argument using a meticulously chosen time function $\tau$; see \S\ref{sec:tau} for its definition). At this point, observe that \emph{all terms on the left-hand side, except for the last term} (involving $f-1$), \emph{are positive for a suitable range of $p$}. More precisely, if
\begin{equation} \label{eq:ideas:p}
    -\frac{\sqrt{17}}{2}+1 < p < \frac{\sqrt{17}}{2}+1,
\end{equation}
then the quadratic form inside $[\cdots]$ in the variables $(r \nab_{4}(r^2 \alpha), r^2 \alpha)$ is positive definite, and the coefficient $(4 - (2p -3))$ in front of the remaining term $|r \nab(r^2 \alpha)|^{2}$ is also positive.
Observe furthermore that, thanks to the integrability in $u$ and the smallness (for $\epsilon$ sufficiently small) of $\frac{1}{r}(f-1)$, \emph{the last term on the left-hand side can be absorbed into the first term}. In conclusion, we arrive at the estimate
\begin{equation} \label{eq:ideas:wave-rp}
\begin{aligned}
    &\sup_{u} \int_{H_u} \abs{(r \nab_4  + 3)(r^{2} \alpha)}^2 w \ud A_{\gamma} \ud r \\
        &\qquad + \int \int_{H_u} \left[ \abs{r\nab_4(r^{2} \alpha)}^2 + \abs{r\nab(r^{2} \alpha)}^2 +\abs{r^{2} \alpha}^2\right]\frac{w}{r} \ud A_{\gamma} \ud r \ud u \aleq \hbox{(Data)} + \cdots.    
\end{aligned}
\end{equation}

Starting from the simple observations that we just made, we obtain progressively stronger estimates on $\alpha$ by making the following modifications:
\begin{enumerate}[-]
    \item We turn the weight $r^{2p-3}$ into a suitable weight of $r + \jb{u}$ away from the centre and $r^{-1}$ near the centre by choosing the weight
    \begin{equation*}
        w(u, r) = \frac{(r+a\jb{u})^{2p_{0}+2p_{\infty}}}{r^{3+2p_{0}}},
    \end{equation*}
    for some exponents $p_{\infty}$, $p_{0}$ and a small positive parameter $a > 0$ to be determined. Observe that this weight coincides with our previous choice with $p = p_\infty$ if $a = 0$; the new exponent $p_{0}$ encodes the strength of the weight $w$ near the centre, whose discussion we postpone until \S\ref{subsubsec:centre}. Essentially the same computation shows that, for $p_{\infty}$ in the range \eqref{eq:ideas:p}, $p_{0}$ in another suitable range (see \eqref{eq:p0bound} below), and $a > 0$ suitably small, then we still have the estimate \eqref{eq:ideas:wave-rp}. 

    We remark that turning the $r$-weight into a corresponding $(r+\jb{u})$-weight is analogous to the dyadic pigeonhole argument in \cite{Dafermos2009ASpacetimes}. However, in our case, we use a simple alternative with a weight that directly grows in $u$ (with a smallness parameter $a > 0$), which works well with centre normalisation.
    
    \item To control other derivatives of $\alpha$ (in particular, $\nab_{3} \alpha$), we use a new weighted multiplier $\tilde{w} r \nab_{3} (r^{2} \alpha)$, where $\tilde{w}$ is comparable to $w$ multiplied by $\frac{\brk{u}}{r + \brk{u}}$ (at least when $p_{\infty} < 3$, which is the case for the proof of Theorem~\ref{thm:main-1}). We obtain:
    \begin{equation} \label{eq:ideas:wave-rp-e3}
\begin{aligned}
    &\sup_{u} \int_{H_u} \left[   \abs{r \nab (r^{2} \alpha)}^2 + \abs{r^{2} \alpha}^2  \right] \frac{\brk{u}}{r + \brk{u}} w \ud A_{\gamma} \ud r \\
        &\qquad + \int \int_{H_u} \abs{r\nab_3 (r^{2} \alpha)}^2 \frac{\brk{u}}{r + \brk{u}} \frac{w}{r} \ud A_{\gamma} \ud r \ud u \aleq \hbox{(Data)} + \cdots.    
\end{aligned}
\end{equation}
    Again, we note that this multiplier is directly effective thanks to centre normalisation.
\end{enumerate}

To prove Theorem~\ref{thm:main-1}, we utilise estimates \eqref{eq:ideas:wave-rp} and \eqref{eq:ideas:wave-rp-e3} with $p_{\infty} = \delta$. Indeed, one may already observe that integrating the transport operator in the first term of \eqref{eq:ideas:wave-rp} leads to an $L^{2}$-estimate for $\alpha$ on each sphere $S_{u, r}$ that would lead to $\alpha = O(\epsilon^{2} (r + \jb{u})^{-2-\delta})$ (assumed in \S\ref{subsubsec:transport}) after we commute with $S-$tangential gradients and apply Sobolev embeddings on $S_{u, r}$ (see \S\S\ref{subsubsec:comm}, \ref{subsubsec:norm-equiv}). For almost sharp Bondi--Sachs peeling, we use $p_{\infty} = 3 - \delta$. We defer the discussion of the choice $p_{0}$ until \S\ref{subsubsec:centre}.

We end the discussion on the estimates for the Teukolsky equation with a few more remarks:
\begin{remark} \label{rem:teukolsky}
\begin{enumerate}[i)]
    \item In contrast to the classical vector field method of Christodoulou--Klainerman \cite{ChrKl90}, there is no need for the construction of any time-like approximately (conformal) Killing vector fields to be used as multipliers. See also Remark~\ref{rem:comm}.ii) for a related comment in the context of commuting vector fields. 
    \item The greater range of $p$ in \eqref{eq:ideas:p} than one expects for the scalar wave equation (see \cite{Dafermos2009ASpacetimes}) reflects the more favourable properties of the Teukolsky equation from the point of view of $r^p$ estimates.
    \item In contrast to the original \cite{Dafermos2009ASpacetimes}, our $r^p$ estimate closes without relying on a separate integrated local energy decay (or Morawetz) estimate. We note that a conceptually similar strategy was recently implemented by Gautam \cite{oG2026} in the context of two-dimensional wave equations, where the standard Morawetz estimate is known to fail due to poor low-frequency behaviour.
\end{enumerate}
\end{remark}

\subsubsection{Higher order derivatives and closing the transport-wave estimates} \label{subsubsec:comm} Next, we turn to the issue of deriving estimates for higher order derivatives of $\alpha$, and discuss how one avoids possible derivative losses when finally closing the $e_{4}$-transport and wave estimates.

Our main commuting operators are $r \nab$ and $r \nab_{4}$, which are weighted derivatives tangent to each $H_{u}$. Let $s_{0}$ be the maximum number of such operators we commute the Teukolsky equation in our argument, and let $\phi$ represent $(r \nab_{4})^{i} (r \nab)^{j} \alpha$ for any $i+j \leqslant s_{0}$. Using an observation going back to \cite{Dafermos2009ASpacetimes} and expanded in \cite{vS2013, gM2016}, estimates analogous to \eqref{eq:ideas:wave-rp} and \eqref{eq:ideas:wave-rp-e3} can be also proved for $(r \nab_{4})^{i} (r \nab)^{j} \alpha$, provided that we have an adequate control of $(r \nab_{4})^{i'} (r \nab)^{j'}$ applied to the right-hand side of \eqref{eq:ideas:teukolsky} for all $i' + j' \leqslant s_{0}$.

In particular, in order to close all the estimates, we need to control appropriate orders of derivatives of the curvature components that arise on the right-hand side of \eqref{eq:ideas:teukolsky}; specifically, we need $s_{0}+1$ many derivatives of $\beta$, and $s_{0}$ many derivatives of $\sigma$ and $K$, the Gauss curvature of the spheres $S_{u, r}$ (see \S\ref{sec:teukolsky} for details). From \eqref{eq:ideas:wave-rp} and \eqref{eq:ideas:wave-rp-e3} for the derivatives of $\alp$ and simply using $e_{4}$-transport equations as in \S\ref{subsubsec:transport}, we control exactly $s_{0}$ derivatives for $\bt$ and $s_{0}-1$ derivatives of $\sigma$ and $K$, which are insufficient. To overcome such possible derivative losses for curvature components, we utilise the following observations:
\begin{enumerate}[-]
    \item The second Bianchi identity relates $\nab_{4} \beta$ and $\nab \hot \beta$ (the symmetric traceless part of $\nab \beta$) with $\div \alpha$ and $\nab_{3} \alpha$, respectively, which are controlled with up to $s_{0}$ many derivatives in the weighted spacetime norms in \eqref{eq:ideas:wave-rp} and \eqref{eq:ideas:wave-rp-e3} (see \eqref{eq:4beta} and \eqref{eq:3alpha}, respectively). Combined with an elliptic estimate on spheres to control $\nab \beta$ via $\nab \hot \beta$, we may control $s_{0}+1$ many derivatives of $\beta$. Observe that at the top order, as for wave equations in general, not all derivatives can be controlled on the null $H_u$ hypersurfaces. As a result, we also use integrated spacetime norms to control certain top derivatives of $\bt$, as well as other quantities that depend on it.
    \item Once we estimate $s_{0}+1$ derivatives of $\bt$, we also control $s_{0}$ derivatives for $\sigma$ and $K$. For $\sigma$, this is immediate from the transport equation; for   $K$, we use the Gauss equation (see \eqref{eq:constrho}) to relate it directly with the null curvature component $\rho$, which is controlled in exactly the same manner as $\sigma$.
\end{enumerate}
Observe from the above discussion that not all curvature components are controlled at the same level of regularity, in contrast to more standard approaches which view the Bianchi equations as a hyperbolic system and treat all curvature components at the same level. Importantly, however, this is consistent with what is need to close the estimates for the Teukolsky equation.

One must also consider the possibility of derivative losses arising from Ricci coefficients, whose discussion we have suppressed so far. Using the transport equations for the Ricci coefficients, one is naturally led to estimates with different orders of derivatives for different Ricci coefficients. Specifically (see definitions in \S\ref{subsec:ricci}), after using the improvements for the curvature components discussed above, we control $s_{0}+1$, $s_{0}$ and $s_{0}-1$ derivatives of $(\chi,\eta)$, $(\omegab,\chib)$ and $\xib$, respectively. Remarkably, this hierarchy of regularity is still consistent with the derivative count needed for the Teukolsky estimates. 

The estimates derived so far allow for further commutation of the Teukolsky equation with $\jb{u} \nab_{u}$, where $\nab_{u} := \frac{1}{2}(\nab_{3} + f \nab_{4})$, but at a lower order. In effect, we may `trade' two copies of $r \nab_{4}, r \nab$ with one $\brk{u} \nab_{u}$. We remark that the commutation with $\jb{u} \nab_{u}$ is \emph{not} needed to close the a-priori estimates above for the $e_{4}$-transport and Teukolsky equations, but it is used to ensure that the spacetime can be smoothly continued (in a different coordinate chart) near the centre; see \S\ref{subsubsec:centre}.

We end this discussion with two more remarks:
\begin{remark} \label{rem:comm}
\begin{enumerate}[i)]
    \item The omitted nonlinear terms in the transport equations and the Teukolsky equation \eqref{eq:ideas:teukolsky} possess several important properties. As mentioned above, the Ricci coefficients appear in a manner that incur no derivative losses (after using the elliptic estimates discussed above). Moreover, the nonlinear terms omit the worst (from the point of view of decay) potential nonlinear interactions. This specific \emph{null structure}, combined with the smallness and $u$-decay of the metric deviations, ensures that the nonlinearities can be treated as error terms in our scheme. For a detailed verification, we refer the reader to \S\ref{sec:teukolsky} and \S\ref{sec:other-geom} below for the Teukolsky and transport equations, respectively.
    \item We point out that $r \nab_{A}$ and $\jb{u} \nab_{u}$ serve as good commuting operators within our scheme, even though the vector fields $r e_{A}$ and $\frac{1}{2}(e_{3} + f e_{4})$ may not even be asymptotically Killing as $r \to +\infty$ (indeed, the $33$-component of the deformation tensor of both may grow as $O(\epsilon r)$).
\end{enumerate}
\end{remark}
\subsubsection{Behaviour near the centre} \label{subsubsec:centre}
So far, we have suppressed the discussion of what happens near the centre. Centre-normalisation necessarily leads to non-smooth estimates near the centre; indeed, $\alpha$ itself cannot be smooth near $\Gamma$ since $e_{4}$ is not regular. Our approach is to show that \emph{the spacetime is nevertheless smooth near $\Gamma$} (in a different coordinate chart than Newman--Unti) by establishing a control on $g$ that is sufficient for the application of the classical (i.e., relying only on energy estimates) low-regularity local well-posedness result, which implies the desired smoothness via persistence of regularity. 

Specifically, we observe that we may take $p_{0} = 1+\delta$ in \S\ref{subsubsec:wave}, which leads to the weight $(\frac{r}{\jb{u}+r})^{-1+\delta}$ in \eqref{eq:main-1:curv}. After transport estimates and commutations, these turn out to imply the following bounds for the induced data $(h, k)$ on each spacelike hypersurface $N$ consisting of geodesics normal to $\Gamma$ (which is smooth):
\begin{equation*}
    |x|^{i} |\nabla^{i} (h - \delta)(x)| + |x|^{i+1} |\nabla^{i} k(x)| \lesssim o_{\epsilon \to 0}(1) |x|^{1+\delta} \quad \hbox{ for $|x|$ small,}
\end{equation*}
for some local coordinates $x$ on $N$ with $x = 0$ at the centre. While it does not directly imply smoothness, this bound nevertheless shows the smallness of $(h - \delta, k)$ in $H^{s} \times H^{s-1}$ near $\Gamma$ with $s > \frac{5}{2}$, which is sufficient for the application of the classical local well-posedness result and persistence of regularity (see, e.g., \cite{Hughes1977Well-posedRelativity}). For details, we refer to \S\ref{sec:centre}.

\subsubsection{Equivalence of norms} \label{subsubsec:norm-equiv}
With our estimates, we are able to always establish the equivalence between the transport/wave estimates discussed above -- which come naturally with norms and derivatives associated with the metric $g$ -- with norms and derivatives defined using the background Minkowski metric; see \S\ref{sec:equiv.norm}. In this fashion, we avoid the need to prove geometric versions of standard analytic inequalities, such as Sobolev embeddings and interpolation.

\subsection{Related works}

\subsubsection{Literature on the stability of the Minkowski spacetime} \label{sec:lit}

%









%

\begin{table}[t]
    \centering
    \begin{tabular}{|c|p{0.25\linewidth}|p{0.16\linewidth}|c|p{0.28\linewidth}|}
    \hline
        Paper & Initial surface & Gauge / scheme & $\nu$ & Dynamical estimates \\
        \hline
        \cite{Friedrich1986OnStructure} & Hyperboloidal  & Conformal method& 3 & Local-in-time energy estimates for conformal eqns\\
        \cite{CK94} & Spacelike & Maximal/null foliation& $\frac{3}{2}$ & Classical vf estimates for Bianchi system \\
        \cite{Lindblad2005GlobalCoordinates} & Spacelike & Wave/harmonic coordinates & $1^+$ & Classical vf estimates for metric components \\
        \cite{Bieri2010AnRelativity} & Spacelike & Maximal/null foliation& $\frac{1}{2}$ & Classical vf estimates for Bianchi system\\
        \cite{Hintz2017StabilityMetric}& Spacelike & Wave-map gauge & $1^+$ &  Classical vf estimates for metric components \\
        \cite{Keir} & Spacelike & Wave/harmonic coordinates & $1^-$ & $r^p$ estimates for metric components \\
        \cite{Graf2020GlobalData} & Spacelike-Characteristic & Maximal/null foliation&   $\frac{3}{2}$ & Classical vf estimates for Bianchi system\\
        \cite{Ionescu2022TheAMS-213} & Spacelike & wave/harmonic coordinates & $1^-$ & Normal form estimates for metric components \\
        \cite{Shen2023GlobalDecay} & Spacelike & Maximal/null foliation& $0^+$ &$r^p$ estimates for Bianchi system \\
        \hline
        This paper& Spacelike & Null foliation& $0^+$ & $r^p$ estimates for Teukolsky wave equation \\
         \hline
    \end{tabular}
    \caption{Aspects of previous proofs.}
    \label{tab:prevwks}
\end{table}

In order to highlight where our approach differs from previous works, we very crudely summarise aspects of past results on stability of Minkowski as a solution of the Einstein-vacuum equations in Table \ref{tab:prevwks}.

Taking the main points of comparison in turn, we first note the character of the surface on which the initial data is prescribed. Our result assumes data on a spacelike surface, however since the gauge is very well adapted to outgoing null cones the results can be adapted with minimal changes to either the initial-characteristic problem (with data for $\alpha$ prescribed on an outgoing cone) or the problem exterior to a sphere, see Appendix \ref{sec:appendix}. Note that in the second case the result will necessarily be semi-global as $\mathcal{I}^+$ will not be complete.

We next record the broad approach to gauge fixing taken within each work. The approach of \cite{Friedrich1986OnStructure} is somewhat sui generis, relying on conformal techniques to render the problem local-in-time. The remaining works divide roughly into two camps: the more geometric approach pioneered by Christodoulou and Klainerman in which the primary estimates are performed at the level of curvature; and the approach based on studying the wave equations satisfied by the metric components in harmonic coordinates due to Lindblad--Rodnianski. The advantage of the first is that the null structure of Einstein's equations is more directly evident.
The second approach connects more directly to the theory of nonlinear wave equations, and can draw on that literature. Our gauge is geometric, but unlike previous works which require more than one foliation, we gain in efficiency by only considering a single null foliation.

Next we note the decay rates assumed in each work. Since different norms are used in the various works to describe the asymptotic decay precisely, the parameter $\nu$ recorded here should be understood in a loose sense. Where the initial data surface does not extend to $i^0$, the value of $\nu$ is inferred from the decay of $\alpha$ near $\mathcal{I}^+$. Our approach enables us to both establish the optimal decay of \cite{Shen2023GlobalDecay}, but also to show stronger peeling results as in \cite{Klainerman2003PeelingEquations}.

Finally we list the mechanism by which the main dynamical estimates are proved in each case. By `classical vf estimates' we refer to results obtained within the general paradigm of the vector field method, \cite{ChrKl90}, making use of the conformal symmetries of Minkowski. By $r^p$ estimates we refer to results in the framework of the `new method' of Dafermos--Rodnianski \cite{Dafermos2009ASpacetimes}. Our use of the $r^p$-method is well adapted to the null foliation as it avoids the need to construct approximate conformal Killing fields.

\subsubsection{Stability of Minkowski space in the presence of matter fields}
There have also been many works on the stability of Minkowski as a solution of the Einstein equations coupled to various matter models. We mention here \cite{LeFloch2016TheFields, Wang2020AnEquations, Ionescu2022TheAMS-213} (Klein--Gordon); \cite{Zipser2000TheEquations, Loizelet2009} (Maxwell); \cite{Speck2010TheCoordinates} (nonlinear electromagnetism); \cite{Kauffman2023GlobalCoordinates} (Maxwell--Klein--Gordon); \cite{Lindblad2020GlobalGauge, Fajman2021TheSystem, Wang2022GlobalSystem} (massive Vlasov); \cite{Taylor2017TheSystem, Bigorgne2021AsymptoticMatter} (massless Vlasov); \cite{huneau2018stability, Huneau2023TheSpacetime, andersson2023global} (Kaluza--Klein).

\subsubsection{Other related works}

(Exterior stability) There are other setups for which the stability of Minkowski is studied. One setting is the stability in the domain of dependence of the exterior of a compact set in the Cauchy data, first proven in \cite{Klainerman2003TheNicolo.}. See also \cite{Hintz.exterior, Shen2023GlobalDecay, dwShen2024}. 

(Scattering problem) Another natural setting for stability is the full scattering problem, where initial data are posed on past null infinity. Near spacelike infinity, this has been recently resolved in \cite{KaKe2025}.

(Regularity of future null infinity) A related problem concerns the regularity of future null infinity, and in particular whether peeling occurs; see \cite{dC2002, hF2018, Klainerman2003PeelingEquations, lmaK2022, Nutzi} for discussions.

(Spherically symmetric problems) Prior to the works \cite{Friedrich1986OnStructure, CK94}, the stability of Minkowski spacetime was first proven in spherical symmetry for the Einstein--scalar field system \cite{dC1986}. In that setting, Christodoulou also established a scale-invariant result \cite{dC1993}; an alternative setting containing some settings beyond scale-invariant data can be found in \cite{LOY1}. Both of these results could be compared to the $\nu=0^+$ case in Table~\ref{tab:prevwks}. See also \cite{gRadR1992, mD2006} for the Einstein--Vlasov system.

(Stability of dispersive spacetimes) Beyond Minkowski spacetime itself, large-data solutions which disperse to Minkowski spacetime are also known to be stable \cite{LO.dispersive}.

Finally, we refer the reader to the surveys \cite{dwShen2026, jS2025} for other related works and further details.


\subsection{Organisation of the paper} The remainder of the paper will be organised as follows. The next few sections contain preliminaries for our setup. In \textbf{\S\ref{sec:geometric.prelim}}, we will discuss the geometry associated to the Newman--Unti gauge and write down the Einstein equations in this gauge. In \textbf{\S\ref{sec:geometry}}, we then introduce the norms that we use in this paper. In \textbf{\S\ref{sec:coords}} we describe how to construct Newman--Unti coordinates starting from a Cauchy surface. In \textbf{\S\ref{sec:assumptions}}, we give the assumptions on the initial data and bootstrap assumptions, and precisely state the main theorem. Our proof is based on a bootstrap argument, and we begin the proof in \textbf{\S\ref{sec:BA.consequences}}, \textbf{\S\ref{sec:transport}} by establishing some immediate consequences of the bootstrap assumptions and establishing bounds on solutions to transport equations. In \textbf{\S\ref{sec:teukolsky}} we analyse the Teukolsky equation, establishing the necessary linear and nonlinear estimates. The estimates for the other geometric quantities are then proven in \textbf{\S\ref{sec:other-geom}}. Finally, we define the bootstrap region in \textbf{\S\ref{sec:tau}}, prove regularity near the centre and future completeness in \textbf{\S\ref{sec:geodesics}}, and conclude the proof in \textbf{\S\ref{sec:proof}}. In Appendix \textbf{\ref{sec:appendix}} we show how to adapt our proof to the case of exterior stability and the spacelike-characteristic initial value problem. In Appendix \textbf{\ref{sec:appendix2}} we connect the initial data assumptions in Definition~\ref{def:AF.data} and Theorem~\ref{thm:main} to those used in the main bootstrap argument formulated in \textbf{\S\ref{sec:assumptions}}.

\subsection*{Acknowledgments} 
JL and CW gratefully acknowledge the hospitality of the Erwin Schr\"odinger Institute, Vienna, where this project was initiated during the programme `Spectral Theory and Mathematical Relativity' in the summer of 2023. We thank Gustav Holzegel and Jacques Smulevici for helpful discussions concerning the Newman--Unti gauge in the adS setting.

JL was partially supported by a National Science Foundation Grant under DMS2304445. SJO was partially supported by a National Science Foundation
CAREER Grant under NSF-DMS-1945615, a National Science Foundation Grant under DMS-2452760, and the Miller Research Professorship. CW acknowledges support from the Simons Collaboration on Black Holes and Strong Gravity.

\section{Geometric preliminaries}\label{sec:geometric.prelim}

\subsection{The Newman--Unti gauge}

\subsubsection{Form of the metric}
Let $\mathcal{W} \subset \{(u, r) \in \mathbb{R}^2 : r>0\}$ be open, $\mathcal{U} = \mathcal{W} \times S^2$, and denote by $S_{u,r}\subset \mathcal{U}$ the spheres of constant $u, r$. Given a smooth family $\gamma: \mathcal{W} \to \Gamma (T^*S^2\times T^*S^2)$ of Riemannian metrics on $S^2$, a smooth family of vector fields $b: \mathcal{W} \to \Gamma T S^2$ and a smooth function $f: \mathcal{U} \to \mathbb R$, we define a Lorentzian metric by:
\begin{equation}
\label{formofthemetric}
g = - (\ud u \otimes \ud r+\ud r\otimes \ud u) - f \ud u \otimes \ud u +\gamma_{AB}(\ud\th^A-b^A \ud u)\otimes (\ud\th^B - b^B \ud u).
\end{equation}
Introduce a time-orientation by stipulating that $\partial_r$ is future directed.
With this choice, $(\mathcal{U}, g)$ is a smooth $4$-dimensional spacetime. We note that the inverse metric is given by
\begin{equation}\label{eq:g.inverse}
\begin{split}
g^{-1} = f \f{\rd}{\rd r}\otimes \f{\rd}{\rd r} - \left (\f{\rd}{\rd u} + b^A\f{\rd}{\rd\th^A}\right )\otimes \f{\rd}{\rd r}- \f{\rd}{\rd r} \otimes \left(\f{\rd}{\rd u} + b^A\f{\rd}{\rd\th^A}\right)+\gamma^{-1},
\end{split}
\end{equation}
where $\gamma^{-1}$ is the inverse of the Riemannian metric $\gamma$. We denote by $D$ the Levi-Civita connection associated to the metric.

\subsubsection{Frames and metric components}\label{someframes}
Let $\f{\rd}{\rd u}$, $\f{\rd}{\rd r}$, $\f{\rd}{\rd\th^A}$ be coordinate vector fields
defined with respect to the $(u, r ,\th^1,\th^2)$ coordinate system that we introduced above. 

\begin{definition}
Define 
\begin{equation}\label{e3.e4.def}
e_4:= \f{\rd}{\rd r}, \quad e_3:= 2 \left (\f{\rd}{\rd u}  +b^A\f{\rd}{\rd\th^A}\right ) - f \f{\rd}{\rd r}, \quad e_A:= \f{\rd}{\rd \theta^A}.
\end{equation}
\end{definition}

The following lemma clarifies the geometric significance of the vector fields $e_3$, $e_4$.
\begin{lemma} \label{Lem:e3e4}
$e_3$, $e_4$ are null and future-directed vector fields. Moreover,
\begin{enumerate}[i)]
\item $e_4= - (D u)^{\sharp}$.
\item $e_4$ is a geodesic vector field, i.e., it satisfies the geodesic equation
\begin{equation}\label{e4.geodesic}
D_{e_4} e_4 = 0.
\end{equation}
\item $e_4$ is orthogonal to $\f{\rd}{\rd \th^A}$, i.e.,
\begin{equation}\label{e4.ortho.th}
g(e_4,e_A) = 0.
\end{equation}
\item $e_3$ is the unique vector field satisfying
\begin{equation}\label{e4.def.metric}
g(e_3,e_3)= g(e_3,e_A) = 0,\quad g(e_3,e_4)=-2.
\end{equation}
\end{enumerate} 
\end{lemma}
\begin{proof}

That $e_4$ is null and future-directed is immediate from \eqref{formofthemetric}.
Similarly, \eqref{e4.ortho.th} follows from \eqref{formofthemetric}. Next, we compute using \eqref{eq:g.inverse} that
$$-(Du)^\sharp = - (g^{-1})^{\mu\nu} (\rd_\mu u) \rd_\nu = -(g^{-1})^{ur} \rd_r = \rd_r = e_4.$$ To see that $e_4$ is geodesic, we compute
$$D_{e_4} e_4 = D_{\rd_r} \rd_ r = \f 12 (g^{-1})^{\mu\nu} (2 g_{r\nu, r} - g_{rr,\nu}) \rd_\nu = 0,$$
since the components $g_{r\nu}$ are all constant. For $e_3$, we compute using \eqref{formofthemetric} that
\begin{equation*}
\begin{split}
g(e_3,e_3) = &\: g \left(2 (\f{\rd}{\rd u}  +b^A\f{\rd}{\rd\th^A}) - f  \f{\rd}{\rd r}, \,2 (\f{\rd}{\rd u}  +b^A\f{\rd}{\rd\th^A}) - f \f{\rd}{\rd r} \right) \\
=&\: 4 g_{uu} +8 b^A g_{uA} + 4 b^A b^B \gamma_{AB}  -4 f g_{ur} = -4 f+ 4|b|_\gamma^2 -8 |b|_\gamma^2 +4|b|_\gamma^2 +4f  = 0,
\end{split}
\end{equation*}
\begin{equation*}
g(e_3, \f{\rd}{\rd\th^A}) = 0,\quad g(e_3, e_4) = 2g_{ur} = -2. \qedhere
\end{equation*}
\end{proof}

\subsection{Ricci coefficients} \label{subsec:ricci}
\subsubsection{Definition of the Ricci coefficients}

Denote $D_{\mu}=D_{e_{\mu}}$, with $\mu\in \{1,2,3,4\}$ and $D_A=D_{e_A}$, with $A\in \{1,2\}$. Define
\begin{equation}\label{eq:Ricci.coeff}
\begin{aligned}
\chi_{AB}=&\:g(D_Ae_4,e_B), \quad \quad \chib_{AB}=g(D_Ae_3,e_B),\\
\omega=&\:\frac{1}{4} g(D_4 e_4,e_3),\quad \quad  \omegab=\frac{1}{4} g(D_3 e_3,e_4),\\
\eta_A=&\:\frac{1}{2}g(D_3e_4,e_A),  \quad \quad \etab_A=\frac{1}{2}g(D_4e_3,e_A),\\
\xi_A=&\:\frac{1}{2}g(D_4e_4,e_A), \quad \quad  \xib_A=\frac{1}{2}g(D_3e_3,e_A), \\
\zeta_A=&\: \frac{1}{2}g(D_Ae_4,e_3).
\end{aligned}
\end{equation}

To separate the trace and traceless parts of $\chi$, we introduce the following notation, where traces are taken with the metric $\gamma$:
\begin{align*}
\chih_{AB}:=&\:\chi_{AB}-\frac{1}{2}\gamma_{AB} \tr \chi,\\
\chibh_{AB}:=&\:\chib_{AB}-\frac{1}{2}\gamma_{AB} \tr \chib.
\end{align*}

From now on, we will raise and lower indices for quantities in \eqref{eq:Ricci.coeff} with respect to $\gamma$.

\subsubsection{Relations for the metric and the Ricci coefficients in the null geodesic gauge}

\begin{proposition}\label{prop:gauge.con}
The following relations between the metric coefficients and the Ricci coefficients hold:
\begin{align}
\omegab&=\f 12 \f{\rd f }{\rd r},\qquad   \etab^A-\eta^A = \f{\rd b^A}{\rd r},\label{gauge.con.metric} 
\qquad \omega= 0, \qquad \xi_A = 0,\\ \zeta_A&=\eta_A=-\etab_A,\qquad \xib_A = \f{ \rd f}{\rd\th^A},\label{gauge.con} \qquad 
2\chi_{AB} = \frac{\partial \gamma_{AB}}{\partial r},\\  2\chib_{AB} &= \left(2\frac{\partial }{\partial u} + 2 b^A\frac{\partial }{\partial \theta^A} - f\frac{\partial }{\partial r}\right)\gamma_{AB} + 2 \frac{\partial b^C}{\partial \theta^A}\gamma_{CB} + 2 \frac{\partial b^C}{\partial \theta^B}\gamma_{AC} .  \label{chib.in.terms.of.gamma} 
\end{align}
\end{proposition}
\begin{proof}
For the first equation in \eqref{gauge.con.metric}, we compute 
$[e_3, e_4] = -2\f{\rd b^A}{\rd r} \f{\rd}{\rd\th^A} + \f{\rd f }{\rd r}\f{\rd}{\rd r}$.
Therefore, 
$$g([e_3,e_4], e_3) = \f{\rd f }{\rd r} g(\f{\rd}{\rd r}, e_3) = \f{\rd f }{\rd r} g(e_3,e_4) = -2\f{\rd f }{\rd r},$$
which implies
$$\omb =-\f 14 g(D_3 e_4, e_3)=-\f 18 e_4 [g( e_3, e_3)] - \f 14 g([e_3,e_4], e_3) = \f 12 \f{\rd f }{\rd r}.$$

The second equation in \eqref{gauge.con.metric} can also be obtained using the commutation formula for $[e_3,e_4]$ above. Namely, by the commutation formula, we have
$$g([e_3,e_4], e_A) = -2\f{\rd b^B}{\rd r}\gamma_{AB},$$
which implies
$$-2\f{\rd b^B}{\rd r}\gamma_{AB} = g([e_3,e_4],e_A)=g(D_3 e_4, e_A)-g(D_4 e_3, e_A)=2\eta_A-2\etab_A.$$

The last two equations in \eqref{gauge.con.metric} follow from the fact that $e_4$ is geodesic. For the first equation in \eqref{gauge.con}, since $e_4=\f{\rd}{\rd r}$ and $e_A=\f{\rd}{\rd\th^A}$, we have $g([e_4,e_A],e_3)=0$. Therefore,
$$\etab_A=-\f 12 g(D_4 e_A, e_3)=-\f 12 g(D_A e_4, e_3)=-\zeta_A.$$
To obtain the second part of this equation, since $e_3 = 2 (\f{\rd}{\rd u}  +b^A\f{\rd}{\rd\th^A}) - f \f{\rd}{\rd r}$ and $e_A=\f{\rd}{\rd\th^A}$, we have 
$$g([e_3, e_A],e_4)=0.$$
As a consequence, this implies
$$\eta_A=-\f 12 g(D_3 e_A,e_4)=-\f 12 g(D_A e_3,e_4) = \f 12 g(e_3, D_A e_4)=\zeta_A.$$

For the second equation in \eqref{gauge.con}, we note that $[e_3,e_A] = \f{\rd f }{\rd\th^A}e_4 - 2\f{\rd b^B}{\rd\th^A} \f{\rd}{\rd\th^B}$. Hence,
$$\xib_A = \f 12 g(D_3 e_3,e_A) = -\f 12 g(e_3,D_3 e_A)=-\f 12 g(e_3,[e_3,e_A]) =  -\f 12 \f{\rd f}{\rd\th^A} g(e_3,e_4) = \f{\rd f}{\rd\th^A}.$$

Next, note that since $g(e_4, e_B)=0$, we have $\chi_{AB} = -g(e_4, D_Ae_B)=\chi_{BA}$ by the metric compatibility and symmetry of the Levi-Civita connection and similarly $\chib_{AB}=\chib_{BA}$. Observing that $[e_4, e_A]=0$ we have
\[
\frac{\partial \gamma_{AB}}{\partial r} = e_4(g(e_A, e_B)) = g(D_4e_A, e_B) + g(e_A D_4, e_B) = g(D_Ae_4, e_B) + g(e_A, D_B e_4) = 2\chi_{AB}
\]
which gives the last equation of \eqref{gauge.con}. Finally, we compute $g([e_3,e_A], e_B) = -2 \partial_A b^C \gamma_{BC}$, and hence
\begin{align*}
    &\left(2\frac{\partial }{\partial u} + 2 b^A\frac{\partial }{\partial \theta^A} - f\frac{\partial }{\partial r}\right)\gamma_{AB} = e_3(g(e_A, e_B)) = g(D_3 e_A, e_B) + g(e_A, D_3 e_B) \\
    &\qquad =g(D_A e_3, e_B) + g(e_A, D_B e_3) + g([e_3, e_A], e_B) + g(e_A, [e_3, e_B])\\
    &\qquad = 2\chib_{AB} - 2 \frac{\partial b^C}{\partial \theta^A}\gamma_{CB} - 2 \frac{\partial b^C}{\partial \theta^B}\gamma_{AC}
\end{align*}
which gives \eqref{chib.in.terms.of.gamma} on rearranging.
\qedhere
\end{proof}

\subsection{Integration by parts formulae}
The following formulae for integration by parts follow from Proposition \ref{prop:gauge.con}, together with the Jacobi formula for the derivative of a determinant.
\begin{lemma}\label{lem:int.fomula} Let $\ud A_\gamma$ be the Riemannian measure on $S_{u, r}$ induced by the metric $\gamma_{AB}$. For a sufficiently regular function $h$ we have
\begin{enumerate}
\item \begin{equation}\label{eq:first.variation}
\f{\rd}{\rd r} \int_{S_{u,r}} h \, \mathrm{dA}_{\gamma}  = \int_{S_{u,r}} (e_4 h + \trch h) \, \mathrm{dA}_{\gamma}. 
\end{equation}
\item \begin{equation}\label{eq:first.variation.3}
2\f{\rd}{\rd u} \int_{S_{u,r}} h \, \mathrm{dA}_{\gamma} - \f{\rd}{\rd r}\int_{S_{u,r}} f h \, \mathrm{dA}_{\gamma}  = \int_{S_{u,r}} (e_3 h + \trchb h -2\omb h) \, \mathrm{dA}_{\gamma}. 
\end{equation}
\end{enumerate}
\end{lemma}

\subsection{Curvature components}
Let $S$ be the scalar curvature of the spacetime metric $g$, and let $K$ be the Gaussian curvature of the Riemannian metric $\gamma$. Define the Weyl curvature tensor by
$$W_{\alp\bt\mu\nu} := R_{\alp\bt\mu\nu} - (g_{\alp[\mu} \mathrm{Ric}_{\nu]\bt} - g_{\bt[\mu} \mathrm{Ric}_{\nu]\alp}) + \f 13 S g_{\alp[\mu} g_{\nu]\bt}.$$

Let $(*W)_{\alpha \beta \mu \nu} = \frac{1}{2} \in_{\alpha \beta \kappa \tau}W^{\kappa \tau}{}_{\nu \nu}$ be the dual Weyl tensor, where $ \in_{\alpha \beta \kappa \tau}$ is the spacetime volume form.

Introduce the following notations for the components of the Weyl curvature.
\begin{definition}[Weyl curvature components]\label{def:weyl}
\begin{align*}
\alpha_{AB}&= W_{A4B4}, &\beta_A&=\frac{1}{2}W_{A434}, &\rho &= \frac{1}{4}W_{4343},\\ \alphab_{AB}&=W_{A3B3}. 
& \betab_A&=\frac{1}{2}W_{A334}, &\sigma&=\frac{1}{4}(^*W)_{4343}.
\end{align*}
\end{definition}

Define also the following notations for the components of the Ricci curvature.
\begin{definition}[Ricci curvature]
Define
$$\slashed{\mathrm{Ric}}_{4A} := \mathrm{Ric}_{4A}, \quad \slashed{\mathrm{Ric}}_{3A} := \mathrm{Ric}_{3A}, \quad \slashed{\mathrm{Ric}}_{AB} := \mathrm{Ric}_{AB},$$
where $\slashed{\mathrm{Ric}}_{4\cdot}$, $\slashed{\mathrm{Ric}}_{3\cdot}$ are now viewed as $S$-tangent $1$-forms, while $\slashed{\mathrm{Ric}}$ is viewed as an $S$-tangent symmetric $2$-tensor.

Introduce also the notation for the trace and trace-free parts of $\slashed{\mathrm{Ric}}$:
$$\tr\slashed{\mathrm{Ric}} = (\gamma^{-1})^{AB} \slashed{\mathrm{Ric}}_{AB},\quad  \widehat{\slashed{\mathrm{Ric}}}:= \slashed{\mathrm{Ric}} - \f 12 \gamma \tr\slashed{\mathrm{Ric}}.$$
\end{definition}

\subsection{Example: spaces of constant sectional curvature}

Let $\mathring{\gamma}$ be the usual round metric on the sphere of radius $r$ with coordinates $\theta, \varphi$. In the $(u,r,\theta,\varphi)$ coordinate system, the metric
\begin{equation}\label{eq:AdS.u.r}
g = -\left (1-\frac{\Lambda}{3} r^2\right ) \, \ud u^2 -2  \, \ud u \, \ud r + \mathring{\gamma}.
\end{equation}
is a solution of Einstein's equation with cosmological constant $\Lambda$. Moreover, $g$ has constant sectional curvature, hence is locally isometric to de Sitter / Minkowski / anti-de Sitter depending on whether $\Lambda >0$ / $\Lambda = 0$ / $\Lambda <0$. The null pair $(e_3, e_4)$ are given by 
$$e_4 := \rd_r, \quad e_3 := 2 \rd_u - \left (1-\frac{\Lambda}{3} r^2\right ) \rd_r.$$
It is easy to check that
$$g(e_3, e_3) = g(e_4, e_4) = 0,\quad g(e_3, e_4) = -2.$$

The non-zero Ricci coefficients are given by 
$$\chi = \f 1 r \mathring{\gamma},\quad \chib = -\f{1-\Lambda r^2/3}r \mathring{\gamma},\quad \omb = -\frac{\Lambda r}{3}.$$
 All components of the Weyl tensor vanish.

\subsection{Differential operators}

\begin{definition}
We say that a tensor field defined on $\mathcal{U}$ is $S-$tangent if it coincides with its $g$-orthogonal projection to $S_{u, r}$. 
\end{definition}
For example a 2-tensor field $T$ defined on $\mathcal{U}$  is $S-$tangent if
\[
T_{3A}=T_{A3} =T_{4A}=T_{A4} = T_{34}=0
\]
for $A=1, 2$. We can identify an $S-$tangent tensor field with its components in the usual way - in the example above we would identify $T$ with the components $T_{AB}$, which are functions of $u, r, \theta^A$. We can raise and lower $S-$tangent indices with $\gamma_{AB}$ and $\gamma^{AB}:=(\gamma^{-1})^{AB}$.

\begin{definition}[Covariant derivatives for $S$-tangent tensor fields]\label{def:nabla}
Let $\nab_3$, $\nab_4$, $\nab_A$ be respectively the projections to $S_{u,r}$ of the covariant derivatives $D_3= D_{e_3}$, $D_4 = D_{e_4}$ and $D_A = D_{e_A}$. Finally, let $\nab_u = \frac{1}{2}(\nab_3 + f \nab_4)$.
\end{definition}

\begin{lemma}\label{lem.Dtonab}
We can express
\begin{align*}
D_Ae_3=&\:\underline{\chi}_{A}{}^{B}e_B+\eta_Ae_3,&
D_Ae_4=&\:\chi_{A}{}^{B}e_B-\eta_Ae_4,\\
D_3e_3=&\: 2\underline{\xi}^Ae_A-2\underline{\omega}e_3,&
D_3e_4=&\: 2\eta^Ae_A+2\underline{\omega}e_4,\\
D_4e_3=&\:-2{\eta}^Ae_A,&
D_4e_4=&\:0,\\
D_Ae_B=&\: \nab_Ae_B+\frac{1}{2}\chi_{AB}e_3+\frac{1}{2}\underline{\chi}_{AB}e_4,&
D_3e_A=&\:\nab_3e_A+\eta_Ae_3+\underline{\xi}_Ae_4,\\
D_4e_A=&\:\nab_4 e_A-{\eta}_Ae_4,
\end{align*}
\end{lemma}
\begin{proof}
These follow directly from the definitions of the Ricci coefficients on taking inner products with the basis vectors.
\end{proof}

The following identities hold for the connections $\nab_3$, $\nab_4$ and $\nab$.
\begin{proposition}\label{diff.formula}
For every covariant tensor field $\phi$ of rank $k$ tangential to the spheres $S_{u,r}$, we have
\begin{equation}\label{nab3.def}
\begin{split}
(\nab_3 \phi)_{A_1 A_2 ... A_k}
=&\: \left [2 \left(\f{\rd}{\rd u}  +b^A\f{\rd}{\rd\th^A}\right) - f \f{\rd}{\rd r}\right] \phi_{A_1 A_2 ... A_k}-\sum_{i=1}^k\left(\chib^B{ }_{A_i}-2\frac{\partial b^B}{\partial\th^{A_i}}\right)\phi_{A_1\dots\hat{A_i}B\dots A_k},
\end{split}
\end{equation}
where $\hat{A_i}$ means that $A_i$ has been omitted. Similarly, we have
\begin{equation}\label{nab4.def}
\begin{split}
(\nab_4 \phi)_{A_1 A_2 ... A_r}
=&\: \frac{\partial}{\partial r} \phi_{A_1 A_2 ... A_r}-\sum_{i=1}^r \chi^B{ }_{A_i}\phi_{A_1\dots\hat{A_i}B\dots A_r}.
\end{split}
\end{equation}
Finally, $\nab$ is given by the Levi-Civita connection associated to the metric $\gamma$, i.e.,
\begin{equation}\label{nab.def}
\nab_B \phi_{A_1 A_2 \dots A_r}=\f{\rd}{\rd \th^B}\phi_{A_1 A_2 \dots A_r}-\sum_{i=1}^r \slashed{\Gamma}_{BA_i}^C \phi_{A_1 A_2\dots \hat{A_i} C\dots A_r},\end{equation}
where $\slashed{\Gamma}$ is given by 
\begin{equation}\label{Gamma.def}
\slashed{\Gamma}_{BA}^C=\f 12(\gamma^{-1})^{CD}\left(\f{\rd}{\rd\th^B}\gamma_{AD}+\f{\rd}{\rd\th^A}\gamma_{BD}-\f{\rd}{\rd\th^D}\gamma_{AB}\right).
\end{equation}

\end{proposition}

\begin{proof}
By Lemma \ref{lem.Dtonab} we have
\begin{align*}
[e_A, e_3] &= 2\frac{\partial b^B}{\partial \theta^A} e_B - \frac{\partial f}{\partial \theta^A}e_4
= \underline{\chi}_{A}{}^{B}e_B- \nab_3e_A-\underline{\xi}_A e_4.
\end{align*}
Rearranging and using Lemma \ref{prop:gauge.con} gives
\[
\nab_3e_A = \underline{\chi}_{A}{}^{B}e_B- 2\frac{\partial b^B}{\partial \theta^A} e_B.
\]
Similarly, considering $[e_A, e_3] $ gives
\[
\nab_4 e_A = \chi_A{}^B e_B.
\]
These two results, together with the Leibniz property and metric compatibility of $D$ give the results for $\nab_3$, $\nab_4$. The result for $\nab_A$ is standard.
\end{proof}

Observe that by Proposition \ref{prop:gauge.con}, $\nab_3, \nab_4, \nab_A, \nab_u$ all annihilate $\gamma_{AB}$ and so commute with raising and lowering $S-$tangent indices. Further $\nab_Ar= \nab_ur=0$.

\subsubsection{Tensor product notation}
Let $\phi, \phi'$ be two $S-$tangent one-forms with components $\phi_A,\phi_A' \in C^{\infty}(\mathcal{U})$. Let $\in_{AB}$ be the induced volume form on the spheres $S_{u, r}$, i.e., $\in_{AB} =\in_{[AB]} $ and $\in_{12} = \sqrt{\det \gamma}$. Then define
\begin{align*}
\div \phi:=&\: \nab_B \phi^A,\\
\curl \phi:=&\:\in^{AB} (\nab_B \phi)_A,\\
(\nabla \hot \phi)_{AB}:=&\:(\nab_A\phi)_B+(\nab_B\phi)_A-\gamma_{AB} \div \phi,\\
\phi\cdot \phi':=&\: \gamma^{AB}\phi_A\phi'_B,\\
(^*\phi)_A:=&\: \in_A{}^{B}\phi_B,\\
\phi\wedge \phi':=&\: \in^{AB}\phi_A\phi'_B=\phi \cdot ^*\phi,\\
(\phi \hot \phi')_{AB}:=&\: \phi_A\phi'_B+\phi_B\phi_A'-\gamma_{AB} \phi \cdot\phi'.\\
\end{align*}

Suppose $\phi, \phi'$ are two $S-$tangent 2-tensors with components $\phi_{AB},\phi'_{AB} \in C^{\infty}(\mathcal{U})$. Assume $\phi$ and $\phi'$ are symmetric and traceless. Then:
\begin{align*}
\phi\cdot \phi'&:=\phi^{AB}\phi'_{AB},\\
(^*\phi)_{AB}:=&\:  \in_B{}^{C}\phi_{AC},\\
(\phi\hat{\odot} \phi')_{AB} :=& \phi_{AC} \phi'_{B}{}^C+\phi_{BC} \phi'_{A}{}^C - \gamma_{AB}\phi\cdot \phi'  \ \\
\phi \wedge \phi':=&\: \in^{AB} \phi_{AC}\phi'_{B}{}^{C}=\phi \cdot ^*\phi'.
\end{align*}

Suppose $\phi$ is an $S-$tangent one-form and $\phi'$ is a symmetric traceless $S-$tangent 2-tensor with components $\phi_A,\phi'_{AB} \in C^{\infty}(\mathcal{U})$. Then:
\begin{equation*}
(\phi \cdot\phi')_A:=\gamma^{BC} \phi_B \phi'_{AC}.
\end{equation*}

\subsection{Commutation formulae}

The following commutation formulae follow from \cite[Lemma~7.3.3]{CK94} after substituting in relations from Proposition~\ref{prop:gauge.con} (which hold for all metrics of the form (\ref{formofthemetric}), without requiring the Einstein equations). 
\begin{lemma}[Commutation formulae] \label{lem:CK.733} For $\phi_{A_1\dots A_r}$ a sufficiently smooth $S$-tangent tensor we have:
\begin{equation}\label{eq:CK.nab4.nab}
\begin{split}
&[ \nab_4,\nab_B] \phi_{A_1\dots A_r}  = \:   -  \chi_{B}{}^{C}\nab_C \phi_{A_1\dots A_r} \\
&\qquad\: +\sum_{i=1}^r (- \chi_{A_i B} \eta^C+ \chi_{B}{}^{C} \eta_{A_i}+ \in_{A_i}{ }^C { }^*(\beta_B-\frac{1}{2}\R_{4B}) )\phi_{A_1\dots \hat{A_i} C\dots A_r},
\end{split}
\end{equation}
\begin{equation}\label{eq:CK.nab3.nab}
 \begin{split}
&[\slashed{\nabla}_3,\slashed{\nabla}_B]\phi_{A_1...A_r}=\:\xib_B \nab_4 \phi_{A_1...A_r} - \chib_{B}{}^{C}\slashed{\nabla}_C\phi_{A_1...A_r} \\
&\qquad +\sum_{i=1}^r (\chib_{A_iB}\eta^{C}-\chib_{B}{}^{C}\eta_{A_i}+\chi_{A_iB}\xib^{C}-\chi_{B}{}^{C}\xib_{A_i}-\in_{A_i}{ }^C{ }^*(\betab_B+\frac{1}{2}\R_{3B}))\phi_{A_1...\hat{A_i}C...A_r},
 \end{split}
\end{equation}
\begin{equation}\label{eq:CK.nab3.nab4}
\begin{split}
[\nab_3,\nab_4]\phi_{A_1...A_r}=&\: 2\omb \nab_4\phi_{A_1...A_r} + 4\eta^C \nab_C\phi_{A_1...A_r} +2 \sum_{i=1}^r  \in_{A_i}{}^{C}\sigma \phi_{A_1...\hat{A}_iC...A_r},
\end{split}
\end{equation}
and
\begin{equation}\label{eq:CK.nabB.nabC}
\begin{split}
[\nab_B,\nab_C]\phi_{A_1...A_r}=&\: K\sum_{i=1}^r  \left(\gamma_{A_i B} \phi_{A_1...\hat{A}_iC...A_r}-\gamma_{A_i C} \phi_{A_1...\hat{A}_iB...A_r} \right),
\end{split}
\end{equation}
where  $\hat{A_i}$ means that $A_i$ has been omitted.

\end{lemma}

We observe that for a symmetric, traceless $S-$tangent 2-tensor we have the identity
\begin{equation}
    \nab \hot \div \phi = \slashed{\Delta}\phi - 2K\phi \label{eq:CK2.2.2a}
\end{equation}
which follows from the $[\nab_A, \nab_B]$ commutation relation together with the definition of $\nab \hot \div.$

\subsection{Null structure equations and null Bianchi equations}\label{sec:equations}

The Ricci coefficients satisfy the following propagation and constraint equations. Notice that we have imposed the conditions in Proposition~\ref{prop:gauge.con} which are satisfied in our gauge. On the other hand, these equations hold for any metric of the form \eqref{formofthemetric}; in particular, we have not yet imposed the Einstein equations. However, from the next section onwards we shall set $\R=0$. We do not include derivations for these formulae, but we note that they have been verified with a computer algebra package.

\subsubsection{Propagation equations}
\begin{align}
\label{eq:3trchib}
\nab_3 \tr \chib +\frac{1}{2}(\tr \chib)^2+ 2\omegab \tr \chib=&\: -|\chibh|^2 +  2\div \xib-4\xib\cdot \eta -\mathrm{Ric}_{33},\\\
\label{eq:4trchi}
\nab_4 \tr \chi +\frac{1}{2}(\tr \chi)^2 =&\: -|\chih|^2 -\mathrm{Ric}_{44},
\end{align}

\begin{align}
\label{eq:4trchib}
\nab_4 \tr \chib +\frac{1}{2}\tr \chi \tr \chib =&\: -2\div \eta - \chih\cdot \chibh +2 |\eta|^2 + 2 \rho - \f 23 S -  \mathrm{Ric}_{34} + \tr\slashed{\mathrm{Ric}},\\\
\label{eq:3trchi}
\nab_3 \tr \chi +\frac{1}{2}\tr \chib \tr \chi -2\omegab \tr \chi=&\: 2 \div  \eta-\chibh \cdot \chih+ 2|\eta|^2 + 2\rho - \f 23 S -  \mathrm{Ric}_{34} +  \tr\slashed{\mathrm{Ric}},
\end{align}


\begin{align}
\label{eq:3chih}
\nab_3 \chih+\frac{1}{2} \tr \chib \chih -2\omegab \chih = &\: -\frac{1}{2}\tr \chi \chibh+\nab \hot \eta+\eta \hot \eta + \widehat{\slashed{\mathrm{Ric}}} ,\\
\label{eq:4chibh}
\nab_4 \chibh+\frac{1}{2} \tr \chi \chibh =&\: -\frac{1}{2}\tr \chib \chih - \nab \hot \eta + \eta \hot  \eta + \widehat{\slashed{\mathrm{Ric}}},
\end{align}

\begin{align}
\label{eq:4chih}
\nab_4 \chih+ \tr \chi \chih=&\: - \alpha,\\
\label{eq:3chibh}
\nab_3 \chibh+\tr \chib \, \chibh +2\omegab \chibh =&\: \nab \hot \xib-2\xib \hot \eta - \alphab,
\end{align}

\begin{align}
\label{eq:3etab}
\nab_3 \eta+{\nab_4 \xib}+ \tr \chib \eta=&\: -2 \chibh \cdot \eta   -\betab +\frac{1}{2}\sR_{3\cdot}, \\
\label{eq:4zeta}
\nab_4\eta+\tr \chi \eta =&\: - 2 \chih\cdot\eta-\beta{-\frac{1}{2} \sR_{4\cdot}},
\end{align}

\begin{equation}
\label{eq:omegab}
\nab_4 \omegab = 3|\eta|^2+\rho + \f{1}{6}S + \f 12 \mathrm{Ric}_{34}.
\end{equation}

\subsubsection{Constraint equations}
\begin{align}
\label{eq:constchih}
\div \chih+\eta \cdot \chih=&\:\frac{1}{2}\nab \tr \chi+\frac{1}{2} \tr \chi \eta-\beta +\f 12 \sR_{4\cdot},\\
\label{eq:constchibh}
\div \chibh-\eta \cdot \chibh=&\: \frac{1}{2}\nab \tr \chib-\frac{1}{2} \tr \chib \eta+\betab +\f 12 \sR_{3\cdot},\\
\label{eq:constzeta}
\curl \eta=&\: -\frac{1}{2}\chih\wedge \chibh+\sigma, \\
\label{eq:constrho}
K -\f 12 \chih\cdot \chibh + \f 14 \trch\trchb 
= &\: -\rho - \f 16 S + \f 12 \tr\slashed{\mathrm{Ric}}.
\end{align}

\subsubsection{Null Bianchi equations} \label{sec:Bianchi}

It will also be convenient to have propagation equations for the Weyl components. These can either be derived from the equations above making use of the commutation relations of Lemma \ref{lem:CK.733} or else by taking components of the Bianchi identity.
\begin{align}
\label{eq:3alpha}
\nab_3\alpha+\frac{1}{2}\tr \chib \alpha{-4\omegab \alpha}-\nab \hot \beta=&\: 5\eta\hot \beta-3(\rho \chih+\sigma ^*\chih)+\mathcal{R}^{\alpha, 3},\\
\label{eq:4alphab}
\nab_4\alphab+\frac{1}{2}\tr \chi \alphab+\nab \hot \betab=&\: 5\eta \hot \betab-3(\rho \chibh-\sigma ^*\chibh)+\mathcal{R}^{\alphab, 4},\\
\label{eq:4beta}
\nab_4\beta+2\tr \chi \beta-\div \alpha=&\: \eta \cdot \alpha+\mathcal{R}^{\beta, 4},\\
\label{eq:3betab}
\nab_3\betab+2\tr \chib \betab{+2\omegab \betab}+\div \alphab=&\: \eta\cdot \alphab-{3(\rho \xib-\sigma ^*\xib)}+\mathcal{R}^{\betab, 3},
\end{align}
\begin{align}
\label{eq:3beta}
\nab_3\beta+\tr \chib \beta {-2\omegab \beta}-\nab\rho-^*\nab \sigma=&\: 2\betab \cdot \chih+3(\rho \eta+\sigma ^*\eta)+{\alpha\cdot \xib}+\mathcal{R}^{\beta, 3},\\
\label{eq:4rho}
\nab_4(-\rho,\sigma)+\frac{3}{2}\tr \chi (-\rho,\sigma)+(\div \beta,-\curl \beta)=&\: \eta\cdot (\beta,^*\beta)+\frac{1}{2}\chibh\cdot (\alpha,^*\alpha)+(-\mathcal{R}^{\rho, 4}, \mathcal{R}^{\sigma, 4}),\\
\label{eq:4betab}
\nab_4\betab+\tr \chi \betab +\nab\rho-^*\nab \sigma=&\: 2\beta \cdot \chibh+3(\rho \eta-\sigma ^*\eta)+\mathcal{R}^{\betab, 4},\\
\label{eq:3rho}
\nab_3(\rho,\sigma)+\frac{3}{2}\tr \chib (\rho,\sigma)+(\div \betab,-\curl \betab)=&\: - \eta\cdot (\betab,^*\betab)+{2\xib \cdot (\beta,-^*\beta)}+\frac{1}{2}\chih\cdot (-\alphab,{}^*\alphab)+(\mathcal{R}^{\rho, 3}, \mathcal{R}^{\sigma, 3}).
\end{align}

Here the terms involving the Ricci tensor (which vanish in vacuum) are given by
\begin{align*}
    \mathcal{R}^{\alpha, 3}&= \frac{1}{2}\nab\hot\slashed{\R}_{4\cdot} - \nab_4 \widehat{\slashed{\R}} - \frac{1}{2}\eta \hot \slashed{\R}_{4\cdot} - \frac{1}{2} \chih \left(\tr\slashed{\R} +\R_{34}\right) - \widehat{\slashed{\R}}\trch - \frac{1}{2}\R_{44} \trchb, \\
    \mathcal{R}^{\alphab, 4}&= \frac{1}{2}\nab\hot\slashed{\R}_{3\cdot} - \nab_3 \widehat{\slashed{\R}} + \frac{1}{2}\eta \hot \slashed{\R}_{3\cdot}+\xib \hot \slashed{\R}_{4\cdot} - \frac{1}{2} \chibh \left(\tr\slashed{\R} +\R_{34}\right) - \widehat{\slashed{\R}}\trchb - \frac{1}{2}\R_{33} \trchb, \\
    \mathcal{R}^{\beta, 4}&= \frac{1}{2}\nab_4 \slashed{\R}_{4\cdot} - \frac{1}{2}\nab \slashed{\R}_{44} - \frac{1}{2}\eta \slashed{\R}_{44}+\chih\cdot \slashed{\R}_{4\cdot} + \frac{1}{2} \trch \slashed{\R}_{4\cdot},\\
    \mathcal{R}^{\betab, 3} &= \frac{1}{2} (\nab - \eta)\slashed{\R}_{33} - \frac{1}{2}(\nab_3 + 2\omegab - \trchb)\slashed{\R}_{3\cdot} - \chibh\cdot \slashed{\R}_{3\cdot} + \widehat{\slashed{\R}}\cdot \xib +\frac{1}{2}\left(\tr\slashed{\R} + \R_{34}\right)\xib +\widehat{\slashed{\R}}\cdot \xib,\\
    \mathcal{R}^{\beta, 3} &=\frac{1}{2} (\nab - \eta)\left(\R_{34}+\frac{S}{3} \right)- \frac{1}{2}(\nab_4 + \frac{1}{2}\trch)\slashed{\R}_{3\cdot} -\widehat{\slashed{\R}}\cdot \eta - \frac{1}{2}\left(\tr\slashed{\R} - \frac{S}{3} \right)\eta\\&\qquad  - \frac{1}{2}\chibh\cdot\slashed{\R}_{4\cdot}- \frac{1}{2}\chih\cdot\slashed{\R}_{3\cdot} - \frac{1}{2}\trchb \slashed{\R}_{4\cdot}, \\
    \mathcal{R}^{\betab, 4} &=\frac{1}{2}\left(\nab_3 - 2\omegab + \frac{1}{2}\trchb\right)\slashed{\R}_{4\cdot} - \frac{1}{2} (\nab + \eta)\left(R_{34}+\frac{S}{3} \right)-\widehat{\slashed{\R}}\cdot \eta - \frac{1}{2}\left(\tr\slashed{\R} - \frac{S}{3} \right)\eta\\&\qquad  + \frac{1}{2}\chibh\cdot\slashed{\R}_{4\cdot}+ \frac{1}{2}\chih\cdot\slashed{\R}_{3\cdot}+\frac{1}{4}\trch \slashed{\R}_{3\cdot} - \frac{1}{2}\R_{44}\xib,
\end{align*}
\begin{align*}
    \mathcal{R}^{\rho, 4}&= \frac{1}{4} \nab_4\left(R_{34}+\frac{S}{3} \right) - \frac{1}{4}(\nab_3 - 4\omegab){\R}_{44} + \frac{3}{2}\eta \cdot \slashed{\R}_{4\cdot}, \\
    \mathcal{R}^{\sigma, 4}&= -\frac{1}{2}\curl\slashed{\R}_{4\cdot} + \frac{1}{2}\eta \wedge \slashed{\R}_{4\cdot} - \frac{1}{2}\chih \wedge \widehat{\slashed{\R}}, \\
    \mathcal{R}^{\rho, 3}&= \frac{1}{4} \nab_3\left(R_{34}+\frac{S}{3} \right) - \frac{1}{4}\nab_4 \R_{33} - \frac{3}{2} \eta \cdot \slashed{\R}_{3\cdot}- \frac{1}{2} \xib \cdot \slashed{\R}_{4\cdot},\\
    \mathcal{R}^{\sigma, 3}&=\frac{1}{2}\curl\slashed{\R}_{3\cdot}+ \frac{1}{2}\eta \wedge \slashed{\R}_{3\cdot}+\frac{1}{2}\chibh \wedge \widehat{\slashed{\R}}.
\end{align*}

From this point onwards we shall assume the vacuum Einstein equations hold, so that $\R = 0$.

\section{Regions, surfaces and norms} \label{sec:geometry}

In order to treat the behaviour of the metric at the axis of the coordinate system, we shall work on $\mathbb{R}^4$ with an axis `blown-up'. On the manifold with boundary $\mathbb{R}_u\times [0, \infty)_r\times S^2_\theta$ with its canonical differentiable structure we define the function $t=u+r$ and let $S$ be the region
\[
S = \{(u, r, \theta)\in \mathbb{R}_u\times [0, \infty)_r \times S^2_\theta : t\geqslant 1\}
\]
Let $\tau = \tau(u,r)$ be the function constructed in \S\ref{sec:tau}.
For  $T>1$ we define the following (sub-)manifolds:
\begin{definition}[Subsets of the spacetime]
\begin{align}
S_T &:= \{(u, r, \theta) \in S :\tau <T \}, \\
\Sigma_{\tau'} &:= \{(u, r, \theta)  \in S_T : \tau = \tau'\}, \\
\tilde{\Sigma}_{t'} &:= \{(u, r, \theta)  \in S_T : t = t'\},\\
S_{u',r'} &:= \{(u, r, \theta)  \in S_T :  u = u',\, r=r'\}, \\
H_{u'} &:= \{(u, r, \theta)  \in S_T :  u = u'\}.
\end{align}
\end{definition}
The initial data surface will be $\tilde{\Sigma}_{1}$. We note that for $1<\tau'<T$ the surface $\Sigma_{\tau'}$ meets $\tilde{\Sigma}_{1}$. Under our bootstrap assumptions, away from $r=0$, $\tilde{\Sigma}_{1}$ will be a smooth spacelike surface, and $\Sigma_{\tau'}$ a $C^{1,1}$, piecewise smooth, spacelike surface (see Figure~\ref{fig:placeholder}, p\pageref{fig:placeholder}). Let $u_T$ be the unique solution to $\tau(u_T,1-u_T) = T$ so that $u=u_T$ on $\Sigma_T \cap \tilde{\Sigma}_1$.

We will work with norms defined relative to a background Minkowski metric for convenience. We fix $\mathring{\gamma}_{AB}$ to be the standard round metric on the sphere of radius $r$ in the coordinates $\theta^A$ and $\mathring{\slashed\nabla}_A$ to be the corresponding connection. For a $S-$tangent tensor, define the background derivatives
\begin{equation}\label{nab40.def}
\begin{split}
(\mathring{\nab}_4 \phi)_{A_1 A_2 ... A_k}{}^{B_1 B_2 ... B_{k'}}
=&\: \left(\frac{\partial}{\partial r} - \frac{k-k'}{r} \right) \phi_{A_1 A_2 ... A_k}{}^{B_1 B_2 ... B_{k'}}, \\
(\mathring{\nab}_u \phi)_{A_1 A_2 ... A_k}{}^{B_1 B_2 ... B_{k'}}
=&\: \frac{\partial}{\partial u} \phi_{A_1 A_2 ... A_k}{}^{B_1 B_2 ... B_{k'}}.
\end{split}
\end{equation}
Note that $r\mathring{\nab}_A, \mathring{\nab}_4, \mathring{\nab}_u$ commute with one another, and annihilate $\mathring{\gamma}$.

We define for $(u,r) \in S_T$ the norms\footnote{We will use $|\cdot|_*$ to denote norms which are pointwise in $(u,r)$ and $\|\cdot\|_*$ to denote norms integrated over subsets of $S$.}
\begin{align*}
|\phi|_0(u,r) &= \left(\int_{S_{u,r}} \left|\phi \right|^2_{\mathring{\gamma}} \ud A_{\mathring{\gamma}} \right)^{\frac{1}{2}}, \\ |\phi|_{1}(u,r) &= \left(|r\mathring{\nab}_4 \phi|^2_0(u,r) +|r\mathring{\nab} \phi|^2_{0}(u,r)+|\phi|_0^2(u,r) \right)^{\frac{1}{2}},
\end{align*}
where $\ud A_{\mathring{\gamma}}$ is the Riemannian volume form on the round sphere of radius $r$, and the pointwise norm is taken with respect to the metric $\mathring{\gamma}$. For $s \in \mathbb Z_{\geqslant 1}$ we recursively define 
\begin{equation}
   |\phi|_{s+1}(u,r) = \left(|r\mathring{\nab}_4 \phi|^2_s(u,r) +|r\mathring{\nab} \phi|^2_{s}(u,r)+|\phi|_s(u,r)^2+|\brk{u} \mathring{\nab}_u\phi|_{s-1}^2(u,r) \right)^{\frac{1}{2}}, \label{norm1}
\end{equation}
where $\jb{u} = \sqrt{1+u^2}$ is the Japanese bracket. Note that as $r\mathring{\nab}_4, r\mathring{\nab},  \brk{u}\mathring{\nab}_u$ commute, we have:
\[
|\phi|_s^{2}(u,r) \sim \sum_{i+j+2k\leqslant s} \abs{(r\mathring{\nab}_4)^i(r \mathring{\nab})^j (\brk{u} \mathring{\nab}_u)^k \phi}^2_0(u,r).
\]
We also observe that if $s$ is \emph{even}, then
\begin{equation}
    |\phi|_{s+1}^{2}(u,r) \sim |r\mathring{\nab}_4 \phi|^2_s(u,r) +|r\mathring{\nab} \phi|^2_{s}(u,r) (u,r)+|\phi|_s^{2}(u,r). \label{est:evennorm}
\end{equation}


The main norms that we shall use are designed to capture $r$-decay at the origin, and $r$ and $u$ decay near infinity. More concretely, define
\begin{align*}
        &|\phi|_{s, p_0, p_\infty}(R) = \sup_{(u,r)\in R} \frac{(\jb{u}+r)^{p_\infty+p_0}}{r^{p_0+1}} |{\phi}|_s(u,r)
\end{align*}
for $R\subset S_T$. Clearly
\[
|\phi|_{s, p_0, p_\infty}(R) \leqslant |\phi|_{s, q_0, q_\infty}(R)
\]
whenever $p_0\leqslant q_0$ and  $p_\infty \leqslant q_\infty$. It will also be useful to note that for $\nu \in \mathbb R$ and $\mu \geqslant 0$,
\begin{align}
    |r^{\nu} \phi|_{s, p_0, p_\i}(R) &= | \phi|_{s, p_0-\nu, p_\i+\nu}(R), \label{eq:norm.r.shift} \\
    |\brk{u}^{\mu} \phi|_{s, p_0, p_\i}(R) & \leqslant | \phi|_{s, p_0, p_\i+\mu}(R). \label{eq:norm.u.shift}
\end{align}
We also require integrated norms:
\begin{align*}
        &  \|\phi\|_{s, p_0, p_\infty}(R) = \left(\int_{R} \left|\phi \right|^2_{s}(u,r)  \frac{(\jb{u}+r)^{2 p_\infty + 2 p_0}}{r^{3+2 p_0}} \ud u \ud r \right)^{\frac{1}{2}},\\
        & \|\phi\|_{s, p_0, p_\infty}(H_u) = \left(\int_{H_u} \left|\phi \right|^2_{s}(u,r)  \frac{(\jb{u}+r)^{2 p_\infty + 2 p_0}}{r^{3+2 p_0}} \ud r \right)^{\frac{1}{2}}, \\
        &  \|\phi\|_{s, p_0, p_\infty}(\tilde{\Sigma}^{r'}_{t'}) = \left(\int_{0}^{r'} \left|\phi \right|^2_{s}(t'-r,r)  \frac{(\langle t'-r\rangle+r)^{2 p_\infty + 2 p_0}}{r^{3+2 p_0}} \ud r \right)^{\frac{1}{2}},
\end{align*}
where $R\subset S_T$ is measurable. We take the convention in all norms that if the domain is not stated then it is $S_T$, i.e.,  $\abs{\phi}_{s, p_0, p_\infty} := \abs{\phi}_{s, p_0, p_\infty}(S_T)$ and $\| \phi\|_{s, p_0, p_\infty} := \| \phi \|_{s, p_0, p_\infty}(S_T)$.

The following Lemma is a straightforward consequence of standard Sobolev estimates for the round sphere.
\begin{lemma}\label{lem:products}
    Suppose $2+i+j+2k\leqslant s$. We have 
    \[
    \left|(r\mathring{\nab}_4)^i(r \mathring{\nab})^j (\brk{u} \mathring{\nab}_u)^k \phi \right|_{\mathring{\gamma}}(u,r, \theta^A) \lesssim \frac{1}{r} |\phi|_s(u,r) \leqslant  \frac{r^{p_0}}{(\jb{u}+r)^{p_0+p_\infty}} |\phi|_{s, p_0, p_\infty}.
    \]
    If $s\geqslant 2$ and $s'\leqslant s$ then we have the Sobolev product formulae (here $\cdot$ denotes any contraction) 
    \begin{align}
    & |\phi \cdot \psi|_{s'} \lesssim \frac{1}{r} |\phi|_s |\psi|_{s'} \lesssim |\phi|_{s,0,0} |\psi|_{s'} \notag \\
    &|\phi \cdot \psi|_{s', p_0+q_0, p_\infty+q_\infty} \lesssim |\phi|_{s', p_0, p_\infty} |\psi|_{s, q_0, q_\infty}. \label{eq:product.for.pointwise}
        \end{align}
        If $s\geqslant 6$ then for exponents $p_{0}, p_{0}', q_{0}, q_{0}', p_{\infty}, p_{\infty}', q_{\infty}, q_{\infty}', p_{u}, p_{u}', q_{u}, q_{u}'$ such that $p_{0} + q_{0} = p_{0}' + q_{0}'$, $p_{\infty} + q_{\infty} = p_{\infty}' + q_{\infty}'$, and $p_{u} + q_{u} = p_{u}'+q_{u}'$, we have				    
        \begin{equation}\label{eq:product.for.integrated}
        \begin{aligned}
    \|\brk{u}^{p_{u}+q_{u}} \phi \cdot \psi\|_{s, p_0+q_0, p_\infty+q_\infty} &\lesssim |\brk{u}^{p_{u}} \phi|_{s-3, p_0, p_\infty} \|\brk{u}^{q_{u}} \psi\|_{s, q_0, q_\infty} \\
    &\qquad + \|\brk{u}^{p_{u}'} \phi\|_{s, p_{0}', p_{\infty}'} |\brk{u}^{q_{u}'} \psi|_{s-3, q_{0}', q_{\infty}'}
    \end{aligned}
    \end{equation}
    where all the implicit constants depend only on $s$, the ranks of the tensors involved and number of contractions.
\end{lemma}

In order to efficiently handle the higher order Teukolsky estimates, it is convenient to introduce a dual norm.
\begin{definition}[Dual norm]
    For $\psi$ an $S$-tangent $k$ tensor on $S_T$ we define
    \[
    \|\psi\|_{0,p_0, p_\infty,*} = \sup_{\phi \in B} \abs{\int_{S_T} \psi \cdot \phi \frac{(r+\jb{u})^{2p_\infty+2p_0}}{r^{3+2p_0}} \ud A_{\gamma} \ud r \ud u}
    \]
    where $B$ is the set of smooth $S$-tangent $k$ tensors on $S_T$ satisfying
    \[
 \sup_{u}\left \|\phi \right\|_{0, p_0, p_{\infty}}^2(H_u) + \|\phi\|_{0, p_0+\frac{1}{2}, p_{\infty}-\frac{1}{2}}^2  \leqslant 1.
    \]
    and $\cdot$ denotes total contraction. Define
    \[
    \|\psi\|_{1,p_0, p_\infty,*} = \|r\mathring{\nab}_4 \psi\|_{0,p_0, p_\infty,*}+\|r\mathring{\nab} \psi\|_{0,p_0, p_\infty,*} + \|\psi\|_{0,p_0, p_\infty,*}
    \]
    and similarly to \eqref{norm1}, for $s \in \mathbb Z_{\geqslant 1}$ we define 
    \[
    \|\psi\|_{s+1,p_0, p_\infty,*} = \|r\mathring{\nab}_4 \psi\|_{s,p_0, p_\infty,*}+\|r\mathring{\nab} \psi\|_{s,p_0, p_\infty,*} + \|\psi\|_{s,p_0, p_\infty,*}+ \|\brk{u} \mathring{\nab}_u\psi\|_{s-1,p_0, p_\infty,*}.
    \]
\end{definition}

\begin{lemma}  \label{lem:dualnormest}
For $s \geqslant 0$ and any $\delta' > 0$, we have
\begin{align} 
	\| \psi \|_{s, p_{0}, p_{\infty},*} &\lesssim_{\delta'} \| \brk{u}^{\frac{1}{2}+\delta'} \psi \|_{s, p_{0}-\frac{1}{2}, p_{\infty}},\label{eq:dualnormest0} \\
	\| \psi \|_{s, p_{0}, p_{\infty},*} &\lesssim \| \psi \|_{s, p_{0}-\frac{1}{2}, p_{\infty}+\frac{1}{2}} ,\label{eq:dualnormest1}
\end{align}
    where all the implicit constants depend only on $s$ and the ranks of the tensors involved.
\end{lemma}
\begin{proof}
First, consider $s = 0$. We start with \eqref{eq:dualnormest1}, which is easier. The Cauchy--Schwarz inequality gives
    \[
    \abs{\int_{S_T} \psi \cdot \phi \frac{(r+\jb{u})^{2p_\infty+2p_0}}{r^{3+2p_0}} \ud A_{\gamma} \ud r \ud u} \leqslant \|\psi\|_{0, p_0-\frac{1}{2}, p_{\infty}+\frac{1}{2}}\|\phi\|_{0, p_0+\frac{1}{2}, p_{\infty}-\frac{1}{2}}
    \]
    so taking a supremum over $\phi \in B$ we have \eqref{eq:dualnormest1} for $s=0$.
    
To prove \eqref{eq:dualnormest0} for $s = 0$, we start with $\phi \in B$ and proceed similarly, but decompose the integration domain as follows:
\begin{align*}
	&\abs{\int_{S_T} \psi \cdot \phi \frac{(r+\jb{u})^{2p_\infty+2p_0}}{r^{3+2p_0}} \ud A_{\gamma} \ud r \ud u} \\
	&\qquad \leqslant 	\int_{S_T \cap \{r > \brk{u}\} } \psi \cdot \phi \frac{(r+\jb{u})^{2p_\infty+2p_0}}{r^{3+2p_0}} \ud A_{\gamma} \ud r \ud u + \int_{S_T \cap \{r \leqslant \brk{u}\} } \psi \cdot \phi \frac{(r+\jb{u})^{2p_\infty+2p_0}}{r^{3+2p_0}} \ud A_{\gamma} \ud r \ud u \\
	&\qquad \leqslant 	\| \brk{u}^{\frac{1}{2}+\de'} \psi \|_{0, p_{0}, p_{\infty}}(S_{T} \cap \{ r > \brk{u} \}) \cdot \sup_{u} \| \phi \|_{0, p_{0}, p_{\infty}}(H_{u} \cap \{ r > \brk{u} \}) \cdot \left(\int_{-\infty}^{\infty} \frac{1}{\brk{u}^{1+2 \delta'}} \, \ud u \right)^{\frac{1}{2}} \\
	&\qquad \qquad + \nrm{\psi}_{0, p_{0}-\frac{1}{2}, p_{\infty}+\frac{1}{2}}(S_{T} \cap \{ r \leqslant \brk{u} \}) \cdot \| \phi \|_{0, p_{0}+\frac{1}{2}, p_{\infty}-\frac{1}{2}}(S_{T} \cap \{ r \leqslant \brk{u} \}) \\
	&\qquad \lesssim \| \brk{u}^{\frac{1}{2}+\de'} \psi \|_{0, p_{0}, p_{\infty}}(S_{T} \cap \{ r > \brk{u} \}) +  \nrm{\psi}_{0, p_{0}-\frac{1}{2}, p_{\infty}+\frac{1}{2}}(S_{T} \cap \{ r \leqslant \brk{u} \}) \\
	& \qquad \lesssim \| \brk{u}^{\frac{1}{2}+\de'} \psi \|_{0, p_{0}-\frac{1}{2}, p_{\infty}}(S_{T}).
\end{align*}
In the last inequality, we used the relative sizes of $r$ and $\brk{u}$ in each region to bound both norms for $\psi$ by the last term. 

Finally, to conclude the proof, observe that the case $s > 0$ follows readily by the inductive definition of the norms involved in the estimates.  \qedhere 
\end{proof}

\section{Construction of coordinates} \label{sec:coords}
Suppose that $(\Sigma, h, k)$ is a smooth set of Cauchy data for the vacuum Einstein equations, with $\Sigma \simeq \mathbb{R}^3$. In this section, we carry out the construction of centre-normalised Newman--Unti gauge.
\subsection{The centre geodesic and centre-normalisation}\label{sec:coordscentre}
 Let $O\in \Sigma$ be a marked point. By assumption we can cover $\Sigma$ by a single coordinate chart $\Phi:\Sigma \to \mathbb{R}^3$ such that without loss of generality $O$ corresponds to the origin and $h_{ij}(0) = \delta_{ij}$. Let $\overline{\theta} = (\overline{\theta}^A)$ be a local\footnote{To cover $S^2$ requires multiple coordinate patches, but we shall suppress this in practice and assume it is handled in the standard way.} choice of coordinates on the unit sphere $S^2 \subset \mathbb{R}^3 $ and $\hat{x} : \overline{\theta} \to S^2$ the corresponding chart. We define $\overline{r} = \sqrt{(x^1)^2+(x^2)^2+(x^3)^2}$, so that $(\overline{r},\overline{\theta}) \mapsto \Phi^{-1}(\overline{r}\hat{x}(\overline{\theta}))$ is a coordinate chart on $\Sigma \setminus O$. Identifying $\overline{\theta}$ with the tangent vector to the curve $\lambda \mapsto \Phi^{-1}(\lambda \hat{x}(\overline{\theta}))$ at $\lambda=0$ we have that $\overline{\theta}$ induces a canonical choice of coordinates $\overline{\theta} \mapsto R(\overline{\theta}) \in S_O\Sigma\subset T_O\Sigma$.

Now we let $(M, g, \iota)$ be the maximal smooth development of the data $(\Sigma, h, k)$, i.e., a $4-$manifold $M$ equipped with a Ricci flat, Lorentzian metric $g$ together with an embedding $\iota:\Sigma \to M$ such that $h, k$ are respectively the corresponding first and second fundamental form and all objects are smooth. We wish to construct Newman--Unti coordinates in a neighbourhood of $\iota(\Sigma)$. Let $N$ be the future directed unit normal to $\iota(\Sigma)$. There exists some $\delta u>0$ such that the geodesic $\Gamma: (1-\delta u, 1+\delta u) \to M$ satisfying $\Gamma(1) = \iota(O)$, $\dot{\Gamma}(1)=N$ exists; this is the \emph{centre geodesic}, or the spacetime axis. Given $\theta^A$ we define the null vector $V(\theta) \in T_{\iota(O)}M$ by $V(1, \theta) = N+\iota_*R(\theta)$ and extend this to a vector field $V(u, \theta)$ along  $\Gamma$ by parallel transport.

We may think of $(u, \theta) \mapsto V(u, \theta)$ as a smooth embedding $V : (1-\delta u, 1+\delta u)\times S^2 \to TM$. From each point $V(u,\theta)$ we follow the geodesic flow on $TM$ a parameter distance $r$, and label this point $\Psi(u, r, \theta)$. By construction, possibly after shrinking $\delta u$, we have that $\Psi : (1-\delta u, 1+\delta u)\times [0, \delta r)\times S^2 \to TM$ is well-defined and smooth for some $0<\delta r <\delta u$. 

To establish that $\Psi$ is an embedding (possibly after again shrinking the domain) it suffices to show that $\Psi_*(\partial_u)$,$\Psi_*(\partial_r)$,$\Psi_*(\partial_{\theta^A})$ span a $4-$dimensional subspace in $T_{\Psi(u, r, \theta)} TM$ for each $(u,r,\theta)$. To see this, let $\gamma_{u, \theta}:[0, r_0) \to M$ be the geodesic with $\gamma_{u, \theta}(0) = \Gamma(u)$, $\dot{\gamma}_{u, \theta}(0) = V(u, \theta)$. If $\Pi:TM \to M$ is the canonical projection, then $\gamma_{u, \theta}(r) = \Pi\circ\Psi(u, r, \theta)$. We have that
\begin{align*}
    \Psi_*\left(\left. \frac{\partial}{\partial u}\right|_{(u, r, \theta)} \right) &= \left(J_{u}(r),  \frac{D J_{u}}{\ud r} (r)\right) \\
    \Psi_*\left(\left. \frac{\partial}{\partial r}\right|_{(u, r, \theta)} \right) &= \left(J_{r}(r),  \frac{D J_{r}}{\ud r} (r)\right) \\
    \Psi_*\left(\left. \frac{\partial}{\partial \theta^A}\right|_{(r, r, \theta)} \right) &= \left(J_{\theta^A}(r),  \frac{D J_{\theta^A}}{\ud r} (r)\right)
\end{align*}
where each $J_*\in \{J_u, J_r, J_{\theta^A}\}$ is a Jacobi field along $\gamma_{u, \theta}$ satisfying the initial conditions:
\begin{align}
    \left(J_{u}(0),  \frac{D J_{u}}{\ud r} (0)\right) &= \left(\dot{\Gamma}(u), 0 \right) \\
    \left(J_{r}(0),  \frac{D J_{r}}{\ud r} (0)\right) &= \left(V(r, \theta), 0 \right) \\
    \left(J_{\theta^A}(0),  \frac{D J_{\theta^A}}{\ud r} (0)\right) &= \left(0, \frac{D}{d\theta^A} V(u, \theta) \right) 
\end{align}
Since $V(u, \theta)$ parameterises a sphere in $T_{\Gamma(u)} M$ to which both $V$ and $\dot{\Gamma}$ are transverse, we deduce that (possibly after shrinking the domain) the map $\Psi : (1-\delta u, 1+\delta u)\times [0, \delta r)\times S^2 \to TM$ is a smooth embedding. Furthermore, letting $\psi = \Pi\circ\Psi$ we note that for $r>0$ and $\delta r$ sufficiently small, the vectors $\psi_*(\partial_u)$,$\psi_*(\partial_r)$,$\psi_*(\partial_{\theta^A})$ span $T_{\psi(u, r, \theta)}M$ and so $\psi: (1-\delta u, 1+\delta u)\times (0, \delta r)\times S^2 \to M$ is a smooth coordinate chart on an open neighbourhood of $\iota(O)$ in $M \setminus \Gamma$.

Now, since $\Pi\circ\Psi$ is smooth we can pull-back the metric tensor $g$ to a smooth, symmetric $2-$tensor defined on $(1-\delta u, 1+\delta u)\times [0, \delta r)\times S^2$. For $r>0$ this tensor will be non-degenerate, but at $r=0$ $g$ will become degenerate. The metric functions 
\[
g_{uu}(u, r, \theta) = \left. (\Pi\circ\Psi)^*g\left(\partial_u, \partial_u \right) \right|_{(u, r, \theta)},\qquad g_{uA}(u, r, \theta) = \left. (\Pi\circ\Psi)^*g\left(\partial_u, \partial_{\theta^A} \right) \right|_{(u, r, \theta)},\quad \ldots
\]
are smooth functions on $(1-\delta u, 1+\delta u)\times [0, \delta r)\times S^2$. Similarly, the volume form of $M$ can be pulled-back to a smooth $4-$form, which will again degenerate at $r=0$.

\begin{remark} The construction above, which `blows-up' the axis $\Gamma$ to the cylinder $\{r=0\}$, enables us to cleanly discuss regularity at $\Gamma$ in terms of smooth objects on a manifold with boundary. This is mainly for convenience to avoid having to express regularity in terms of limits as $r\to 0$.
\end{remark}

\subsection{Form of the metric}\label{sec:coords.form}It follows from our construction above that
\[
(\Pi \circ \Psi)_*\left(\left. \frac{\partial}{\partial r}\right|_{(u, r, \theta)} \right) = \dot{\gamma}_{u, \theta}(r) = J_r(r)
\]
We further observe that for each of the Jacobi fields $J_*$ above we have
\[
g\left(\frac{D J_*}{\ud r}(0), V(u, \theta)\right)  = 0,
\]
and hence $g(J_*(r), \dot{\gamma}_{u, \theta}(r))$ is constant along $\gamma_{u, \theta}$, since the Jacobi equation implies $\frac{\ud^2}{\ud r^2}g(J_*(r), \dot{\gamma}_{u, \theta}(r)) = 0$. We deduce that
\begin{equation}
    g_{rr} = g_{rA} = 0, \qquad g_{ru}=g_{ur}=1. \label{metcomps}
\end{equation}

Next, we consider $g_{uu} = g(J_u(r), J_u(r))$. We have
\[
g_{uu}|_{r=0} = g(J_u(0), J_u(0)) = 1, \qquad \partial_r g_{uu}|_{r=0} = 2g\left(J_u(0), \frac{D J_u}{\ud r}(0)\right) = 0
\]
and thus $g_{uu} = 1 + O(r^2)$, where we can naively differentiate the error term as a smooth function. In principle, using the Jacobi equation we can develop the series in $r$ in terms of components of the Riemann curvature and its derivatives. Next we note that the Jacobi equation implies that $\frac{D^2 J_{\theta^A}}{\ud r^2}(0) = 0$. We compute
\begin{align*}
    g_{AB}|_{r=0} &= g(J_{\theta^A}(0), J_{\theta^B}(0)) = 0, \\
    \partial_r g_{AB}|_{r=0} &= g\left(\frac{D J_{\theta^A}}{\ud r}(0), J_{\theta^B}(0)\right)+ g\left(J_{\theta^A}(0), \frac{D J_{\theta^B}}{\ud r}(0)\right) = 0 \\
    \partial_r^2 g_{AB}|_{r=0} &= g\left(\frac{D^2 J_{\theta^A}}{\ud r^2}(0), J_{\theta^B}(0)\right)+ 2 g\left(\frac{D J_{\theta^A}}{\ud r}(0), \frac{D J_{\theta^B}}{\ud r}(0)\right)+ g\left(J_{\theta^A}(0), \frac{D^2 J_{\theta^B}}{\ud r^2}(0)\right) = 2\mathring{\gamma}_{AB}(1) \\
    \partial_r^3 g_{AB}|_{r=0} &= g\left(\frac{D^3 J_{\theta^A}}{\ud r^3}(0), J_{\theta^B}(0)\right)+ 3 g\left(\frac{D^2 J_{\theta^A}}{\ud r^2}(0), \frac{D J_{\theta^B}}{\ud r}(0)\right)\\
    &\qquad+ 3 g\left(\frac{D J_{\theta^A}}{\ud r}(0), \frac{D^2 J_{\theta^B}}{\ud r^2}(0)\right) + g\left(J_{\theta^A}(0), \frac{D^3 J_{\theta^B}}{\ud r^3}(0)\right) = 0
\end{align*}
where we have used the fact that $V(u, \theta)$ is a parameterisation of a round unit sphere in $T_{\Gamma(u)}M$. We deduce that $g_{AB} = \mathring{\gamma}_{AB} + O(r^4)$. Finally we consider $g_{uA}$:
\begin{align*}
    g_{uA}|_{r=0} &= g(J_{u}(0), J_{\theta^A}(0)) = 0 \\
    \partial_r g_{uA}|_{r=0} &= g\left(\frac{D J_{u}}{\ud r}(0), J_{\theta^a}(0)\right)+ g\left(J_{u}(0), \frac{D J_{\theta^A}}{\ud r}(0)\right) = 0 \\
    \partial_r^2 g_{uB}|_{r=0} &= g\left(\frac{D^2 J_{u}}{\ud r^2}(0), J_{\theta^A}(0)\right)+ 2 g\left(\frac{D J_{u}}{\ud r}(0), \frac{D J_{\theta^A}}{\ud r}(0)\right)+ g\left(J_{u}(0), \frac{D^2 J_{\theta^A}}{\ud r^2}(0)\right) = 0
\end{align*}
so that $g_{uA} = O(r^3)$.

\subsection{Extending away from the axis} We now wish to extend the coordinate system to cover a neighbourhood of $\iota(\Sigma)$. Recall from above that $(\overline{r}, \overline{\theta}) \mapsto \Phi^{-1}(\overline{r} \hat{x}(\overline{\theta}))$ gives a coordinate chart on $\Sigma \setminus \{O\}$. We define $\tilde{\iota}$ by $\tilde{\iota}(\overline{r}, \overline{\theta}) =\iota\circ \Phi^{-1}(\overline{r} \hat{x}(\overline{\theta}))$, so that $\tilde{\iota}:(0, \infty)\times S^2 \to \iota(\Sigma \setminus \{O\})$ is a diffeomorphism. We can take $\delta \overline{r}$ to be sufficiently small so that  $\hat{\iota}:=\Psi \circ \psi^{-1}\circ \tilde{\iota}:(0, \delta\overline{r}) \times S^2 \to TM$ is well defined. This map extends to a smooth map on $[0, \delta\overline{r}) \times S^2$ by setting $\hat{\iota}(0, \overline{\theta}) = V(1,\overline{\theta})$. By construction, $\hat{\iota}: [0,\delta \overline{r})\times S^2$ parameterises a smooth surface in $TM$, and we can write
\[
(u_0(\overline{r}, \overline{\theta}), r_0(\overline{r}, \overline{\theta}), \theta_0(\overline{r}, \overline{\theta}))    = \Psi^{-1}\circ \hat{\iota}(\overline{r}, \overline{\theta})
\]
for smooth $(u_0, r_0, \theta_0): [0,\delta \overline{r})\times S^2 \to (1-\delta u, 1+\delta u)\times [0, \delta r)\times S^2$. These functions satisfy
\[
u_0(0, \overline{\theta}) = 1, \quad r_0(0, \overline{\theta}) = 0, \quad {\theta}^A_0(0, \overline{\theta}) = \overline{\theta}^A, \quad \partial_{\overline{r}}r_0(0, \overline{\theta}) = 1 = -\partial_{\overline{r}}u_0(0, \overline{\theta})
\]
Let $\chi : [0, \infty) \to [0,1]$ be a smooth monotonic function satisfying 
$\chi(\overline{r}) = 1$ for $\overline{r}<\frac{1}{4}\delta \overline{r}$ and 
$\chi(\overline{r}) = 0$ for $\overline{r}>\frac{3}{4}\delta \overline{r}$. We extend $(u_0, r_0, \theta_0)$ to the domain $[0, \infty) \times S^2$ by setting:

\begin{align*}
u_1(\overline{r}, \overline{\theta}) &= \chi(\overline{r}) u_0(\overline{r}, \overline{\theta}) + (1-\chi(\overline{r})) (1-\overline{r}) \\
r_1(\overline{r}, \overline{\theta}) &= \chi(\overline{r}) r_0(\overline{r}, \overline{\theta}) + (1-\chi(\overline{r})) \overline{r} \\
\theta^A_1(\overline{r}, \overline{\theta}) &= \chi(\overline{r}) \theta^A_0(\overline{r}, \overline{\theta}) + (1-\chi(\overline{r})) \overline{\theta}^A 
\end{align*}
Taking $\delta \overline{r}$ sufficiently small, the map $F:(\overline{r}, \overline{\theta})\mapsto (u_{1}, \theta_{1})$ is a smooth diffeomorphism of $[0, \infty) \times S^2$ onto $(-\infty, 1] \times S^2$. Now, suppose we are given a smooth future directed null vector-field, $\tilde{W}:\iota(\Sigma \setminus \{O\}) \to TM$ such that $\Pi \circ \tilde{W} = I_{\iota(\Sigma \setminus \{O\})}$, let $W = \tilde{W}\circ F^{-1}$. Given $(u, \theta) \in (-\infty, 1)\times S^2$ we flow along the geodesic vector field starting at $W(u,r)$ a distance $r-r_1\circ F^{-1}(u,\theta)$. For $r$ sufficiently close to $r_1\circ F^{-1}(u,r)$ this will be well defined and we call this point $\tilde{\Psi}(u, r, \theta)$. In the neighbourhood of a point $(u_1(\overline{r}, \overline{\theta}), r_1(\overline{r}, \overline{\theta}), \theta_1(\overline{r}, \overline{\theta}))$ the map $\tilde{\Psi}$ is smooth. Further we have
\begin{align*}
    \tilde{\Psi}_*\left(\left. \frac{\partial}{\partial u}\right|_{(u, r, \theta)} \right) &= \left(J_{u}(r-r_1),  \frac{D J_{u}}{\ud r} (r-r_1)\right) \\
    \tilde{\Psi}_*\left(\left. \frac{\partial}{\partial r}\right|_{(u, r, \theta)} \right) &= \left(J_{r}(r-r_1),  \frac{D J_{r}}{\ud r} (r-r_1)\right) \\
    \tilde{\Psi}_*\left(\left. \frac{\partial}{\partial \theta^A}\right|_{(u, r, \theta)} \right) &= \left(J_{\theta^A}(r-r_1),  \frac{D J_{\theta^A}}{\ud r} (r-r_1)\right)
\end{align*}
where $r_1 = r_1\circ F^{-1}(u,\theta)$ and each $J_*\in \{J_u, J_r, J_{\theta^A}\}$ is a Jacobi field along the geodesic $\tilde{\gamma}_{u, \theta}(\lambda) = \Pi \circ \tilde{\Psi}(u, \lambda+r_1, \theta)$ satisfying the initial conditions:
\begin{align}
    \left(J_{u}(0),  \frac{D J_{u}}{\ud r} (0)\right) &= \left((\tilde{\iota}\circ F^{-1})_*\left(\frac{\partial}{\partial u} \right)-W\frac{\partial}{\partial u}r_1\circ F^{-1} , \frac{DW}{\ud u} \right) \label{eq:Judef}\\
    \left(J_{r}(0),  \frac{D J_{r}}{\ud r} (0)\right) &= \left(W(u, \theta), 0 \right) \\
    \left(J_{\theta^A}(0),  \frac{D J_{\theta^A}}{\ud r} (0)\right) &= \left((\tilde{\iota}\circ F^{-1})_*\left(\frac{\partial}{\partial \theta^A} \right)-W\frac{\partial}{\partial \theta^A}r_1\circ F^{-1} , \frac{DW}{\ud \theta^A} \right)\label{eq:Jthetadef}
\end{align}
Since $F$ is a diffeomorphism, $(\tilde{\iota}\circ F^{-1})_*\partial_u, (\tilde{\iota}\circ F^{-1})_*\partial_{\theta^A}$ span $T\iota(\Sigma)\subset TM$, and as $\iota(\Sigma)$ is spacelike, $W$ must be transverse to this subspace. As a result we see that the Jacobi fields are initially linearly independent and span $TM$ and so $\tilde{\Psi}$ is locally an embedding and $\Pi \circ \tilde{\Psi}$ a local diffeomorphism. We further observe that there is a unique choice of null $W$ such that
\[
g(J_u(0), J_r(0)) = -1,\qquad  g(J_r(0), J_r(0)) = 0, \qquad g(J_{\theta^A}(0), J_r(0)) = 0. 
\]
Making this choice of $W$ we see that near $\{r=0\}$ our two embedding maps $\tilde{\Psi}$ and ${\Psi}$ agree, so we can use $\tilde{\Psi}$ to smoothly extend ${\Psi}$. Further, our previous argument establishing \eqref{metcomps} can be extended to this region.

On $\iota(\Sigma \cap \{ \overline{r} > \frac{3}{4} \delta \overline{r} \})$, where $(u, r, \theta) = (1 - \overline{r}, \overline{r}, \overline{\theta})$, the null pair $e_{3}, e_{4}$ admits a simple expression as follows. Since $e_{4} = W$, by the above characterisation of $W$, we have
\begin{equation}\label{eq:e4initdef}
    e_{4} = W = \ell(N + \iota_{\ast} \overline{N}),
\end{equation}
where $\overline{N}$ is the outward unit normal of the spheres of constant $\overline{r}$ in $\Sigma$ and $\ell^{-1} = h(\rd_{\overline{r}}, \overline{N})$. From the characterisation \eqref{e4.def.metric}, it follows that $e_{3} = \ell^{-1} (N - \iota_{\ast} \overline{N})$ on the same subset of $\iota(\Sigma)$.  

We summarise this discussion in the following theorem

\begin{theorem}\label{thm:coords}
    Given $(\Sigma, h, k)$ a smooth set of Cauchy data for the vacuum Einstein equations, with $\Sigma \simeq \mathbb{R}^3$ and $O\in \Sigma$ a marked point, let $(M, g, \iota)$ be its maximal Cauchy development, and $\Gamma$ the future directed geodesic normal to $\iota(\Sigma)$ at $O$, parameterised by proper time. Given $\delta\overline{r}>0$ sufficiently small there exists $U$, open in $S$, a neighbourhood of $\tilde{\Sigma}_1$ and a smooth map $\Psi: U \to TM$ such that
    \begin{enumerate}[i)]
    \item $\Pi\circ \Psi: U \setminus \{r=0\} \to M \setminus \Gamma$ is a smooth diffeomorphism onto its image
    \item $\Psi^*g$ takes the form
    \[
    \Psi^*g = - (\ud u \otimes \ud r+\ud r\otimes \ud u) - f \ud u \otimes \ud u +\gamma_{AB}(\ud\th^A-b^A \ud u)\otimes (\ud\th^B - b^B \ud u)
    \]
    where $f, b^A, \gamma_{AB}$ are smooth functions on $U$ satisfying
    \[
    f = 1 + O(r^2), \qquad b^A = O(r), \qquad \gamma_{AB} = \mathring{\gamma}_{AB} + O(r^4)
    \]
    \item $\Pi\circ \Psi(\tilde{\Sigma}_1)$ and $\iota(\Sigma)$ are coincident outside the spacetime region $\Pi\circ \Psi((1-2\delta\overline{r}, 1+2\delta\overline{r})\times [0, \delta \overline{r})\times S^2$). Moreover, we have
    \begin{equation} \label{eq:e4-ids}
       e_{3} = \ell^{-1} (N - \iota_{\ast} \overline{N}), \qquad
       e_{4} = \frac{\rd}{\rd r} = \ell(N + \iota_{\ast} \overline{N}) \quad \hbox{ on } \iota(\Sigma \cap \{\overline{r} \geq \tfrac{3}{4} \delta r \}),
    \end{equation}
    where $\overline{N}$ is the outward unit normal of the spheres of constant $\overline{r}$ in $\Sigma$ and $\ell^{-1} = h(\rd_{\overline{r}}, \overline{N})$.
    
    In particular, any bounds assumed on the metric and a finite (but arbitrary) number of derivatives may be transferred by continuity to hold on $\tilde{\Sigma}_1$ with arbitrarily small loss.
    \end{enumerate}
\end{theorem}

\begin{corollary}\label{cor:loc.exist.bounds}
    For any $s$, we have
    \[
    \sup_{(u,r)\in S_T} r^{-3}\left[ \abs{f-1}_s + \abs{b}_s +\abs{\gamma - \mathring{\gamma}}_{s} \right] <\infty
    \]
    The following quantities are well defined and smooth on $S_T$
    \begin{align*}
        \chih, (\trch - 2r^{-1}),\chibh, (\trchb + 2r^{-1}), \omegab, \xib, \eta, \alpha, \beta, \betab, \rho, \sigma 
    \end{align*}
    and satisfy
    \begin{align*}
        \sup_{(u,r)\in S_T} &r^{-1}\Bigg[r^{-1}\abs{\chih}_{s}+r^{-1}\abs{\trch - 2r^{-1}}_{s} +r^{-1}\abs{\chibh}_{s}+r^{-1}\abs{\trchb + 2r^{-1}}_{s}  \\
        &\qquad + r^{-1}\abs{\omegab}_s +r^{-1}\abs{\xib}_s + r^{-1}\abs{\eta}_s +\abs{\alpha}_s+\abs{\beta}_s  +\abs{\betab}_s +\abs{\rho}_s +\abs{\sigma}_s\Bigg](u,r) <\infty
    \end{align*}
    for any $s$.
\end{corollary}

\section{Bootstrap assumptions and statement of the main theorem}\label{sec:assumptions}

We fix $s_0 \geqslant 20$ to be \emph{even} (see footnote \ref{fn:regularity}). Fix some $\delta>0$ which we will later choose sufficiently small ($\delta<1/20$ will suffice). For compactness of notation we set $p^\pm:=p \pm \delta$. We will later take $0< \epsilon_0 \ll 1$ to be sufficiently small depending on $s_0$ and $\de$. From now on, we assume that $0 < \ep \leqslant \ep_0$.

\subsection{Initial data assumptions}

\subsubsection{Weak initial data assumptions}
We assume that $\alp$ initially satisfies
\begin{equation}\label{eq:data.alp.weak}
\begin{aligned}
&    \| (r \nab_{4} \alp, r \nab \alp, \alp) \|_{s_0, (-1)^+, 2^+}(\tilde{\Sigma}_1)
      + \|r^{-2} \nab_3 (r^2\alp)\|_{s_0, (-2)^+, 3^+}(\tilde{\Sigma}_1) 
+\abs{\alp}_{s_0, (-1)^+, 2^+}(\tilde{\Sigma}_1) < \ep^2
\end{aligned}
\end{equation}
and that the other geometric quantities initially satisfy the pointwise(-in-$(u, r)$) bounds
\begin{equation}\label{eq:data.weak}
    \begin{split}
        &|\hat{\chi}, \trch - 2r^{-1}|_{s_0, 0^+, 1^+}(\tilde{\Sigma}_1)+ |\gamma-\mathring{\gamma} |_{s_0, 1^+, 0^{+}}(\tilde{\Sigma}_1) + |\beta|_{s_0-1, (-1)^+, 2^+}(\tilde{\Sigma}_1)\\
    &\qquad + |\eta|_{s_0-1, 0^+, 1^+}(\tilde{\Sigma}_1) + |\rho, \sigma|_{s_0-2, (-1)^+, 2^+}(\tilde{\Sigma}_1) + |\chibh, \trchb + 2r^{-1}|_{s_0-2, 0^+, 1^{+}}(\tilde{\Sigma}_1) \\& \qquad + |\betab|_{s_0-3, (-1)^+, 2^{+}}(\tilde{\Sigma}_1) + |b|_{s_0-1, 1^+, 0^{+}}(\tilde{\Sigma}_1) + \abs{\omegab}_{s_0-2, 0^+, 1^+}(\tilde{\Sigma}_1) + \abs{f- 1}_{s_0-2, 1^+, 0^+}(\tilde{\Sigma}_1) <\epsilon^2 .  
    \end{split}
\end{equation}
and the integrated bounds
\begin{equation}\label{eq:data.weak.int}
    \begin{split}
        &\| \hat{\chi}, \trch - 2r^{-1}\|_{s_0+1, 0^+, 1^+}(\tilde{\Sigma}_1)+ \|\gamma-\mathring{\gamma} \|_{s_0+1, 1^+, 0^{+}}(\tilde{\Sigma}_1) + \| \beta \|_{s_0+1, (-1)^+, 2^+}(\tilde{\Sigma}_1)\\
    &\qquad + \| \eta \|_{s_0+1, 0^+, 1^+}(\tilde{\Sigma}_1) + \| \rho, \sigma \|_{s_0, (-1)^+, 2^+}(\tilde{\Sigma}_1) + \|\chibh, \trchb + 2r^{-1}\|_{s_0, 0^+, 1^{+}}(\tilde{\Sigma}_1) \\& \qquad + \| \betab \|_{s_0-1, (-1)^+, 2^{+}}(\tilde{\Sigma}_1) + \| \omegab \|_{s_0, 0^+, 1^+}(\tilde{\Sigma}_1) + \| f- 1 \|_{s_0, 1^+, 0^+}(\tilde{\Sigma}_1) <\epsilon^2 .  
    \end{split}
\end{equation}

\begin{remark} In this case, the exponents $p_{0}, p_{\infty}$ are determined entirely by the object type: for curvature components, $(p_{0}, p_{\infty}) = ((-1)^{+}, 2^{+})$; for connection coefficients, $(p_{0}, p_{\infty}) = (0^{+}, 1^{+})$; and for metric components, $(p_{0}, p_{\infty}) = (1^{+}, 0^{+})$. The exponents $s$ match with those in our bootstrap assumptions \eqref{eq:BA.weak.pointwise} and \eqref{eq:BA.strong.integrated} below.
\end{remark}

\subsubsection{Strong initial data assumptions}
We also consider a version of the theorem with strong initial data assumptions, in which the $s$ and $p_{0}$ exponents remain the same but the $p_{\infty}$ exponents differ for some quantities. In this case, we assume that $\alp$ initially satisfies
\begin{equation}\label{eq:data.alp.strong}
    \|(r \nab_{4} \alp, r \nab \alp, \alp)\|_{s_0+1, (-1)^+, 5^{-}(\tilde{\Sigma}_1)} + \|r^{-2}\nab_3 (r^2\alp)\|_{s_0, (-2)^+, 6^{-}}(\tilde{\Sigma}_1)  +\abs{\alp}_{s_0, (-1)^+, 5^{-}}(\tilde{\Sigma}_1) < \ep^2
\end{equation}
and that the other geometric quantities initially satisfy the following bounds in addition to \eqref{eq:data.weak}:
\begin{equation}\label{eq:data.strong}
    \begin{split}
        &|\hat{\chi}, \trch - 2r^{-1}|_{s_0, 0^+, 2}(\tilde{\Sigma}_1) 
        + |\beta|_{s_0-1, (-1)^+, 4}(\tilde{\Sigma}_1) + |\eta|_{s_0-1, 0^+, 2}(\tilde{\Sigma}_1) + |\rho, \sigma|_{s_0-2, (-1)^+, 3}(\tilde{\Sigma}_1)  <\epsilon^2 ,  
    \end{split}
\end{equation}
\begin{equation}\label{eq:data.strong.int}
    \begin{split}
        &\| \hat{\chi}, \trch - 2r^{-1} \|_{s_0+1, 0^+, 2^{-}}(\tilde{\Sigma}_1) 
        + \| \beta \|_{s_0+1, (-1)^+, 4^{-}}(\tilde{\Sigma}_1) + \|\eta \|_{s_0+1, 0^+, 2^{-}}(\tilde{\Sigma}_1) +\|\rho, \sigma\|_{s_0, (-1)^+, 3^{-}}(\tilde{\Sigma}_1)  <\epsilon^2 . 
    \end{split}
\end{equation}
To relate these conditions to the assumptions on initial data in the introduction we have the following result, whose proof we defer to Appendix \ref{sec:appendix2}. Note that the loss of derivatives here is likely not sharp.

\begin{theorem}\label{thm:data.lemma}
    Fix $0<\epsilon \ll 1$ and $\nu>\delta$. Suppose that $(\mathbb{R}^3, h, k)$ is a smooth asymptotically flat initial data set in the centre of mass frame with maximal future Cauchy development $(M, g, \iota)$. Then there exists $\mathfrak{e}>0$ such that if 
    \begin{equation}\label{eq:frakedef}
         \mathfrak{n}:=\|h_{ij} - \delta_{ij}\|_{H^{s'+1, \nu-\frac{3}{2}}} + \|k_{ij}\|_{H^{s', \nu-\frac{1}{2}}}< \mathfrak{e}
    \end{equation}
    for some $s'>s_0+6$ we can construct coordinates in a neighbourhood of $\iota(\mathbb{R}^3)$, as in Theorem \ref{thm:coords}, such that \eqref{eq:data.alp.weak}, \eqref{eq:data.weak}, \eqref{eq:data.weak.int} hold with $\epsilon^2\lesssim_{\mathfrak{e}, s', s_0, \delta, \nu} \mathfrak{n}$. 

    If, moreover, $\nu>3-\delta$ and there exist $R_0>0, M>0, J\in \mathbb{R}^3$ such that 
     \begin{align*}
    \mathfrak{n}'&:= \left\|h_{ij} - (1 + \tfrac{2M}{|x|} + \tfrac{3 M^{2}}{2 |x|^{2}}) \delta_{ij}\right\|_{H^{s'+1, \nu-\frac{3}{2}}(\{|x|>R_0\})}\\&\qquad + \left\|k_{ij}-3 (\in_{i k \ell} J^{\ell} \tfrac{x_{j} x^{k}}{|x|^{5}} + \in_{j k \ell} J^{\ell} \tfrac{x_{i} x^{k}}{|x|^{5}})\right\|_{H^{s', \nu-\frac{1}{2}}(\{|x|>R_0\})}  <\infty,
    \end{align*}
    then there exists $\mathfrak{e}'>0$ such that if
    \[
    \mathfrak{n}+\mathfrak{n}'+M+|J| =: \mathfrak m<\mathfrak{e}'
    \]
    then \eqref{eq:data.alp.strong}, \eqref{eq:data.strong}, \eqref{eq:data.strong.int} hold with $\epsilon^2\lesssim_{\mathfrak{e}', s', s_0, \delta, \nu} \mathfrak m$.
\end{theorem}

\begin{remark}[Cancellation properties of the leading order asymptotics in Definition~\ref{def:AF.data}] \label{rem:AF.data.cancel}
    The validity of the strong initial data assumptions \eqref{eq:data.alp.strong}, \eqref{eq:data.strong} and \eqref{eq:data.strong.int} (specifically the control on $\alpha$ and $\beta$) implies nontrivial cancellation properties of the leading-order spatial asymptotics prescribed in Definition~\ref{def:AF.data}.iii). Naively, a generic $r^{-1}$ term in the metric tensor would typically induce $r^{-3}$ terms in the curvature tensor, which would be inconsistent with the $p_{\infty}$ exponents for the norms of $\alpha$ and $\beta$ in these assumptions. Indeed, they reflect the fact that the initial data for $\beta$ is entirely devoid of terms below order $r^{-4}$, and also that $\alpha$ does not have terms below order $r^{-5^{-}}$.

    In Appendix~\ref{sec:appendix2}, we verify this cancellation by showing that the leading-order spatial asymptotics in Definition~\ref{def:AF.data} are exactly those of an initial data set for the Kerr metric. The problem then reduces to checking that \eqref{eq:data.alp.strong}, \eqref{eq:data.strong} and \eqref{eq:data.strong.int} hold for the exact Kerr metric when equipped with our Newman--Unti gauge and restricted outside a large ball, which can be done by explicit computation. We further note that if $J = 0$ (corresponding to the Schwarzschild metric), there are more cancellations: both $\beta$ and $\alpha$ completely vanish for the exact Schwarzschild metric equipped with our Newman--Unti gauge (adapted to symmetry spheres).
\end{remark}

\subsection{Bootstrap assumptions}
\subsubsection{Weak bootstrap assumptions}\label{sec:WBS}

 We say that the \emph{weak bootstrap assumptions} are satisfied if the region $S_T$ is equipped with a metric of the form \eqref{formofthemetric}, satisfying the following bounds on geometrical quantities for some $0<\epsilon\leqslant \epsilon_0$: 
\begin{enumerate}
    \item (Pointwise bounds)
    \begin{equation}\label{eq:BA.weak.pointwise}
    \begin{split}
        &|\alpha|_{s_0, (-1)^+, 2^+}+  |\hat{\chi}, \trch - 2r^{-1}|_{s_0, 0^+, 1^+}+ |\brk{u}^{0^{+}}(\gamma-\mathring{\gamma} )|_{s_0, 1^+, 0} + |\beta|_{s_0-1, (-1)^+, 2^+}\\
    &\qquad + |\eta|_{s_0-1, 0^+, 1^+} + |\rho, \sigma|_{s_0-2, (-1)^+, 2^+} + |\brk{u}^{0^{+}}\betab|_{s_0-3, (-1)^+, 2} + |\brk{u}^{0^{+}}(\chibh, \trchb + 2r^{-1})|_{s_0-2, 0^+, 1}\\
    & \qquad  
    	+ |\brk{u}^{0^{+}} (K - r^{-2}) |_{s_0-2, (-1)^+, 2} + |\brk{u}^{1^{+}}  b|_{s_0-1, 1^+, -1} + \abs{\jb{u}^{1^+}\omegab}_{s_0-2, 0^+, 0} 
    \\ &\qquad  
    + \abs{\jb{u}^{1^+}(f- 1)}_{s_0-2, 1^+, -1}
    + \abs{\jb{u}^{0^+}(f-2\omegab r - 1)}_{s_0-2, 1^+, 0}+ \abs{\jb{u}^{1^+}\xib}_{s_0-3, 0^+, 0} <\epsilon .  
    \end{split}
\end{equation}
    \item (Integrated bounds) 
    \begin{equation}\label{eq:BA.weak.integrated}
    \begin{split}
        &\nrm{\alpha}_{s_0+1, (-\frac{1}{2})^+, \frac{3}{2}^+} + \nrm{\brk{u}^{\frac{1}{2}} \beta}_{s_0+1, (-\frac{1}{2})^+, 1^+} \\
        &\qquad + \nrm{\chih, \trch - 2r^{-1}}_{s_0+1, \frac{1}{2}^+, \frac{1}{2}^{+}} + \nrm{\gamma - \mathring{\gamma}}_{s_0+1, \frac{3}{2}^+, (-\frac{1}{2})^{+}} + \nrm{\brk{u}^{\frac{1}{2}} \eta}_{s_0+1, \frac{1}{2}^+, 0^+}\\
    &\qquad + \|\brk{u}^{\frac{1}{2}}(\rho, \sigma)\|_{s_0, (-\frac{1}{2})^+, 1^+} + \|\brk{u}^{\frac{1}{2}} \betab\|_{s_0-1, (-\frac{1}{2})^+, 1^{+}}  + \|\brk{u}^{\frac{1}{2}}(\chibh,\trchb + 2r^{-1})\|_{s_0, \frac{1}{2}^+, 0^{+}} \\
    &\qquad + \|\brk{u}^{\frac{1}{2}} (K-r^{-2})\|_{s_0, (-\frac{1}{2})^+, 1^+}
    + \| \brk{u}^{\frac{1}{2}+\frac{3}{2}\delta} \omegab \|_{s_0,\frac{1}{2}^+ , -\frac{1}{2}\delta} 
    +\|\brk{u}^{\frac{1}{2}+\frac{3}{2}\delta} (f-1)\|_{s_0,\frac{3}{2}^+ , -1-\frac{1}{2}\delta} \\
    &\qquad
    + \|\brk{u}^{\frac{1}{2}}(f-2\omegab r - 1)\|_{s_0,\frac{3}{2}^+ , (-1)^+}
    +  \|\brk{u}^{\frac{1}{2}+\frac{3}{2}\delta}\xib \|_{s_0-1,\frac{1}{2}^+ , -\frac{1}{2}\delta} \\
    &\qquad + \left(\int_{u_T}^T(\jb{u}^{\frac{1}{2}^{+}})^{2}(\abs{\omb}^2_{s_0,\frac{1}{2}^+ , 0}(H_u) + \abs{f-1 }^2_{s_0,\frac{3}{2}^+ , -1}(H_u)) \ud u\right)^{\frac{1}{2}} <\epsilon.
    \end{split}
    \end{equation}
\end{enumerate}

\begin{remark}
In this case, the exponent $p_{0}$ depends only on the object type, matching the exponents used in the initial data assumptions for pointwise bounds and shifted by $\frac{1}{2}$ for integrated bounds. The exponent $p_{\infty}$ may vary between objects of the same type, but the total power of $\brk{u}$ and $(r + \brk{u})$ (represented by $p_{\infty}$) remains fixed for each object type, consistent with our previous assumptions.
\end{remark} 

\subsubsection{Strong bootstrap assumptions}\label{sec:BS}
We say the \emph{strong bootstrap assumptions} hold if the following improved pointwise bounds hold in addition to the weak assumptions:
\begin{equation}\label{eq:BA.strong.pointwise}
\begin{aligned}
    &|\alpha|_{s_0, (-1)^+, 5^{-}}+  |\jb{u_{+}}^{2^{-}}(\hat{\chi},\trch - 2r^{-1})|_{s_0, 0^+, 2} \\
    &\qquad +|\jb{u_{+}}^{1^{-}} \beta|_{s_0-1, (-1)^+, 4}+ |\jb{u_{+}}^{2^{-}}\eta|_{s_0-1, 0^+, 2} + |\jb{u_{+}}^{2^{-}}(\rho, \sigma)|_{s_0-2, (-1)^+, 3}  <\epsilon
\end{aligned}
\end{equation}
together with the following improved integrated bounds
\begin{equation}\label{eq:BA.strong.integrated}
    \begin{split}
        &\|\alpha\|_{s_0+1, (-\frac{1}{2})^+, \frac{9}{2}^{-}} + \|\beta\|_{s_0+1, (-\frac{1}{2})^+, \frac{7}{2}^-}
        +\|\chih,\trch - 2r^{-1}\|_{s_0+1, \frac{1}{2}^+, \frac{3}{2}^-} \\
        &\qquad + \|\eta\|_{s_0+1, \frac{1}{2}^+, \frac{3}{2}^-} 
        + \|\rho, \sigma\|_{s_0, (-\frac{1}{2})^+, \frac{5}{2}^-} + \|K - r^{-2}\|_{s_{0}, (-\frac{1}{2})^{+}, \frac{3}{2}^{+}} \\
    &\qquad 
    + \|\chibh, \trchb + 2r^{-1}\|_{s_0, \frac{1}{2}^+, \frac{1}{2}^{+}} + \|f-2\omegab r - 1\|_{s_0,\frac{3}{2}^+ , (-\frac{1}{2})^{+}}  <\epsilon.
    \end{split}
\end{equation}

\begin{remark} 
    Note that the $s$ and $p_{0}$ exponents are the same as the weak bootstrap case; the only difference lies in the $p_{\infty}$ exponents. All $p_{\infty}$ exponents in \eqref{eq:BA.strong.pointwise} are sharp in view of the structure of the transport equation the geometric quantities satisfy (dictated by $\lambda_{0}$ in Theorem~\ref{thm:transest} below), with the exception of $\alpha$ (for which $p_{\infty} = 5^{-}$ is $\delta$ away from the sharp value $5$). On the other hand, the values of $p_{\infty}$ in \eqref{eq:BA.strong.integrated} are mostly determined by the initial data assumptions, specifically, \eqref{eq:data.strong.int} for $\alpha, \beta, \chih, \trch-2r^{-1}, \eta, \rho, \sigma$ and \eqref{eq:data.weak.int} for $\chibh, \trchb + 2 r^{-1}, f - 2 \omegab r - 1$. For the remaining quantity $K - r^{-2}$, $p_{\infty}$ is determined through the others by the constraint equation \eqref{eq:constrho} (see also \eqref{eq:constK}).
\end{remark}

\subsection{Precise formulation of the main theorem}\label{sec:theorem}

We are now in a position to precisely state the main result of the paper.
\begin{theorem}\label{thm:main}
    Fix $s_0, \delta$ as above. Suppose $(\Sigma, h, k)$ is a smooth set of Cauchy data for the vacuum Einstein equations with $\Sigma \simeq \mathbb{R}^3$, with maximal future Cauchy development $(M, g, \iota)$. Choose Newman--Unti coordinates in a neighbourhood of $\iota(\Sigma)$ by the process described in \S\ref{sec:coords}. There exists $\epsilon_0=\epsilon_0(s_0, \delta)>0$ such that if that the resulting metric satisfies the weak initial data assumptions \eqref{eq:data.alp.weak}, \eqref{eq:data.weak} for $0<\epsilon<\epsilon_0$ the following conclusions hold.
    \begin{enumerate}[i)]
        \item $(M, g)$ is globally smooth and future complete and may be covered by a single Newman--Unti coordinate chart, regular away from the axis.
        \item $(M, g)$ is asymptotically flat in the sense that all curvature components with respect to the frame $(e_A, e_3, e_4)$ decay to zero along any inextendible future-directed causal curve.
        \item The metric components, the Ricci coefficients, and the curvature components satisfy the bounds \eqref{eq:BA.weak.pointwise}, \eqref{eq:BA.weak.integrated} globally, with $\epsilon$ replaced by $C\epsilon^2$ for some $C = C(s_0, \delta) >0$.
    \end{enumerate}
    Furthermore, there exists $\epsilon_0'=\epsilon_0(s_0, \delta)>0$ such that if the metric satisfies the strong initial data assumptions \eqref{eq:data.alp.weak}, \eqref{eq:data.weak}, \eqref{eq:data.alp.strong} and \eqref{eq:data.strong} for $0<\epsilon<\epsilon_0'$ then the bounds \eqref{eq:BA.strong.pointwise}, \eqref{eq:BA.strong.integrated} hold globally, again with $\epsilon$ replaced by $C\epsilon^2$. In particular, the spacetime exhibits the almost sharp Bondi--Sachs peeling property, in the sense that \eqref{eq:main-1:imp} holds.
\end{theorem}
We defer the proof until \S\ref{sec:proof}.

\section{Consequences of the bootstrap assumptions}\label{sec:BA.consequences}
\subsection{Equivalence of norms} \label{sec:equiv.norm}
Here we establish the equivalence of the norms defined in \S\ref{sec:geometry} (defined using background geometry) with their counterparts defined using the actual metric $g$. Let
\begin{align*}
|\phi|_{0, \gamma}(u,r) &= \left(\int_{S_{u,r}} \left| \phi \right|^2_{{\gamma}} \ud A_{\gamma} \right)^{\frac{1}{2}} \\
|\phi|_{1, \gamma}(u,r) &= \left(|r{\nab}_4 \phi|^2_{0,\gamma}(u,r) +|r{\nab} \phi|^2_{0, \gamma}(u,r) (u,r)+|\phi|^2_{0, \gamma}(u,r)  \right)^{\frac{1}{2}}, 
\end{align*}
and for $s \in \mathbb Z_{>1}$
\begin{equation}
   |\phi|_{s+1, \gamma}(u,r) = \left(|r{\nab}_4 \phi|^2_{s,\gamma}(u,r) +|r{\nab} \phi|^2_{s, \gamma}(u,r) (u,r)+|\phi|^2_{s, \gamma}(u,r) +  |\brk{u} \nab_u\phi|^2_{s-1, \gamma}(u,r)  \right)^{\frac{1}{2}}. \label{ norm2}
\end{equation}
Note that the \emph{geometric norm} $\abs{\phi}_{s, \gamma}$ controls $(r\nab_4)^i (r\nab)^j (\brk{u} \nab_u)^{k}$ acting on $\phi$, with the derivatives taken in any order, as long as $i+j+2 k\leqslant s$. For measurable $R \subset S_{T}$, we also introduce the norms 
\begin{align*}
	| \phi |_{s, p_{0}, p_{\infty}, \gamma}(R) &= \sup_{(u, r) \in R} \frac{(\brk{u} + r)^{p_{\infty} + p_{0}}}{r^{p_{0} + 1}}  \abs{\phi}_{s, \gamma}(u, r), \\
	\| \phi \|_{s, p_{0}, p_{\infty}, \gamma}(R) &= \left(\int_{R} \abs{\phi}_{s, \gamma}^{2}(u, r) \frac{(\brk{u} + r)^{2 p_{\infty} + 2 p_{0}}}{r^{3 + 2 p_{0}}} \ud u \ud r \right)^{\frac{1}{2}}.
\end{align*}

For this section we shall raise and lower indices explicitly, to avoid confusion over which metric is being used. Let $\delta\gamma_{AB} = \gamma_{AB} - \mathring\gamma_{AB}$ be the difference between the true metric and the reference metric, and assume the bootstrap assumptions hold. Since $|r^{-1}\mathring{\gamma}|_s = (8\pi)^{\f 12}$, 
this implies that  $|r^{-1}\gamma|_{s} < (8\pi)^{\f 12} + C\epsilon$. 
Let $G_A{}^B: =\delta \gamma_{AC}(\mathring{\gamma}^{-1})^{CB}$. We note that Sobolev embedding gives
\[
\abs{G}_{\ring{\gamma}} = \abs{\delta \gamma}_{\ring{\gamma}} <c_S |\delta \gamma|_{2, 0, 0} < c_S \epsilon_0.
\]
For some universal constant $c_S$. It follows from a Neumann series expansion  that for $\epsilon_0$ sufficiently small $I + G$ is invertible and $(I+G)^{-1} = I + \tilde{G}$ with $\abs{\tilde{G}}_{\ring{\gamma}} \leqslant 2 \abs{G}_{\ring{\gamma}}$. Observing that
\[
\gamma_{AB} =  \left(\delta_A{}^C + G_A{}^C \right)\mathring{\gamma}_{CB}
\]
we deduce that $\gamma_{AB}$ is invertible, with
\[
(\gamma^{-1})^{AB} = (\mathring{\gamma}^{-1})^{AC}\left(\delta_C{}^B + \tilde{G}_C{}^B \right) = (\mathring{\gamma}^{-1})^{AB} +\delta({\gamma}^{-1})^{AB}, 
\]
where $\delta({\gamma}^{-1})^{AB} = (\mathring{\gamma}^{-1})^{AC} \tilde{G}_C{}^B$ satisfies
\[
\abs{\delta({\gamma}^{-1})}_{\mathring{\gamma}} <2 c_S \epsilon_0.
\]
Moreover, we have that
\[
\det \gamma = \det \mathring{\gamma} \det(I + G)
\]
and, since the determinant is locally Lipshitz continuous about $I$, for $\epsilon_0$ sufficiently small there exists $c$ such that
\[
\abs{\det(I + G)-1}<c\abs{G}_{\ring{\gamma}} < c c_S \epsilon_0.
\]

Combining the observations above, we have
\begin{lemma}\label{lem:gammacont}
    For $\epsilon_0$ sufficiently small,
    and any $k, l=0, 1, 2,\ldots$ there exists a constant $c = c( k, l)$ such that 
    \[
   c^{-1} |\phi|_{0, \gamma}(u,r) \leqslant |\phi|_0(u,r) \leqslant c |\phi|_{0, \gamma}(u,r)
    \]
    holds for all $S$-tangent $(k,l)-$tensors $\phi$ and all $(u,r) \in S_T$.
\end{lemma}

Now we wish to consider the equivalence of the higher order norms associated to $\gamma$ and $\mathring{\gamma}$. A short computation establishes that
\[
\delta\slashed{\Gamma}^C_{BA}:= {\slashed{\Gamma}}^C_{BA}-\mathring{\slashed{\Gamma}}^C_{BA} = \frac{1}{2}(\gamma^{-1})^{CD}\left(\mathring{\nab}_{B}\delta\gamma_{DA} +\mathring{\nab}_{A}\delta\gamma_{DB}- \mathring{\nab}_{D}\delta\gamma_{AB}\right)
\]
and
\begin{align*}
\nab_A \phi^{B_1\ldots B_k}{}_{C_1 \ldots C_l} &= \mathring{\nab}_A \phi^{B_1\ldots B_k}{}_{C_1 \ldots C_l} + \delta\slashed{\Gamma}^{B_1}_{AD}\phi^{DB_2\ldots B_k}{}_{C_1 \ldots C_l} + \ldots+ \delta\slashed{\Gamma}^{B_k}_{AD}\phi^{B_1\ldots B_{k-1}D}{}_{C_1 \ldots C_l} \\ & \qquad - \delta\slashed{\Gamma}^{D}_{AC_1}\phi^{B_1\ldots B_k}{}_{DC_2 \ldots C_l} - \ldots- \delta\slashed{\Gamma}^{D}_{AC_l}\phi^{B_1\ldots B_{k}}{}_{C_1 \ldots C_{l-1}D}.
\end{align*}
Moreover, we have
\begin{align*}
\nab_4 \phi^{B_1\ldots B_k}{}_{C_1 \ldots C_l}& = \mathring{\nab_4}  \phi^{B_1\ldots B_k}{}_{C_1 \ldots C_l} - \chih^D{}_{C_1} \phi^{B_1\ldots B_k}{}_{DC_2 \ldots C_l} -\ldots - \chih^D{}_{C_l} \phi^{B_1\ldots B_k}{}_{C_1 \ldots C_{l-1}D} \\&\qquad + \chih^{B_1}{}_{D} \phi^{DB_2\ldots B_k}{}_{C_1 \ldots C_l} + \cdots+  \chih^{B_k}{}_{D} \phi^{B_1\ldots B_{k-1}D}{}_{C_1 \ldots C_l} + \frac{k-l}{2}\left(\trch - \frac{2}{r}\right) \phi^{B_1\ldots B_k}{}_{C_1 \ldots C_l}
\end{align*}
and
\begin{align*}
2\nab_u \phi^{B_1\ldots B_k}{}_{C_1 \ldots C_l}& = 2\mathring{\nab}_u  \phi^{B_1\ldots B_k}{}_{C_1 \ldots C_l} + 2 b^A \mathring{\nab}_A \phi^{B_1\ldots B_k}{}_{C_1 \ldots C_l}  \\ &\qquad - (\chibh^D{}_{C_1}+f \chih^D{}_{C_1} - 2\mathring{\nab}_{C_1} b^D) \phi^{B_1\ldots B_k}{}_{DC_2 \ldots C_l} -\\&\qquad  \ldots - (\chibh^D{}_{C_l} + f\chih^D{}_{C_l}- 2\mathring{\nab}_{C_l} b^D) \phi^{B_1\ldots B_k}{}_{C_1 \ldots C_{l-1}D}+ \\&\qquad + (\chibh^{B_1}{}_{D} +f\chih^{B_1}{}_{D}- 2\mathring{\nab}_{B} b^{B_1}) \phi^{DB_2\ldots B_k}{}_{C_1 \ldots C_l} +\\&\qquad \cdots+  (\chibh^{B_k}{}_{D}+f\chih^{B_k}{}_{D}- 2\mathring{\nab}_{B} b^{B_k}) \phi^{B_1\ldots B_{k-1}D}{}_{C_1 \ldots C_l}\\&\qquad
+ \frac{k-l}{2}\left(\trch - \frac{2}{r} + \trchb + \frac{2}{r}+ (f-1)\trch\right) \phi^{B_1\ldots B_k}{}_{C_1 \ldots C_l}.
\end{align*}
\begin{lemma} \label{Hs equivalence}
    Suppose that the weak bootstrap assumptions hold and that $\epsilon_0$ is sufficiently small that the hypothesis of Lemma \ref{lem:gammacont} is satisfied. Then, taking $\ep_0$ smaller if necessary, for any $k, l=0, 1, \ldots$ there exists a constant $c = c(s_0, k, l)>0$ such that for $0 \leqslant s \leqslant s_{0}$, we have
    \begin{equation} \label{eq:equivalence-pointwise}
	c^{-1} |\phi|_{s, \gamma}(u,r) \leqslant |\phi|_{s}(u,r) \leqslant c |\phi|_{s, \gamma}(u,r)
\end{equation}
    for all sufficiently regular $S$-tangent $(k,l)-$tensors $\phi$. 

For $s = s_{0}+1$, we have 
\begin{equation} \label{eq:equivalence-integrated}
\begin{aligned}
	& c^{-1} ( \| (r \nab_{4} \phi, r \nab \phi, \phi) \|_{s_{0}, p_{0}, p_{\infty}, \gamma}(R) + \epsilon_{0} | \phi |_{s_{0}-2, p_{0}^{-} - \frac{3}{2}, p_{\infty}^{-} + \frac{1}{2}, \gamma}(R) ) \\ & \qquad \leqslant \| \phi \|_{s_{0}+1, p_{0}, p_{\infty}}(R) + \epsilon_{0} | \phi |_{s_{0}-2, p_{0}^{-} - \frac{3}{2}, p_{\infty}^{-} + \frac{1}{2}} (R)   \\
	& \qquad \qquad  \leqslant c ( \| (r \nab_{4} \phi, r \nab \phi, \phi) \|_{s_{0}, p_{0}, p_{\infty}, \gamma}(R) + \epsilon_{0} | \phi |_{s_{0}-2, p_{0}^{-} - \frac{3}{2}, p_{\infty}^{-} + \frac{1}{2}, \gamma}(R) ).
\end{aligned}
\end{equation}
\end{lemma}
\begin{proof}
  We begin with the first equivalence. We work by induction on $s$, and drop the explicit dependence on $(u,r)$ for notational convenience. We make free use of the Cauchy--Schwarz inequality and bounds on the metric already derived, and all estimates below have implicit constants that are permitted to depend on $s, s_0, k, l$. The case $s=0$ is already established in Lemma~\ref{lem:gammacont}. Suppose that the result holds for $s\leqslant s_0-1$. By the Sobolev product formula we have
    \begin{align*}
        |r(\mathring{\nab}_4 - \nab_4)\phi|_{s} \lesssim (|\chih|_{s_0-1, -1, 1} +|\trch - 2r^{-1}|_{s_0-1, -1, 1} ) |\phi|_{s}\lesssim \epsilon_0 |\phi|_{s}
    \end{align*}
    and similarly,
    \begin{align*}
        |r(\mathring{\nab} - \nab)\phi|_{s} \lesssim  |\delta\gamma|_{s_0, 0, 0}  |\phi|_{s} \lesssim \epsilon_0 |\phi|_{s},
    \end{align*}
    which we note is the estimate that necessitates $s \leqslant s_{0}$ for the pointwise norms.
  It immediately follows that for $\epsilon_0$ sufficiently small there exists a $c>0$ such that
    \begin{align*}
    &c^{-1}(|r{\nab}_4\phi|^2_{0, \gamma} + |r{\nab}\phi|^2_{0, \gamma} + |\phi|^2_{0, \gamma}) \\&\qquad \leqslant |r\mathring{\nab}_4\phi|_{0}^2 + |r\mathring{\nab}\phi|_{0}^2 + |\phi|_{0}^2\\&\qquad \qquad \leqslant c(|r{\nab}_4\phi|_{0, \gamma}^2 + |r{\nab}\phi|_{0, \gamma}^2 + |\phi|_{0, \gamma}^2)
    \end{align*}
    which implies the result holds for $s=1$. For $1<s\leqslant s_0-1$ we further observe that
     \begin{align*}
        |\brk{u}(\mathring{\nab}_u - \nab_u)\phi|_{s-1} &\lesssim (|\brk{u} b|_{s_0-1, 1, -1} +|\chih|_{s_{0}-2, 0, 1} +|\brk{u}(f-1)|_{s_{0}-2, 1, -1}+|\chibh|_{s_{0}-2, 0, 1}\\&\qquad +|\trch - 2r^{-1}|_{s_0-2, 0, 1} +|\trchb + 2r^{-1}|_{s_{0}-2, 0, 1} ) |\phi|_{s}\lesssim \epsilon_0 |\phi|_{s}
    \end{align*}
    so that by the induction hypothesis we deduce that for $\epsilon_0$ sufficiently small there exists $c>0$ such that 
    \begin{align*}
    &c^{-1}(|r{\nab}_4\phi|^2_{s, \gamma} + |r{\nab}\phi|^2_{s, \gamma} + |\phi|^2_{s, \gamma}+ |\brk{u} \nab_u\phi|^2_{s-1, \gamma}) \\&\qquad \leqslant |r\mathring{\nab}_4\phi|_{s}^2 + |r\mathring{\nab}\phi|_{s}^2 + |\phi|_{s}^2 + |\brk{u} \mathring{\nab}_u\phi|_{s-1}^2 \\&\qquad \qquad \leqslant c(|r{\nab}_4\phi|^2_{s, \gamma} + |r{\nab}\phi|^2_{s, \gamma} + |\phi|^2_{s, \gamma}+ |\brk{u} \nab_u\phi|^2_{s-1, \gamma}).
    \end{align*}

Next, we turn to the proof of \eqref{eq:equivalence-integrated}. Repeating the previous argument but using the product formula of Lemma~\ref{lem:products}, we obtain
\begin{align*}
	\| r(\mathring{\nab}_4 - \nab_4)\phi \|_{s_{0}, p_{0}, p_{\infty}} &\lesssim |(\chih, \trch - 2r^{-1})|_{s_0-3, -1, 1} \|\phi\|_{s_{0}, p_{0}, p_{\infty}} \\
	&\qquad + \| (\chih, \trch - 2r^{-1}) \|_{s_0, \frac{1}{2}^{+}, \frac{1}{2}^{+}} | \phi |_{s_{0}-3, p_{0}^{-} - \frac{3}{2}, p_{\infty}^{-}+\frac{1}{2}} \\
	&\lesssim \epsilon_0 ( \|\phi\|_{s_{0}, p_{0}, p_{\infty}} + | \phi |_{s_{0}-3, p_{0}^{-} - \frac{3}{2}, p_{\infty}^{-}+\frac{1}{2}} ),
\end{align*}
and similarly
\begin{align*}
	\| r(\mathring{\nab} - \nab) \phi \|_{s_{0}, p_{0}, p_{\infty}} &\lesssim |\delta \gamma|_{s_0-2, 0, 0} \|\phi\|_{s_{0}, p_{0}, p_{\infty}} 
	+ \| \delta \gamma \|_{s_{0}+1, \frac{3}{2}^{+}, (-\frac{1}{2})^{+}} | \phi |_{s_{0}-3, p_{0}^{-} - \frac{3}{2}, p_{\infty}^{-}+\frac{1}{2}} \\
	&\lesssim \epsilon_0 ( \|\phi\|_{s_{0}, p_{0}, p_{\infty}} + | \phi |_{s_{0}-3, p_{0}^{-} - \frac{3}{2}, p_{\infty}^{-}+\frac{1}{2}} ).
\end{align*}
We do not need to estimate $\brk{u} (\mathring{\nab}_{u} - \nab_{u})$, since $s_{0}$ is even. The desired conclusion now follows. \qedhere
\end{proof}

As long as it is applicable (i.e., for $s \leqslant s_{0}$), we shall freely make use of \eqref{eq:equivalence-pointwise} to recast estimates involving norms constructed from $\gamma$ into estimates involving norms constructed from $\mathring{\gamma}$ and vice versa without making explicit notice of the fact. 
\subsection{Commutator estimates}
We combine the bootstrap assumptions with the commutation formulae, Lemma~\ref{lem:CK.733}, and record the results.
\begin{theorem}\label{thm:comest}
    Let $\phi$ be a suitably regular $S$-tangent tensor field. Under the weak pointwise bootstrap assumptions \eqref{eq:BA.weak.pointwise}, we have the following pointwise estimates for $s\leqslant s_0-1$:
    \begin{align}
         \left  |[\nab_4, r\nab_A]\phi \right |_{s}(u,r) &\lesssim \frac{\epsilon }{(\jb{u}+r)^{1^+}} |\phi|_{s+1}(u,r), \label{4Acommest}\\
          \left  | [ \nab_u, r\nab_4]\phi \right |_{s-1}(u,r)  &\lesssim \frac{\epsilon }{(\jb{u}+r)^{1^+}} |\phi|_{s}(u,r), \label{u4commest} \\
       \left  |[ \nab_3, r\nab_4]\phi + \nab_4 \phi\right |_{s-1}(u,r) &\lesssim \frac{\epsilon }{(\jb{u}+r) \jb{u}^{0+}} |\phi|_{s}(u,r), \label{34comest} \\
        \left  | [ r\nab_A, r\nab_B]\phi \right |_{s-1}(u,r) &\lesssim |\phi|_{s-1}(u,r), \label{ABcomest}\\
       \left  | [ \nab_3, r\nab_A]\phi \right |_{s-2}(u,r) &\lesssim \frac{\epsilon}{\brk{u}^{1+}} |\phi|_{s-1}(u,r),  \\
       \left  | [ \nab_u, r\nab_A]\phi \right |_{s-2}(u,r) &\lesssim \frac{\epsilon}{\brk{u}^{1+}} |\phi|_{s-1}(u,r). \label{uAcomest}
    \end{align}
    Under the strong bootstrap assumptions, we have the improved estimates
    \begin{align}
         \left  |[\nab_4, r\nab_A]\phi \right |_{s}(u,r) &\lesssim \frac{\epsilon }{(\jb{u}+r)^{2}} |\phi|_{s+1}(u,r), \label{4AcommestStr}\\
          \left  | [ \nab_u, r\nab_4]\phi \right |_{s-1}(u,r)  &\lesssim \frac{\epsilon }{(\jb{u}+r)^{2}} |\phi|_{s}(u,r). \label{u4commestStr}
    \end{align}
\end{theorem}
\begin{proof}
In this proof, we freely use Lemma~\ref{Hs equivalence} above, whose proof is independent of this theorem, to pass back and forth geometric derivatives $\nab_{4}, \nab_{3}, \nab_{u}, \nab$ and background derivatives (implicitly present in the norms).

\pfstep{Step~1} Using \eqref{eq:CK.nab4.nab}, schematically we have
\begin{equation} \label{eq:comm-nab4rnab}
        \begin{aligned}
         \relax [ \nab_4, r\nab]\phi &\sim (\trch-2r^{-1}) (r \nab \phi) + \chih \cdot (r \nab \phi) +  r  \chih\cdot \eta \cdot \phi\\&\qquad  + \eta \cdot \phi +  r (\trch - 2r^{-1})  \eta \cdot \phi+ r  \beta \cdot \phi,
        \end{aligned}
\end{equation}
        where $\beta \cdot \phi$ means a sum of terms involving contractions of $\beta, \phi$, etc. Estimating each term using the bootstrap assumptions gives \eqref{4Acommest}--\eqref{u4commest} (weak) and \eqref{4AcommestStr}--\eqref{u4commestStr} (strong).
        
\pfstep{Step~2}  For the second estimate, we compute, making use of the fact that $\nab_u r = 0$
         \begin{align}
             [\nab_u, r \nab_4] \phi &= \frac{r}{2}[\nab_3 + f \nab_4 , \nab_4]\phi = \frac{r}{2}[\nab_3, \nab_4]\phi - r\omegab \nab_4 \phi \label{eq:comm-naburnab4}\\
             & \sim \eta \cdot(r\nab \phi) + r\sigma \cdot \phi \notag
         \end{align}
         and we can again estimate.
         The third estimate follows immediately from this, together with the identity
         \begin{equation} \label{eq:comm-nab3rnab4}
         [\nab_3, r\nab_4] +\nab_4 = 2[\nab_u, r\nab_4] - \frac{f - 2\omegab r -1}{r} r \nab_4.
         \end{equation}
\pfstep{Step~3}  The fourth estimate follows from the relation
        \begin{equation} \label{eq:comm-rnabrnab}
        [r\nab_A, r\nab_B]\phi  \sim r^2 K \phi,
        \end{equation}
        and $r^{2} K = 1 + r^{2} (K - r^{-2})$.
        
\pfstep{Step~4}  Next, we have schematically
	 \begin{equation} \label{eq:comm-nab3rnab}
       \begin{aligned}
          \relax  [ \nab_3, r\nab]\phi  &\:  \sim \xib r\nab_4 \phi+ \chibh \cdot r\nab \phi + (\trchb + 2r^{-1}) r \nab \phi + \frac{f-1}{r} r \nab \phi + r\chibh \cdot \eta \cdot \phi+ r\chih \cdot \xib \cdot \phi \\
            &\qquad +r(\trchb + 2r^{-1}) \eta \cdot \phi+r(\trch - 2r^{-1}) \xib \cdot \phi + \eta \cdot \phi+  \xib \cdot \phi +r \betab \cdot \phi,
        \end{aligned}
	\end{equation}
        which we can estimate for the fifth estimate. The final estimate follows by observing
        \begin{equation} \label{eq:comm-naburnab}
	[ \nab_u, r\nab] = \frac{1}{2}[ \nab_3, r\nab] +\frac{1}{2}[\nab_4, r\nab] + \frac{f-1}{2}[\nab_4, r\nab] - \frac{1}{2}\xib r\nab_4.
	\end{equation}
\end{proof}

\begin{theorem}\label{thm:comestST}
    Let $\phi$ be a suitably regular $S$-tangent tensor field. Under the weak bootstrap assumptions \eqref{eq:BA.weak.pointwise}--\eqref{eq:BA.weak.integrated}, we have  the following spacetime estimates for $s\leqslant s_0$:
    \begin{align}
       \left  \|[\nab_4, r\nab]\phi \right \|_{s, p_0, p_\infty} &\lesssim \epsilon(\|\phi\|_{s+1, p_0^{-}, p_\infty^--1}+ |\brk{u}^{-\frac{1}{2}} \phi|_{s-2, p_0^{-}-\frac{1}{2}, p_\infty^-} ),\label{4AcommestST}\\
       \left  \|[ \nab_u, r\nab_4]\phi\right \|_{s, p_0, p_\infty} &\lesssim \epsilon(\|\phi\|_{s+1, p_0^{-}, p_\infty^--1}+ |\brk{u}^{-\frac{1}{2}} \phi|_{s-2, p_0^{-}-\frac{1}{2}, p_\infty^-}), \label{u4commestST} \\
       \left  \|[ \nab_3, r\nab_4]\phi + \nab_4 \phi\right \|_{s, p_0, p_\infty} &\lesssim \epsilon(\|\brk{u}^{0^{-}} \phi\|_{s+1, p_0^{-}, p_\infty-1}+ |\brk{u}^{-\frac{1}{2}} \phi|_{s-2, p_0^{-}-\frac{1}{2}, p_\infty^-}), \label{34commestST} \\
       \left  \|[ r\nab_A, r\nab_B]\phi\right \|_{s, p_0, p_\infty} &\lesssim \|\phi\|_{s, p_0, p_\infty} 
 	+ \epsilon \left( \|\brk{u}^{0^{-}}\phi\|_{s, p_0^{-} - 1, p_\infty} + |\brk{u}^{-\frac{1}{2}} \phi|_{s-3, p_0^{-}-\frac{3}{2}, p_\infty^-+1} \right), \label{ABcomestST} \\
       \left  \|[ \nab_3, r\nab]\phi\right \|_{s-1, p_0, p_\infty} &\lesssim \epsilon(\| \brk{u}^{(-1)^{+}} \phi\|_{s, p_0^{-}, p_\infty}+ |\brk{u}^{-\frac{1}{2}-\frac{3}{2}\delta} \phi|_{s-3, p_0^{-}-\frac{1}{2}, p_\infty+\frac{1}{2}\delta}), \\
       \left  \|[ \nab_u, r\nab]\phi\right \|_{s-1, p_0, p_\infty} &\lesssim \epsilon(\|\brk{u}^{(-1)^{+}}\phi\|_{s, p_0^{-}, p_\infty}+ |\brk{u}^{-\frac{1}{2}-\frac{3}{2}\delta} \phi|_{s-3, p_0^{-}-\frac{1}{2}, p_\infty+\frac{1}{2}\delta}).\label{uAcommestST}
    \end{align}
    If the strong bootstrap assumptions hold we have the improved estimates: 
     \begin{align}
       \left  \|[\nab_4, r\nab]\phi \right \|_{s, p_0, p_\infty} &\lesssim \epsilon(\|\phi\|_{s+1, p_0^{-}, p_\infty-2}+ |\phi|_{s-2, p_0^{-}-\frac{1}{2}, p_\infty^{+}-\frac{3}{2}}), \label{4AcommestSTstr}\\
       \left  \|[ \nab_u, r\nab_4]\phi\right \|_{s, p_0, p_\infty} &\lesssim \epsilon(\|\phi\|_{s+1, p_0^{-}, p_\infty-2}+ |\phi|_{s-2, p_0^{-}-\frac{1}{2}, p_\infty^{+}-\frac{3}{2}}), \label{u4commestSTstr} \\
       \left  \|[ \nab_3, r\nab_4]\phi + \nab_4 \phi\right \|_{s, p_0, p_\infty} &\lesssim \epsilon(\|\brk{u}^{0^{-}} \phi\|_{s+1, p_0^{-}, p_\infty-1}+ |\phi|_{s-2, p_0^{-}-\frac{1}{2}, p_\infty^{-}-\frac{1}{2}}), \label{34commestSTstr} \\
       \left  \|[ r \nab_{A}, r \nab_{B}]\phi \right\|_{s, p_0, p_\infty} &\lesssim \|\phi\|_{s, p_{0}, p_{\infty}} + \epsilon(\|\phi\|_{s, p_0, p_\infty}+ |\phi|_{s-3, p_0^{-}-\frac{3}{2}, p_\infty^{-}+\frac{1}{2}}). \label{ABcomestSTstr}
    \end{align}
\end{theorem}
\begin{proof}
    Repeat the argument of the previous Theorem, estimating instead using the product formula of Lemma \ref{lem:products} together with the bootstrap assumptions.
\end{proof}

\begin{remark} Since the commutation relations \eqref{eq:CK.nab4.nab}--\eqref{eq:CK.nabB.nabC} only have derivatives tangential to the cones of constant $u$ appearing on the right hand side, any powers of $\brk{u}$ can be freely commuted with all the commutators. Moreover, in any estimate above which loses a derivative we can replace the norm on the right with the norm of all tangential derivatives at one lower order, e.g.\ in \eqref{4Acommest}
\[
\left  |[\nab_4, r\nab_A]\phi \right |_{s}(u,r) \lesssim \frac{\epsilon }{(\jb{u}+r)^{1^+}} \left(|r\nab_4\phi|_{s}(u,r)+|r\nab\phi|_{s}(u,r)+|\phi|_{s}(u,r) \right)
\]
or in \eqref{4AcommestSTstr}
\begin{align*}
\left  \|[\nab_4, r\nab]\phi \right \|_{s, p_0, p_\infty} &\lesssim \epsilon(\|r\nab_4\phi\|_{s, p_0^{-}, p_\infty-2}+\|r\nab\phi\|_{s, p_0^{-}, p_\infty-2}+\|\phi\|_{s, p_0^{-}, p_\infty-2}\\
&\qquad \qquad + |r\nab_4\phi|_{s-2, p_0^{-}-\frac{1}{2}, p_\infty^--\frac{3}{2}}+ |r\nab\phi|_{s-2, p_0^{-}-\frac{1}{2}, p_\infty^--\frac{3}{2}}+ |\phi|_{s-2, p_0^{-}-\frac{1}{2}, p_\infty^--\frac{3}{2}})
\end{align*}
and so on. 
\end{remark}

\subsection{Elliptic estimates}
For convenience, we extend the definition of $\nab \hot $ to act on a $k+1$ tensor $\phi_{C_1\ldots C_k A}$ as
\[
(\nab \hot \phi)_{C_1\ldots C_k AB} = \nab_A \phi_{C_1\ldots C_k B} + \nab_B \phi_{C_1\ldots C_k A} - \gamma_{AB} \nab^D \phi_{C_1\ldots C_k D}.
\]
\begin{lemma}\label{lem:elliptic}
    Suppose the weak bootstrap assumptions hold for a sufficiently small $\epsilon_0$. If $s \leqslant s_0-1$ and $\phi$ is a sufficiently regular $k+1$ tensor then
\begin{equation} \label{eq:elliptic-pointwise}
    |r\mathring{\nab}\phi|_{s}(u,r) \lesssim_{k} |r\nab \hot \phi|_{s}(u,r) + |r\nab_4\phi|_s(u,r) + |\phi|_s(u,r).
\end{equation}
Moreover, for $s = s_{0}$, we have the integrated bound
\begin{equation} \label{eq:elliptic-integrated}
    \| \phi \|_{s_0+1, p_{0}, p_{\infty}} \lesssim_{k} \| r\nab \hot \phi \|_{s_0, p_{0}, p_{\infty}} + \|r\nab_4\phi \|_{s_0, p_{0}, p_{\infty}} + \| \phi \|_{s_0, p_{0}, p_{\infty}} + \epsilon_{0} | \brk{u}^{-\frac{1}{2}} \phi |_{s_{0}-3, p_{0}^{-}-\frac{1}{2}, p_{\infty}^{-}+1}.
 \end{equation}
 If the strong bootstrap assumptions hold, then after possibly reducing $\epsilon_0$, we have
 \begin{equation} \label{eq:elliptic-integrated-strong}
   \| \phi \|_{s_0+1, p_{0}, p_{\infty}} \lesssim_{k} \| r\nab \hot \phi \|_{s_0, p_{0}, p_{\infty}} + \|r\nab_4\phi \|_{s_0, p_{0}, p_{\infty}} + \| \phi \|_{s_0, p_{0}, p_{\infty}} + \epsilon_{0} | \phi |_{s_{0}-3, p_{0}^{-}-\frac{1}{2}, p_{\infty}^{+}+\frac{1}{2}}.
 \end{equation}
\end{lemma}
\begin{proof}
    We start with \eqref{eq:elliptic-pointwise}, which is pointwise in $(u, r)$. We work by induction on $s$. For the base case, we require the identity
    \[
    2|r\nab \phi|^2 = |r\nab \hot \phi|^2  + 2 \phi^A\cdot [r\nab_B, r\nab_A] \phi^B - 2 \nab^A(r\phi^B \cdot r\nab_B \phi_A - r\phi_A \cdot r\nab^B\phi_B)
    \]
    where $\cdot$ indicates contraction on the indices $C_1, \ldots C_k$. This identity can be verified by expanding out the right hand side. Integrating over $S_{u,r}$ and making use of \eqref{ABcomest} we have
    \[
    |r\nab \phi|^2_{0, \gamma}(u,r) \lesssim |r\nab \hot \phi|^2_{0, \gamma}(u,r)+ |\phi|^2_{0, \gamma}(u,r)
    \]
    which implies our result with $s=0$ by Lemma \ref{Hs equivalence}. 
    
    For the induction step we note
    \begin{align*}
    &r\nab_{C_1} (r \nab \hot \phi)_{C_2\ldots C_kAB} - (r \nab \hot r\nab \phi)_{C_1\ldots C_kAB} \\&= [r\nab_{C_1}, r\nab_A] \phi_{C_2\ldots C_k B} + [r\nab_{C_1}, r\nab_B] \phi_{C_2\ldots C_k A} - \gamma_{AB} [r\nab_{C_1}, r\nab^D] \phi_{C_2\ldots C_k D} 
    \end{align*}
    so that for $s\leqslant s_0-2$
    \[
    |r\nab (r\nab \hot \phi) - r\nab \hot (r \nab \phi)|_s \lesssim |\phi|_s.
    \]
    by \eqref{ABcomest}. Similarly we have for $s\leqslant s_0-2$
    \begin{align*}
    |r\nab_4 (r\nab \hot \phi) - r\nab \hot (r \nab_4 \phi)|_s &\lesssim |\phi|_{s+1} \\
    |\brk{u} \nab_u (r\nab \hot \phi) - r\nab \hot ( \brk{u}  \nab_u \phi)|_{s-1} &\lesssim |\phi|_{s}
    \end{align*}
    by \eqref{4Acommest}, \eqref{uAcomest}.
    
    Now suppose that the result holds for some $s\leqslant s_0-2$. We have
    \begin{align*}
        |r\nab(r\nab \phi)|_{s} &\lesssim |r\nab \hot (r\nab \phi)|_{s} + |r\nab_4 r\nab \phi|_s + |r \nab \phi|_s \\
        & \lesssim |r\nab (r\nab \hot  \phi)|_{s}+  | r\nab r\nab_4 \phi|_s + |r \nab \phi|_s \\
        &\qquad + |r\nab (r\nab \hot \phi) - r\nab \hot (r \nab \phi)|_s + |[r\nab_4, r\nab]\phi|_s \\
        & \lesssim |r\nab \hot  \phi|_{s+1} + |r\nab_4 \phi|_{s+1} + |\phi|_{s+1}
    \end{align*}
    a similar argument gives
    \begin{align*}
        |r\nab_4(r\nab\phi)|_{s} &\lesssim  |r\nab \hot  \phi|_{s+1} + |r\nab_4 \phi|_{s+1} + |\phi|_{s+1}\\
        |\brk{u} \nab_u(r\nab\phi)|_{s-1} &\lesssim  |r\nab \hot  \phi|_{s+1} + |r\nab_4 \phi|_{s+1} + |\phi|_{s+1}
    \end{align*}
    and this closes the induction. 
    
 Next, we turn to the proof of \eqref{eq:elliptic-integrated} and \eqref{eq:elliptic-integrated-strong}. We first assume the weak bootstrap assumptions and consider \eqref{eq:elliptic-integrated}. By integrating \eqref{eq:elliptic-pointwise}, we immediately have
\begin{equation*}
\| r \nab \phi \|_{s_0-1, p_{0}, p_{\infty}} \lesssim \| r\nab \hot \phi \|_{s_0-1, p_{0}, p_{\infty}} + \|r\nab_4\phi \|_{s_0-1, p_{0}, p_{\infty}} + \| \phi \|_{s_0-1, p_{0}, p_{\infty}}.
\end{equation*}
Using \eqref{4AcommestST}, \eqref{ABcomestST}, and \eqref{uAcommestST} instead of \eqref{4Acommest}, \eqref{ABcomest}, and \eqref{uAcomest}, respectively, and repeating the preceding argument, we obtain
\begin{align*}
\| r\nab (r\nab \hot \phi) - r\nab \hot (r \nab \phi) \|_{s_{0}-1, p_{0}, p_{\infty}} &\lesssim \| \phi \|_{s_{0}-1, p_{0}, p_{\infty}} + \epsilon_{0} | \brk{u}^{-\frac{1}{2}} \phi |_{s_{0}-4, p_{0}^{-}-\frac{3}{2}, p_{\infty}^{-}+1}, \\
\| r\nab_4 (r\nab \hot \phi) - r\nab \hot (r \nab_4 \phi) \|_{s_{0}-1, p_{0}, p_{\infty}} &\lesssim \| \phi \|_{s_{0}, p_{0}, p_{\infty}} + \epsilon_{0} | \brk{u}^{-\frac{1}{2}} \phi |_{s_{0}-3, p_{0}^{-}-\frac{3}{2}, p_{\infty}^{-}+1}, \\
\| \brk{u} \nab_u (r\nab \hot \phi) - r\nab \hot (\brk{u} \nab_u \phi) \|_{s_{0}-2, p_{0}, p_{\infty}} &\lesssim \| \phi \|_{s_{0}-1, p_{0}, p_{\infty}} + \epsilon_{0} | \brk{u}^{\frac{1}{2}-\frac{3}{2}\delta} \phi |_{s_{0}-3, p_{0}^{-}-\frac{1}{2}, p_{\infty}+\frac{1}{2}\delta}.
\end{align*}
Note that all the pointwise norms in these estimates are bounded by $| \brk{u}^{-\frac{1}{2}} \phi |_{s_{0}-3, p_{0}^{-}-\frac{1}{2}, p_{\infty}^{-}+1}$. Proceeding as before, we have
\begin{align*}
	& \| r \nab (r \nab \phi) \|_{s_{0}-1, p_{0}, p_{\infty}} + \| r \nab_{4} (r \nab \phi) \|_{s_{0}-1, p_{0}, p_{\infty}} + \| \brk{u} \nab_{u} (r \nab \phi) \|_{s_{0}-2, p_{0}, p_{\infty}}\\
	& \qquad \aleq \| r \nab \hot \phi \|_{s_{0}, p_{0}, p_{\infty}} + \| r \nab_{4} \phi \|_{s_{0}, p_{0}, p_{\infty}} + \| r \nab \phi \|_{s_{0}, p_{0}, p_{\infty}} + \epsilon_{0} | \brk{u}^{-\frac{1}{2}} \phi |_{s_{0}-3, p_{0}^{-}-\frac{1}{2}, p_{\infty}^{-}+1},
\end{align*}    
where, by \eqref{eq:equivalence-pointwise}, the left hand side controls $\| r \nab \phi \|_{s_{0}, p_{0}, p_{\infty}}$. Since we assumed $s_{0}$ to be even, adding $\| r \nab_{4} \phi \|_{s_{0}, p_{0}, p_{\infty}} + \| r \nab \phi \|_{s_{0}, p_{0}, p_{\infty}}$ to both sides makes the left hand side proportional to $\| \phi \|_{s_{0}+1, p_{0}, p_{\infty}, \gamma}$, while the right hand side remains unchanged (up to modifying the constant). Applying \eqref{eq:equivalence-integrated} and observing that the pointwise norm in this equivalence is also bounded by $| \brk{u}^{-\frac{1}{2}} \phi |_{s_{0}-3, p_{0}^{-}-\frac{1}{2}, p_{\infty}^{-}+1}$, we obtain \eqref{eq:elliptic-integrated}. 

Finally, \eqref{eq:elliptic-integrated-strong} follows by repeating the above argument, but using \eqref{4AcommestSTstr}, and \eqref{ABcomestSTstr} in place of \eqref{4AcommestST}, and \eqref{ABcomestST}, respectively. In this case, all pointwise norms that arise in the proof are bounded by $|\phi |_{s_{0}-3, p_{0}^{-}-\frac{1}{2}, p_{\infty}^{+}+\frac{1}{2}}$.\qedhere
\end{proof}

\section{Transport equations} \label{sec:transport}
Under the bootstrap assumptions, we now establish estimates for the transport equation along $e_{4}$.
The main result is as follows:
\begin{theorem}\label{thm:transest}
    For any fixed choice of indices for the norms below, there exists $\ep_0$ sufficiently small such that the following holds whenever the weak bootstrap assumptions \eqref{eq:BA.weak.pointwise}--\eqref{eq:BA.weak.integrated} hold for $\ep \leqslant \epsilon_0$. %
    Suppose $\phi$, a rank $k$ $S-$tangent field satisfies the transport equation
    \[
    {\nab}_4 \phi + \lambda_0 \trch \phi  = F.
    \]
    Then there exist implicit constants depending only on the indices for the norms such that the following estimates hold:
    \begin{enumerate}[i)]
    \item If $s \leqslant s_{0}$, $2\lambda_0 + q_0+1>0$, $\tilde{q}_{\infty} \leqslant q_{\infty}$, and for $u\geqslant 1$ we have
    \[
    \lim_{r\to 0} r^{2\lambda_0-1} |\phi|_s(u,r) =0,
    \]
    then we can estimate
    \begin{align}
    \abs{\jb{u}^{\tilde{q}_\infty-p_\infty-1} \jb{u_{+}}^{p_{u_{+}}} \phi}_{s, q_0+1, p_\infty} &\lesssim \abs{\phi}_{s, q_0+1, \tilde{q}_\infty-1}(\tilde{\Sigma}_1) +  |\jb{u_{+}}^{q_{u_{+}}} F|_{s, q_0, q_\infty}, \label{eq:main.sphere.transest}
    \end{align}
    where
    \[
    p_\infty = \left \{ \begin{array}{cc}
        2\lambda_0 & \tilde{q}_\infty >2\lambda_0 + 1, \\
         \tilde{q}_\infty-1 & \tilde{q}_\infty < 2\lambda_0 + 1.
    \end{array}\right. , \qquad
    p_{u_{+}} = q_{u_{+}} + q_{\infty} - \tilde{q}_{\infty}.
    \]
    Furthermore, we have
    \begin{align}
       \int_{u_T}^T \abs{\jb{u}^{\tilde{q}_\infty-p_\infty-1} \jb{u_{+}}^{q_{\infty} - \tilde{q}_{\infty}} \phi}_{s, q_0+1, p_\infty}^2(H_u) \ud u &\lesssim \abs{\phi}^2_{s, q_0+1, \tilde{q}_\infty-1}(\tilde{\Sigma}_1) + \int_{u_T}^T |F|_{s, q_0, q_\infty}^2(H_u) \, \ud u, \label{eq:22.sphere.transest} \\
    \int_{u_T}^T \abs{\jb{u}^{\tilde{q}_\infty-p_\infty-1}\jb{u_{+}}^{q_{\infty}-\tilde{q}_{\infty}} \phi}_{s, q_0+1, p_\infty}^2(H_u) \ud u &\lesssim \|{\phi}\|^2_{s, q_0+1, \tilde{q}_\infty-1}(\tilde{\Sigma}_1) +  \|F\|_{s, q_0, q_\infty}^2.\label{eq:22.sphere.transest.2} 
    \end{align}
    \item If $ 2\lambda_0  + q_0 +\frac{1}{2}\geqslant 0$, and
    \[
    \lim_{r\to 0} \int_1^T r^{-2(q_0+2)} |\phi|^2_s(u,r) \ud u =0,
    \]
    then for $s \leqslant s_{0}$, we have
    \begin{equation}\label{eq:main.int.transest}
        \|\brk{u}^{q_{\infty} - \tilde{p}_{\infty} - 1} \phi\|_{s, q_0+1, \tilde{p}_{\infty}} \lesssim \|\phi\|_{s, q_0+\frac{1}{2}, q_\i -\frac{1}{2}}(\tilde{\Sigma}_1) + \| F\|_{s, q_0, q_\infty},
    \end{equation}
    where
    \[
    \tilde{p}_\infty = \left \{ \begin{array}{cc}
        2\lambda_0 - \frac{1}{2} \delta & q_\infty >2\lambda_0 + 1, \\
         q_{\infty}-1 & q_{\infty} < 2\lambda_0 + 1.
    \end{array}\right.
    \]
    
    If $s = s_{0}+1$, then we have
    \begin{equation}\label{eq:main.int.transest.top}
    \begin{aligned}
        \|\brk{u}^{q_{\infty} - \tilde{p}_{\infty} - 1} \phi\|_{s_{0}+1, q_0+1, \tilde{p}_{\infty}} &\lesssim \|\phi\|_{s_{0}+1, q_0+\frac{1}{2}, q_\i -\frac{1}{2}}(\tilde{\Sigma}_1) + \| F\|_{s, q_0, q_\infty} \\
        &\qquad + \epsilon_{0}|\brk{u}^{q_{\infty} - \tilde{p}_{\infty} - 1} \phi|_{s_{0}-2, q_0+\frac{1}{2}, \tilde{p}_{\infty}+\frac{1}{2}}.
    \end{aligned}
    \end{equation}
    \end{enumerate}
   
\end{theorem}

We note that we may add extra $\jb{u_{+}}$-weights for $F$ in \eqref{eq:22.sphere.transest} and \eqref{eq:22.sphere.transest.2}, but it is not needed below.

The first step of the proof is to consider transport equations on background geometry. It is straightforward to verify that
\[
\frac{\partial}{\partial r} \int_{S_{u,r}} h \ud A_{\mathring{\gamma}} =  \int_{S_{u,r}} \left(\frac{\partial h}{\partial r} + \frac{2}{r} h \right)\ud A_{\mathring{\gamma}}, \qquad \frac{\partial }{\partial r} |\phi|_{\mathring{\gamma}}^2 = 2 \langle\phi, \mathring{\nab}_4\phi\rangle_{\mathring{\gamma}}.
\]
\begin{lemma}\label{Lem:bgtransport}
    Suppose $\phi$, a rank $k$ $S$-tangent field, satisfies the transport equation
    \[
    \mathring{\nab}_4 \phi + \frac{2 \lambda_0}{r}\phi  = F.
    \]
    Suppose $\lambda, \lambda'$ satisfy $\max\{\lambda, \lambda+\lambda'\}\leqslant 2 \lambda_0-1$. Then provided $S_{u, r_0}$ and $S_{u,r_1}$ are both in $S_T$ with $r_0<r_1$ we have
    \begin{equation}\label{eq:bgtransport.1}
        \sup_{r_0\leqslant r \leqslant r_1} r^\lambda (\jb{u}+r)^{\lambda'}|\phi|_s(u,r) \lesssim r_0^\lambda (\jb{u}+r_0)^{\lambda'}|\phi|_s(u,r_0) + \int_{r_0}^{r_1} {r}^\lambda (\jb{u}+r)^{\lambda'} |F|_s(u,r) \ud r    
    \end{equation}
    and for any $\nu \in [0,1)$ the Hardy estimate
    \begin{equation}\label{eq:bgtransport.2}
        \begin{split}
            &(1-\nu) \int_{r_0}^{r_1}  {r}^{2\lambda-2} (\jb{u}+r)^{2\lambda'+ \nu} |\phi|^2_s(u,r) \ud r \\
            &\qquad \lesssim \: r_0^{2\lambda-1} (\jb{u}+r_0)^{2\lambda'+\nu}|\phi|^2_s(u,r_0) + (1-\nu)^{-1}\int_{r_0}^{r_1} {r}^{2\lambda} (\jb{u}+r)^{2\lambda'+\nu} |F|^2_s(u,r) \ud r
        \end{split}
    \end{equation}
    for any $s$, where the implicit constants depend only on $s$.
\end{lemma}
\begin{proof}
    First we consider the case $s=0$. Let $w = r^\lambda (\jb{u}+r)^{\lambda'}$, and observe that
    \[
    \partial_r w^2 = \frac{2}{r} w^2 \left( \lambda + \lambda'\frac{r}{\jb{u}+r}\right) \leqslant \max\{\lambda, \lambda+\lambda'\}\frac{2}{r}w^2.
    \]
    Hence
    \begin{align*}
        \frac{\partial}{\partial r}(w^2 |\phi|^2_{\mathring{\gamma}^2}) +\frac{2}{r}w^2 |\phi|^2_{\mathring{\gamma}^2} &\leqslant (1+\max\{\lambda, \lambda+\lambda'\})\frac{2}{r}w^2 |\phi|^2_{\mathring{\gamma}^2} + 2 w^2 \langle\phi, \mathring{\nab}_4\phi\rangle_{\mathring{\gamma}} \\
        &= (\max\{\lambda, \lambda+\lambda'\} - \lambda_1)\frac{2}{r}w^2 |\phi|^2_{\mathring{\gamma}^2} + 2 \langle w\phi, wF \rangle_{\mathring{\gamma}}.
    \end{align*}
    Integrating over the sphere and using the formula above with the Cauchy--Schwarz inequality, we deduce that
    \begin{equation*}
        \begin{split}
            \frac{\partial}{\partial r} |w\phi|^2_0 =&\:  \int_{S_{u,r}} \Big( \frac{\partial}{\partial r}(w^2 |\phi|^2_{\mathring{\gamma}^2}) +\frac{2}{r}w^2 |\phi|^2_{\mathring{\gamma}^2} \Big) \ud A_{\mathring{\gamma}}  \\
            \leqslant &\: \frac{2}{r} (\max\{\lambda, \lambda+\lambda'\} - \lambda_1) |w\phi|^2_0 + 2|w\phi|_0 |wF|_0\leqslant 2|w\phi|_0 |wF|_0
        \end{split}
    \end{equation*}
    so that for any $\varepsilon>0$
    \[
    \frac{\partial}{\partial r} \left(|w\phi|^2_0 + \varepsilon\right)^\frac{1}{2} \leqslant |wF|_0.
    \]
    Integrating between $r_0$ and $r$ and sending $\varepsilon \to 0$ gives the first result after taking a supremum over $r\leqslant r_1$.

    For the Hardy estimate, first note that if $\nu \in [0,1)$, then
    $\f{\rd}{\rd r}\f{(\brk{u}+r)^\nu}{r} = \f{(\brk{u}+r)^{\nu-1}}{r^2}(\nu r - (\brk{u}+r))  \f{(1-\nu)(\brk{u}+r)^{\nu}}{r^2}$, which implies
    $$\f{(1-\nu)(\brk{u}+r)^\nu)}{r^2} \leqslant -\f{\rd}{\rd r} \f{(\brk{u}+r)^\nu}{r}.$$
    Then 
    \begin{align*}
        &(1-\nu) \int_{r_0}^{r_1} \f{(\brk{u}+r)^{\nu}}{r^2} |w\phi|_0^2(u,r) \ud r \\
        &\qquad \leqslant   -\int_{r_0}^{r_1} \frac{\partial}{\partial r}\left( \f{(\brk{u}+r)^\nu}{r} \right)|w\phi|_0^2(u,r)  \ud r \\
        &\qquad =  \left[\f{(\brk{u}+r)^\nu}{r}|w\phi|_0^2(u,r) \right]^{r_0}_{r_1} + \int_{r_0}^{r_1} \f{(\brk{u}+r)^\nu}{r} \frac{\partial}{\partial r}|w\phi|_0^2(u,r) \ud r \\
        &\qquad \leqslant \f{(\brk{u}+r_0)^\nu}{r_0}|w\phi|_0^2(u,r_0)  + \frac{1-\nu}{2}\int_{r_0}^{r_1} \frac{(\brk{u}+r_0)^\nu}{{r}^2} |w\phi|_0^2(u,r) \ud r+ \frac{2}{1-\nu} \int_{r_0}^{r_1} (\brk{u}+r_0)^\nu |wF|_0^2(u,r) \ud r .
    \end{align*}
    
    Rearranging we have the required result. 

    The result for higher $s$ is an immediate consequence of the fact that by commuting the equation we have
    \begin{align*}
    &\mathring{\nab}_4(r\mathring{\nab} \phi) + \frac{2 \lambda_0}{r} (r\mathring{\nab} \phi) =r \mathring{\nab} F, \qquad \mathring{\nab}_4(\jb{u}\mathring{\nab}_u \phi) + \frac{2 \lambda_0}{r} (\jb{u}\mathring{\nab}_u \phi) = \jb{u}\mathring{\nab}_u F, \\
   & \qquad \textrm{and}\quad  \mathring{\nab}_4(r\mathring{\nab}_4 \phi) + \frac{2 \lambda_0}{r} (r\mathring{\nab}_4 \phi) =r \mathring{\nab}_4 F + F.
    \end{align*}
    A simple induction completes the proof.
\end{proof}

\begin{corollary} \label{cor:bgtrans}
Estimates \eqref{eq:main.sphere.transest}, \eqref{eq:22.sphere.transest}, \eqref{eq:22.sphere.transest.2}, and \eqref{eq:main.int.transest} in Theorem~\ref{thm:transest} hold for the background transport equation $\mathring{\nab}_{4} \phi + \frac{2 \lambda_{0}}{r} \phi = F$ (without any restrictions on $s$).   
\end{corollary}
\begin{proof}
    We apply Lemma~\ref{lem:basicest} for suitably chosen $\lambda, \lambda'$.
 
    \pfstep{Proof of i)} Set $\lambda = 2\lambda_0-1$ and $\lambda' = 0$, which is admissible by our assumptions. In the case of \eqref{eq:main.sphere.transest}, we may divide both sides of the transport equation by $\jb{u_{+}}^{q_{u_{+}}}$ to reduce to the case $q_{u_{+}} = 0$ and $p_{u_{+}} = q_{\infty} - \tilde{q}_{\infty}$. If $u\geqslant 1$ then we apply the estimate with $r_0=0$ and can discard a term by the assumption near the axis. We then find
    \begin{align*}
     r^{2\lambda_0-1} |\phi|_s(u,r) &\lesssim \int_0^r {r'}^{2\lambda_0-1} |F|_s(u, r') \ud r' \leqslant |F|_{s, q_0, q_\infty}\int_0^r {r'}^{2\lambda_0+q_0} {(\jb{u}+{r'})^{-q_0-q_\infty}} \ud r' \\
     &\lesssim |F|_{s, q_0, q_\infty} I(r),
    \end{align*}
    where
    \[
    I(r) = \left \{ \begin{array}{cc}
        \jb{u}^{1+2\lambda_0-q_\infty}(r (\jb{u}+r)^{-1})^{2\lambda_0+q_0+1} & q_\infty>2\lambda_0 + 1, \\
         r^{2\lambda_0+q_0+1} (\jb{u}+{r})^{-q_0-q_\infty} & q_\infty<2\lambda_0 + 1.
    \end{array}\right.
    \]
    Rearranging and taking the supremum over $r, u$ gives the result for $u\geqslant 1$. For $u<1$, we replace $q_{\infty}$ by $\tilde{q}_{\infty}$ and apply the estimate with $r_0 = 1-u$ to obtain
    \begin{align*}
     r^{2\lambda_0-1} |\phi|_s(u,r) &\lesssim r_0^{2\lambda_0-1} |\phi|_s(u,r_0) + \int_{r_0}^r {r'}^{2\lambda_0-1} |F|_s(u, r') \ud r'\\ 
     & \leqslant r_0^{2\lambda_0-1}  |\phi|_s(u,r_0)+ |F|_{s, q_0, q_\infty}I(r),
    \end{align*}
    where we used $\tilde{q}_{\infty} \leqslant q_{\infty}$. Hence,
    \[
    \frac{r^{2 \lambda_0-1}}{I(r)} |\phi|_s(u,r) \lesssim \frac{r_0^{2 \lambda_0-1}}{I(r_0)} |\phi|_s(u,r_0) + |F|_{s, q_0, q_\infty},
    \]
    where we have used that $I(r)$ is positive and monotone increasing so that $I(r_0)\leqslant I(r)$. Again taking suprema we arrive at \eqref{eq:main.sphere.transest}, recalling that $\jb{u}\sim \jb{r_0}$. Alternatively, taking the supremum in $r$ and then $L^2$ in $u$, we obtain \eqref{eq:22.sphere.transest}.
    
    For the other variants of the estimate, we instead use the bound
    \begin{align*}
      \int_{r_0}^r {r'}^{2\lambda_0-1} |F|_s(u, r') \ud r' & \leqslant \left(\int_{r_0}^r  |F|^2_s(u, r') \frac{(r'+\jb{u})^{2q_0 + 2q_\infty}}{r'{}^{3 + 2q_0}} \ud r' \right)^{\frac{1}{2}}\left(\int_{r_0}^r   \frac{r'{}^{4\lambda_0 + 2q_0+1}} {(r'+\jb{u})^{2q_0 + 2q_\infty}}\ud r' \right)^{\frac{1}{2}} \\
      &\lesssim \|F\|_{s, q_0, q_\infty}(H_u) \times  I(r).
    \end{align*}
    to deduce
    \[
    \frac{r^{2 \lambda_0-1}}{I(r)} |\phi|_s(u,r) \lesssim \frac{r_0^{2 \lambda_0-1}}{I(r_0)} |\phi|_s(u,r_0) + \|F\|_{s, q_0, q_\infty}(H_u).
    \]
    Squaring, taking the supremum over $r$ and integrating over $u$ gives \eqref{eq:22.sphere.transest.2}.

   \pfstep{Proof of \eqref{eq:main.int.transest}} In the case $q_\infty \in [2\lambda_0+\f 12, 2\lambda_0+1)$, we set $\lambda = -\frac{3}{2}-q_0$, $\lambda' = 2\lambda_0 + \frac{1}{2} + q_0$, and $\nu = 2 q_\i -4\lambda_0 -1 \in [0,1)$. In the case $q_\i < 2\lambda_0+\f 12$, we set  $\lambda = -\frac{3}{2}-q_0$, $\lambda' = q_\i + q_0$, and $\nu = 0$.
    Then, in these two cases, the Hardy estimate \eqref{eq:bgtransport.2}  gives
    \begin{align*}
        &\: \int_{r_0}^{r_1} |\phi|_s^2(u,r) \frac{(\jb{u}+r)^{2q_\i + 2q_0}}{r^{3 + 2 (q_0+1)}} \ud r \\
        \lesssim &\: |\phi|_s^2(u,r_0) \frac{(\jb{u}+r_0)^{2q_\i + 2q_0}}{r_0^{3 + 2 (q_0+1/2)}} + \int_{r_0}^{r_1} |F|_s^2(u,r) \frac{(\jb{u}+r)^{2q_\i + 2q_0}}{r^{3 + 2 q_0}} \ud r \\
        \lesssim &\: |\phi|_s^2(u,r_0) \frac{(\jb{u}+r_0)^{2(q_\i-1/2) + 2(q_0+1/2)}}{r_0^{3 + 2 (q_0+1/2)}} + \int_{r_0}^{r_1} |F|_s^2(u,r) \frac{(\jb{u}+r)^{2q_\infty + 2q_0 }}{r^{3 + 2 q_0}} \ud r.
    \end{align*}
    Setting $r_0 = \max\{1-u, 0\}$ and choosing $r_1(u)$ such that $\tau(u, r_1(u)) = T$, we can then integrate over $u$ and discard the contribution at $r=0$ to get the result for $q_{\infty} < 2 \lambda_0 + 1$. When $q_\infty > 2 \lambda_0 + 1$, then we simply use $\brk{u} \leqslant r + \brk{u}$ to estimate 
    \begin{equation*}
        \|\brk{u}^{q_{\infty}-\tilde{p}_{\infty}-1} \phi\|_{s, q_0+\frac{1}{2}, \tilde{p}_{\infty}+\frac{1}{2}}(\tilde{\Sigma}_1)
        +\|\brk{u}^{q_{\infty}-\tilde{p}_{\infty}-1} F\|_{s, q_0, \tilde{p}_{\infty}+1} \leqslant 
        \|\phi\|_{s, q_0+\frac{1}{2}, q_\i -\frac{1}{2}}(\tilde{\Sigma}_1) 
        + \| F\|_{s, q_0, q_\infty}, 
    \end{equation*}
    where $\tilde{p}_{\infty} = 2 \lambda_{0} + \frac{1}{2} \delta$, and use the case $q_{\infty} = 2 \lambda_0 + 1 - \frac{1}{2} \delta$. \qedhere
\end{proof}

We are now ready to complete the proof of Theorem~\ref{thm:transest}.

\begin{proof}[Proof of Theorem~\ref{thm:transest}]
As before, in the case of \eqref{eq:main.sphere.transest}, we divide both sides of the transport equation by $\jb{u_{+}}^{q_{u_{+}}}$ to reduce to the case $q_{u_{+}} = 0$ and $p_{u_{+}} = q_{\infty} - \tilde{q}_{\infty}$.
We work under the assumption that $\tilde{q}_{\infty} = q_{\infty} \leqslant 2 \lambda_{0} +1 +\frac{1}{2} \delta$. The general case is reduced to this case by trading powers of $r + \brk{u}$ with appropriate powers of $\brk{u}$ and $\jb{u_{+}}$ in the norms on the right hand sides of the estimates. More precisely, we (i)~estimate $\jb{u_{+}}^{q_{\infty} - \tilde{q}_{\infty}} \leqslant (r + \brk{u})^{q_{\infty} - \tilde{q}_{\infty}}$ in the norm for $F$ to reduce to the case $q_{\infty} = \tilde{q}_{\infty}$, and then (ii)~if $q_{\infty} > 2 \lambda_{0} + 1 + \frac{1}{2} \delta$, estimate further $\brk{u}^{q_{\infty} - 2 \lambda_{0} +1 +\frac{1}{2} \delta} \leqslant (r + \brk{u})^{q_{\infty} - 2 \lambda_{0} +1 +\frac{1}{2} \delta}$ in the norms for $\phi$ on $\tilde{\Sigma_{1}}$ and $F$ to reduce to the case the other inequality also holds. The benefit of restricting to this range is that 
\begin{equation} \label{eq:transest-qinfty}
    q_{\infty}^{-}-1 \leqslant p_{\infty}, \quad q_{\infty}^{-}-1 \leqslant \tilde{p}_{\infty}
\end{equation}
which will simplify our perturbation argument.

We re-write the equation as
    \[
    \mathring{\nab}_4 \phi + \frac{2\lambda_0}{r} \phi  = F + \lambda_0\left(\frac{2}{r} -\trch \right) \phi + (\mathring{\nab}_4 - \nab_4)\phi =: F + F_e.
    \]
    We note that schematically we have
    \[
    F_e \sim \left(\trch - 2 r^{-1} \right) \phi + \chih \cdot \phi
    \]
    By Lemma~\ref{lem:products} and the pointwise weak bootstrap assumption \eqref{eq:BA.weak.pointwise}, we have, for $s \leqslant s_{0}$,
    \begin{equation} \label{eq:transest-pert-1}
        |(\trch - 2 r^{-1}) \phi|_{s, q_{0}, q_{\infty}}(u, r) + |\chih \phi|_{s, q_{0}, q_{\infty}}(u, r) \aleq \epsilon_{0} |\phi|_{s, q_{0}^{-}, q_{\infty}^{-}-1}(u, r).
    \end{equation}
    where $|\cdot|_{s, q_{0}, q_{\infty}}(u, r)$ is a shorthand for $|\cdot|_{s, q_{0}, q_{\infty}}(\{ (u, r) \})$. This estimate immediately implies \eqref{eq:main.sphere.transest}, \eqref{eq:22.sphere.transest}, \eqref{eq:22.sphere.transest.2}, and \eqref{eq:main.int.transest} for the geometric transport equation when $s \leqslant s_{0}$. For instance, in the case of \eqref{eq:main.sphere.transest}, we use \eqref{eq:transest-pert-1} to estimate
    \begin{align*}
        |F_e|_{s, q_{0}, q_{\infty}} \aleq \epsilon_{0} |\phi|_{s, q_{0}^{-}, q_{\infty}^{-}-1}
         \aleq \epsilon_{0} | \phi |_{s, q_{0}+1, p_{\infty}},
    \end{align*}
    where we used the first inequality in \eqref{eq:transest-qinfty}. Applying \eqref{eq:main.sphere.transest} to the background transport operator (which holds by Corollary~\ref{cor:bgtrans}), we see that the contribution of $F_{e}$ can be absorbed if $\epsilon_{0}$ is sufficiently small. This proves the desired geometric transport estimate. A similar argument establishes \eqref{eq:22.sphere.transest}, \eqref{eq:22.sphere.transest.2}, and \eqref{eq:main.int.transest}; for the last estimate, we use the second inequality in \eqref{eq:transest-qinfty}.
    
    It remains to prove \eqref{eq:main.int.transest.top}. By Lemma~\ref{lem:products} and both weak bootstrap assumptions \eqref{eq:BA.weak.pointwise}--\eqref{eq:BA.weak.integrated}, we have
    \begin{equation*} 
        \|(\trch - 2 r^{-1}) \phi\|_{s_{0}+1, q_{0}, q_{\infty}} + \|\chih \phi\|_{s_{0}+1, q_{0}, q_{\infty}} \aleq \epsilon_{0} \|\phi\|_{s_{0}+1, q_{0}^{-}, q_{\infty}^{-}-1} + \epsilon_{0}|\phi|_{s_{0}-2, q_{0}^{-}-\frac{1}{2}, q_{\infty}^{-}-\frac{1}{2}}.
    \end{equation*}
    We now apply \eqref{eq:main.sphere.transest} to the background transport operator, which holds even for $s = s_{0} + 1$ by Corollary~\ref{cor:bgtrans}. The contribution of the first term can be handled as before. We bound the second term by the last term on the right hand side of \eqref{eq:main.int.transest.top}, which is possible thanks to the second inequality in \eqref{eq:transest-qinfty}. \qedhere
\end{proof}
\begin{remark} Restricting our proof of \eqref{eq:main.sphere.transest} to $u \geqslant 1$, we obtain the following variant: if the assumptions of Theorem~\ref{thm:transest} $i)$ hold, then for $u \geqslant 1$, we have
\begin{equation} \label{eq:main.sphere.transest.3}
|u^{q_{\infty} - p_{\infty}-1} \phi|_{s, q_0+1, p_\infty}(H_{u}) \lesssim |F|_{s, q_0, q_\infty}(H_{u}),
\end{equation}
where $p_{\infty} = 2 \lambda_{0}$ for $q_{\infty} > 2 \lambda_{0} + 1$, and $p_{\infty} = q_{\infty} - 1$ for $q_{\infty} < 2 \lambda_{0} + 1$.
\end{remark}
\section{The Teukolsky equation} \label{sec:teukolsky}

\subsection{Derivation of the Teukolsky equation}
We now derive a wave equation satisfied by the components $\alpha$ of the Weyl tensor.

\begin{proposition}\label{prop:Teukolsky.nonlinear.control}
    There exists $\ep_0>0$ sufficiently small such that the following holds for some implicit constant depending only on $s_0$. The Weyl tensor components $\alpha$ satisfy the Teukolsky equation
    \begin{equation} \label{eq:Teukolsky}
       \nab_3(r \nab_4(r^2\alpha)) + 2 \nab_4 (r^2 \alpha) + 3 \nab_3 (r^2 \alpha) + \frac{4}{r} (r^2 \alpha) -  r \slashed\Delta (r^2 \alpha) = \mathcal{E},
    \end{equation}
    where if the weak bootstrap assumptions are satisfied, the error term $\mathcal{E}$ satisfies the estimate
        \begin{equation}
        \| \brk{u}^{\frac{1}{2}^{+}} \mathcal{E} \|_{s_0,\frac{1}{2}^+, 0^+}\lesssim \epsilon^2.
    \end{equation}
    If additionally the strong bootstrap assumptions hold then
    \begin{equation}
        \|\mathcal{E}\|_{s_0,\frac{1}{2}^+,\frac{7}{2}^{-} }\lesssim \epsilon^2.
    \end{equation}
\end{proposition}
\begin{proof}
We start with \eqref{eq:3alpha} and multiply by $r^2$ to find
\begin{align*}
    \nab_3(r^2\alpha ) + 2 f r \alpha - 4 \omegab r^2 \alpha + \frac{r^2}{2} \trchb \alpha - \nab \hat{\otimes} (r^2 \beta) = 5 r^2 \eta \hat{\otimes} \beta - 3r^2(\rho \chih+\sigma {}^* \chih)
\end{align*}
so that
\begin{align}
\nab_3(r^2\alpha ) + r\alpha - \nab \hat{\otimes} (r^2 \beta) &= -2(f-1-2\omegab r) r\alpha - \frac{1}{2}(2+r \trchb) r \alpha + 5 r^2 \eta \hat{\otimes} \beta - 3r^2(\rho \chih+\sigma {}^* \chih) \nonumber\\
&=: \mathcal{E}_1 \label{eq:3alpha v2}
\end{align}
Next we similarly multiply \eqref{eq:4beta} by $r^2$ to obtain
\begin{equation}
    \nab_4(r^2 \beta) + 2r \beta - \div(r^2 \alpha) = 2 (2-r \trch) r \beta + r^2 \eta\cdot \alpha =:\mathcal{E}_2. \label{eq:4beta v2}
\end{equation}
Acting on \eqref{eq:3alpha v2} with $r\nab_4+3$ yields
\begin{align*}
   & r\nab_4\nab_3(r^2\alpha ) + \nab_4(r^2\alpha)+3\nab_3(r^2\alpha) +2 r\alpha - r\nab \hat{\otimes} \nab_4(r^2 \beta) - 2 r\nab \hat{\otimes} (r \beta) \\&\qquad \qquad = (r\nab_4 +3) \mathcal{E}_1 + \nab_4\left[r\nab \hat{\otimes} (r^2 \beta)\right]-r\nab \hat{\otimes} \nab_4(r^2 \beta).
\end{align*}
Inserting \eqref{eq:4beta v2} this becomes
\begin{align*}
    &r\nab_4\nab_3(r^2\alpha ) + \nab_4(r^2\alpha)+3\nab_3(r^2\alpha) +4 r\alpha  - r \nab \hat{\otimes} \div(r^2 \alpha) -2 r^2K r\alpha 
    \\ &\qquad \qquad = (r\nab_4 +3) \mathcal{E}_1 + r \nab \hat{\otimes} \mathcal{E}_2 + 2 (1-Kr^2) r\alpha+ \nab_4\left[r\nab \hat{\otimes} (r^2 \beta)\right]-r\nab \hat{\otimes} \nab_4(r^2 \beta). 
\end{align*}
Recalling that $r\nab \hat{\otimes} \div (r^2 \alpha) = r\slashed \Delta (r^2 \alpha) - 2 r^2K r \alpha$ and commuting the derivatives in the leading term we find
\begin{equation}\label{eq:Teukolsky.with.explicit.error}
    \begin{split}
        &\nab_3 (r\nab_4 (r^2\alpha))  + 2 \nab_4(r^2 \alpha) +3\nab_3(r^2 \alpha) + \frac{4}{r}(r^2\alpha) - \frac{1}{r} r^2\slashed\Delta (r^2 \alpha) \\
& \qquad= \underbrace{(r\nab_4 +3) \mathcal{E}_1}_{=:I} + \underbrace{r \nab \hat{\otimes}\mathcal{E}_2}_{=:II} + \underbrace{2 (1-Kr^2) r\alpha}_{=:III}\\
& \qquad\qquad + \underbrace{\nab_4\left[r\nab \hat{\otimes} (r^2 \beta)\right]-r\nab \hat{\otimes} \nab_4(r^2 \beta)}_{=:IV} + \underbrace{[\nab_3, r\nab_4](r^2\alpha) + \nab_4(r^2\alpha)}_{=:V},
    \end{split}
\end{equation}
and we now define $\mathcal E$ in \eqref{eq:Teukolsky} to be the right-hand side of \eqref{eq:Teukolsky.with.explicit.error}. We need to estimate the terms appearing on the right-hand side. For the commutator terms, we use \eqref{4AcommestST} and \eqref{34commestST} to obtain 
\begin{align*}
    \| \brk{u}^{\frac{1}{2}^{+}} IV\|_{s_0, \f 12{}^+,0^+} \ls &\: \ep \Big(\| \brk{u}^{\frac{1}{2}^{+}} r^2 \bt \|_{s_0+1, \f 12{}^+,-1} + |\brk{u}^{0^{+}} r^2 \bt |_{s_0-2,0^+,0} \Big) \\
    \ls &\: \ep \Big(\| \brk{u}^{\frac{1}{2}^{+}} \bt \|_{s_0+1, (-\f 32){}^+, 1} + |\brk{u}^{0^{+}} \bt |_{s_0-2,(-2)^+,2} \Big) \ls \ep^2, \\
    \| \brk{u}^{\frac{1}{2}^{+}} V \|_{s_0, \f 12{}^+, 0^{+}} \ls &\: \ep \Big(\| \jb{u}^{\frac{1}{2}} r^2 \alp \|_{s_0+1, \f 12{}^+, -1^+} + |\brk{u}^{0^{+}} r^2 \alp |_{s_0-2,0^+,0} \Big) \\
    \ls  &\: \epsilon \Big( \| \jb{u}^{\frac{1}{2}} \alp \|_{s_0+1, (-\f 32){}^+, 1^+} + |\brk{u}^{0^{+}} \alp |_{s_0-2,(-2)^+, 2} \Big) \ls \ep^2.
\end{align*}
Under the strong bootstrap assumption, we argue similarly using \eqref{4AcommestSTstr} and \eqref{34commestSTstr} in place of \eqref{4AcommestST} and \eqref{34commestST}, respectively, to obtain the stronger bounds
\begin{equation}
    \| IV\|_{s_0, \f 12{}^+,\f 72{}^-} + \| V\|_{s_0, \f 12{}^+,\f 72{}^-} \ls \ep^2.
\end{equation}


For the remaining terms, we first rewrite $I$ using some transport equations:
\begin{align*}
    (r\nab_4 +3) \mathcal{E}_1 &= 4 (3 \abs{\eta}^2 + \rho) r^3 \alpha -2 (f-1-2\omegab r)( r^2 \nab_4 \alpha +4 r \alpha) \\
    & \qquad -\frac{1}{2} \left(\frac{1}{2} r \trchb (2-r \trch) - 2 r^2 \div \eta - r^2 \chih\cdot \chibh  + 2 r^2 \abs{\eta}^2 + 2 r^2 \rho\right)(r \alpha)\\&\qquad  -\frac{1}{2}(2+r \trchb) (r^2\nab_4 \alpha + 4 r\alpha) + (r\nab_4 +3) 5r^2 \eta\hat{\otimes} \beta \\
    &\qquad -3 r^3 \left( \frac{5}{2r}(2 - r \trch) \rho + \div \beta - \eta\cdot \beta - \frac{1}{2} \chibh \cdot \alpha\right) \chih  - 3r^3 \alpha \rho,
\end{align*}
where we have used the transport equations satisfied by $f, \omegab, \trchb, \rho, \chih$.
Observe that in $I$, $II$ and $III$, there are three types of terms: (1) nonlinear terms with a factor of $\alp$, $r\nab_4\alp$ or $r\nab \alp$, (2) nonlinear terms with a factor of $\bt$ or $r\nab \bt$, and (3) the term $r^2 (2-r\trch)\rho\chih$. For types (1) and (2), it suffices to check that hardest term, namely $(f-1-2\omb r) r^2\nab_4 \alp$ and $r^3\eta \widehat{\otimes}\nab_4 \bt$ or $r^3 \nab_4 \eta \widehat{\otimes} \bt$, since in all the other terms $\alp$, $\bt$ (or their derivatives) are multiplied by terms with either the same or better bounds according to the bootstrap assumptions.

We check these terms in Table~\ref{table:Teukolsky}. The table is to be understood as follows. We need to check the parameters $s,p_0,p_\i$ for which the term can be controlled in $\| \cdot \|_{s,p_0,p_\i}$ norms. Each row checks a different parameter; the $p_\i$ (w) and $p_\i$ (s) rows checks the $p_\i$ parameter under the weak and strong bootstrap assumptions, respectively. We need that every entry satisfies the goal stated on the far-right column. When checking the admissibility of the indices, we use the product estimates \eqref{eq:product.for.integrated} and the bootstrap assumptions, and thus naturally we need to consider sums of indices. Notice also that we need to consider putting different factors in the integrated and pointwise norms, which is why we need to take the minimum of multiple possibilities. Moreover, under the weak bootstrap assumptions, there is no need to separately check the powers of the weaker weight $\brk{u}$ since the total sum of all powers of $\brk{u}$ and $r + \brk{u}$ (represented by the $p_{\infty}$ parameter) remain constant in the product estimates and the bootstrap assumptions.
\begin{table}[!ht]
    \centering
    \caption{Checking the nonlinear terms in the Teukolsky equation}
    \label{table:Teukolsky}
        
        \begin{tabular}{l|ccc|r}
            \toprule
             & $(f-1-2\omb r) r^2\nab_4 \alp$ & $r^2\eta \widehat{\otimes}(r\nab_4) \bt, r^2 (r\nab_4) \eta \widehat{\otimes} \bt$ & $r^2 (2-r\trch)\chih\rho$ & Goal \\
            \midrule
            $s$ & $s_0 \wedge s_0$ & $(s_0+1) \wedge s_0$ & $s_0\wedge s_0\wedge (s_0+1)$ & $\geqslant s_0$ \\
            $p_0$ & \tiny{$\min \left\{ \begin{smallmatrix} 1+1^++(-\f 12)^+ \\ 1+\f 32{}^+ + (-1)^+ \end{smallmatrix}  \right\}$} & \tiny{$\min \left\{ \begin{smallmatrix} 2+0^++(-\f 12)^+ \\ 2+\f 12{}^+ + (-1)^+ \end{smallmatrix}  \right\}$} & \tiny{$\min \left\{ \begin{smallmatrix} 3+0^+ + 0^+ + (-\f 12)^+ \\ 3+0^+ + \f 12{}^+ + \f 52{}^+ \end{smallmatrix}  \right\}$} & $\geqslant \f 12{}^+$\\
            $p_\infty$ (\hbox{w}) & \tiny{$\min \left\{\begin{smallmatrix}
                (-1)+ 0+\frac{3}{2}^+\\ (-1)+(-1)^++ 2^+
            \end{smallmatrix}
            \right\}$}  & \tiny{$\min \left\{\begin{smallmatrix} (-2)+ 1^+ + 1^+ \\ (-2)+ \f 12{}^++2^+ \end{smallmatrix} \right\}$} & \tiny{$\min \left\{\begin{smallmatrix}
                (-3) + 1^+ + 1^+ + 1^{+}  \\
                (-3)+ 1^+ + \frac{1}{2}^+ + 2^{+} 
            \end{smallmatrix}  \right\}$} & $\geqslant 0^+$ \\
            $p_\infty$ (\hbox{s}) & \tiny{$\min \left\{\begin{smallmatrix}
                (-1)+ 0+\frac{9}{2}^{-} \\ (-1)+(-\f 12)^++5^{-}
            \end{smallmatrix}
            \right\}$}  & \tiny{$\min \left\{\begin{smallmatrix} (-2)+ 2 +\frac{7}{2}^- \\ (-2)+ \frac{3}{2}^{-}+4 \end{smallmatrix} \right\}$} & \tiny{$\min \left\{\begin{smallmatrix}
                (-3) + 2 + 2 + \frac{5}{2}^{-}  \\
                (-3)+ 2 + \frac{3}{2}^{-} + 3 
            \end{smallmatrix}  \right\}$} & $\geqslant \frac{7}{2}^{-}$ \\
            \bottomrule
        \end{tabular}

\end{table}

\end{proof}

\subsection{Estimates for the linearised Teukolsky operator} \label{TeukEsts}

The central estimate in our proof is a Morawetz and $r^p$-type estimate for the linearised Teukolsky operator, which will permit us to establish control of the curvature components $\alpha$. 
\begin{definition}
    Let $\phi$ be a sufficiently regular $S$-tangent tensor. The linearised Teukolsky operator acting on $\phi$ is given by:
    \[
    \mathcal{T}\phi =  \nab_3(r \nab_4 \phi) + 2  \nab_4 \phi + 3 \nab_3 \phi + \frac{4}{r} \phi - \frac{1}{r} r^2 \slashed \Delta \phi.
    \]
    When commuting to control higher derivatives, we will also require the modified operator
    \[
     \mathcal{T}_l\phi =  \mathcal{T}\phi + \frac{l}{r}(r \nab_4 \phi + 3 \phi)
    \]
    where $l\geqslant 0$.
\end{definition}

\subsubsection{The basic estimate}\label{sec:Teu.basic}
We write the Teukolsky operator as
\[
\mathcal{T}\phi =  \nab_3(r \nab_4 \phi + 3 \phi) + 2  \nab_4 \phi + \frac{4}{r} \phi - \frac{1}{r} r^2 \slashed \Delta \phi
\]
and multiply by $(r \nab_4 \phi + 3 \phi)$ to find
\begin{align*}
    \mathcal{T}\phi \cdot \left(r\nab_4 \phi + 3 \phi\right) &= \nab_3 \left[\frac{1}{2 }\abs{r \nab_4 \phi + 3 \phi}^2 \right] + 2 r \abs{\nab_4 \phi}^2+10 \phi \cdot \nab_4 \phi + \frac{12}{r}\abs{\phi}^2\\&\qquad -  r^2 (\slashed \Delta \phi) \cdot \nab_4 \phi - \frac{3}{r}  \phi \cdot r^2\slashed \Delta \phi  
     \\    &= \left( \nab_3  + \trchb - 2 \omegab \right)  
     \left[\frac{1}{2}\abs{r \nab_4 \phi + 3 \phi}^2\right] + \frac{1}{r}\abs{r \nab_4 \phi + 3 \phi}^2 + 2 r \abs{\nab_4 \phi}^2
     \\ &\qquad + \left[ 2r \omegab -(2 + r \trchb) \right]\frac{1}{2r}\abs{r \nab_4 \phi + 3 \phi}^2 \\
     &\qquad +(\nab_4 + \trch)\left[\abs{\phi}^2 \right] + \frac{10}{r} \abs{\phi}^2 + 8 \phi \cdot \nab_4 \phi+ (2 - r \trch) \frac{1}{r} \abs{\phi}^2 \\
     &\qquad - \nab^A \left( (r\nab_A \phi)\cdot(r \nab_4 \phi) + \frac{3}{r} r\phi \cdot r\nab_A \phi  \right) + \left[  r\nab_A, \nab_4\right]\phi \cdot (r\nab^A \phi) \\
     &\qquad + ( \nab_4 + \trch)\left[ \frac{1}{2} \abs{r\nab\phi}^2\right] + \frac{2}{r} \abs{r\nab\phi}^2 + (2 - r \trch) \frac{1}{2r} \abs{r\nab\phi}^2.
\end{align*}
Integrating over $S_{u,r}$ we find
\begin{align}
   & \frac{\partial}{\partial r} \int_{S_{u, r}}\left[ \frac{1}{2} \abs{r\nab\phi}^2 +  \abs{\phi}^2 - \frac{f}{2}\abs{r \nab_4 \phi + 3 \phi}^2  \right] \ud A_{\gamma} + \frac{\partial}{\partial u} \int_{S_{u, r}}\abs{r \nab_4 \phi + 3 \phi}^2  \ud A_{\gamma}\nonumber \\ \label{eq:nab4mult}
    & \qquad + \int_{S_{u,r}} \left[ \frac{1}{r}\abs{r \nab_4 \phi + 3 \phi}^2 + 2 r\abs{\nab_4 \phi}^2 + \frac{2}{r} \abs{r\nab\phi}^2+ \frac{8}{r}\phi\cdot (r\nab_4 \phi) + \frac{10}{r} \abs{\phi}^2\right] \ud A_{\gamma} \\
    &= \int_{S_{u,r}} \bigg[ \mathcal{T}\phi \cdot \left(r \nab_4 \phi + 3 \phi\right) + \left\{ (2 + r \trchb) - 2\omegab r\right\}\frac{1}{2r}\abs{r \nab_4 \phi + 3 \phi}^2\nonumber\\&\qquad \qquad  +(r \trch-2) \left\{ \frac{1}{2r} \abs{r\nab\phi}^2+ \frac{1}{r}\abs{\phi}^2\right\}+  \left[  \nab_4, r\nab_A\right]\phi \cdot (r\nab^A \phi) \bigg] \ud A_{\gamma}. \nonumber
\end{align}

Before we integrate this estimate, we first multiply by a suitable function to improve the weights near the origin, and simultaneously to establish decay in $u$. To this end, we introduce
\[
w = \frac{(r+a\jb{u})^{2p_0 + 2p_\infty}}{r^{3+2p_0}},
\]
where $0<a<1$ is a small constant to be chosen later, and $p_0, p_\infty$ control decay near $r=0$, $r=\infty$ respectively. Our weight function
satisfies
\[
\frac{r \partial_r w}{w} = (2p_\infty-3 )\frac{r}{r+ a \jb{u}} -(3+2p_0)\frac{a\jb{u} }{ r+ a \jb{u}}, \qquad \frac{r \partial_u w}{w} =  (2p_0+2p_\infty)\frac{a r}{r+a \jb{u}} \frac{u}{\jb{u}} 
\]
Multiplying \eqref{eq:nab4mult} by $w$ we have
\begin{align*}
   & \frac{\partial}{\partial r} \int_{S_{u, r}}\left[ \frac{1}{2} \abs{r\nab\phi}^2 + \abs{\phi}^2 - \frac{f}{2}\abs{r \nab_4 \phi + 3 \phi}^2  \right] w\ud A_{\gamma} + \frac{\partial}{\partial u} \int_{S_{u, r}}\abs{r \nab_4 \phi + 3 \phi}^2  w\ud A_{\gamma}\nonumber \\ 
    & \qquad + \int_{S_{u,r}} \bigg[ \frac{1}{r}\left(1 + \frac{f}{2}\frac{r \partial_r w}{w}  - \frac{r \partial_u w}{w}\right)\abs{r \nab_4 \phi + 3 \phi}^2 + 2 r\abs{\nab_4 \phi}^2 + \frac{8}{r} \phi \cdot(r \nab_4 \phi)\\
    &\qquad + \frac{1}{2 r}\left( 4 -\frac{r \partial_r w}{w}\right) \abs{r\nab\phi}^2 + \frac{1}{r}\left(10 -\frac{r \partial_r w}{w} \right) \abs{\phi}^2\bigg] w\ud A_{\gamma} \\
    &= \int_{S_{u,r}} \bigg[ \mathcal{T}\phi \cdot \left(r \nab_4 \phi + 3 \phi\right) + \left\{ (2 + r \trchb) - 2\omegab r\right\}\frac{1}{2r}\abs{r \nab_4 \phi + 3 \phi}^2\nonumber\\&\qquad \qquad  +(r \trch-2) \left\{ \frac{1}{2r} \abs{r\nab\phi}^2+ \frac{1}{r}\abs{\phi}^2\right\}+  \left[  \nab_4, r\nab_A\right]\phi \cdot (r\nab^A \phi) \bigg] w\ud A_{\gamma} \nonumber
\end{align*}

Introducing
\[
A =  1 + \frac{1}{2}\frac{r \partial_r w}{w}  - \frac{r \partial_u w}{w}, \quad B =  10 -\frac{r \partial_r w}{w} , \quad D =  2 -\frac{r \partial_r w}{2w}
\]
we can rewrite the identity as
\begin{align}
   & \frac{\partial}{\partial r} \int_{S_{u, r}}\left[ \frac{1}{2} \abs{r\nab\phi}^2 + \abs{\phi}^2 - \frac{f}{2}\abs{r \nab_4 \phi + 3 \phi}^2  \right] w \ud A_{\gamma} + \frac{\partial}{\partial u} \int_{S_{u, r}}\abs{r \nab_4 \phi + 3 \phi}^2  w \ud A_{\gamma}  \nonumber\\ 
    & \qquad + \int_{S_{u,r}} \bigg[ A \abs{r \nab_4 \phi + 3 \phi}^2 + 2 \abs{r\nab_4 \phi}^2+ 8 \phi \cdot(r \nab_4 \phi) +B  \abs{\phi}^2 + D\abs{r\nab\phi}^2 \bigg] \frac{w}{r} \ud A_{\gamma} \label{eq:nab4multwt}\\
    &= \int_{S_{u,r}} \bigg[ \mathcal{T}\phi \cdot \left( r\nab_4 \phi + 3 \phi\right) + \left\{ (1-f)\frac{r \partial_r w}{w} +(2 + r \trchb) - 2\omegab r \right\}\frac{1}{2r}\abs{r \nab_4 \phi + 3 \phi}^2\nonumber\\&\qquad \qquad  +(r \trch-2) \left\{ \frac{1}{2r} \abs{r\nab\phi}^2+ \frac{1}{r} \abs{\phi}^2\right\}+  \left[  \nab_4, r\nab_A\right]\phi \cdot (r\nab^A \phi) \bigg] w \ud A_{\gamma} \nonumber
\end{align}
The terms quadratic in $\phi$ on the right we shall treat as errors. We have that
\[
D = \frac{ (7+2p_0)a  \jb{u}+(7-2p_\infty)  r}{2 (a \jb{u}+r)} \sim 1
\]
provided $p_0>-\frac{7}{2}$ and $p_\infty< \frac{7}{2}$. We must additionally consider the positivity of the quadratic form with matrix
\[
 M=       \begin{pmatrix}
            2 + A & 4+3 A \\
            4+3 A & 9A+B \\
        \end{pmatrix}.
\]
We compute (setting $\theta = r/a\jb{u}\geqslant 0$):
\begin{align*}
(1+\theta) \textrm{Tr}\,M &= (10-8p_0) + (10+8 p_\infty)\theta - 20a \frac{u}{\jb{u}} (p_0 + p_\infty)\theta ,\\
(1+\theta)^2 \det M &= \left(\frac{17}{2} - 2 (p_0+1)^2 \right) + \left(4(p_0+1)(p_\infty-1)+17\right)\theta +\left(\frac{17}{2} - 2 (p_\infty-1)^2 \right)\theta^2\\ 
&\qquad - 2 (p_0+p_\infty) a \frac{u}{\jb{u}} ((7+2p_0)\theta + (7-2p_\infty)\theta^2).
\end{align*}
We observe that when $a = 0$, provided 
\begin{align}
    &-\frac{\sqrt{17}}{2} -1 <p_0 <\frac{\sqrt{17}}{2} -1 \label{eq:p0bound} \\
    &-\frac{\sqrt{17}}{2} +1 <p_\infty <\frac{\sqrt{17}}{2} +1\label{eq:pibound}
\end{align}
we have that both $\det M, \textrm{Tr}\,M$ are bounded below by a positive constant. By continuity this holds also when $a$ is small since the terms involving $a$ are of no higher order in $\theta$ than those independent of $a$. We have established:
\begin{proposition} \label{prop:rp-pos}
    Suppose that the bounds \eqref{eq:p0bound}, \eqref{eq:pibound} hold. Then we may choose $a>0$ sufficiently small that  
    \[
A \abs{r \nab_4 \phi + 3 \phi}^2 + 2 \abs{r\nab_4 \phi}^2 + 8 \phi \cdot(r \nab_4 \phi)  +B  \abs{\phi}^2 + D\abs{r\nab\phi}^2 \sim \abs{r\nab_4\phi}^2 + \abs{r\nab\phi}^2 +\abs{\phi}^2.
\]
\end{proposition}
Note in particular that if $\delta < 1/20$ the bounds are satisfied for $p_0 = 1+\delta$ and $p_\infty \in [0, 3+\delta]$\footnote{In this section, we include the case $p_{\infty} > 3$ in view of the applications discussed in Remarks~\ref{rem:main-1}.ii), \ref{rem:strong-peeling-rp}, and \ref{rem:strong-peeling-spnull}. For the proof of the main theorem, however, we only require $p_{\infty} = 3 - \delta$, which simplifies some of the analysis here.}. We are now ready to establish
\begin{lemma}\label{lem:basicest}
    Assume the weak bootstrap assumptions. Suppose \eqref{eq:p0bound}, \eqref{eq:pibound} hold and $\phi$ is a sufficiently regular $S$-tangent tensor satisfying the regularity condition at the axis:
    \[
    \lim_{r \to 0}\sup_{1-r<u<T} r^{-p_0-\frac{3}{2}} |\phi|_{1} (u, r) = 0
    \]
    Then for $\epsilon_0, a$ sufficiently small there exist constants $c, C>0$ depending on $p_0, p_\infty$ such that 
    \begin{align*}
        &\sup_{u_T<u<T} c\int_{H_u} \abs{r \nab_4 \phi + 3 \phi}^2 w \ud A_{\gamma} \ud r + c\int_{S_T} \left[ \abs{r\nab_4\phi}^2 + \abs{r\nab\phi}^2 +\abs{\phi}^2\right]\frac{1}{r} w \ud A_{\gamma} \ud r \ud u \\
        &\qquad \leqslant  \int_{S_T} \mathcal{T}\phi \cdot \left( r\nab_4 \phi + 3 \phi\right) w \ud A_{\gamma} \ud r \ud u + C\int_{\tilde{\Sigma}_1}  \left[ \abs{r\nab_4\phi}^2 + \abs{r\nab\phi}^2 +\abs{\phi}^2\right] w \ud A_{\gamma} \ud r. 
    \end{align*}
\end{lemma}
\begin{proof}
    For $u_0<T$ and $r_0>0$ we integrate \eqref{eq:nab4multwt} over the region $R=\{(r,u)\in S_T : r_0<r,\, u<u_0\}$, with $a$ chosen sufficiently small as above. The boundary of this region consists of the components $\mathfrak S_1 := \overline{R} \cap \{r = r_0\}$, $\mathfrak S_2 :=H_{u_0}\cap \{\max\{r_0, 1-u_0\}<r\}$, $\mathfrak S_3:=\Sigma_T \cap \{u<u_0\}$, $\mathfrak S_4:=\tilde{\Sigma}_1\cap \{u_T < u < \min \{u_0, 1-r_0\} \}$. The integral receives flux contributions $\mathcal{F}_i$ from the surfaces $\mathfrak S_i$ and two bulk contributions $\mathcal{B}_i$ where we write
    \[
    \mathcal{F}_1 +\mathcal{F}_2+\mathcal{F}_3+\mathcal{B}_1 = \mathcal{F}_4+\mathcal{B}_2.
    \]
    Working through these terms,
    \[
    \mathcal{F}_1 = -\int_{1-r_0}^T \int_{S_{u, r_0}}\left[ \frac{1}{2} \abs{r\nab\phi}^2 + \abs{\phi}^2 - \frac{f}{2}\abs{r \nab_4 \phi + 3 \phi}^2  \right] w \ud A_{\gamma} \ud u.
    \]
    We can estimate $|\mathcal{F}_1| \lesssim T \sup_{1-r_0<u<T} r_0^{-2-2p_0} |\phi|_{1}^2 (u, r_0)$ so that the contribution from $\mathcal{F}_1$ vanishes as $r_0 \to 0$. Next we have
    \[
    \mathcal{F}_2  =  \int_{\max\{r_0, 1-u_0\}} \int_{S_{u_0, r}}\abs{r \nab_4 \phi + 3 \phi}^2  w \ud A_{\gamma} \ud r \to \int_{H_{u_0}}\abs{r \nab_4 \phi + 3 \phi}^2  w \ud A_{\gamma} \ud r
    \]
    as $r_0 \to 0$. The term from the future spacelike boundary is 
    \begin{align*}
        \mathcal{F}_3  = \int_{u_T}^{u_0} \int_{S_{u, r_T(u)}}\left[ \frac{1}{2} \abs{r\nab\phi}^2 + \abs{\phi}^2+ \left( \frac{\partial_u \tau}{\partial_r \tau} - \frac{f}{2}\right)\abs{r \nab_4 \phi + 3 \phi}^2  \right] w \ud A_{\gamma} \ud u \geqslant 0
    \end{align*}
     where $r_T(u)$ is the unique solution to $\tau(u,r_T(u)) = T$. This term is non-negative, and hence can be discarded, due to the inequality $2\partial_u \tau > f \partial_r \tau$ which follows from the construction of $\tau$ in \S\ref{sec:tau}. The first bulk term is
    \begin{align*}
    \mathcal{B}_1 &= \int_{R} \bigg[ A \abs{r \nab_4 \phi + 3 \phi}^2 + 2 \abs{r\nab_4 \phi}^2+ 8 \phi \cdot(r \nab_4 \phi) +B  \abs{\phi}^2 + C\abs{r\nab\phi}^2 \bigg] \frac{w}{r} \ud A_{\gamma} \ud u \ud r \\
    &\gtrsim \int_{R} \left[  \abs{r\nab_4 \phi}^2  +\abs{r\nab\phi}^2+\abs{\phi}^2  \right] \frac{w}{r} \ud A_{\gamma} \ud u \ud r
    \end{align*}
    by our previous discussion when constructing $w$. In the limit $r_0\to 0$ we can replace the region of integration with $S_T\cap \{u<u_0\}$. Moving now to the terms on the other side, the flux term is
    \begin{align*}
        \mathcal{F}_4  &= \int_{r_0}^{1-u_T} \int_{S_{1-r, r}}\left[ \frac{1}{2} \abs{r\nab\phi}^2 + \abs{\phi}^2+ \left( 1 - \frac{f}{2}\right)\abs{r \nab_4 \phi + 3 \phi}^2  \right] w \ud A_{\gamma} \ud r \\
        &\lesssim \int_{0}^{1-u_T} \int_{S_{1-r, r}}\left[ \abs{r \nab_4 \phi}^2+ \abs{r\nab\phi}^2 + \abs{\phi}^2  \right] w \ud A_{\gamma} \ud r,
    \end{align*}
     since $|f-1|\lesssim \epsilon$ on $\tilde{\Sigma}_1$. It remains to consider the second bulk term.
     \begin{align*}
         \mathcal{B}_2 &=\int_{R} \bigg[ \mathcal{T}\phi \cdot \left( r\nab_4 \phi + 3 \phi\right) + \left\{ (1-f)\frac{r \partial_r w}{w} +(2 + r \trchb) - 2\omegab r \right\}\frac{1}{2r}\abs{r \nab_4 \phi + 3 \phi}^2\nonumber\\&\qquad \qquad  +(r \trch-2) \left\{ \frac{1}{2r} \abs{r\nab\phi}^2+ \frac{5}{r} \abs{\phi}^2\right\}+  \left[  \nab_4, r\nab_A\right]\phi \cdot (r\nab^A \phi) \bigg] w \ud A_{\gamma} \ud u \ud r.
     \end{align*}
     The first contribution can be left alone. We observe that
     \[
     \abs{\frac{1-f}{2}\frac{r \partial_r w}{w} +(2 + r \trchb) - 2\omegab r } \lesssim \epsilon \left( \frac{r}{\jb{u}^{1^+}} + 1\right)
     \]
     which together with bounds on $r\trch-2$ and the commutator estimate \eqref{4Acommest} implies that
     \begin{align*}
         \mathcal{B}_2 &\leqslant \int_{S_T}  \mathcal{T}\phi \cdot \left( r\nab_4 \phi + 3 \phi\right) w \ud A_{\gamma} \ud u \ud r  + C \epsilon_0 \int_{R} \left[  \abs{r\nab_4 \phi}^2  +\abs{r\nab\phi}^2+\abs{\phi}^2  \right] \frac{w}{r} \ud A_{\gamma} \ud u \ud r \\
         & \qquad + C \epsilon_0 \int_{R} \frac{1}{\jb{u}^{1^+}}\abs{r \nab_4 \phi + 3 \phi}^2 w \ud A_{\gamma} \ud u \ud r.
     \end{align*}
     We note that
     \begin{align*}
     \int_{R} \frac{1}{\jb{u}^{1^+}}\abs{r \nab_4 \phi + 3 \phi}^2 w \ud A_{\gamma} \ud u \ud r \leqslant \left(\int_{-\infty}^{\infty} \frac{1}{\jb{u}^{1^+}} \ud u\right) \times \sup_{u_T<u<T} \int_{H_{u}}\abs{r \nab_4 \phi + 3 \phi}^2  w \ud A_{\gamma} \ud r 
     \end{align*}
     where the $u$-integral converges. Thus, we may send $r_0 \to0$, take a supremum over $u_0$ in our estimate and then absorb the two error terms proportional to $\epsilon_0$, provided this is sufficiently small, which completes the proof.
\end{proof}

\begin{corollary}
    Under the assumptions of Lemma \ref{lem:basicest} we have
    \begin{align*}
       & \sup_{u}\left \|r \nab_4 \phi + 3 \phi \right\|_{0, p_0, p_\infty}^2(H_u) + \|\phi\|_{1, p_0+\frac{1}{2}, p_\infty-\frac{1}{2}}^2 \\ \qquad &\lesssim \| \mathcal{T}_l \phi \|^2_{0,p_0, p_\infty,*} + \| (r \nab_{4} \phi, r \nab \phi, \phi)\|_{0, p_0, p_\infty}^2(\tilde{\Sigma}_1)
    \end{align*}
    for any $l\geqslant 0$.
\end{corollary}
\begin{proof}
    We insert the identity $\mathcal{T}\phi = \mathcal{T}_l\phi -\frac{l}{r} (r\nab_4 \phi + 3\phi)$ into the result of Lemma \ref{lem:basicest}. On the right hand side we estimate for $\varepsilon>0$:
    \begin{align*}
    &\int_{S_T} \left(\mathcal{T}_l\phi \cdot \left( r\nab_4 \phi + 3 \phi\right) - \frac{l}{r}\abs{r\nab_4 \phi + 3 \phi}^2 \right)w \ud A_{\gamma} \ud r \ud u \\
    &\qquad = \int_{S_T} \left(\mathcal{T}_l\phi \cdot \left( r\nab_4 \phi + 3 \phi\right) - \frac{l}{r}\abs{r\nab_4 \phi + 3 \phi}^2 \right)\frac{(r+a \jb{u})^{2p_0+2p_\infty}}{r^{3+2p_0}} \ud A_{\gamma} \ud r \ud u \\
    &\qquad \leqslant  \|\mathcal{T}_l\phi\|_{0,p_0, p_\infty*} \left( \sup_{u}\left \|r \nab_4 \phi + 3 \phi \right\|_{0, p_0, p_\infty}^2(H_u) + \|r \nab_4 \phi + 3 \phi\|_{1, p_0+\frac{1}{2}, p_\infty-\frac{1}{2}}^2 \right)^{\frac{1}{2}} \\
    &\qquad \leqslant \varepsilon \left( \sup_{u}\left \|r \nab_4 \phi + 3 \phi \right\|_{0, p_0, p_\infty}^2(H_u) + \|\phi\|_{1, p_0+\frac{1}{2}, p_\infty- \frac{1}{2}}^2 \right) + \frac{1}{\varepsilon} \|\mathcal{T}_l\phi\|_{0,p_0, p_\infty,*}^2,
    \end{align*}
    where we used $l \geqslant 0$ to drop a non-positive term.
    Rewriting the integrals in terms of the weighted norms we are done after observing that for $\varepsilon$ sufficiently small we can absorb the first term. 
\end{proof}

\subsubsection{Control of further derivatives} We observe that our estimates above only give bulk control of derivatives of $\phi$ tangent to the surfaces of constant $u$, and only control the flux on $H_u$ of a (twisted) radial derivative. To gain bulk control of derivatives transverse to the lightcone, and also additional flux terms, we make use of a multiplier constructed from $\nab_3\phi$. Because of the error terms generated by this, we pay the price of a weaker $r$-weight near infinity in comparison to the basic estimate, but get compensated by the corresponding $u$-weight. We first introduce
\[
\tilde{\mathcal{T}} \phi = \mathcal{T}\phi - [\nab_3, r \nab_4]\phi - \nab_4 \phi = \nab_4(r \nab_3 \phi) + \nab_4 \phi + 2 \nab_3 \phi +\frac{4}{r} \phi - \frac{1}{r} r^2 \slashed \Delta \phi.
\]
Multiplying by $r\nab_3 \phi$, we find
\begin{align*}
   r\nab_3\phi \cdot  \tilde{\mathcal{T}}\phi   &= \left(\nab_4 + \trch\right) \frac{1}{2} \abs{r\nab_3 \phi}^2  + \frac{1}{r} \abs{r\nab_3 \phi}^2 +\frac{1}{r} (r\nab_3 \phi) \cdot (r\nab_4 \phi)+ (2-\trch r)\frac{1}{2r } \abs{r\nab_3 \phi}^2  \\&\qquad + \left(\nab_3 + \trchb - 2\omegab \right)\left[ \frac{1}{2} \abs{r \nab \phi}^2 + 2\abs{\phi}^2 \right] \\&\qquad +\left(2 \omegab r - \trchb r\right) \left[ \frac{1}{2r} \abs{r \nab \phi}^2 + \frac{2}{r}\abs{\phi}^2 \right] \\
   &\qquad - \nab^A \left((r \nab_A \phi)\cdot (r\nab_3) \phi \right) + \left[ r \nab_A ,\nab_3\right]\phi \cdot (r\nab^A \phi).
\end{align*}
Integrating over $S_{u,r}$ gives
\begin{align}
   & \frac{\partial}{\partial r} \int_{S_{u, r}}\left[ \frac{1}{2} \abs{r\nab_3 \phi}^2 - f\left( \frac{1}{2} \abs{r \nab \phi}^2 + 2\abs{\phi}^2 \right)  \right] \ud A_{\gamma} + \frac{\partial}{\partial u} \int_{S_{u, r}}\left(  \abs{r \nab \phi}^2 + 4\abs{\phi}^2 \right)   \ud A_{\gamma}\nonumber \\ \label{eq:nab3mult}
    & \qquad + \int_{S_{u,r}} \frac{1}{r} \abs{r\nab_3 \phi}^2  \ud A_{\gamma} \\
    &= \int_{S_{u,r}} \bigg[ r\nab_3\phi \cdot  \tilde{\mathcal{T}}\phi- \frac{1}{r} (r\nab_3 \phi) \cdot (r\nab_4 \phi) -\left(2 \omegab r - \trchb r\right) \left[ \frac{1}{2r} \abs{r \nab \phi}^2 + \frac{2}{r}\abs{\phi}^2 \right] \nonumber\\&\qquad \qquad  +(r \trch-2)\frac{1}{2r} \abs{r\nab_3\phi}^2+  \left[  \nab_3, r\nab_A\right]\phi \cdot (r\nab^A \phi) \bigg] \ud A_{\gamma}. \nonumber
\end{align}
We introduce the weight
\[
\tilde{w}= \frac{(r+ \jb{u})^{2p_0+2p_\infty'} \brk{u}^{2 p_{\infty} - 2 p_{\infty}'}}{r^{3+2p_0}} ,
\]
where we shall assume that $p_0, p_\infty$ are as in \S\ref{sec:Teu.basic} above, and $p_\infty'\leqslant p_\infty$. The weight satisfies
\[
\frac{r \partial_r \tilde{w}}{\tilde{w}} = (2p_\infty'-3 )\frac{r}{r+ \jb{u}} -(3+2p_0)\frac{\jb{u} }{ r+ \jb{u}}, \qquad \frac{r \partial_u \tilde{w}}{\tilde{w}} =  (2p_0+2p_\infty')\frac{r}{r+ \jb{u}} \frac{u}{\jb{u}} + (2 p_{\infty} - 2 p_{\infty}') \frac{r u}{\jb{u}^{2}} .
\]
Since $r + a \brk{u} \sim r + \brk{u}$ (with $a$ fixed as in Proposition~\ref{prop:rp-pos}), we have
\begin{equation} \label{eq:tildew-w}
	\tilde{w} \sim w \left(\frac{\brk{u}}{r + \brk{u}}\right)^{2 p_{\infty} - 2 p_{\infty}'}.
\end{equation}
Multiplying \eqref{eq:nab3mult} by $\tilde{w}$ we find
\begin{align}
   & \frac{\partial}{\partial r} \int_{S_{u, r}}\left[ \frac{1}{2} \abs{r\nab_3 \phi}^2 - f\left( \frac{1}{2} \abs{r \nab \phi}^2 + 2\abs{\phi}^2 \right)  \right] \tilde{w} \ud A_{\gamma} + \frac{\partial}{\partial u} \int_{S_{u, r}}\left(  \abs{r \nab \phi}^2 + 4\abs{\phi}^2 \right)  \tilde{w} \ud A_{\gamma}\nonumber \\ \label{eq:nab3multwt}
    & \qquad + \int_{S_{u,r}} \frac{1}{2r} \left( 2 - \frac{r\partial_r \tilde{w}}{\tilde{w}}\right) \abs{r\nab_3 \phi}^2 \tilde{w} \ud A_{\gamma} \\
    &= \int_{S_{u,r}} \bigg[ r\nab_3\phi \cdot  \tilde{\mathcal{T}}\phi- \frac{1}{r} (r\nab_3 \phi) \cdot (r\nab_4 \phi) -\left(2 \omegab r  + f \frac{r\partial_r \tilde{w}}{\tilde{w}} - 2 \frac{r\partial_u \tilde{w}}{\tilde{w}}- \trchb r\right) \left[ \frac{1}{2r} \abs{r \nab \phi}^2 + \frac{2}{r}\abs{\phi}^2 \right] \nonumber\\&\qquad \qquad  +(r \trch-2)\frac{1}{2r} \abs{r\nab_3\phi}^2+  \left[  \nab_3, r\nab_A\right]\phi \cdot (r\nab^A \phi) \bigg]\tilde{w}\ud A_{\gamma}. \nonumber
\end{align}
We note that
\[
2 - \frac{r\partial_r \tilde{w}}{\tilde{w}} = (5 - 2p_\infty') \frac{r}{r+\jb{u}} + (5+2p_0) \frac{ \jb{u}}{r+ \jb{u}} \sim 1
\]
provided $p_\infty'<\frac{5}{2}$ and $p_0 >-\frac{5}{2}$. We can now mostly repeat our previous argument to obtain an estimate --- we focus only on where the argument requires adjustment.

\begin{lemma}\label{lem:nab3est}
    Fix $0<\delta < 1/20$ and suppose $p_0 = 1+\delta$ and $p_\infty \in [0, 3+\delta]$. Let 
    \[
    p_\infty' =\left \{ \begin{array}{cc}
         p_\infty-\frac{1}{2} & p_\infty <3-\frac{\delta}{2} \\
         \frac{5}{2}-\frac{\delta}{2} & p_\infty \geqslant 3-\frac{\delta}{2}.
    \end{array} \right.
    \]
    Suppose that the weak bootstrap assumptions hold and $\phi$ is a sufficiently regular $S$-tangent tensor satisfying the regularity condition at the axis:
    \[
    \lim_{r \to 0}\sup_{1-r<u<T} r^{-p_0-\frac{3}{2}} \left(|\phi|_{1}+ |r\nab_3 \phi|_0\right) (u, r) = 0.
    \]
    Then for $\epsilon_0$ sufficiently small and any $l\geqslant 0$ we have the estimate
    \begin{align*}
        &\sup_{u_T<u<T} \int_{H_u} \left(  \abs{r \nab \phi}^2 + \abs{\phi}^2 \right)  \tilde{w} \ud A_{\gamma} \ud r + \int_{S_T} \abs{r\nab_3\phi}^2 \frac{1}{r} \tilde{w} \ud A_{\gamma} \ud r \ud u \\
        &\qquad \ls \int_{S_T} \abs{\mathcal{T}_l\phi}^2 r \tilde{w} \ud A_{\gamma} \ud r \ud u + \int_{\tilde{\Sigma}_1}  \left(\abs{r\nab_3\phi}^2+\abs{r\nab\phi}^2 +\abs{\phi}^2 \right)\tilde{w} \ud A_{\gamma} \ud r \\
        &\qquad \qquad+ \int_{S_T} \left[ \abs{r\nab_4\phi}^2 + \abs{r\nab\phi}^2 +\abs{\phi}^2\right]\frac{w}{r} \ud A_{\gamma} \ud r \ud u,
    \end{align*}
    where the implicit constant depends on $\delta, l$.
\end{lemma}
\begin{proof}
    We integrate over the same region as in the proof of Lemma \ref{lem:basicest} and similarly split the contributions into fluxes $\mathcal{F}_i, \mathcal{B}_i$. The flux terms $\mathcal{F}_1, \mathcal{F}_2, \mathcal{F}_3, \mathcal{F}_4$ and bulk term $\mathcal{B}_1$ are handled in the same way, and give rise to the terms on the left of the estimate and the initial data term. The flux term $\mathcal{F}_4$ we bound by
    \begin{align*}
        \mathcal{F}_4  &= \int_{r_0}^{1-u_T} \int_{S_{1-r, r}}\left[ \frac{1}{2} \abs{r\nab_3\phi}^2 +  \left( 1 - \frac{f}{2}\right)\left( \abs{r\nab\phi}^2 + 4\abs{\phi}^2 \right)  \right] \tilde{w} \ud A_{\gamma} \ud r \\
        &\lesssim \int_{0}^{1-u_T} \int_{S_{1-r, r}}\left[ \abs{r \nab_3 \phi}^2+ \abs{r\nab\phi}^2 + \abs{\phi}^2  \right] \tilde{w} \ud A_{\gamma} \ud r.
    \end{align*}
    Turning now to the bulk error terms, we can estimate for any $\varepsilon>0$
    \[
    \left( \abs{\frac{1}{r} (r\nab_3\phi)\cdot(r\nab_4\phi)}+\abs{r\nab_3 \phi\cdot (\tilde{\mathcal{T}}-\mathcal{T}_l)\phi} \right)\tilde{w} \lesssim \varepsilon  \frac{1}{r} \abs{r\nab_3\phi}^2  \tilde{w}+ \frac{1}{\varepsilon} \frac{1}{r} \left(\abs{r\nab_4\phi}^2 +\abs{\phi}^2)\right) \tilde{w},
    \]
    using identities in the proof of Theorem~\ref{thm:comest}.
    Further, making use of the estimates \eqref{eq:tildew-w},
    \[
    \abs{r\trch - 2}\lesssim \epsilon_0
    \]
    and
    \[
    \abs{2 \omegab r  + f \frac{r\partial_r \tilde{w}}{\tilde{w}} - 2 \frac{r\partial_u \tilde{w}}{\tilde{w}}- \trchb r} \tilde{w} \lesssim \left( \frac{r}{\brk{u}} + 1 \right) w \left(\frac{\brk{u}}{r + \brk{u}}\right)^{2 p_{\infty} - 2 p_{\infty}'} \lesssim w,
    \]
    which follow from the bootstrap assumptions \eqref{eq:BA.weak.pointwise} and $p_{\infty}' \leqslant p_{\infty} - \frac{1}{2}$, together with
    \[
    \abs{[\nab_3, r\nab_A]\phi \cdot r\nab^A \phi}  \tilde{w} \lesssim \left(\frac{1}{\brk{u}} + \frac{1}{r}\right) \left[ \abs{r\nab_4\phi}^2 + \abs{r\nab\phi}^2 +\abs{\phi}^2\right] \tilde{w} \lesssim \frac{1}{r}\left[ \abs{r\nab_4\phi}^2 + \abs{r\nab\phi}^2 +\abs{\phi}^2\right] w
    \]
    which follows from \eqref{eq:comm-nab3rnab} in the proof of Theorem \ref{thm:comest}, \eqref{eq:tildew-w} and again $p_{\infty}' \leqslant p_{\infty} - \frac{1}{2}$, we deduce that
    \begin{align*}
    \mathcal{B}_2  &\lesssim
     \int_{S_T} \left[\frac{\epsilon_0+\varepsilon}{r} \abs{r\nab_3 \phi}^2 +\frac{r}{\varepsilon}\abs{\mathcal{T}_l\phi}^2\right]\tilde{w} \ud A_\gamma \ud r \ud u \\
    &\qquad +  \int_{S_T} \frac{1}{r}\left[ \abs{r\nab_4\phi}^2 + \abs{r\nab\phi}^2 +\abs{\phi}^2\right] w\ud A_\gamma \ud r \ud u.
    \end{align*} 
    The first term can be absorbed for $\epsilon_0, \varepsilon$ sufficiently small and we leave the remaining terms as they are after fixing $\varepsilon$. 
\end{proof}

\begin{corollary}\label{cor:BasicEst}
    Under the assumptions of the previous lemma we have     
    \begin{align*}
        &\sup_{u_T<u<T}\left[\left \|r \nab_4 \phi + 3 \phi \right\|_{0, p_0, p_\infty}^2(H_u) +\|\brk{u}^{p_{\infty} - p_{\infty}'}\phi\|_{1, p_0, p_\infty'}^2(H_u) \right]  \\
        & \qquad +  \|\phi\|_{1, p_0 + \frac{1}{2}, p_\infty-\frac{1}{2}}^2+  \|\brk{u}^{p_{\infty} - p_{\infty}'}\nab_3\phi\|_{0, p_0-\frac{1}{2}, p_\infty'+\frac{1}{2}}^2 +\abs{\phi}_{0, p_0, p''_\infty}^2\\
        &\qquad \lesssim \| \mathcal{T}_l \phi \|^2_{0,p_0, p_\infty,*} + \| \brk{u}^{\frac{1}{2}}\mathcal{T}_l \phi \|^2_{0,p_0-\frac{1}{2}, p_\infty}
        + \| (r \nab_{4} \phi, r \nab \phi, \phi)\|_{0, p_0, p_\infty}^2(\tilde{\Sigma}_1) \\
        &\qquad \qquad + \|\nab_3 \phi\|_{0, p_0-1, p_\infty+1}^2(\tilde{\Sigma}_1)+\abs{\phi}_{0, p_0, p_\infty''}^2(\tilde{\Sigma}_1)
    \end{align*}
    where 
    \[
    p''_\infty = \left \{ \begin{array}{ccc}
        3 , & \quad &p_\infty >3, \\
        p_\infty, && p_\infty <3.
    \end{array} \right.
    \]
\end{corollary}
\begin{proof}
    Combining the results of Lemmas \ref{lem:basicest}, \ref{lem:nab3est} we control all of the terms on the left in terms of the data terms on the right, after rewriting the integrals as norms. To control the norm of $\mathcal{T}_{l} \phi$ in Lemma~\ref{lem:nab3est}, we used \eqref{eq:tildew-w} and $p_{\infty} - p_{\infty}' \geqslant \frac{1}{2}$. It remains to establish control of $\abs{\phi}_{0, p_0, p_\infty''}$. For this, we observe that
    \[
    \nab_4 (r^3\phi) = r^{2} (r \nab_4 \phi + 3 \phi) =:F
    \]
    where we control $\sup_u \|F\|_{0, p_0+2, p_\infty-2}(H_u)$. We apply Theorem \ref{thm:transest} $i)$ with $\lambda_0 = 0$ to deduce that
    \[
    |r^3\phi|_{0, p_0+3, p_\infty''-3} \lesssim |r^3\phi|_{0, p_0+3, p_\infty''-3}(\tilde{\Sigma}_1)+\sup_u \|F\|_{0, p_0+2, p_\infty-2}(H_u)
    \]
    which gives the result.
\end{proof}

\begin{remark} By Lemma~\ref{lem:dualnormest}, we have two simple ways of estimating the two norms for $\mathcal{T}_{l} \phi$ in Corollary~\ref{cor:BasicEst}: for any $\delta' > 0$,
\begin{align} 
	\| \psi \|_{s,p_0, p_\infty,*} + \| \brk{u}^{\frac{1}{2}} \psi \|_{s, p_0-\frac{1}{2}, p_\infty} & \lesssim_{\delta'} 
	\| \brk{u}^{\frac{1}{2}+\delta'} \psi \|_{s, p_0-\frac{1}{2}, p_\infty}, \label{eq:forcingest0} \\
	\| \psi \|_{s,p_0, p_\infty,*} + \| \brk{u}^{\frac{1}{2}} \psi \|_{s, p_0-\frac{1}{2}, p_\infty} & \lesssim 
	\| \psi \|_{s, p_0-\frac{1}{2}, p_\infty+\frac{1}{2}}. \label{eq:forcingest1}
\end{align}
Estimate \eqref{eq:forcingest0}, which involves a weaker $r$-weight, will be used for all nonlinear terms. The virtue of estimate \eqref{eq:forcingest1} is that no powers of $\brk{u}$ or $r + \brk{u}$ are lost; it will be used to bound linear contributions. 
\end{remark}

\subsubsection{Higher order estimates}

For compactness of notation, let us write
\begin{align*}
\mathcal{N}_s[\phi] &:= \sup_{u_T<u<T}\left[\left \|r \nab_4 \phi + 3 \phi \right\|_{s, p_0, p_\infty}(H_u) +\|\brk{u}^{p_{\infty} - p_{\infty}'} r\nab \phi\|_{s, p_0, p_\infty'}(H_u)+\|\brk{u}^{p_{\infty} - p_{\infty}'} \phi\|_{s, p_0, p_\infty'}(H_u) \right]\\
&\qquad+\|r\nab_4\phi\|_{s, p_0+\frac{1}{2}, p_\infty-\frac{1}{2}}+\|r\nab \phi\|_{s, p_0+\frac{1}{2}, p_\infty-\frac{1}{2}}+  \|\phi\|_{s, p_0+\frac{1}{2}, p_\infty-\frac{1}{2}} \\
&\qquad+  \|\brk{u}^{p_{\infty} - p_{\infty}'} \nab_3\phi\|_{s,  p_0-\frac{1}{2}, p_\infty'+\frac{1}{2}} +\abs{\phi}_{s, p_0, p_{\i}''}.
\end{align*}
For the proofs in this subsection, we also fix some $\delta' \in (0, \frac{1}{2} \delta]$ and write $p^{\pm\prime} := p \pm \delta'$.
\begin{lemma}\label{lem:nabAcomm}
    Let $s\leqslant s_0-1$ and $p_\i \in [0,3 + \frac{1}{2} \delta ] \setminus \{3\}$. Assume that the weak bootstrap assumptions hold. Then
    \begin{align*}
    \left\| \mathcal{T}_{l}(r\nab \phi) - r\nab \mathcal{T}_l\phi \right\|_{s, p_0, p_\infty, *} + \left\| \brk{u}^{\frac{1}{2}} (\mathcal{T}_{l}(r\nab \phi) - r\nab \mathcal{T}_l\phi )\right\|_{s, p_0-\frac{1}{2}, p_\infty} &\lesssim \epsilon \mathcal{N}_{s+1}[\phi] + \mathcal{N}_{s}[\phi].
    \end{align*}
\end{lemma}
\begin{proof}
    We first require the identity
    \begin{align}
        \left[r\nab, \nab_3(r\nab_4 + 3) \right]\phi &= [r\nab, r\nab_4]\nab_3 \phi- [r\nab, \nab_4]\phi  + r\nab([\nab_3, r\nab_4]+ \nab_4) \phi \nonumber\\&\qquad   - ([\nab_3, r\nab_4]+ \nab_4)r\nab \phi + (r\nab_4+3) [r\nab, \nab_3]\phi\label{eq:nabAcommPr1}
    \end{align}
    which can be established by expanding out the commutators on the right hand side. Taking the terms on the first line one at a time, using the estimates of Theorem \ref{thm:comestST}, together with the remark following it, and the weak bootstrap assumptions, we have
    \begin{align*}
        \|\brk{u}^{\frac{1}{2}^{+\prime}}[r\nab, r\nab_4]\nab_3 \phi \|_{s, p_0-\frac{1}{2}, p_\infty} &\lesssim \epsilon \|\brk{u}^{\frac{1}{2}^{+\prime}}\nab_3 \phi\|_{s+1, p_0^{-}-\f 32, p_\infty^-} \lesssim \epsilon \mathcal{N}_{s+1}[\phi], \\
        \|\brk{u}^{\frac{1}{2}^{+\prime}} [r\nab, \nab_4]\phi\|_{s, p_0-\frac{1}{2}, p_\infty} &\lesssim \epsilon \|\brk{u}^{\frac{1}{2}^{+\prime}} \phi\|_{s+1, p_0^{-}-\frac{1}{2}, p_\infty^{-}-1} \lesssim \epsilon \mathcal{N}_{s+1}[\phi], \\
        \|\brk{u}^{\frac{1}{2}^{+\prime}} r\nab([\nab_3, r\nab_4]+ \nab_4) \phi\|_{s, p_0-\frac{1}{2}, p_\infty} &\lesssim \|\brk{u}^{\frac{1}{2}^{+\prime}} ([\nab_3, r\nab_4]+ \nab_4) \phi\|_{s+1, p_0-\frac{1}{2}, p_\infty} \\
        & \lesssim \epsilon (\|\brk{u}^{\frac{1}{2}-(\delta-\delta')} r\nab_4\phi\|_{s+1, p_0^{-}-\frac{1}{2}, p_\infty-1}\\
        &\qquad+\|\brk{u}^{\frac{1}{2}-(\delta-\delta')} r\nab\phi\|_{s+1, p_0^{-}-\frac{1}{2}, p_\infty-1}
        \\&\qquad +\|\brk{u}^{\frac{1}{2}-(\delta-\delta')} \phi\|_{s+1, p_0^{-}-\frac{1}{2}, p_\infty-1}+| \brk{u}^{0^{+\prime}} \phi|_{s-1, p_0^{-}-1, p_\infty^-} )\\&\lesssim \epsilon \mathcal{N}_{s+1}[\phi], \\
        \|\brk{u}^{\frac{1}{2}^{+\prime}} ([\nab_3, r\nab_4]+ \nab_4) r\nab\phi\|_{s, p_0-\frac{1}{2}, p_\infty} & \lesssim \epsilon \| \jb{u}^{\frac{1}{2}} r\nab \phi\|_{s+1, p_0^{-}-\frac{1}{2}, p_\infty-1}\lesssim \epsilon \mathcal{N}_{s+1}[\phi] .
    \end{align*}
    Here, we have used $p_{\infty}-(\delta-\delta') \leqslant p_{\infty}''$ in our range of $p_{\infty}$ to bound $| \brk{u}^{0^{+\prime}} \phi|_{s-1, p_0^{-}-1, p_\infty^-} \aleq | \phi|_{s-1, p_0^{-}-1, p_\infty''}$.
    
    For the final term of \eqref{eq:nabAcommPr1} we expand out to find schematically
    \begin{align}
            (r\nab_4+3)[ \nab_3, r\nab]\phi  & \sim \xib (r\nab_4+3) r\nab_4 \phi  + (r\nab_4 \xib) r\nab_4 \phi+\xib(r\nab_4+3)  \phi  + (r\nab_4 \xib)  \phi \nonumber\\
            &\qquad + (r\nab_4+3)\Bigg [\chibh \cdot r\nab \phi + (\trchb + 2r^{-1}) r \nab \phi + r\chibh \cdot \eta \cdot \phi+ r\chih \cdot \xib \cdot \phi \label{eq:nabAcommPr2}\\
            &\qquad +r(\trchb + 2r^{-1}) \eta \cdot \phi+r(\trch - 2r^{-1}) \xib \cdot \phi + \eta \cdot \phi+ r \betab \cdot \phi \Bigg]. \nonumber
    \end{align}
    
	To control the first term in \eqref{eq:nabAcommPr2}, we first note that since $\xib = \nab f$,
    \begin{align*}
	\left( \int_{u_T}^T \left( \brk{u}^{\frac{1}{2}^{+}}  |\xib|_{s_0-1, \frac{1}{2}^+, 0}(H_u) \right)^{2} \ud u \right)^{\frac{1}{2}}
	&\leqslant \left( \int_{u_{T}}^{T} \left( \brk{u}^{\frac{1}{2}^{+}}\abs{f - 1}_{s_{0}, \frac{3}{2}^{+}, -1}(H_{u}) \right)^{2} \ud u \right)^{\frac{1}{2}}	\aleq \epsilon,
	\end{align*}
	so we can estimate
\begin{align*}
    \| \brk{u}^{\frac{1}{2}^{+}} \xib (r\nab_4+3) r\nab_4 \phi\|_{s,p_0-\frac{1}{2}, p_\infty} \lesssim \epsilon \sup_{u_T<u<T}\left \|r \nab_4 \phi + 3 \phi \right\|_{s+1, p_0, p_\infty}(H_u) \lesssim \epsilon \mathcal{N}_{s+1}[\phi].
\end{align*}

    To control the second term in \eqref{eq:nabAcommPr2} we use the following identity which follows from $\xib = \nab f, \nab_4f = 2\omegab$
    \[
    r\nab_4 \xib = ([\nab_4, r\nab] - \nab)(f-2\omegab r +1) + 2 r [\nab_4, r\nab]\omegab,
    \]
    to deduce that under the weak bootstrap assumptions we have 
\begin{equation*}
    \|\brk{u}^{\frac{1}{2}+\frac{1}{2}\delta} r\nab_4 \xib\|_{s_0-1, \frac{1}{2}^+, \frac{1}{2}\delta}\lesssim \epsilon. 
\end{equation*}
    Hence, using the spacetime integral of \eqref{eq:product.for.pointwise}, 
    \begin{align*}
\|\brk{u}^{\frac{1}{2}^{+\prime}} (r\nab_4 \xib) r\nab_4 \phi\|_{s, p_0-\frac{1}{2}, p_\infty} &\lesssim \|\brk{u}^{\frac{1}{2}+\frac{1}{2}\delta} r\nab_4 \xib\|_{s_0-1, \frac{1}{2}^+, \frac{1}{2}\delta} |\brk{u}^{-(\frac{1}{2}\delta - \delta')}r\nab_4 \phi|_{s, p_0^--1, p_\infty-\frac{1}{2}\delta} 
    \lesssim \epsilon \mathcal{N}_{s+1}[\phi] .
\end{align*}
Here, we have used $p_{\infty}-\frac{1}{2}\delta \leqslant p_{\infty}''$ in our range of $p_{\infty}$ to bound $|r\nab_4 \phi|_{s, p_0^--1, p_\infty-\frac{1}{2}\delta} \aleq | \phi|_{s+1, p_0^{-}-1, p_\infty''}$.
    The next two terms of \eqref{eq:nabAcommPr2} can be bounded in exactly the same way as the first two.
    
    The remaining terms can all be bounded in the $\| \brk{u}^{\frac{1}{2}^{+\prime}} (\cdot) \|_{s, p_0-\frac{1}{2}, p_\infty}$ norm by a multiple of $\epsilon\mathcal{N}_{s+1}[\phi] $ after making use of the various transport equations. We have thus established
    \[
    \left\| \brk{u}^{\frac{1}{2}^{+\prime}} \left[r\nab, \nab_3(r\nab_4 + 3) \right]\phi\right\|_{s, p_0-\frac{1}{2}, p_\infty} \lesssim \epsilon \mathcal{N}_{s+1}[\phi] .
    \]
    To complete the proof, note that $r \nab$ commutes with $\frac{1}{r}$, and the commutator $[r \nab, \nab_{4}]$ has already been treated. Finally, to treat $[r\nab, r\slashed \Delta]$, we first observe the schematic expressions 
    \begin{align*}
	[r\nab, r\slashed \Delta]\phi \sim \nab [r \nab, r \nab] \phi + [r \nab, r \nab] \nab \phi \sim \nab \phi + r^{2} (K - r^{-2}) \nab \phi + r^{2} \nab (K - r^{-2}) \phi
\end{align*}
Using \eqref{eq:forcingest1}, we estimate the first term as follows:
    \begin{align*}
& \|\nab \phi\|_{s,p_0-\frac{1}{2}, p_\infty+\frac{1}{2}} \lesssim\|\phi\|_{s+1,p_0+\frac{1}{2}, p_\infty-\frac{1}{2}} \lesssim  \mathcal{N}_{s}[\phi].
    \end{align*}
    For the other two terms, we use \eqref{eq:forcingest0} and proceed as follows:
    \begin{align*}
	&\| \brk{u}^{\frac{1}{2}^{+\prime}} (r^{2} (K - r^{-2}) \nab \phi + r^{2} \nab (K - r^{-2}) \phi) \|_{s, p_{0}-\frac{1}{2}, p_{\infty}} \\
	& \qquad \aleq  \| \brk{u}^{\frac{1}{2}} r^{2} (K - r^{-2}) \|_{s_{0}, \frac{3}{2}^{+}, (-1)^{+}} |\brk{u}^{0^{+\prime}} \phi|_{s-2, p_{0}^{-}-1, p_{\infty}^{-}} \\
	&\qquad \qquad + |\brk{u}^{0^{+}} r^{2} (K - r^{-2}) |_{s-2, 1^{+}, 0} \| \brk{u}^{\frac{1}{2}-(\delta-\delta')} \phi \|_{s+1, p_{0}^{-}-\frac{1}{2}, p_{\infty}-1}  \lesssim \epsilon \mathcal{N}_{s+1}[\phi]. \qedhere
\end{align*}
\end{proof}

\begin{lemma}\label{lem:nab4comm}
    Let $s\leqslant s_0-1$ and $p_\i \in [0,3+\frac{1}{2}\delta] \setminus \{3\}$. Assume the weak bootstrap assumptions hold. Then
    \begin{align*}
    \left\|{\mathcal{T}_{l+1}(r\nab_4 \phi) - r\nab_4 \mathcal{T}_l\phi}\right\|_{s, p_0, p_\infty, *}
    +\left\| \brk{u}^{\frac{1}{2}} (\mathcal{T}_{l+1}(r\nab_4 \phi) - r\nab_4 \mathcal{T}_l\phi)\right\|_{s, p_0-\frac{1}{2}, p_\infty}&\lesssim \epsilon \mathcal{N}_{s+1}[\phi]+ \mathcal{N}_{s}[r\nab\phi] + \mathcal{N}_{s}[\phi].
    \end{align*}
\end{lemma}

\begin{proof}
    We begin with the identity
    \begin{align*}
        &r\nab_4\left(\nab_3+\frac{l}{r}\right)(r\nab_4+3)\phi - \left(\nab_3+\frac{l+1}{r}\right)(r\nab_4+3)r\nab_4\phi\\
        &\qquad \qquad= ([r\nab_4, \nab_3] - \nab_4) (r\nab_4+3)\phi -\frac{l}{r} (r\nab_4+3)\phi.
    \end{align*}
    For the first term on the last line, we use \eqref{eq:forcingest0} and estimate
    \begin{align*}
        &\| \brk{u}^{\frac{1}{2}^{+\prime}}([r\nab_4, \nab_3] - \nab_4) (r\nab_4+3)\phi \|_{s, p_{0}-\frac{1}{2}, p_{\infty}} \\
        &\qquad \aleq \epsilon (\|\brk{u}^{\frac{1}{2}-(\delta-\delta')} (r \nab_{4} \phi, r \nab \phi, \phi) \|_{s+1, p_{0}^{-}-\frac{1}{2}, p_{\infty}-1} + |\brk{u}^{0^{-\prime}} (r \nab_{4} \phi, r \nab \phi, \phi) |_{s-2, p_{0}^{-}-1, p_{\infty}^{-}}) \\
        &\qquad \aleq \epsilon \mathcal{N}_{s+1}[\phi],
    \end{align*}
    where we used Theorem~\ref{thm:comestST} and the remark following it for the first inequality, $p_{\infty} - (\delta - \delta') \leqslant p_{\infty}''$ for the second. For the other term, use \eqref{eq:forcingest1} and bound
    \begin{equation} \label{eq:lem:nab4comm-1}
        \| r^{-1} (r \nab_{4} \phi, \phi) \|_{s, p_{0}-\frac{1}{2}, p_{\infty}+\frac{1}{2}} \aleq \mathcal{N}_{s}[\phi].
    \end{equation}
    Next, we observe that
    \begin{align*}
    [r\nab_4, 2\nab_4 + \frac{4}{r}] \phi &= - 2\nab_4\phi - \frac{4}{r}\phi, \\
    [r\nab_4, \frac{1}{r} r^2\slashed \Delta]\phi &= - \frac{1}{r} r^2\slashed \Delta\phi + \frac{1}{r} [r\nab_4, r\nab^A]r\nab_A \phi + \frac{1}{r} r\nab_A [r\nab_4, r\nab^A] \phi.
    \end{align*}
    The contribution of $[r\nab_4, 2\nab_4 + \frac{4}{r}] \phi$ is handled again using \eqref{eq:lem:nab4comm-1}. For $r \slashed \Delta\phi$, we use
    \begin{align*}
        \| r \slashed{\Delta} \phi \|_{s, p_{0}-\frac{1}{2}, p_{\infty}+\frac{1}{2}} \aleq \mathcal{N}_{s}[r \nab \phi].
    \end{align*}
    Finally, for the terms involving $[r \nab_{4}, r \nab_{A}]$, we use \eqref{eq:forcingest0}. Again by Theorem~\ref{thm:comestST} and the remark after it,
    \begin{align*}
        &\| \brk{u}^{\frac{1}{2}^{+\prime}} r^{-1}([r\nab_4, r\nab^A]r\nab_A \phi + r\nab_A [r\nab_4, r\nab^A] \phi)\|_{s, p_{0}-\frac{1}{2}, p_{\infty}} \\
        &\qquad \aleq \epsilon (\|\brk{u}^{\frac{1}{2}^{+\prime}} (r \nab_{4} \phi, r \nab \phi, \phi) \|_{s+1, p_{0}^{-}-\frac{1}{2}, p_{\infty}^--1} + |\brk{u}^{0^{-\prime}} (r \nab_{4} \phi, r \nab \phi, \phi) |_{s-2, p_{0}^{-}-1, p_{\infty}^{-}}) \\
        &\qquad \aleq \epsilon \mathcal{N}_{s+1}[\phi]. \qedhere
    \end{align*}
\end{proof} 

\begin{lemma}\label{lem:nabucomm}
    Let $s\leqslant s_0-1$ and $p_\i \in [0,3+\frac{1}{2}\delta] \setminus \{3\}$. Assume that the weak bootstrap assumptions hold. Then
    \begin{align*}
    & \left\|{ \mathcal{T}_{l}(\brk{u} \nab_u \phi) - (\brk{u} \nab_u + \tfrac{u}{\brk{u}})\mathcal{T}_l\phi}\right\|_{s-1,p_0, p_\infty, *} \\
    &\qquad + \|{\jb{u}^{\frac{1}{2}} (\mathcal{T}_{l}(\brk{u} \nab_u \phi) - (\brk{u} \nab_u + \tfrac{u}{\brk{u}})\mathcal{T}_l\phi})\|_{s-1, p_0-\frac{1}{2}, p_\infty'+\frac{1}{2}} \lesssim \epsilon \mathcal{N}_{s+1}[\phi] + \mathcal{N}_{s}[\phi].
    \end{align*}
\end{lemma}

\begin{proof}
    We first write out the expression on the left hand side schematically as follows:
\begin{equation} \label{eq:lem:nabucommid-main}
    \begin{aligned}
	&(\brk{u} \nab_u + \tfrac{u}{\brk{u}})\mathcal{T}_l\phi - \mathcal{T}_{l}(\brk{u} \nab_u \phi) \\
	&\qquad = [\brk{u} \nab_{u}, \mathcal{T}_{l}] \phi + \tfrac{u}{\brk{u}} \mathcal{T}_{l} \phi \\
	&\qquad \sim \left(\left[ \brk{u} \nab_{u}, \left(\nab_3+\frac{l}{r}\right)(r\nab_4 + 3 ) \right] + \tfrac{u}{\brk{u}} \nab_{3} (r \nab_{4} + 3) \right) \phi \\
	&\qquad \qquad + [\brk{u} \nab_{u}, \nab_{4}] \phi + [\brk{u} \nab_{u}, r \slashed{\Delta}] \phi + \tfrac{u}{\brk{u}}\nab_{4} \phi + \tfrac{u}{\brk{u}} r^{-1} \phi + \tfrac{u}{\brk{u}} r \slashed{\Delta} \phi.
    \end{aligned}
\end{equation}
    We first require the identity
\begin{equation} \label{eq:lem:nabucommid}
    \begin{aligned}
        & \left[\brk{u} \nab_u, \left(\nab_3+\frac{l}{r}\right)(r\nab_4 + 3 )\right]\phi + \tfrac{u}{\brk{u}} \nab_{3} (r \nab_{4} + 3) \phi \\
        &= ([\brk{u} \nab_u, \nab_3] + \tfrac{u}{\brk{u}} \nab_{3}) (r\nab_4 \phi + 3 \phi ) + \nab_{3} [\brk{u} \nab_u, r\nab_4]\phi
        + \brk{u} r^{-1} l [\nab_u, r \nab_4]\phi. 
    \end{aligned}
\end{equation}
    Let us consider the first term on the right hand side of \eqref{eq:lem:nabucommid}. Since $\nab_{3} = 2 \nab_{u} - f \nab_{4}$, we have 
    \begin{align*}
	[\brk{u} \nab_{u}, \nab_{3}] + \tfrac{u}{\brk{u}} \nab_{3} &= - \brk{u} [\nab_{u}, f \nab_{4}] + [\brk{u}, 2 \nab_{u}] \nab_{u} + 2 \tfrac{u}{\brk{u}} \nab_{u} - \tfrac{u}{\brk{u}} f \nab_{4} \\
	&= - (\brk{u} \nab_{u} f) \nab_{4} -  \brk{u} r^{-1}  f [\nab_{u}, r \nab_{4}] - \tfrac{u}{\brk{u}} f \nab_{4},
\end{align*}
where we notice that the terms involving $\tfrac{u}{\brk{u}} \nab_{u}$ cancel. In fact, the purpose of adding $\mathcal{T}_{l} \phi$ to the commutator was exactly to enforce this cancellation.

Recalling \eqref{eq:comm-naburnab4}, we have, schematically, 
\begin{align*}
	([\brk{u} \nab_u, \nab_3] + \nab_{3}) (r\nab_4 \phi + 3 \phi ) 
	& \sim (\brk{u} \nab_u f) \nab_4 (r\nab_4 \phi + 3 \phi ) + \brk{u} f (\eta\cdot \nab)(r\nab_4 \phi + 3 \phi )  \\
	& \qquad + \brk{u} f \sigma \cdot (r\nab_4 \phi + 3 \phi ) + \tfrac{u}{\brk{u}} f \nab_{4} (r \nab_{4} \phi + 3 \phi).
\end{align*}
	The first term may be dealt with in exactly the same fashion as the first term of \eqref{eq:nabAcommPr2} in Lemma \ref{lem:nabAcomm}, noting that $r^{-1} \brk{u} \nab_u f$ satisfies the same bounds as $\xib$, at one lower order of differentiability. The fourth term is similar but easier, since $r^{-1} f$ satisfies the same bounds as $\xib$ at a higher order of differentiability. For the second and third terms, we have
	\begin{align*}
	 &\nrm{\brk{u}^{\frac{1}{2}^{+\prime}} (\brk{u} f (\eta\cdot \nab)(r\nab_4 \phi + 3 \phi ) + \brk{u} f \sigma \cdot (r\nab_4 \phi + 3 \phi ))}_{s-1, p_{0}-\frac{1}{2}, p_{\infty}}	\\
	 &\qquad \aleq \epsilon (\left \| r\nab \phi \right\|_{s,p_0 +\frac{1}{2}, p_\infty-\frac{1}{2}}+\left \| \phi \right\|_{s, p_0 + \frac{1}{2}, p_\infty-\frac{1}{2}}) \lesssim \epsilon \mathcal{N}_{s}[\phi].
	\end{align*}
	
	We now turn to the second term on the right hand side of \eqref{eq:lem:nabucommid}. By \eqref{eq:comm-naburnab4}, we have, schematically
\begin{align*}
    \nab_3([\brk{u} \nab_u, r\nab_4]\phi) &\sim \brk{u} (\nab_3\eta) \cdot r\nab\phi + \eta \cdot r\nab\phi + \brk{u} \eta\cdot  r\nab (\nab_3\phi)+\brk{u} \eta\cdot [\nab_3, r\nab]\phi \\
    &\qquad + \brk{u} \nab_3(r\sigma) \phi + r\sigma \phi + \brk{u} r \sigma \nab_3 \phi
\end{align*}
    so that, by the product estimates in Lemma~\ref{lem:products}, we may estimate
    \begin{align*}
    &\| \brk{u}^{\frac{1}{2}^{+\prime}} \nab_3([\brk{u} \nab_u, r\nab_4]\phi) \|_{s-1, p_0-\frac{1}{2}, p_\i} \\
    &\qquad \lesssim \epsilon\Big(  \| \brk{u}^{\frac{1}{2}^{+\prime}} (r \nab_{4} \phi , r \nab \phi, \phi) \|_{s,p_0^{-}-\frac{1}{2}, p_{\infty}^{-}-1} + \| \brk{u}^{\frac{3}{2}^{+\prime}} \nab_3\phi \|_{s, p_0^{-}-\frac{1}{2}, p_\i^{-}-1}+|\brk{u}^{\frac{1}{2}^{+\prime}} \phi|_{s, p_0^{-}, p_\infty^- - \frac{1}{2}}\Big) \lesssim \epsilon \mathcal{N}_{s}[\phi].
    \end{align*}
    We remark in particular that controlling $\nab_3(r\sigma)$ requires using the top order integrated norm, which explains the need for using $|\brk{u}^{\frac{1}{2}^{+\prime }} \phi|_{s, p_0^{-}, p_\infty^- - \frac{1}{2}}$. Since $p_{\infty}-(\delta-\delta') \leqslant p_{\infty}''$ in our range of $p_{\infty}$, this term can be bounded by $|\phi|_{s, p_0, p_\infty''}$.

	The last term in \eqref{eq:lem:nabucommid} can be bounded in $\nrm{\brk{u}^{\frac{1}{2}^{+\prime}}(\cdot)}_{s-1, p_{0}-\frac{1}{2}, p_{\infty}}$ using \eqref{eq:comm-naburnab4} as follows:
\begin{align*}
	\nrm{\brk{u}^{\frac{1}{2}^{+\prime}} \brk{u} r^{-1} [\nab_{u}, r \nab_{4}] \phi}_{s-1, p_{0}-\frac{1}{2}, p_{\infty}} & \aleq \nrm{\brk{u}^{\frac{1}{2}^{+\prime}}[\nab_{u}, r \nab_{4}] \phi}_{s-1, p_{0}^{+}+\frac{1}{2}, p_{\infty}} \\
	&\aleq \nrm{\brk{u}^{\frac{1}{2}^{+\prime}} \eta \cdot r \nab \phi}_{s-1, p_{0}^{+}+\frac{1}{2}, p_{\infty}}
	+ \nrm{\brk{u}^{\frac{1}{2}^{+\prime}} r \sigma \phi}_{s-1, p_{0}^{+}+\frac{1}{2}, p_{\infty}} \\
	&\aleq \epsilon (\left \| r\nab \phi \right\|_{s-1,p_0 +\frac{1}{2}, p_\infty-\frac{1}{2}}+\left \| \phi \right\|_{s-1, p_0 + \frac{1}{2}, p_\infty-\frac{1}{2}}) \lesssim \epsilon \mathcal{N}_{s}[\phi].
\end{align*}

	Returning to \eqref{eq:lem:nabucommid-main}, the preceding estimate also handles $[\brk{u} \nab_{u}, \nab_{4}] \phi = \brk{u} r^{-1} [\nab_{u}, r \nab_{4}] \phi$. For $[\brk{u} \nab_{u}, r \slashed{\Delta}] \phi$, we have
    \[
    [\brk{u} \nab_u, r \slashed \Delta]\phi = \frac{\brk{u}}{r} [\nab_u, r\nab^A]r\nab_A \phi + \frac{\brk{u}}{r} r\nab_A [\nab_u, r\nab^A] \phi.
    \]
    We can estimate using \eqref{uAcommestST} as follows:
    \begin{align*}
    \left\|[\brk{u} \nab_u, r \slashed \Delta]\phi \right\|_{s-1, p_0-\frac{1}{2}, p_\infty +\frac{1}{2}} &\lesssim \epsilon \Big (  \left \| \brk{u}^{0^{-}} r\nab_4 \phi \right\|_{s,  p_0+\frac{1}{2}, p_\infty -\frac{1}{2}} + \left \| \brk{u}^{0^{-}} r\nab \phi \right\|_{s, p_0+\frac{1}{2}, p_\infty-\frac{1}{2}}\\&\qquad  +\left \| \brk{u}^{0^{-}}\phi \right\|_{s,p_0+\frac{1}{2}, p_\i-\frac{1}{2}} +  \left | \brk{u}^{\frac{1}{2}-\frac{3}{2}\delta} \phi \right |_{s-2, p_0^{-}+\frac{1}{2}, p_\infty-\frac{1}{2}+\frac{1}{2}\delta} \Big)\lesssim \epsilon \mathcal{N}_{s}[\phi].
    \end{align*}
    Finally, we bound the last three terms in \eqref{eq:lem:nabucommid-main} using \eqref{eq:forcingest1} as follows:
    \begin{equation*}
	\| (\nab_{4} \phi, r^{-1} \phi, r \slashed{\Delta} \phi) \|_{s-1, p_{0}-\frac{1}{2}, p_{\infty}+\frac{1}{2}}
	\aleq \| (r \nab \phi, \phi) \|_{s, p_{0}+\frac{1}{2}, p_{\infty}-\frac{1}{2}} \aleq \mathcal{N}_{s}[\phi]. 
	\qedhere
\end{equation*}
\end{proof}

\begin{theorem}\label{thm:main.Teukolsky}
    Fix $0<\delta < 1/20$ and let $p_0 = 1+\delta, p_\infty \in[0, 3+\frac{1}{2}\delta] \setminus \{ 3 \}$.  Suppose the weak bootstrap assumptions hold. Assume that $\phi$ is a sufficiently regular $S$-tangent tensor satisfying the regularity condition at the axis:
    \[
    \lim_{r \to 0}\sup_{1-r<u<T} r^{-p_0 - \frac{3}{2}} \left(|\phi|_{s+1}+ |r\nab_3 \phi|_s\right) (u, r) = 0.
    \]
    Then for $\epsilon_0$ sufficiently small, $s\leqslant s_0$ and any $l\geqslant 0$ we have the estimate 
    \begin{align*}
        &\sup_{u_T<u<T}\left[\left \|r \nab_4 \phi + 3 \phi \right\|_{s, p_0, p_\infty}(H_u)+\|\brk{u}^{p_{\infty}-p_{\infty}'} r\nab \phi\|_{s, p_0, p_\infty'}(H_u) +\|\brk{u}^{p_{\infty}-p_{\infty}'}\phi\|_{s, p_0, p_\infty'}(H_u) \right]\\&\qquad   +  \|(r\nab_4 \phi, r \nab \phi, \phi)\|_{s, p_0+\frac{1}{2}, p_\infty- \frac{1}{2}}+  \|\brk{u}^{p_{\infty}-p_{\infty}'}\nab_3\phi\|_{s, p_0+\frac{1}{2}, p_\infty'+ \frac{1}{2}} +\abs{\phi}_{s, p_0, p_\infty''}\\
        &\qquad\qquad \lesssim \left \| \mathcal{T}_l \phi\right\|_{s,p_0, p_\infty*} + \| \brk{u}^{\frac{1}{2}} \mathcal{T}_l \phi \|_{s, p_0-\frac{1}{2}, p_\infty} \\
        &\qquad \qquad \qquad  + \|\phi\|_{s+1, p_0, p_\infty}(\tilde{\Sigma}_1)+ \|\nab_3 \phi\|_{s, p_0-1, p_\infty+1}(\tilde{\Sigma}_1)  +\abs{\phi}_{s, p_0, p_\infty''}(\tilde{\Sigma}_1).
    \end{align*}
\end{theorem}
\begin{proof}
    Let
    \[
    \mathcal{D}_s[\phi] =  \|\phi\|_{s+1, p_0, p_\infty}(\tilde{\Sigma}_1)+ \|\nab_3 \phi\|_{s, p_0-1, p_\infty'+1}(\tilde{\Sigma}_1)  +\abs{\phi}_{s, p_0, p_\infty''}(\tilde{\Sigma}_1).
    \]
    We wish to show that
    \[
    \mathcal{N}_s[\phi] \lesssim \left \| \mathcal{T}_l \phi\right\|_{s,p_0, p_\infty*}+ \| \brk{u}^{\frac{1}{2}} \mathcal{T}_l \phi \|_{s, p_0-\frac{1}{2}, p_\infty} +\mathcal{D}_s[\phi].
    \]
    We work by induction on $s$. Corollary \ref{cor:BasicEst} gives the base case. Suppose the result has been established for some $s\leqslant s_0-1$. We apply the induction hypothesis and Lemma \ref{lem:nabAcomm} to deduce that
    \begin{align*}
    \mathcal{N}_s[r\nab\phi] &\lesssim \left \| \mathcal{T}_l r\nab\phi\right\|_{s,p_0, p_\infty, *}+ \|\brk{u}^{\frac{1}{2}} \mathcal{T}_l r\nab\phi \|_{s, p_0-\frac{1}{2}, p_\infty} +\mathcal{D}_s[r\nab\phi] \\
    &\lesssim \left \| \mathcal{T}_l \phi\right\|_{s+1,p_0, p_\infty, *} \|\brk{u}^{\frac{1}{2}} \mathcal{T}_l \phi \|_{s+1, p_0-\frac{1}{2}, p_\infty} +\mathcal{D}_{s+1}[\phi]\\&\qquad + \left\|{\mathcal{T}_{l}(r\nab \phi) - r\nab \mathcal{T}_l\phi}\right\|_{s,p_0, p_\infty, *}+ \|\brk{u}^{\frac{1}{2}} (\mathcal{T}_{l}(r\nab \phi) - r\nab \mathcal{T}_l\phi) \|_{s, p_0-\frac{1}{2}, p_\infty} \\
    &\lesssim \epsilon \mathcal{N}_{s+1}[\phi]  + \left \| \mathcal{T}_l \phi\right\|_{s+1,p_0, p_\infty, *} +\|\brk{u}^{\frac{1}{2}} \mathcal{T}_l \phi \|_{s+1, p_0-\frac{1}{2}, p_\infty} +\mathcal{D}_{s+1}[\phi]
    \end{align*}
    where on the final line we have used the induction hypothesis to control the $\mathcal{N}_{s}[\phi]$ arising from the commutator. An identical argument using Lemma \ref{lem:nab4comm} gives
    \[
    \mathcal{N}_s[r\nab_4\phi] \lesssim \epsilon \mathcal{N}_{s+1}[\phi] + \mathcal{N}_s[r\nab \phi] + \left \| \mathcal{T}_l \phi\right\|_{s+1,p_0, p_\infty, *} +\|\brk{u}^{\frac{1}{2}} \mathcal{T}_l \phi \|_{s+1, p_0-\frac{1}{2}, p_\infty} +\mathcal{D}_{s+1}[\phi],
    \]
    and (for $s\geqslant 1$) from Lemma \ref{lem:nabucomm}:
    \[
    \mathcal{N}_{s-1}[\brk{u} \nab_3\phi] \lesssim \epsilon \mathcal{N}_{s+1}[\phi] + \left \| \mathcal{T}_l \phi\right\|_{s+1,p_0, p_\infty, *} +\|\brk{u}^{\frac{1}{2}} \mathcal{T}_l \phi \|_{s+1, p_0-\frac{1}{2}, p_\infty} +\mathcal{D}_{s+1}[\phi].
    \]
    Adding a sufficiently large multiple of the first inequality to the second we can absorb the error terms on the right hand side, and adding the third we have established that
    \[
    \mathcal{N}_s[r\nab \phi]+\mathcal{N}_s[r\nab_4 \phi]+\mathcal{N}_{s-1}[\brk{u} \nab_u \phi] \lesssim  \left \| \mathcal{T}_l \phi\right\|_{s+1,p_0, p_\infty, *} +\|\brk{u}^{\frac{1}{2}} \mathcal{T}_l \phi \|_{s+1, p_0-\frac{1}{2}, p_\infty} +\mathcal{D}_{s+1}[\phi].\qedhere
    \]
\end{proof}

Applying Theorem~\ref{thm:main.Teukolsky} with $\phi = r^2 \alp$, $p_0 = 1^+$, $p_\i = \begin{cases}
    0^+ & \hbox{under weak bootstrap} \\
    3^{-} & \hbox{under strong bootstrap} 
\end{cases}$, and using the initial data assumptions in \eqref{eq:data.alp.weak}--\eqref{eq:data.alp.strong}, the estimates in Proposition~\ref{prop:Teukolsky.nonlinear.control} to control the inhomogeneous term, the estimates in Corollary \ref{cor:loc.exist.bounds} to establish the regularity condition at the axis holds, and the estimate \eqref{eq:equivalence-integrated} for equivalence of norms (observe that the pointwise norm in \eqref{eq:equivalence-integrated} can be absorbed by taking $\epsilon_{0}$ small), we obtain the following estimates.

\begin{corollary}\label{cor:alpha.final}
    There exists $\ep_0, \ep'_0 >0$ sufficiently small depending only on $s_0$ and $\de$ such that the following holds for some implicit constants depending only on $s_0$ and $\de$. 
    
    Suppose that for some $\ep \leqslant \ep_0$, the weak bootstrap assumptions hold, together with the weak bounds on the initial conditions. Then
    \begin{align*}
        &\sup_{u_T<u<T}\left[\left \|r \nab_4 \alpha + 5 \alpha \right\|_{s_0, (-1)^+, 2^{+}}(H_u)+\|\jb{u}^{\frac{1}{2}} \alpha\|_{s_0+1, (-1)^+, \frac{3}{2}^{+}}(H_u) \right]\\&\qquad   +   \|\alpha\|_{s_0+1, (-\frac{1}{2})^+, \frac{3}{2}^+}+  \|\jb{u}^{\frac{1}{2}} \nab_3(r^2\alpha)\|_{s_0, \frac{3}{2}^+, 0^+} +\abs{\alpha}_{s_0, (-1)^+, 2^+} \lesssim \epsilon^2.
    \end{align*}
    If for some $\ep \leqslant \ep'_0$, the strong bootstrap assumptions hold together with the strong bounds on the initial conditions, then
    \begin{align*}
        &\sup_{u_T<u<T}\left[\left \|r \nab_4 \alpha + 5 \alpha \right\|_{s_0, (-1)^+, 5^{-}}(H_u)+\|\brk{u}^{\frac{1}{2}} \alpha\|_{s_0+1, (-1)^+, \frac{9}{2}^-}(H_u) \right]\\&\qquad   +   \|\alpha\|_{s_0+1, (-\frac{1}{2})^+, \frac{9}{2}^{-}}+  \|\brk{u}^{\frac{1}{2}} \nab_3(r^2\alpha)\|_{s_0, \frac{3}{2}^+, 3^-} +\abs{\alpha}_{s_0, (-1)^{+}, 5} \lesssim \epsilon^2.
    \end{align*}
\end{corollary}

\begin{remark}[Propagation of Bondi--Sachs peeling, I] \label{rem:strong-peeling-rp}
The purpose of this remark is to give more details of the propagation of Bondi--Sachs peeling in the context of spacelike-characteristic initial value problem alluded to in Remark~\ref{rem:strong-peeling-spnull}. Note that we are also allowed to choose a value of $p$ that exceeds $3$, e.g., $p = 3 + \frac{1}{2} \delta$. Using the notation in Appendix~\ref{sec:appendix}, if we adapt the argument to the future $M$ of $\Sigma \cup H_{0}$, we may obtain
\begin{equation}\label{eq:rp.verystrong}
    \begin{split}
        &\sup_{u}\left[\left \|r \nab_4 \alpha + 5 \alpha \right\|_{s_0, (-1)^+, 5+\frac{1}{2}\delta}(H_u)+\|\brk{u}^{\frac{1}{2}+\frac{3}{2}\delta} \alpha\|_{s_0+1, (-1)^+, \frac{9}{2}^-}(H_u) \right]\\&\qquad   +   \|\alpha\|_{s_0+1, (-\frac{1}{2})^+, \frac{9}{2}+\frac{1}{2}\delta}+  \|\brk{u}^{\frac{1}{2}+\frac{3}{2}\delta} \nab_3(r^2\alpha)\|_{s_0, \frac{3}{2}^+, 3^-} +\abs{\alpha}_{s_0, 1^+, 5} \lesssim \mathfrak n_{\Sigma} + \mathfrak n_{H} + \| \mathcal{E} \|_{s_{0}, \frac{1}{2}^{+}, \frac{7}{2}+\frac{1}{2} \delta},
    \end{split}
    \end{equation}
    where $\mathfrak n_{H}$ is defined with $\nu = 3+\frac{1}{2} \delta$ and all norms are defined on $M$ unless otherwise specified. The last term can be bounded under the following modification of \eqref{eq:BA.strong.integrated}:
    \begin{equation}\label{eq:BA.verystrong.integrated}
    \begin{split}
        &\|\alpha\|_{s_0+1, (-\frac{1}{2})^+, \frac{9}{2}+\frac{1}{2}\delta}
        + \|\beta\|_{s_0+1, (-\frac{1}{2})^+, \frac{7}{2}+\frac{1}{2}\delta}
        + \|\chih,\trch - 2r^{-1}\|_{s_0+1, \frac{1}{2}^+, \frac{3}{2}+\frac{1}{2}\delta} \\
        &\qquad + \|\eta\|_{s_0+1, \frac{1}{2}^+, \frac{3}{2}+\frac{1}{2}\delta} 
        + \|\rho, \sigma\|_{s_0, (-\frac{1}{2})^+, \frac{5}{2}+\frac{1}{2}\delta} 
        + \|K - r^{-2}\|_{s_{0}, (-\frac{1}{2})^{+}, \frac{3}{2}^{+}} \\
    &\qquad 
    + \|\chibh, \trchb + 2r^{-1}\|_{s_0, \frac{1}{2}^+, \frac{1}{2}^{+}} + \|f-2\omegab r - 1\|_{s_0,\frac{3}{2}^+ , (-\frac{1}{2})^{+}}  <\epsilon.
    \end{split}
    \end{equation}
    To complete the sketch of the propagation argument, it only remains to close \eqref{eq:BA.verystrong.integrated}. We defer this discussion until Remark~\ref{rem:strong-peeling-geom}.
\end{remark}

\section{Estimates for the other geometric quantities} \label{sec:other-geom}

In this section, we start with the estimates for $\alp$ in Corollary~\ref{cor:alpha.final} and prove estimates for all the other geometric quantities. In particular, we improve the bootstrap assumptions \eqref{eq:BA.weak.pointwise}--\eqref{eq:BA.strong.integrated}.

We first consider the bounds on the $2$-sphere.
\begin{proposition}\label{prop:fixed.sphere}
    \begin{enumerate}
        \item Assume that the weak bootstrap assumptions hold. Assume in addition that $\alp$ obeys the following bounds:
    \begin{equation}\label{eq:fixed.sphere.alp.input.1}
        |\alp|_{s_0,(-1)^+,2^{+}} \ls \ep^2.
    \end{equation}
    Then, for $\ep_0>0$ sufficiently small, in fact the bounds in the bootstrap assumptions \eqref{eq:BA.weak.pointwise} hold with $\ep$ replaced by $C\ep^2$.
        \item Suppose, in addition that the strong bootstrap assumptions hold and that
    \begin{equation}\label{eq:fixed.sphere.alp.input.2}
        |\alp|_{s_0,(-1)^+,5^{-}} \ls \ep^2.
    \end{equation}
    Then, for $\ep_0>0$ sufficiently small, in fact the bounds in the strong bootstrap assumptions \eqref{eq:BA.strong.pointwise} hold with $\ep$ replaced by $C\ep^2$.
    \end{enumerate}
    In both instances, $C>0$ depends only on $s_0$, $\de$ and the implicit constants in \eqref{eq:fixed.sphere.alp.input.1}, \eqref{eq:fixed.sphere.alp.input.2}.
\end{proposition}
\begin{proof}
    We prove the desired bounds using the transport estimate \eqref{eq:main.sphere.transest} in Theorem~\ref{thm:transest} and the initial data bound \eqref{eq:data.weak}, plus \eqref{eq:data.strong} for the strong bootstrap. There are four things to check: (1) total number of derivatives, (2) weights near $r=0$, (3) weights near $r = \infty$, and (4) the smallness constant. 
    
    Let us first discuss the smallness constant. We control the terms in an order such that all the error terms are either quadratic (so that the term is of size $\ep^2$ by the bootstrap assumption) or are linear terms that have been controlled in the previous step (so that the term is again of size $\ep^2$ because this has been proven).

    The sequence in which we choose to estimate the geometric quantities is as follows:
    \begin{equation}\label{eq:Ricci.sequence}
        \alp \to (\chih, \trch, \bt) \to (\gamma, \eta, \rho,\sigma, K) \to (b,  \chibh,\trchb,\omb,\betab) \to (f,\xib),
    \end{equation}
    meaning that when estimating $\chih$, $\trch$, or $\bt$, the only linear terms are given in terms of $\alp$; when estimating $\gamma$, $\eta$, $\rho$, or $\sigma$, the linear terms only involve $\alp$, $\chih$, $\trch$, or $\bt$, etc. This structure can be readily observed in the equations.

    We now consider (1)--(3) for each of the quantities involved. 

    Here is the procedure that we will go through. We use the transport estimate \eqref{eq:main.sphere.transest} with parameters $(s,\lambda_0,q_0,p_\infty)$ satisfying the following admissibility conditions which are required in Theorem~\ref{thm:transest}:
    \begin{equation}\label{eq:admissibility}
        2\lambda_0 +q_0+1>0,\quad p_\i \leqslant 2 \lambda_0.
    \end{equation}
    The parameters will also be chosen so that the initial data term arising in \eqref{eq:main.sphere.transest} will be bounded by \eqref{eq:data.weak} and \eqref{eq:data.strong}.
    
    We start with $\trch$ and $\chih$. For $\trch$, we begin with \eqref{eq:4trchi} and rewrite it as 
    \begin{equation}\label{eq:4trchi.2}
        \nab_4(\trch - \f 2r) + \trch (\trch - \f 2r) = \f 12 (\trch - \f 2r)^2 - |\chih|^2. 
    \end{equation}
    We apply \eqref{eq:main.sphere.transest} to $\trch- \f 2r$ and $\chih$ using \eqref{eq:4trchi.2} and \eqref{eq:4chih}, respectively. In other to get the desired estimates, we need to check that every term on the right-hand side of the equation can be controlled in a suitable norm. This is checked in Table~\ref{table:chi}.

    \begin{table}[h!]
    \centering
    \captionsetup[subtable]{justification=centering}
    \caption{Table for $\chih$ and $\trch-\f 2r$}\label{table:chi}
    
    \begin{subtable}{0.46\textwidth}
        \centering
        \caption{Weak bootstrap:\\ $(\lambda_0,s,q_0,p_\i) = (1,s_0,(-1)^+,1^+)$}
        \begin{tabular}{l|ccc|r}
            \toprule
             & $(\trch - \f 2r)^2$ & $|\chih|^2$ & $\alp$ & Goal \\
            \midrule
            $s$ & $s_0 \wedge s_0$ & $s_0 \wedge s_0$ & $s_0$ & $\geqslant s_0$ \\
            $q_0$ & $0^++0^+$ & $0^+ + 0^+$ & $(-1)^+$ & $\geqslant (-1)^+$\\
            $q_\infty$ & $1^+ + 1^+$  & $1^+ + 1^+$ & $2^+$ & $\geqslant 2^+$ \\
            \bottomrule
        \end{tabular}
    \end{subtable}
    \hfill 
    \begin{subtable}{0.48\textwidth}
        \centering
        \caption{Strong bootstrap:\\ 
        $(\lambda_0,s,q_0,p_\i,p_{u_{+}}) = (1,s_0,(-1)^+,2, 2^{-})$}
        \begin{tabular}{l|ccc|r}
            \toprule
             & $(\trch - \f 2r)^2$ & $|\chih|^2$ & $\alp$ & Goal \\
            \midrule
            $s$ & $s_0 \wedge s_0$ & $s_0 \wedge s_0$ & $s_0$ &  $\geqslant s_0$ \\
            $q_0$ & $0^++0^+$ & $0^+ + 0^+$ & $(-1)^+$ & $\geqslant (-1)^+$\\
            $q_\infty$ & $2+2$  & $2+2$ & $5^{-}$ & $> 3$ \\
            {\tiny $q_{u_{+}}+q_\infty$} & $4^{-}+4^{-}$  & $4^{-}+4^{-}$ & $5^{-}$ & $\geqslant 5^{-}$ \\
            \bottomrule
        \end{tabular}
    \end{subtable}
\end{table}
    The subtable on the left is used for the weak bootstrap (resp.~the subtable on the right is to be used under for the strong bootstrap). The parameters $(\lambda_0,s,q_0,p_\i)$ (resp. $(\lambda_0,s,q_0,p_\i,p_{u_+})$) for which \eqref{eq:main.sphere.transest} is applied are stated in the header of the subtable. For each term that appears on the right-hand side of the equation, we check the values of $(s,q_0,q_\infty)$ for which the term can be put in the $|\cdot|_{s,q_0,q_\infty}$ norm (resp.~$(s, q_{0}, q_{\infty}, q_{u})$ for which the term can be put in the $|\jb{u_{+}}^{q_{u_{+}}} (\cdot)|_{s,q_0,q_\infty}$ norm for the strong bootstrap), and tabulate them. The goal we need to achieve in order to have the desired estimate is listed in the far-right column. When dealing with nonlinear terms, we implicitly use the product estimate in Lemma~\ref{lem:products}. We also freely use \eqref{eq:norm.r.shift} and \eqref{eq:norm.u.shift} to trade in $r$ and $\jb{u}$ weights to shift the exponents in the norms.
    
    We use similar tables to check the bounds for $\bt$ using \eqref{eq:4beta}; see Table~\ref{table:bt}.

    \begin{table}[h!]
    \centering
    \captionsetup[subtable]{justification=centering}
    \caption{Table for $\beta$}
    \label{table:bt}
    \begin{subtable}{0.48\textwidth}
        \centering
        \caption{Weak bootstrap\\ $(\lambda_0,s,q_0,p_\i) = (2,s_0-1,(-2)^+,2^+)$}
    \begin{tabular}{l|cc|r} 
        \toprule 
         & $\eta\cdot \alp$ & $\div\alp$ & Goal \\
        \midrule 
        $s$ & $(s_0-1)\wedge s_0$ & $s_0-1$ & $\geqslant s_0-1$ \\
        $q_0$ & $0^++(-1)^+$ & $(-2)^+$ & $\geqslant (-2)^+$ \\
        $q_\infty$ & $1^+ + 2^+$  & $1+2^+$ & $\geqslant 3^+$\\
        \bottomrule 
    \end{tabular}
    \end{subtable}
    \hfill 
    \begin{subtable}{0.48\textwidth}
        \centering
        \caption{Strong bootstrap\\ $(\lambda_0,s,q_0,p_\i,p_{u_{+}}) = (2,s_0-1,(-2)^+,4,1^{-})$}
        \begin{tabular}{l|cc|r} 
        \toprule 
         & $\eta\cdot \alp$ & $\div\alp$ & Goal \\
        \midrule 
        $s$ & $(s_0-1)\wedge s_0$ & $s_0-1$ & $\geqslant s_0-1$\\
        $q_0$ & $0^++(-1)^+$ & $(-2)^+$ & $\geqslant (-2)^+ $\\
        $q_\infty$ & $2 + 5^{-}$  & $1+5^{-}$ & $>5$ \\
        {\tiny $q_{u_{+}}+q_\infty$} & $4^{-}+5^{-}$  & $1+5^{-}$ & $\geqslant 6^{-}$ \\
        \bottomrule 
    \end{tabular}
    \end{subtable}
\end{table}

    We now turn to the next group in \eqref{eq:Ricci.sequence}, i.e., $(\gamma,\eta,\rho,\sigma)$. For $\gamma$, we start with $\f{\rd}{\rd r} \gamma = 2\chi$ and $\f{\rd}{\rd r} \mathring{\gamma} = \f 2r \mathring{\gamma}$ to obtain
    $$\f{\rd}{\rd r} (\gamma - \mathring{\gamma})_{AB} = 2\chih_{AB} + 2 (\trch - \f 2r )\gamma_{AB} + \f 2 r (\gamma - \mathring{\gamma})_{AB},$$
    which then translates to 
    \begin{equation}\label{eq:gamma.diff}
        \nab_4 (\gamma - \mathring{\gamma})  = 2\chih - \chih \hat{\odot} (\gamma - \mathring{\gamma}) - \chih \cdot (\gamma - \mathring{\gamma}) + 2 (\trch - \f 2r )\gamma - (\trch - \f 2 r) (\gamma - \mathring{\gamma}).
    \end{equation}
     For $\eta$ and $(\rho,\sigma)$, we use the equations \eqref{eq:4zeta} and \eqref{eq:4rho}. The relevant parameters for $(\gamma, \eta, \rho,\sigma)$ are then checked in Table~\ref{table:gamma.eta.rho.sigma}. For \eqref{eq:gamma.diff}, in view of the fact that $\chih$ and $\trch - \f 2r$ satisfy the same estimates, the linear terms  $\chih$, $(\trch - \f 2r )\gamma$ can be treated similarly, while the nonlinear terms $\chih \hat{\odot} (\gamma - \mathring{\gamma})$, $\chih \cdot (\gamma - \mathring{\gamma})$, $ (\trch - \f 2 r) (\gamma - \mathring{\gamma})$ are also similar. We will therefore only take one term from each group. Moreover, observe that for $\gamma - \mathring{\gamma}$, we need to obtain extra $\brk{u}^{-0^+}$ decay. This would come from the $\brk{u}^{q_\i-p_\i-1}$ factor in \eqref{eq:main.sphere.transest} (which we have not used so far); our goal is stated to take into account the needed $\brk{u}$-decay. Note also that there is no strong bootstrap for $\gamma - \mathring{\gamma}$. Finally, the desired weak bootstrap for $K-\f 1{r^2}$ then follows from that for $\gamma- \mathring{\gamma}$, since the Gauss curvature can be computed by
        \begin{equation}\label{Gauss.def}
\gamma_{BC} K=\f{\rd}{\rd\th^A}\slashed{\Gamma}^{A}_{BC}-\f{\rd}{\rd\th^C}\slashed{\Gamma}^A_{BA}+\slashed{\Gamma}^A_{AD}\slashed{\Gamma}^D_{BC}-\slashed{\Gamma}^A_{CD}\slashed{\Gamma}^D_{BA},
\end{equation}
        where $\slashed{\Gamma}$ is as in \eqref{Gamma.def}, and the Gauss curvature for the metric $\mathring{\gamma}$ is exactly $r^{-2}$.

    \begin{table}[h!]
    \centering
    \captionsetup[subtable]{justification=centering}
    \caption{Table for $\gamma, \eta, \rho, \sigma$}
    \label{table:gamma.eta.rho.sigma}

    \begin{subtable}{0.48\textwidth}
        \centering
        \caption{For $\gamma$:\\ $(\lambda_0,s,q_0,p_\i) = (0,s_0,0^+,0)$}
        \label{table:gamma}
        \begin{tabular}{l|cc|r}
            \toprule 
             & $(\trch - \f 2r)(\gamma-\mathring{\gamma})$ & $(\trch - \f 2r)\gamma$ & Goal \\
            \midrule 
            $s$ & $s_0\wedge s_0$ & $s_0$ & $\geqslant s_0$ \\
            $q_0$ & $0^++1^+$ & $0^+$ & $\geqslant 0^+$\\
            $q_\infty$ & $1^+ + 0$ & $1^+$ & $\geqslant 1^{+}$ \\
            \bottomrule 
        \end{tabular}
    \end{subtable}
    \bigskip 

    \begin{subtable}{0.48\textwidth}
        \centering
        \caption{For $\eta$, weak bootstrap:\\ $(\lambda_0,s,q_0,p_\i) = (1,s_0-1,(-1)^+,1^{+})$}
        \begin{tabular}{l|cc|r}
            \toprule 
             & $\chih\cdot \eta$ & $\bt$ & Goal \\
            \midrule 
            $s$ & $s_0 \wedge (s_0-1)$ & $s_0-1$ & $\geqslant s_0-1$ \\
            $q_0$ & $0^++0^+$ & $(-1)^+$ & $\geqslant (-1)^+$\\
            $q_\infty$ & $1^+ + 1^+$ & $2^+$ & $\geqslant 2^+$ \\
            \bottomrule 
        \end{tabular}
    \end{subtable}
    \hfill
    \begin{subtable}{0.48\textwidth}
        \centering
        \caption{For $\eta$, strong bootstrap:\\ $(\lambda_0,s,q_0,p_\i,p_{u_{+}}) = (1,s_0-1,(-1)^+,2,2^{-})$}
        \begin{tabular}{l|cc|r}
            \toprule 
             & $\chih\cdot \eta$ & $\bt$ & Goal  \\
            \midrule 
            $s$ & $s_0 \wedge (s_0-1)$ & $s_0-1$ & $\geqslant s_0-1$ \\
            $q_0$ & $0^++0^+$ & $(-1)^+$ & $\geqslant (-1)^+$\\
            $q_\infty$ & $2+2$ & $4$ & $>3$ \\
            {\tiny $q_{u_{+}}+q_\infty$} & $4^{-}+4^{-}$ & $5^{-}$ & $\geqslant 5^{-}$ \\
            \bottomrule 
        \end{tabular}
    \end{subtable}

    \bigskip 
    
    \begin{subtable}{0.48\textwidth}
        \centering
        \caption{For $(\rho,\sigma)$, weak bootstrap:\\ $(\lambda_0,s,q_0,p_\i) = (\f 32,s_0-2,(-2)^+,2^{+})$}
        \begin{tabular}{l|ccc|r}
            \toprule 
            & $\eta \cdot \bt$ & $\chih\cdot \alp$ & $\nab \bt$ & Goal  \\
            \midrule 
            $s$ & \tiny{$(s_0-1) \wedge (s_0-1)$} & \tiny{$(s_0-1)\wedge s_0$} & $s_0-2$ & {\tiny $\geqslant s_0-2$} \\
            $q_0$ & \tiny{$0^++(-1)^+$} & \tiny{$0^++(-1)^+$} & $(-2)^+$  & {\tiny $\geqslant (-2)^+$}\\
            $q_\infty$ & $1^+ + 2^+$ & $1^++2^+$ & $1+2^+$ & $\geqslant 3^+$ \\
            \bottomrule 
        \end{tabular}
    \end{subtable}
    \hfill
    \begin{subtable}{0.48\textwidth}
        \centering
        \caption{For $(\rho,\sigma)$, strong bootstrap:\\ $(\lambda_0,s,q_0,p_\i,p_{u_{+}}) = (\f 32,s_0-2,(-2)^+,3,2^{-})$}
        \begin{tabular}{l|ccc|r}
            \toprule 
             & $\eta \cdot \bt$ & $\chih\cdot \alp$ & $\nab \bt$ & Goal  \\
            \midrule 
            $s$ & \tiny{$(s_0-1) \wedge (s_0-1)$} & \tiny{$(s_0-1)\wedge s_0$} & $s_0-2$ & {\tiny $\geqslant s_0-2$} \\
            $q_0$ & \tiny{$0^++(-1)^+$} & \tiny{$0^++(-1)^+$} & $(-2)^+$  & {\tiny $\geqslant (-2)^+$}\\
            $q_\infty$ & $2+4$ & $2+5^{-}$ & $1+4$ & $>4$ \\
            {\tiny $q_{u_{+}}+q_\infty$} & $4^{-}+5^{-}$ & $4^-+5^{-}$ & $1+5^{-}$ & $\geqslant 6^{-}$ \\
            \bottomrule 
        \end{tabular}
    \end{subtable}
    
\end{table}

    We finally move to the last two groups in \eqref{eq:Ricci.sequence}, namely, $(b, f-1, \xib,\chibh,\trchb+\f 2r,\omb,\betab)$. There are more quantities to consider in these two groups, but there are no strong bootstrap assumptions involved.

    The connection coefficients $\trchb + \f 2r$, $\chibh$ and $\omb$ can be treated together. For $\omb$, we directly use the transport equation \eqref{eq:omegab}. For $\trchb + \f 2r$, we rewrite \eqref{eq:4trchib} into the following form:
    \begin{equation}
        \begin{split}
            \nab_4 (\tr \chib + \f 2r)  +\frac{1}{2}\tr \chi (\tr \chib +\f 2r) = \f 1r (\trch-\f 2r) -2\div \eta - \chih\cdot \chibh +2 |\eta|^2 + 2 \rho. 
        \end{split}
    \end{equation}
    For $\chibh$, we use \eqref{eq:4chibh}, except for rewriting $\trchb \chih = (\trchb + \f 2r) \chih - \f 2r \chih$. For $\omb$, we again need to obtain extra $\brk{u}^{-(1^+)}$ decay, which also comes from the $\brk{u}^{q_\i-p_\i-1}$ factor in \eqref{eq:main.sphere.transest} after choosing the parameters appropriately.

    \begin{table}[h!]
    \centering
    \label{table:connection.the.rest}

        \centering
        \caption{For $\chibh$, $\trchb + \f 2r$: $(\lambda_0,s,q_0,p_\i) = (1,s_0-2,(-1)^+,1)$;\\ for $\omb$: $(\lambda_0,s,q_0,p_\i) = (0,s_0-2,(-1)^+,0)$}
        \label{table:chibh}
        \begin{tabular}{l|ccccc|r}
            \toprule 
             & $\eta \cdot\eta,\eta\hot \eta$ & $\chibh \cdot\chih, (\trchb + \f 2r) \chih$ & $\f 1r \chih, \f 1r (\trch - \f 2r)$ & $\nab \eta$ & $\rho$ & Goal \\
            \midrule 
            $s$ & $(s_0-1)\wedge (s_0-1)$ & $(s_0-2)\wedge s_0$ & $s_0$ & $s_0-2$ & $s_0-2$ & $\geqslant s_0-2$ \\
            $q_0$ & $0^++0^+$ & $0^++0^+$ & $0^+-1$ & $0^+-1$ & $(-1)^+$ & $\geqslant (-1)^+$\\
            $q_\infty$ & $1^+ + 1^+$ & $1+ 1^+$ & $1+ 1^+$ & $1+ 1^+$ & $2^+$ & $\geqslant 2^{+}$ \\
            \bottomrule 
        \end{tabular}

\end{table}
    Using \eqref{eq:constchibh}, we now write $\betab$ in terms of geometric quantities that we have already estimated to obtain the desired bound for $\betab$.

    For the metric component $b$, we combine \eqref{gauge.con.metric}, \eqref{gauge.con} and \eqref{nab4.def} to obtain
    \begin{equation}\label{eq:b.eqn}
        \nab_4 b - \f 12 \trch \cdot b = \chih \cdot b+ 2 \eta.
    \end{equation}
    Using \eqref{eq:b.eqn}, we check the indices in Table~\ref{table:b}.
    
    \begin{table}[h!]
    \centering
    \caption{For $\jb{u} b$: $(\lambda_0,s,q_0,p_\i) = (-\f 12,s_0-1,0^+,-1)$}
    \label{table:b}
        \begin{tabular}{l|cc|r}
            \toprule 
             & $\chih\cdot \brk{u} b$ & $\brk{u}\eta$ & Goal  \\
            \midrule 
            $s$ & $s_0 \wedge (s_0-1)$ & $s_0-1$ & $\geqslant s_0-1$ \\
            $q_0$ & $0^++1^+$ & $0^+$ & $\geqslant 0^+$\\
            $q_\infty$ & $1^++(-1)$ & $-1+1^+$ & $\geqslant 0^{+}$ \\
            \bottomrule 
        \end{tabular}
    \end{table}
    
    For $f$, we have the equation $\nab_4 f = 2 \omb$ in \eqref{gauge.con.metric}, together with
    \begin{equation}\label{eq:f-2ombr-1}
        \nab_4 (f - 2\omb r - 1) = -2r \nab_4 \omb = -6r|\eta|^2 - 2r\rho.    
    \end{equation}
    Notice that we again need to obtain $\brk{u}$-decay. For $f -2\omb r -1$, this is achieved by using the factor $\brk{u}^{q_\i-p_\i-1}$ in \eqref{eq:main.sphere.transest} after choosing the parameters appropriately. For $b$ and $f-1$, we apply \eqref{eq:main.sphere.transest} to $\brk{u} b$ and $\brk{u}^{1^{+}} (f-1)$, in which case the desired $u$-decay is achieved again using the factor $\brk{u}^{q_\i-p_\i-1}$. In the tables, we freely trade non-negative powers of $\brk{u}$ with $r+\brk{u}$ using \eqref{eq:norm.u.shift}.
     \begin{table}[h!]
    \centering
    \caption{Tables for $f-1$, $f -2\omb r -1$}
    \label{table:metric.the.rest}

    \begin{subtable}{0.48\textwidth}
        \centering
        \caption{For $f - 2\omb r -1$: $(\lambda_0,s,q_0,p_\i) = (0,s_0-2,0^+,0)$}
        \begin{tabular}{l|cc|r}
            \toprule 
             & $r|\eta|^2$ & $r\rho$ & Goal \\
            \midrule 
            $s$ & $(s_0-1) \wedge (s_0-1)$ & $s_0-2$ & $\geqslant s_0-2$ \\
            $q_0$ & $1+0^++0^+$ & $1+(-1)^+$ & $\geqslant 0^+$\\
            $q_\infty$ & $-1+1^+ +1^+$ & $-1+2^+$ & $\geqslant 1^+$ \\
            \bottomrule 
        \end{tabular}
    \end{subtable}
    \hfill
    \begin{subtable}{0.48\textwidth}
        \centering
        \caption{{\tiny For $\jb{u}^{1^{+}}(f-1)$: $(\lambda_0,s,q_0,p_\i) = (0,s_0-2,0^+,-1)$}}
        \begin{tabular}{l|c|r}
            \toprule 
             & $\jb{u} \omb$ & Goal  \\
            \midrule 
            $s$ & $s_0-2$ & $\geqslant s_0-2$ \\
            $q_0$ & $0^+$ & $\geqslant 0^+$\\
            $q_\infty$ & $0$ & $\geqslant 0$ \\
            \bottomrule 
        \end{tabular}
    \end{subtable}
\end{table}

    Finally, for $\xib$ we use $\xib = \nab f$ (by \eqref{gauge.con}) to obtain the desired bounds from those that have already been established. \qedhere
\end{proof}

    Next, we prove the integrated decay for the connection coefficients.
    \begin{proposition}\label{prop:integrated}
    There exist $\ep_0$ sufficiently small depending on $s_0$ and $\de$ such that the following hold.
        \begin{enumerate}
            \item Under the assumptions of Proposition~\ref{prop:fixed.sphere}.(1), assume in addition that 
            \begin{equation}\label{eq:integrated.alp.assumption.weak}
                \|\alp\|_{s_0+1,-\f 12{}^+,\frac{3}{2}^+} + \| \jb{u}^{\frac{1}{2}}\nab_3 (r^2 \alp) \|_{s_0,\f 32{}^+,0^+} \ls \ep^2.
            \end{equation}
            Then the bootstrap assumptions \eqref{eq:BA.weak.integrated} hold with $\ep$ replaced by $C\ep^2$.
            \item Under the assumptions of Proposition~\ref{prop:fixed.sphere}.(2), assume in addition that 
            \begin{equation}\label{eq:integrated.alp.assumption.strong}
                \|\alp\|_{s_0+1,-\f 12{}^+,\f 92+\frac{1}{2}\delta} + \| \brk{u}^{\frac{1}{2}+\frac{3}{2}\delta} \nab_3 (r^2 \alp) \|_{s_0,\f 32{}^+,3^{-}} \ls \ep^2.
            \end{equation}
            Then the bootstrap assumptions \eqref{eq:BA.strong.integrated} hold with $\ep$ replaced by $C\ep^2$.
        \end{enumerate}
        In both instances, $C>0$ depends only on $s_0$, $\de$ and the implicit constants in \eqref{eq:integrated.alp.assumption.weak} and \eqref{eq:integrated.alp.assumption.strong}.
    \end{proposition}
    \begin{proof} 
    In this proof, we use the shorthand $p^{\pm\prime} := p\pm\frac{1}{2}\delta$.
    
        We will derive our estimates in the order
        \begin{equation}\label{eq:Ricci.sequence.2}
        \alp \to \bt \to (\chih, \trch) \to (\gamma, \eta, \rho,\sigma) \to (K, \chibh,\trchb,\omb,\betab) \to (f, \xib)
    \end{equation}
        (cf.~\eqref{eq:Ricci.sequence}). In particular, the same observation as in Proposition~\ref{prop:fixed.sphere} takes care of the powers of $\epsilon$. One difference between \eqref{eq:Ricci.sequence} and \eqref{eq:Ricci.sequence.2} is that now $\bt$ is treated first and separately. This is to deal with the top order, where we need a combination of transport and elliptic estimates in order to avoid a derivative loss. This will be first treated in Step~1 before we return to the transport estimates for the other quantities in Step~2.

        \pfstep{Step~1: Elliptic estimates and control for $\bt$}  Applying Lemma~\ref{lem:elliptic} with $p_{0}=(-\frac{1}{2})^{+}$ to $\brk{u}^{p_{u}} \beta$, we have 
        \begin{equation}\label{eq:bt.top.main.elliptic}
        \begin{aligned}
            \|\brk{u}^{p_{u}} \bt\|_{s_0+1,(-\f 12{})^+,p_\i} &\lesssim \|\brk{u}^{p_{u}} (r\nab \hot \bt, r \nab_{4} \bt, \bt)\|_{s_0,(-\f 12{})^+,p_\i} + \epsilon_{0} \mathcal{Q}[\beta],
        \end{aligned}
        \end{equation}
        where 
        \begin{equation*}
            \mathcal{Q}[\beta] = \begin{cases}
                |\brk{u}^{p_{u}-\frac{1}{2}} \beta|_{s_{0}-3, -1, p_{\infty}^{-}+1} & \hbox{for weak bootstrap}, \\
                |\brk{u}^{p_{u}} \beta|_{s_{0}-3, -1, p_{\infty}^{+}+\frac{1}{2}} & \hbox{for strong bootstrap}.
            \end{cases}
        \end{equation*}
        We control each term on the right-hand side of \eqref{eq:bt.top.main.elliptic}, with $(p_\i, p_{u}) = (1^+, \frac{1}{2})$ in the case of weak bootstrap, and $(p_\i, p_{u}) = (\frac{7}{2}^{-}, 0)$ in the case of strong bootstrap.

        First, we control $\epsilon_{0} \mathcal{Q}[\beta]$. For both the weak and strong bootstrap cases, note that $\mathcal{Q}[\beta]$ can be controlled by the pointwise bootstrap assumptions; since we have closed them in Proposition~\ref{prop:fixed.sphere}, we control the whole term by $C \epsilon_{0} \epsilon^{2}$. 
        
        Next, we control $\|\brk{u}^{p_{u}} \bt\|_{s_0,(-\frac{1}{2})^{+},p_\i}$. For this, we apply the transport estimate in \eqref{eq:main.int.transest} to $\brk{u}^{p_{u}} \beta$. We check the indices in Table~\ref{table:integrated.bt}. Since we are now estimating the integrated norms $\|\cdot\|$, for the product of two factors, e.g.~$\eta\cdot \alp$, we now use \eqref{eq:product.for.integrated} so that one of the factors is put in a $\|\cdot\|$ norm, while the other is put in a $|\cdot|$ norm (with an appropriate distribution of the $\brk{u}$ weights, which will be clear in the tables). In order to check the exponents, we need to consider the minimum of two sums, corresponding to putting a different factor in the $\|\cdot\|$ norm. See, e.g., the $(0^+ + (-\f 12){}^+)\wedge (\f 12{}^++ (-1)^+)$ entry in the table. As before, we freely use \eqref{eq:norm.u.shift} to bound non-negative powers of $\brk{u}$ by decreasing $q_{\infty}$. 

        \begin{table}[h!]
    \centering
    \captionsetup[subtable]{justification=centering}
    \caption{Table for $\beta$}
    \label{table:integrated.bt}
    \begin{subtable}{0.48\textwidth}
        \centering
        \caption{Weak bootstrap for $\brk{u}^{\frac{1}{2}} \beta$:\\ $(\lambda_0,s,q_0,\tilde{p}_{\infty}) = (2,s_0,(-\f 32)^+,1^+)$}
    \begin{tabular}{l|cc|r} 
        \toprule 
         & $\jb{u}^{\frac{1}{2}} \eta\cdot \alp$ & $\jb{u}^{\frac{1}{2}} \div\alp$ & Goal \\
        \midrule 
        $s$ & $(s_0+1)\wedge(s_0+1)$ & $s_0$ & $\geqslant s_0$ \\
        $q_0$ & \tiny{$(0^+ + (-\f 12){}^+)\wedge (\f 12{}^++ (-1)^+)$} & $(-\f 32)^+$ & \tiny{$\geqslant (-\f 32)^+$}\\
        $q_\infty$ & \tiny{$(-\frac{1}{2}+1^+ + \frac{3}{2}^+)\wedge (0^+ + 2^+)$}  & $-\frac{1}{2}+\f 52{}^+$ & $\geqslant 2^+$\\
        \bottomrule 
    \end{tabular}
    \end{subtable}
    \hfill 
    \begin{subtable}{0.48\textwidth}
        \centering
        \caption{Strong bootstrap for $\beta$:\\ $(\lambda_0,s,q_0,\tilde{p}_{\i}) = (2,s_0,(-\f 32)^+, \frac{7}{2}^{-})$}
        \begin{tabular}{l|cc|r} 
        \toprule 
         & $\eta\cdot \alp$ & $\div\alp$ & Goal \\
        \midrule 
        $s$ & $(s_0+1)\wedge(s_0+1)$ & $s_0$ & $\geqslant s_0$ \\
        $q_0$ & \tiny{$(0^+ + (-\f 12){}^+)\wedge (\f 12{}^++ (-1)^+)$} & $(-\f 32)^+$ & \tiny{$\geqslant (-\f 32)^+$} \\
        $q_\infty$ & \tiny{$(2+\frac{9}{2}^{-})\wedge (\frac{3}{2}^- + 5^{-})$}  & $1+\frac{9}{2}^{-}$ & $\geqslant \frac{9}{2}^-$ \\
        \bottomrule 
    \end{tabular}
    \end{subtable}
\end{table}

        We now turn to the other two terms in \eqref{eq:bt.top.main.elliptic}. To bound $\|\jb{u}^{\frac{1}{2}} r\nab \hot \bt\|_{s_0,(-\frac{1}{2})^{+},1^+}$, we use \eqref{eq:3alpha v2} to write $\nab \hat{\otimes} \bt = r^{-2} \nab \hat{\otimes} (r^2 \bt)$ in terms of other quantities. In Table~\ref{table:bt.top.angular.weak}, we check all the other terms in this equation to then obtain $\|\jb{u}^{\frac{1}{2}} r\nab \hot \bt\|_{s_0,(-\frac{1}{2})^{+},1^+} \ls \ep^2$. Notice that all the terms taking the form of a connection coefficient contracted with a curvature component would give the same indices for $q_0$ as in Table~\ref{table:integrated.bt} and so we skip that row for those terms.
        \begin{table}[h!]
    \centering
        \caption{Weak bootstrap for $\jb{u}^{\frac{1}{2}} \nab \widehat{\otimes} \bt$}
        \label{table:bt.top.angular.weak}
        \begin{subtable}{1.0\textwidth}
        \centering
        \begin{tabular}{l|ccc|r}
            \toprule
             & $\jb{u}^{\frac{1}{2}} r^{-2}\nab_3(r^2\alp)$ & $\jb{u}^{\frac{1}{2}} r^{-1} \alp$ & $\jb{u}^{\frac{1}{2}} r^{-1}(f-1-2\omb r) \alp$  & Goal \\
            \midrule
            $s$ & $s_0$ & $s_0+1$ & $(s_0+1) \wedge (s_0+1)$ & $\geqslant s_0$ \\
            $q_0$ & $(-\f 12){}^+$ & $(-\f 32){}^+$ & \tiny{$(-1+1^+ + (-\f 12)^+) \wedge (-1+ \f 32{}^++(-1)^+)$ }  & $\geqslant (-\f 32){}^+$\\
            $q_\infty$ & $2^+$  & \tiny{$-\frac{1}{2}+1+\frac{3}{2}^+$} & \tiny{$((-\frac{1}{2})^{+} + 1 + 0 + \frac{3}{2}^+)\wedge (1 + (-1)^+ + 2^+)$}  & $\geqslant 2^+$ \\
            \bottomrule
        \end{tabular}
        \end{subtable}

        \bigskip
        \begin{subtable}{1.0\textwidth}
        \centering
        \begin{tabular}{l|ccc|r}
            \toprule
             & $\jb{u}^{\frac{1}{2}} (\trchb+\f 2r) \alp$ & $\jb{u}^{\frac{1}{2}} \eta \widehat{\otimes} \bt$ & $\jb{u}^{\frac{1}{2}} (\chih\rho,\, {}^*\chih\sigma)$ & Goal \\
            \midrule
            $s$ &  $s_0 \wedge (s_0+1)$ & $(s_0+1) \wedge (s_0+1)$ & $(s_0+1) \wedge s_0$ & $\geqslant s_0$ \\
            $q_\infty$ & \tiny{$((-\frac{1}{2})^{+} + 1 + \frac{3}{2}^+)\wedge (0^{+} + 2^+)$}  & \tiny{$(1^{+} + 1^{+})\wedge (0^+ + 2^+)$} & \tiny{$(1^{+} + 1^{+}) \wedge (-\frac{1}{2} + \frac{1}{2}^+ + 2^+)$} & $\geqslant 2^+$ \\
            \bottomrule
        \end{tabular}
        \end{subtable}
\end{table}

        We also check the indices for the strong bootstrap, in which case we obtain $\|r\nab \hot \bt\|_{s_0,(-\frac{1}{2})^{+},\f 72^{-}} \ls \ep^2$. The indices are checked in Table~\ref{table:bt.top.2}. Notice that the $s$ and $q_0$ row would have been identical to the weak bootstrap case, and thus we only check the index $q_\i$.
        \begin{table}[h!]
    \centering
        \caption{Strong bootstrap for $\nab \widehat{\otimes} \bt$}
        \label{table:bt.top.2}
        \begin{subtable}{1.0\textwidth}
        \centering
        \begin{tabular}{l|cccc|r}
            \toprule
             & $r^{-2}\nab_3(r^2\alp)$ & $r^{-1} \alp$ & $r^{-1}(f-1-2\omb r) \alp$ & $(\trchb+\f 2r) \alp$ & Goal \\
            \midrule
            $q_\infty$ & $5^-$  & $1+\frac{9}{2}^{-}$ & \tiny{$(1 + 0 + \frac{9}{2}^{-})\wedge (1 + (-\frac{1}{2})^{+} + 5^{-})$} & \tiny{$(1 + \frac{9}{2}^{-})\wedge (\frac{1}{2}^{+} + 5^{-})$}  & $\geqslant \frac{9}{2}^-$ \\
            \bottomrule
        \end{tabular}
        \end{subtable}

        \bigskip
        \begin{subtable}{1.0\textwidth}
        \centering
        \begin{tabular}{l|cc|r}
            \toprule
             & $\eta \widehat{\otimes} \bt$ & $\chih\rho,\, {}^*\chih\sigma$ & Goal \\
            \midrule
            $q_\infty$ & \tiny{$(2 + \frac{7}{2}^{-})\wedge (\frac{3}{2}^{-} + 4)$} & \tiny{$(2 + \frac{5}{2}^{-}) \wedge (\frac{3}{2}^{-} + 3)$} & $\geqslant \frac{9}{2}^-$ \\
            \bottomrule
        \end{tabular}
        \end{subtable}
\end{table}

        We finally turn to the term $\|\jb{u}^{p_{u}} r\nab_4\bt \|_{s_0,(-\frac{1}{2})^{+},p_\i}$ in \eqref{eq:bt.top.main.elliptic}. For this we use \eqref{eq:4beta} to write $\nab_4 \bt$ as terms bounded by \eqref{eq:integrated.alp.assumption.weak}, \eqref{eq:integrated.alp.assumption.strong} and the bootstrap assumptions. We check the necessary exponents in Table~\ref{table:nab4bt.top}. We skip the row for $q_0$, as they are similar to those in Tables~\ref{table:integrated.bt} and \ref{table:bt.top.angular.weak} depending on the type of the term. Remark also that for the linear $r^{-1}\bt$ term, we use the bound for $\|\bt \|_{s_0,(-\frac{1}{2})^{+},p_\i} \ls \ep^2$ that we just established above (instead of the bootstrap assumption) so that we obtain a $C\ep^2$ (instead of $C\ep$) upper bound.

        \begin{table}[h!]
            \centering
        \caption{Weak bootstrap for $\jb{u}^{\frac{1}{2}} \nab_4 \bt$}
        \label{table:nab4bt.top}
        \begin{tabular}{l|cccc|r}
            \toprule
             & $\jb{u}^{\frac{1}{2}} r^{-1} \bt$ & $\jb{u}^{\frac{1}{2}} (\trch - \f 2r) \bt$ & $\jb{u}^{\frac{1}{2}} \div \alp$ & $\jb{u}^{\frac{1}{2}}\eta\cdot \alp$  & Goal \\
            \midrule
            $s$ & $s_0$ & $(s_0+1)\wedge s_0$ & $s_0$ &  $s_0 \wedge (s_0+1)$ & $\geqslant s_0$ \\
            $q_\infty$ & $1+1^+$ & \tiny{$(1^++ 1^+)\wedge(-\frac{1}{2}+ \frac{1}{2}^+ + 2^+)$} & $-\frac{1}{2}+1+\frac{3}{2}^+$ & \tiny{$(-\frac{1}{2}+1^+ + \frac{3}{2}^+) \wedge (0^+ + 2^+)$} & $\geqslant 2^+$ \\
            \bottomrule
        \end{tabular}
\end{table}

        We then consider strong bootstrap for $\nab_4\bt$. The only difference is the $q_\i$ exponent and we only check that row. This is carried out in Table~\ref{table:nab4bt.top.strong}.

        \begin{table}[h!]
            \centering
        \caption{Strong bootstrap for $\nab_4 \bt$}
        \label{table:nab4bt.top.strong}
        \begin{tabular}{l|cccc|r}
            \toprule
             & $r^{-1} \bt$ & $(\trch - \f 2r) \bt$ & $\div \alp$ & $\eta\cdot \alp$  & Goal \\
            \midrule
            $q_\infty$ & $1+\frac{7}{2}^-$ & $(2 + \frac{7}{2}^{-})\wedge(\frac{3}{2}^{-} + 4)$ & $1+\frac{9}{2}^{-}$ & $(2 + \frac{9}{2}^{-}) \wedge (\f 32^- + 5^{-})$ & $\geqslant \frac{9}{2}^-$ \\
            \bottomrule
        \end{tabular}
\end{table}

        Combining all the above, we obtain that the right-hand side of \eqref{eq:bt.top.main.elliptic} is $\ls \ep^2$.

        \pfstep{Step~2: Transport estimates for the weak bootstrap} We now bound all the other quantities in \eqref{eq:BA.weak.integrated}. For most of the estimates, we use the transport estimate \eqref{eq:main.int.transest} and the initial data bound \eqref{eq:data.weak.int}.

        We follow \eqref{eq:Ricci.sequence.2} and begin with $\chih$, $\trch+\f 2r$. In this case, $s = s_{0}+1$, so we need to use the top-order transport estimate \eqref{eq:main.int.transest.top}. We immediately observe that the last (extra pointwise) term in \eqref{eq:main.int.transest.top} can be bounded by $C \epsilon_{0} \epsilon^{2}$ using Proposition~\ref{prop:fixed.sphere} in each case; hence, it only remains to check the term $\| F \|_{s, q_{0}, q_{\infty}}$ as usual. This is checked in Table~\ref{table:int.chi}.
        \begin{table}[h!]
    \centering
    \caption{Table for $\chih$, $\trch+\f 2r$: $(\lambda_0,s,q_0,\tilde{p}_\i) = (1,s_0+1,(-\f 12){}^+,\frac{1}{2}^+)$}
    \label{table:int.chi}
        
        \begin{tabular}{l|ccc|r}
            \toprule
             & $(\trch - \f 2r)^2$ & $|\chih|^2$ & $\alp$ & Goal \\
            \midrule
            $s$ & $(s_0+1) \wedge (s_0+1)$ & $(s_0+1) \wedge (s_0+1)$ & $s_0+1$ & $\geqslant s_0+1$ \\
            $q_0$ & $0^++\f 12{}^+$ & $0^++\f 12{}^+$ & $(-\f 12){}^+$ & $\geqslant (-\f 12){}^+$\\
            $q_\infty$ & $1^+ + \frac{1}{2}^+ $  & $1^+ + \frac{1}{2}^+ $ & $\frac{3}{2}^+$ & $\geqslant \frac{3}{2}^{+}$ \\
            \bottomrule
        \end{tabular}

\end{table}
    
We next consider $\gamma, \eta$ in Table~\ref{table:integrated.gamma.eta}. Here, we again need to use the top-order transport estimate \eqref{eq:main.int.transest.top}. As before, the last (extra pointwise) term in \eqref{eq:main.int.transest.top} can be bounded by $C \epsilon_{0} \epsilon^{2}$ using Proposition~\ref{prop:fixed.sphere} in each case, so again it only remains to check the term $\| F \|_{s, q_{0}, q_{\infty}}$. Since the numerology for the $q_0$ row is the same for all connection coefficients, it will be omitted in the table for $\eta$ (and also the table for $\chibh$, $\trchb+\f 2r$ and $\omb$ later.)
\begin{table}[h!]
    \centering
    \caption{Table for $\gamma, \eta$}
    \label{table:integrated.gamma.eta}
    \begin{subtable}{0.48\textwidth}
    \centering
        \caption{For $\gamma$: $(\lambda_0,s,q_0, \tilde{p}_\i) = (0,s_0+1,\frac{1}{2}^+,(-\frac{1}{2})^{+})$}
        \begin{tabular}{l|cc|r}
            \toprule 
             & $(\trch - \f 2r)(\gamma-\mathring{\gamma})$ & $(\trch - \f 2r)\gamma$ & Goal \\
            \midrule 
            $s$ & $(s_0+1)\wedge (s_0+1)$ & $s_0+1$ & $\geqslant s_0+1$ \\
            $q_0$ & $0^++\frac{1}{2}^+$ & $\frac{1}{2}^+$ & $\geqslant \frac{1}{2}^+$\\
            $q_\infty$ & \tiny{$(1^+ + (-\frac{1}{2})^{+})\wedge (\f 12^+ + 0)$} & $\f 12^+$ & $\geqslant \frac{1}{2}^{+}$ \\
            \bottomrule 
        \end{tabular}
    \end{subtable}
    \hfill
    \begin{subtable}{0.48\textwidth}
        \centering
        \caption{For $\jb{u}^{\frac{1}{2}} \eta$: $(\lambda_0,s,q_0,\tilde{p}_\i) = (1,s_0+1,(-\f 12)^+,0^{+})$}
        \begin{tabular}{l|cc|r}
            \toprule 
             & $\jb{u}^{\frac{1}{2}}\chih\cdot \eta$ & $\jb{u}^{\frac{1}{2}}\bt$ & Goal \\
            \midrule 
            $s$ & $(s_0+1) \wedge (s_0+1)$ & $s_0+1$ & \tiny{$\geqslant s_0+1$} \\
            $q_\infty$ & \tiny{$(1^+ + 0^+) \wedge (-\frac{1}{2}+\frac{1}{2}^+ + 1^+)$} & $1^+$ & $\geqslant 1^+$ \\
            \bottomrule 
        \end{tabular}
    \end{subtable}
\end{table}

Next, we consider $\rho,\sigma$ in Table~\ref{table:integrated.rho.sigma}. From now on, we only need to apply \eqref{eq:main.int.transest} since $s \leqslant s_{0}$.
    
   \begin{table}[h!]
    \centering
    \caption{Table for $\jb{u}^{\frac{1}{2}}(\rho, \sigma)$: $(\lambda_0,s,q_0,\tilde{p}_\i) = (\f 32,s_0,(-\f 32)^+,1^{+})$}
    \label{table:integrated.rho.sigma} 
        \begin{tabular}{l|ccc|r}
            \toprule 
            & $\jb{u}^{\frac{1}{2}} \eta \cdot \bt$ & $\jb{u}^{\frac{1}{2}}\chih\cdot \alp$ & $\jb{u}^{\frac{1}{2}} \nab \bt$ & Goal  \\
            \midrule 
            $s$ & $(s_0+1) \wedge (s_0+1)$ & $(s_0+1)\wedge (s_0+1)$ & $s_0$ & $\geqslant s_0$ \\
            $q_0$ & $(0^++(-\f 12)^+)\wedge (\f 12{}^++(-1)^+)$ & $(0^++(-\f 12)^+)\wedge (\f 12{}^++(-1)^+)$ & $(-\f 32)^+$  & $\geqslant (-\f 32)^+$\\
            $q_\infty$ & $(1^+ + 1^+)\wedge (0^+ + 2^+)$ & $(-\frac{1}{2}+1^++\frac{3}{2}^+)\wedge(-\frac{1}{2}+\frac{1}{2}^+ + 2^+)$ & $2^+$ & $\geqslant 2^{+}$ \\
            \bottomrule 
        \end{tabular}
    \end{table}

    We now turn to the connection coefficients $\chibh$, $\trchb + \f 2r$ and $\omb$ in Table~\ref{table:integrated.connection.the.rest}. From these estimates, we also obtain the bounds for $\betab$ using \eqref{eq:constchibh}. 
    
     \begin{table}[h!]

        \centering
        \caption{For $\jb{u}^{\frac{1}{2}}(\chibh, \trchb + \f 2r)$: $(\lambda_0,s,\tilde{p}_\i) = (\f 12,s_0,(-\f 12)^+,0^+)$; \\ for $\jb{u}^{\frac{1}{2}} \omb$: $(\lambda_0,s,q_0,\tilde{p}_\i) = (0,s_0,(-\f 12)^+, 0^{-\prime})$}
        \label{table:integrated.connection.the.rest}
        \begin{tabular}{l|ccccc|r}
            \toprule 
             & $\jb{u}^{\frac{1}{2}}(\eta \cdot\eta,\eta\hot \eta)$ & $\jb{u}^{\frac{1}{2}}(\chibh \cdot\chih, (\trchb + \f 2r) \chih)$ & $\jb{u}^{\frac{1}{2}}(\f 1r \chih, \f 1r (\trch - \f 2r))$ & $\jb{u}^{\frac{1}{2}}\nab \eta$ & $\jb{u}^{\frac{1}{2}}\rho$ & Goal \\
            \midrule 
            $s$ & $s_0\wedge s_0$ & $s_0\wedge (s_0+1)$ & $s_0+1$ & $s_0$ & $s_0$ & $\geqslant s_0$ \\
            $q_\infty$ & $1^+ + 0^+$ & \tiny{$((-\frac{1}{2})^{+} + 1 + \frac{1}{2}^{+})\wedge (0^+ + 1^+)$} & $-\frac{1}{2} + 1 + \frac{1}{2}^{+}$ & $1+ 0^+$ & $1^+$ & $\geqslant 1^+$ \\
            \bottomrule 
        \end{tabular}

\end{table}

    To estimate $K - r^{-2}$, we do \emph{not} use transport equations, but rather use the following equation (which is a restatement of \eqref{eq:constrho}):
    \begin{equation} \label{eq:constK}
    K - r^{-2} = -\rho + \f 12 \chih \cdot \chibh - \f 14 (\trch-\f 2r) (\trchb +\f 2r) + \f 1{2r} (\trch- \f 2r) - \f 1{2r} (\trchb+\f 2r).
    \end{equation}
    At this point, by \eqref{eq:constK} and our estimates so far, the desired estimate for $K-r^{-2}$ follows, specifically, $\| \jb{u}^{\frac{1}{2}} (K - r^{-2})\|_{s_0,(-\f 12)^+,1^+} \ls \ep^2$.

    We finally consider the bounds for $f-1$, $f-2\omb r - 1$, $\xib$. For $f-1$, $f-2\omb r - 1$, we use \eqref{gauge.con.metric} and \eqref{eq:f-2ombr-1}, respectively. The numerology is checked in Table~\ref{table:top.metric.the.rest}. Recalling \eqref{gauge.con}, we obtain the desired estimates for $\xib$ from those for $f-1$.

    \begin{table}[h!]
    \centering
    \caption{Tables for $f-1$, $f -2\omb r -1$}
    \label{table:top.metric.the.rest}

    \begin{subtable}{0.48\textwidth}
        \centering
        \caption{For $\jb{u}^{\frac{1}{2}} (f - 2\omb r -1)$: \\ $(\lambda_0,s,q_0,\tilde{p}_\i) = (0,s_0,\f 12{}^+,(-1)^+)$}
        \begin{tabular}{l|cc|r}
            \toprule 
             & $\jb{u}^{\frac{1}{2}} r|\eta|^2$ & $\jb{u}^{\frac{1}{2}} r\rho$ & Goal \\
            \midrule 
            $s$ & $(s_0+1) \wedge (s_0+1)$ & $s_0$ & $\geqslant s_0$ \\
            $q_0$ & $1+0^++\f 12{}^+$ & $1+(-\f 12)^+$ & $\geqslant \f 12{}^+$\\
            $q_\infty$ & $-1+1^{+}+0^{+}$ & $-1+1^+$ & $\geqslant 0^+$ \\
            \bottomrule 
        \end{tabular}
    \end{subtable}
    \hfill
    \begin{subtable}{0.48\textwidth}
        \centering
        \caption{For $\jb{u}^{\frac{1}{2}+\frac{3}{2}\delta} (f-1)$: \\ $(\lambda_0,s,q_0,\tilde{p}_\i) = (0,s_0,\f 12{}^+,(-1)^{-\prime})$}
        \begin{tabular}{l|c|r}
            \toprule 
             & $\jb{u}^{\frac{1}{2}+\frac{3}{2}\delta} \omb$ & Goal  \\
            \midrule 
            $s$ & $s_0$ & $\geqslant s_0$ \\
            $q_0$ & $\f 12{}^+$ & $\geqslant \f 12{}^+$\\
            $q_\infty$ & $0^{-\prime}$ & $\geqslant 0^{-\prime}$ \\
            \bottomrule 
        \end{tabular}
    \end{subtable}
\end{table}
    
    We use \eqref{eq:22.sphere.transest} and \eqref{eq:22.sphere.transest.2} to prove the mixed norms estimates for $\omb$ and $f-1$, i.e., we bound the last terms in \eqref{eq:BA.weak.integrated}. The estimate for $\omb$ is carried out using \eqref{eq:22.sphere.transest.2}, and the numerology is checked in  Table~\ref{table:top.L2Li}. Finally, using $\nab_4 (f-1) = 2\omb$ (see \eqref{gauge.con.metric}), \eqref{eq:22.sphere.transest} and the mixed norm bounds for $\omb$ that we just established, we obtain the desired mixed norm bounds for $f-1$.

    \begin{table}[h!]
    \centering
    \caption{Mixed norms estimates for $\jb{u}^{\frac{1}{2}} \omb$ using \eqref{eq:22.sphere.transest.2}: $(\lambda_0,s,q_0,p_\i) = (0,s_0,(-\f 12)^+,0)$}
    \label{table:top.L2Li}

        \begin{tabular}{l|cc|r}
            \toprule 
             & $\jb{u}^{\frac{1}{2}} |\eta|^2$ & $\jb{u}^{\frac{1}{2}} \rho$ & Goal \\
            \midrule 
            $s$ & $s_0$ & $s_0$ & $\geqslant s_0$ \\
            $q_0$ & $0^++\f 12{}^+$ & $(-\f 12)^+$ & $\geqslant (-\f 12)^+$\\
            $q_\infty$ & $1^+ + 0^+$ & $1^+$ & $\geqslant 1^{+}$ \\
            \bottomrule 
        \end{tabular}
\end{table}

    \pfstep{Step~3: Transport estimates for the strong bootstrap} Finally, we bound the remaining quantities in \eqref{eq:BA.strong.integrated}. Our primary tool is again the transport estimate \eqref{eq:main.int.transest}, but we now use the initial data bound \eqref{eq:data.strong.int} in addition to \eqref{eq:data.weak.int} (the latter is used only for $\chib$, $f - 2 \omegab r - 1$; see Table~\ref{table:chib.f}). Since the argument is similar to the previous ones, we omit the repetitive details and only sketch the main differences. 
    
    First, we observe that the $s$ and $q_0$ rows would be exactly the same as for the weak bootstrap in Step~2; hence, it suffices to check only the $q_\i$ row. For the $q_\i$ row, we begin with two observations:
     \begin{enumerate}[i)]
            \item The last index in the $\| \cdot \|$ norm for each geometric quantity in \eqref{eq:BA.strong.integrated} is at most $\frac{1}{2}^{+}$ less than the value of the last index in the $| \cdot |$ norm of the same quantity in \eqref{eq:BA.strong.pointwise} and \eqref{eq:BA.weak.pointwise} (for $\chibh, \trchb + 2 r^{-1}, f - 2 \omegab r - 1$).
            \item Except for $\alpha$ (which we already closed) and $K-r^{-2}$, $\chibh, \trchb + 2 r^{-1}, f - 2 \omegab r - 1$ (which will be treated later), the difference is exactly $\frac{1}{2}^{+}$. 
    \end{enumerate}
    
    With these observations, we claim that the desired estimates for all quantities in \eqref{eq:BA.strong.integrated} except for $K - r^{-2}, \chibh, \trchb + 2 r^{-1}, f - 2 \omegab r - 1$ follow from revisiting the $q_\i$ rows in Tables~\ref{table:chi}.(B), \ref{table:bt}.(B), \ref{table:gamma.eta.rho.sigma}.(C), \ref{table:gamma.eta.rho.sigma}.(E) in the proof of Proposition~\ref{prop:fixed.sphere}. Indeed, in view of ii), the goal $q_{\i}$ for each of these quantities is shifted down by $\frac{1}{2}^{+}$. Moreover, when we apply \eqref{eq:product.for.integrated} in lieu of \eqref{eq:product.for.pointwise} to estimate products, the $p_{\i}$ exponent of the quantity estimated in $\| \cdot \|$ (which is only one at a time) is shifted down by at most $\frac{1}{2}^{+}$ in view of i), which meets the goal. 

    To conclude the proof, it remains to estimate $K - r^{-2}, \chibh, \trchb + 2 r^{-1}, f - 2 \omegab r - 1$. The estimate for $K - r^{-2}$ follows from those proved so far from \eqref{eq:constK}; we note that $(2r)^{-1} (\trchb + 2 r^{-1})$ contributes the main term. Finally, we check the $q_\i$ exponent for the remaining quantities in Table~\ref{table:chib.f} below. We remark that the goal $q_\i$ exponents are determined by the initial data assumption \eqref{eq:BA.weak.integrated}, and the nonlinear estimates (tracked by the table) exceed the goals by a large margin.
    
    \begin{table}[h!]
    \centering
    \caption{Tables for $\chibh$, $\trchb + 2r^{-1}$, $f -2\omb r -1$}
    \label{table:chib.f}
     \begin{subtable}{0.96\textwidth}
        \centering
        \caption{Strong bootstrap for $\chibh, \trchb + 2r^{-1}$}
        \begin{tabular}{l|ccccc|r}
            \toprule 
             & $(\eta \cdot\eta,\eta\hot \eta)$ & $(\chibh \cdot\chih, (\trchb + \f 2r) \chih)$ & $(\f 1r \chih, \f 1r (\trch - \f 2r))$ & $\nab \eta$ & $\rho$ & Goal \\
            \midrule 
            $q_\infty$ & $2 + \frac{3}{2}^{-}$ & $(1+\frac{3}{2}^{-}) \wedge (\frac{1}{2}^{+} + 2)$ & $1 + \frac{3}{2}^{-}$ & $1 + \frac{3}{2}^{-}$ & $\frac{5}{2}^{-}$ & $\geqslant \frac{3}{2}^{+}$ \\
            \bottomrule 
        \end{tabular}
    \end{subtable}
    \bigskip 

    \begin{subtable}{0.48\textwidth}
        \centering
        \caption{Strong bootstrap for $f - 2\omb r -1$}
        \begin{tabular}{l|cc|r}
            \toprule 
             & $r|\eta|^2$ & $r\rho$ & Goal \\
            \midrule 
            $q_\infty$ & $(-1)+2+\frac{3}{2}^{-}$ & $(-1)+\frac{5}{2}^{-}$ & $\geqslant \frac{1}{2}^{+}$ \\
            \bottomrule 
        \end{tabular}
    \end{subtable}
\end{table}
    \end{proof}

\begin{remark}[Propagation of Bondi--Sachs peeling, II] \label{rem:strong-peeling-geom}
Together with Remark~\ref{rem:strong-peeling-rp}, the purpose of this remark is to provide more details of the propagation of Bondi--Sachs peeling in the context of the spacelike-characteristic initial value problem alluded to in Remark~\ref{rem:strong-peeling-spnull}. Observe that the modified bootstrap assumption \eqref{eq:BA.verystrong.integrated} follows from the same argument as in Step~3 of the preceding proposition. Here, the important difference is that the initial data for the transport estimates are only needed on $\Sigma = B_{1}(0)$, which allows us to use $q_{\infty}$ not permitted by \eqref{eq:data.strong.int} in the Cauchy data case.
\end{remark}

\section{Construction of the time function}\label{sec:tau}

Fix $\tilde{\delta}>0$ and $u_0< 1$. Let
\[
F(u,r) = \int_0^{u-u_0} \frac{dv}{\jb{v}^{1+\tilde{\delta}}} + \log(r + (u-u_0)/2).
\]
Recalling that $u+r\geqslant 1$, this is well defined in $S$. We compute
\[
\partial_r F = \frac{1}{r + (u-u_0)/2}, \qquad \partial_u F = \frac{1}{\jb{u-u_0}^{1+\tilde{\delta}}} + \frac{1}{2(r + (u-u_0)/2)}.
\]
Clearly in $u+r\geqslant 1$ we have $\partial_r F >0$ and
\[
f\partial_r F - 2\partial_u F = - \frac{2}{\jb{u-u_0}^{1+\tilde{\delta}}}\left( 1 - (f-1)\frac{\jb{u-u_0}^{1+\tilde{\delta}}}{r + (u-u_0)/2} \right).
\]
Observing that in $u+r\geqslant 1$
\[
\frac{\jb{u-u_0}^{1+\tilde{\delta}}}{r + (u-u_0)/2} \lesssim \frac{\jb{u}^{1+\tilde{\delta}}}{r}
\]
we see that $f\partial_r F - 2\partial_u F<0$ if $\jb{u}^{1+\tilde{\delta}}|f-1|_{2, 1^+, -1} < \epsilon$ for $\epsilon$ sufficiently small (depending on $\tilde{\delta}, u_0$). Thus the level sets of $F$ are spacelike in $u+r \geqslant 1$ under this assumption on $f$.

We wish to modify $F$ near $r=0$ so that the level sets agree with those of $u+r$ in this region. To do this, we observe that the level sets of $F$ are tangent to the level sets of $u+r$ when $\partial_r F = \partial_u F$, which occurs on the curve $\mathcal{C}$ given by $2r = \jb{u-u_0}^{1+\tilde{\delta}}-(u-u_0)$. We will choose $u_0$ so that $\mathcal{C}$ is tangent to $u+r=1$, our initial data surface. A brief calculation shows that if $u_1$ is the unique solution to 
\[
u_1\jb{u_1}^{\tilde{\delta}-1} = -(1+\tilde{\delta})^{-1}
\]
then taking
\[
u_0 = 1-\frac{1}{2}(u_1 + \jb{u_1}^{1+\tilde{\delta}}) <1,
\]
$\mathcal{C}$ meets $u+r=1$ tangentially at $u=u_0+u_1$, $r = 1-u_0-u_1 = (\jb{u_1}^{1+\tilde{\delta}}-u_1)/2>0$.

In the region defined by
\[
R_1 = \{(u, r)\in S:  u\geqslant u_0+u_1, \quad \jb{u-u_0}^{1+\tilde{\delta}}-(u-u_0) \geqslant 2r  \}
\]
we define $\tau = u+r$. Since in this region $r\lesssim \jb{u}^{1+\tilde{\delta}}$ we can bound $\abs{f-1} \lesssim \epsilon$ under the same assumption as above and so $\tau$ has spacelike level sets for $\epsilon$ sufficiently small.

We next wish to continue $\tau$ continuously to the region
\[
R_2 = S \setminus R_1
\]
such that the level sets of $\tau$ agree with those of $F$ in $R_2$, and $\tau$ is continuous across the boundary. Doing so will necessarily result in a $C^{1,1}_{loc.}$ function, smooth on $R_1\cup R_2 \setminus \mathcal{C}$. The line $u+r = \tau$ meets $\mathcal{C}$ when 
\[
2\tau = \jb{u-u_0}^{1+\tilde{\delta}}+u+u_0
\]
We can verify that  $\frac{d\tau}{du}>0$ for $u>u_0+u_1$, so we may invert this relation to give a map $U : [1, \infty) \to [u_0 +u_1, \infty)$ which is smooth on $(1, \infty)$ such that
\[
2\tau = \jb{U(\tau)-u_0}^{1+\tilde{\delta}}+U(\tau)+u_0
\]
We let
\[
G(\tau) = F(U(\tau), \tau -U(\tau)).
\]
Recalling that on $\mathcal{C}$ we have $\partial_r F = \partial_u F$, we can verify that
\[
G'(\tau) = \frac{\partial F}{\partial r}(U(\tau), \tau -U(\tau)) >0
\]
so that $G:[1,\infty) \to [F(u_0+u_1, 1-u_0-u_1), \infty)$ is a bijection, and $G^{-1}$ is smooth on $(F(u_0+u_1, 1-u_0-u_1), \infty)$. We also observe that since $\partial_rF-\partial_uF$ vanishes only on $\mathcal{C}$, and considering the limit as $r\to \infty$, we must have
\[
\frac{\ud}{\ud r}F(1-r, r) >0
\]
for $r> 1-u_0-u_1$. Together with the fact that $\partial_r F>0$ this implies that $F(u,r) \geqslant F(u_0+u_1, 1-u_0-u_1)$ for $(u,r) \in R_2$, with equality only when $(u,r)=(u_0+u_1, 1-u_0-u_1)$. Accordingly, in $R_2$ we can define $\tau$ by $G(\tau) = F(u,r)$, and by construction $\tau$ and its first derivative will be continuous across $\mathcal{C}$. We also note that $F(1-r, r) \to \infty$ as $r\to \infty$, so that each level set of $\tau = \tau_0>1$ meets the initial data surface at exactly one value of $r>0$.

\begin{figure}
    \centering
    \includegraphics[width=12cm]{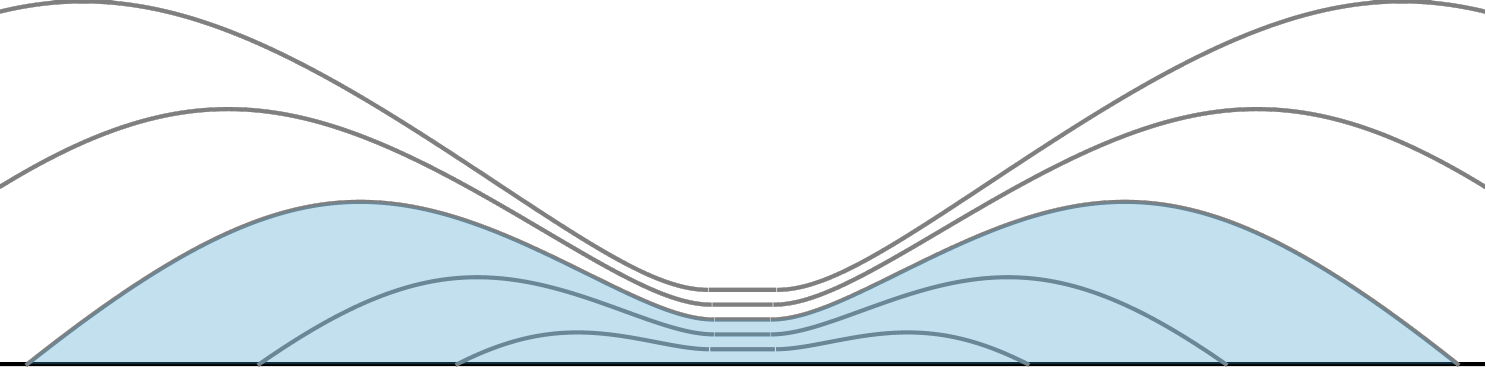}
    \caption{Level sets of $\tau$ and the region $S_T$ for $\tilde{\delta}=1$}
    \label{fig:placeholder} 
\end{figure}

\begin{lemma}
    For $\tilde{\delta}>0$, the function $\tau$ as constructed above is continuous on $\{u+r\geqslant 1\}$, $C^{1,1}_{loc.}$ on $\{u+r> 1\}$ and smooth on $\{u+r> 1\}\setminus \mathcal{C}$. For $T>1$, the region $S_T = \{(u,r) : u+r\geqslant 1, \tau \leqslant T\}$ is a simply connected lenticular domain with spacelike boundaries if $\jb{u}^{1+\tilde{\delta}}|f-1|_{2, 1^+, 1} < \epsilon$ for $\epsilon$ sufficiently small, depending on $\tilde{\delta}$.
\end{lemma}

\section{Control of geodesics} \label{sec:geodesics}
\subsection{Geodesic surfaces emanating from the axis} \label{sec:centre}

Fix $u_0 \in (1, T)$. Let $\dot{\Gamma}(u_0)^\perp = \{v \in T_{\Gamma(u_0)}M : g(v, \dot{\Gamma}(u_0)) = 0\}$ be the linear subspace of directions orthogonal to the axis at $\Gamma(u_0)$. Taking $U \subset \dot{\Gamma}(u_0)^\perp$ a sufficiently small open neighbourhood of the origin, it follows from elementary properties of the exponential map that $\exp_{\Gamma(u_0)}:U \to M$ is a smooth embedding, and that $N=\exp_{\Gamma(u_0)}(U)\subset M$ is a smooth spacelike hypersurface. In this section we prove:
\begin{theorem}\label{thm:geomcont1}
    Assume the bootstrap assumptions hold for $\epsilon$ sufficiently small and let $U = \{v \in \dot{\Gamma}(u_0)^\perp : g(v,v)<1/2\}$. Then  $N=N(u_0) = \exp_{\Gamma(u_0)}(U)$ is a smoothly embedded submanifold for as long as it remains in the image of $S_T$. Moreover, if $h, k$ are the first and second fundamental forms of the embedding then
    \[
    \|h-\delta\|_{H^{\frac{5}{2}^+}} + \|k\|_{H^{\frac{3}{2}^+}} \lesssim \epsilon^{\frac{1}{2}},
    \]
    with the implicit constant depending on $T$.
\end{theorem}

\subsubsection{$S^2-$adapted calculus on embedded surfaces} Suppose that $j: [0, \lambda_0)_\lambda \times S^2_{\varphi} \to S_T$ is a smooth embedding with coordinate maps $u = u(\lambda, \varphi^a), r = r(\lambda, \varphi^a), \theta^A = \theta^A(\lambda, \varphi^a)$, and let $N = j([0, \lambda_0)_\lambda \times S^2_{\varphi})$ be the image. We will use lower-case letters to denote indices associated to the $S^2_{\varphi}$ spheres in the domain and capital letters to denote 
indices associated to the $S^2_{\theta}$ spheres in the image. We write
\begin{align*}
j_*\left(\frac{\partial}{\partial \lambda} \right) &= \slashed{v}^Ae_A + v^3e_3 +v^4 e_4 \\
j_*\left(\frac{\partial}{\partial \varphi^a} \right) &= \slashed{w}^A_ae_A + w^3_ae_3 +w^4_a e_4
\end{align*}
so that
\begin{align}
    \slashed{v}^A &= \frac{\partial \theta^A}{\partial \lambda} -b^A \frac{\partial u}{\partial \lambda}  &v^3&=\frac{1}{2} \frac{\partial u}{\partial \lambda} & v^4&=\frac{\partial r}{\partial \lambda} +\frac{f}{2} \frac{\partial u}{\partial \lambda} \label{eq:Nembedd1}\\ \label{eq:Nembedd2}
    \slashed{w}^A_a &= \frac{\partial \theta^A}{\partial \varphi^a} -b^A \frac{\partial u}{\partial \varphi^a}  &w^3_a&=\frac{1}{2} \frac{\partial u}{\partial \varphi^a} & w^4_a&=\frac{\partial r}{\partial \varphi^a} +\frac{f}{2} \frac{\partial u}{\partial \varphi^a} 
\end{align}
For an $S$-tangent tensor field $\phi$ defined on $N\subset S_T$ we let
\begin{align}
    \frac{\slashed{D}}{\partial\lambda} \phi &= \slashed{v}^A \nab_A \phi + v^3 \nab_3 \phi + v^4 \nab_4 \phi\nonumber \\
    \frac{\slashed{D}}{\partial\varphi^a}  \phi &= \slashed{w}^A_a \nab_A \phi + w^3_a \nab_3 \phi + w^4_a \nab_4 \phi \label{eq:sDdef}
\end{align}
where the derivatives can be computed on any smooth extension of $\phi$ away from $N$ and will not depend on this extension. We can compute from \eqref{eq:Nembedd1}, \eqref{eq:Nembedd2} that\footnote{Strictly we should write $\eta^A\circ j$ instead of $\eta^A$, etc.\ here and below, but for notational clarity we will assume the embedding map can be inserted from context.}
\begin{align}
    \frac{\slashed{D}}{\partial\varphi^a} \slashed{v}^A - \frac{\slashed{D}}{\partial\lambda}  \slashed{w}^A_a &= 4(v^3 w^4_a - v^4 w^3_a) \eta^A + (w^3_a \slashed{v}^B - v^3 \slashed{w}_a^B)\chib_B{}^A+ (w^4_a \slashed{v}^B - v^4 \slashed{w}_a^B)\chi_B{}^A \nonumber\\
    \frac{\slashed{D}}{\partial\varphi^a} v^3 - \frac{\slashed{D}}{\partial\lambda}  w^3_a &=0 \label{eq:vwcom}\\
    \frac{\slashed{D}}{\partial\varphi^a} v^4 - \frac{\slashed{D}}{\partial\lambda}  w^4_a &= 2 (v^3 w^4_a - v^4 w^3_a)\omegab - (w^3_a \slashed{v}^B - v^3 \slashed{w}_a^B)\xib_B \nonumber
\end{align}
Making use of the identity
\begin{align*}
    \left(\frac{\slashed{D}}{\partial\varphi^a}\frac{\slashed{D}}{\partial\lambda} -\frac{\slashed{D}}{\partial\lambda} \frac{\slashed{D}}{\partial\varphi^a} \right) \phi &= \left(\frac{\slashed{D}}{\partial\varphi^a} \slashed{v}^A - \frac{\slashed{D}}{\partial\lambda}  \slashed{w}^A_a \right) \nab_A \phi + \left( \frac{\slashed{D}}{\partial\varphi^a} v^3 - \frac{\slashed{D}}{\partial\lambda}  w^3_a\right) \nab_3 \phi + \left( \frac{\slashed{D}}{\partial\varphi^a} v^4 - \frac{\slashed{D}}{\partial\lambda}  w^4_a\right) \nab_4 \phi \\
    &\qquad + (w^3_a v^4 - w_a^4 v^3) [\nab_3, \nab_4]\phi+ (w^3_a \slashed{v}^A - \slashed{w}_a^A v^3) [\nab_3, \nab_A]\phi \\
    &\qquad +  (w^4_a \slashed{v}^A - \slashed{w}_a^A v^4) [\nab_4, \nab_A]\phi+  (\slashed{w}^A_a \slashed{v}^B - \slashed{w}_a^B \slashed{v}^A) [\nab_A, \nab_B]\phi
\end{align*}
which follows from \eqref{eq:sDdef} we deduce
\begin{align}
    \left(\frac{\slashed{D}}{\partial\varphi^a}\frac{\slashed{D}}{\partial\lambda} -\frac{\slashed{D}}{\partial\lambda} \frac{\slashed{D}}{\partial\varphi^a} \right) \phi &= (w^3_a v^4 - w_a^4 v^3) \left([\nab_3, \nab_4]\phi - 2 \omegab \nab_4 \phi - 4 \eta^A \nab_A \phi \right) \nonumber \\
    &\qquad + (w^3_a \slashed{v}^A - \slashed{w}_a^A v^3) \left( [\nab_3, \nab_A]\phi - \xib_A \nab_4\phi +\chib_A{}^B\nab_B\phi \right)\label{eq:sDcomm1} \\
    &\qquad + (w^4_a \slashed{v}^A - \slashed{w}_a^A v^4)\left( [\nab_4, \nab_A]\phi +\chi_A{}^B\nab_B\phi\right)+  (\slashed{w}^A_a \slashed{v}^B - \slashed{w}_a^B \slashed{v}^A) [\nab_A, \nab_B]\phi. \nonumber
\end{align}
We note that all terms involving derivatives of $\phi$ in fact cancel on the right-hand side by Lemma \ref{lem:CK.733}, so that the commutator depends only on the $0-$jet of $\phi$.

In order to control the embedding map it will be convenient to extend our covariant derivatives to act on objects carrying $S^2_{\varphi}$-tangent indices as well. Suppose that $\phi^{a_1\ldots a_k}{}_{b_1 \ldots b_l}{}^{A_1\ldots A_K}{}_{B_1\ldots B_L}$ is a tensor carrying both $S^2_{\varphi}$ and $S^2_{\theta}$ tangent indices. Let $\mathring{\gamma}(\lambda) = \mathring{\gamma}_{ab}(\lambda) \ud \varphi^a \ud \varphi^b$ be the metric on the round sphere of radius $\lambda$, and $\mathring{\Gamma}_a{}^{b}{}_c$ the associated Christoffel symbols. We define (suppressing the $S^2_{\theta}$ tangent indices):
\begin{align}
    (\slashed{D}_\lambda \phi)^{a_1\ldots a_k}{}_{b_1 \ldots b_l} &= \frac{\slashed D}{\partial\lambda}\phi ^{a_1\ldots a_k}{}_{b_1 \ldots b_l} + \frac{k-l}{\lambda} \phi^{a_1\ldots a_k}{}_{b_1 \ldots b_l} \\
    (\slashed{D}_a \phi)^{a_1\ldots a_k}{}_{b_1 \ldots b_l} &= \frac{\slashed D}{\partial\varphi^a}\phi ^{a_1\ldots a_k}{}_{b_1 \ldots b_l} + \mathring{\Gamma}_a{}^{a_1}{}_c\phi ^{ca_2\ldots a_k}{}_{b_1 \ldots b_l}+\ldots+   \mathring{\Gamma}_a{}^{a_k}{}_c\phi ^{a_1\ldots a_{k-1}c}{}_{b_1 \ldots b_l}\nonumber \\
    &\qquad - \mathring{\Gamma}_a{}^{c}{}_{b_1} \phi^{a_1\ldots a_k}{}_{cb_2 \ldots b_l}-\ldots - \mathring{\Gamma}_a{}^{c}{}_{b_l} \phi^{a_1\ldots a_k}{}_{b_1 \ldots b_{l-1}c}
\end{align}
By construction these derivatives commute with contraction of $S^2_{\varphi}$-tangent indices among themselves and contraction of $S^2_{\theta}$-tangent indices among themselves, and furthermore annihilate $\mathring{\gamma}_{ab}(\lambda)$ and $\gamma_{AB}$. We also note that for any $\phi$ with both $S^2_{\varphi}$ and $S^2_{\theta}$ tangent indices
\begin{equation}
\slashed{D}_\lambda(\lambda \slashed{D}_a \phi) - \lambda \slashed{D}_a(\slashed{D}_\lambda\phi) = \lambda\left(\frac{\slashed{D}}{\partial\varphi^a}\frac{\slashed{D}}{\partial\lambda} -\frac{\slashed{D}}{\partial\lambda} \frac{\slashed{D}}{\partial\varphi^a} \right) \phi \label{eq:sDcomm2}
\end{equation}
and we have already determined the right hand side above.

We introduce the pointwise norm
\begin{align*}
|\phi|^N_0(\lambda, \varphi) &= \Bigg(\phi^{a_1\ldots a_k}{}_{b_1 \ldots b_l}{}^{A_1\ldots A_K}{}_{B_1\ldots B_L}\times \phi^{a'_1\ldots a'_k}{}_{b'_1 \ldots b'_l}{}^{A'_1\ldots A'_K}{}_{B'_1\ldots B'_L} \times \mathring{\gamma}_{a_1a'_1} \times\cdots \times \mathring{\gamma}_{a_ka'_k}\times\\
&\qquad \times \mathring{\gamma}^{b_1b'_1} \times\cdots \times \mathring{\gamma}^{b_lb'_l}\times \mathring{\gamma}_{A_1A'_1} \times\cdots \times \mathring{\gamma}_{A_KA'_K}\times \mathring{\gamma}^{B_1B'_1} \times\cdots \times \mathring{\gamma}^{B_LB'_L}\Bigg)^{\frac{1}{2}},
\end{align*}
where all contractions are with respect to the metric $\mathring{\gamma}_{ab}$. For $s\geqslant 1$ and $\phi$ sufficiently regular we define
\[
|\phi|^N_s(\lambda, \varphi) = \left(|\lambda \slashed{D}_\lambda \phi|^N_{s-1}+|\lambda \slashed{D} \phi|^N_{s-1} + | \phi|^N_{s-1}\right)^{\frac{1}{2}}.
\]
Finally, for $s\geqslant 0$ and $p_0 \in \mathbb{R}$ we define
\[
    |\phi|^N_{s,p_0}(\lambda_0) = \sup_{(0, \lambda_0)\times S^2} \frac{|\phi|^N_s(\lambda, \varphi) }{\lambda^{p_0}}.
\]

\begin{remark} We can formalise the discussion above in the language of vector bundles. $TS$ has a smooth subbundle consisting of vectors tangent to the surfaces $S_{u,r}$, call this bundle $\mathbb{S}S$. $S^2_\theta-$tangent vector fields are sections of this bundle. Forming tensor products of $\mathbb{S}S$ and its dual $\mathbb{S}^*S$ we have the bundle $\mathbb{S}^K{}_LS$ of rank $(K, L)$ $S^2_\theta-$tangent tensors. We define a connection $\underline{\nab}$
on $\mathbb{S}^K{}_LS$ by setting
\[
\underline{\nab}_X \phi = X^A\nab_A \phi + X^3\nab_3 \phi + X^4\nab_4 \phi
\]
where $X = X^Ae_A + X^3e_3+X^4e_4 \in T_pS$ and $\phi\in \Gamma\mathbb{S}^K{}_LS$. Similarly, if $R = \mathbb{R}_\lambda \times S^2_{\varphi}$ then we can define the bundle $\mathbb{S}^k{}_lR$ of rank $(k, l)$ $S^2_{\varphi}-$tangent tensors over $R$. We equip $\mathbb{S}^k{}_lR$ with the connection obtained by projecting Levi-Civita connection associated to the Euclidean metric $d\lambda^2 + \mathring{\gamma}(\lambda)$ onto the $S^2_{\varphi}$ spheres. Given a (local) immersion $j:R\to S$, we can construct the pullback bundle $j^*\mathbb{S}^K{}_LS$, which carries a natural connection, the pullback of $\underline{\nab}$. Objects on $R$ with mixed $S^2_{\varphi}$ and $S^2_{\theta}$ tangent indices should be thought of as local sections of $\mathbb{S}^k{}_lR \times j^*\mathbb{S}^K{}_LS$, which carries the product connection\footnote{Given this, the observation immediately following \eqref{eq:sDcomm1} should not be surprising.} $\underline{\slashed{D}}$, which is related to our previous definitions by
\[
\underline{\slashed{D}}_Y\phi = Y^\lambda \slashed{D}_\lambda \phi + Y^a\slashed{D}_a \phi
\]
where $Y = Y^\lambda \partial_\lambda +Y^a \partial_{\varphi^a}\in T_qR$ and $\phi \in \Gamma(\mathbb{S}^k{}_lR \times j^*\mathbb{S}^K{}_LS)$. Finally, we endow $\mathbb{S}^k{}_lR \times j^*\mathbb{S}^K{}_LS$ with the bundle metric induced by $\mathring{\gamma}(\lambda)$ in each fibre. Note that with this definition $\underline{\slashed{D}}$ is not a metric connection.
\end{remark}

\subsubsection{Geodesic equations}

 We wish to control the embedding functions of $N$ in our $(u, r, \theta)$ coordinate system.

Let $\hat{\gamma}_{u_0, \varphi}:[0, \lambda_0)\to TM$ be the integral curve of the geodesic vector field with $\hat{\gamma}_{u_0, \varphi}(0) = (\Gamma(u_0), V(u_0, \varphi) - \dot{\Gamma}(u_0)) \in TM$. The map taking $(\lambda, \varphi) \in (0, \lambda_0)\times S^2$ to $\gamma_{u_0, \varphi}(\lambda) = \Pi\circ \hat{\gamma}_{u_0, \varphi}(\lambda)$ for $\lambda_0$ sufficiently small is a smooth parameterisation of a punctured neighbourhood of $\Gamma(u_0)$ in $N$. For $(\lambda, \varphi) \in (0, \lambda_0)\times S^2$ we let $j(\lambda, \varphi)= \psi^{-1}\circ\Pi\circ \hat{\gamma}_{u_0, \varphi}(\lambda) \in S_T$. This smoothly extends to a map $j: [0, \lambda_0)\times S^2 \to S_T$ by setting $j(0, \varphi) = (u_0, 0, \varphi)$.

The condition that $\lambda \mapsto j(\lambda, \varphi)$ is geodesic for fixed $\varphi$ implies the equations
\begin{align}
    \slashed{D}_\lambda \slashed{v} &= -  (\chib \cdot \slashed{v}) v^3-  (\chi \cdot \slashed{v}) v^4 - 2 \xib v^3 v^3, \label{eq:geodsvA} \\
    \slashed{D}_\lambda v^3  &=- \frac{1}{2}\slashed{v}\cdot\chi\cdot \slashed{v} - 2 \slashed{v}\cdot\eta v^3  + 2\omegab v^3 v^3, \label{eq:geodsv3} \\
    \slashed{D}_\lambda v^4 &= - \frac{1}{2}\slashed{v}\cdot \chib \cdot\slashed{v} + 2 \slashed{v}\cdot \eta v^4- 2\omegab v^3 v^4 - \slashed{v}\cdot \xib v^3,\label{eq:geodsv4} 
\end{align}
which form a closed system of ODEs together with \eqref{eq:Nembedd1}. The initial conditions for the geodesics imply
\begin{align*}
\theta^A(0, \varphi) &= \varphi^A \quad u(0, \varphi) = u_0, \quad r(0, \varphi) = 0,  \\
\slashed{v}(0, \varphi) &= 0 \quad v^3(0, \varphi) = -\frac{1}{2} \quad v^4(0, \varphi) = \frac{1}{2}
\end{align*}
and we are guaranteed a solution with these initial conditions exists and depends smoothly on $\lambda, \varphi$ provided the curve does not exit $S_T$.

\begin{theorem}\label{thm:geomcont}
    Assume the bootstrap assumptions hold. Let 
    \[\lambda_{0} = \sup \{0<\lambda'\leqslant 1/2: j(\lambda, \varphi)\in S_T, \, \forall\, \lambda <\lambda', \varphi \in S^2\}
    \]
    and assume $s_1 \leqslant s_0/2-2$. Then for $\epsilon<\epsilon_0$ sufficiently small the estimates
    \begin{align}
        \abs{\frac{r}{\lambda}-1}^N_{s_1, 1^+}(\lambda_0)&\leqslant \epsilon^\frac{1}{2} & \abs{\frac{u-u_0}{\lambda}+1}^N_{s_1, 1^+}(\lambda_0)&\leqslant \epsilon^\frac{1}{2}  & \sup_{(0,\lambda_0)\times S^2} \lambda^{-2-\delta} d_{S^2, \mathring{\gamma}}(\theta, \varphi) &\leqslant \epsilon^\frac{1}{2} \nonumber \\
        |\slashed{v}|^N_{s_1, 1^+}(\lambda_0) &\leqslant \epsilon^\frac{1}{2} & \abs{v^3 +\frac{1}{2}}^N_{s_1, 1^+}(\lambda_0) &\leqslant \epsilon^\frac{1}{2} & \abs{v^4 -\frac{1}{2}}^N_{s_1, 1^+}(\lambda_0) &\leqslant \epsilon^\frac{1}{2}\label{eq:geodests1}\\
        |\slashed{w} - \delta|^N_{s_1-1, 1^+}(\lambda_0) &\leqslant \epsilon^\frac{1}{2} & \abs{w^3}^N_{s_1-1, 1^+}(\lambda_0) &\leqslant \epsilon^\frac{1}{2} &\abs{w^4}^N_{s_1-1, 1^+}(\lambda_0) &\leqslant \epsilon^\frac{1}{2}\nonumber
    \end{align}
    hold. Hence for $u_0<T-\epsilon^{\frac{1}{2}}/2$, we have $\lambda_0 = 1/2$.
\end{theorem}
Proof is by a further bootstrap argument and requires some intermediate lemmata. Let
\[
\Lambda = \sup\{0<\lambda'\leqslant \lambda_0: \eqref{eq:geodests1} \textrm{ hold with $\lambda_0$ replaced by $\lambda'$} \}.
\]
Examining the right-hand side of the geodesic equations, by smoothness we see that $\Lambda >0$.

\begin{lemma}
    Suppose $\phi = \phi^{A_1\ldots A_K}{}_{B_1\ldots B_L}$ is a sufficiently regular $S^2_{\theta}$-tangent tensor on $S_T$ and $s\leqslant s_1$, then
    \[
    \abs{\phi}^N_{s, p_0}(\Lambda) \lesssim \abs{\phi}_{2s+2, p_0, p_\infty}
    \]
    for any $p_\infty$, where the implicit constant depends on $T, s, p_\infty, \epsilon_0$.
\end{lemma}
\begin{proof}
        Work by induction on $s$. For $s=0$ we note that
        \[
        \frac{\abs{\phi}^N_{0}(\lambda, \varphi)}{\lambda^{p_0}} = \frac{1}{\lambda^{p_0}}\left(\frac{r}{\lambda} \right)^{K-L} \abs{\phi}_{\mathring{\gamma}}(r, u, \theta) \lesssim \left(\frac{r}{\lambda} \right)^{K-L+p_0}\frac{1}{(\jb{u_0}+r)^{p_0+p_\infty}} \abs{\phi}_{2, p_0, p_\infty}\lesssim \abs{\phi}_{2, p_0, p_\infty}.
        \]
        Taking the supremum over $\lambda, \varphi$ gives the $s=0$ case. For the induction step we note that
        \begin{align*}
            \lambda \slashed{D}_\lambda \phi &= \frac{\lambda}{r} \left(\slashed{v}^A r\nab_A \phi + v^3 r\nab_3 \phi + v^4 r\nab_4 \phi \right), \\
            \lambda \slashed{D}_a \phi &= \frac{\lambda}{r} \left(\slashed{w}^A_a r\nab_A \phi + w^3_a r\nab_3 \phi + w^4_a r\nab_4 \phi \right).
        \end{align*}
        Suppose for induction the result holds for some $s<s_1$. We have
        \[
        \abs{\lambda \slashed{D}_\lambda \phi}^N_{s, p_0}(\Lambda)+\abs{\lambda \slashed{D} \phi}^N_{s, p_0}(\Lambda) \lesssim \abs{r\nab_3 \phi}_{2s+2, p_0, p_\infty}+\abs{r\nab_4 \phi}_{2s+2, p_0, p_\infty}+\abs{r\nab \phi}_{2s+2, p_0, p_\infty}\lesssim \abs{ \phi}_{2s+4, p_0, p_\infty}
        \]
        which recovers the induction hypothesis for $s+1$.
\end{proof}
\begin{corollary}\label{cor:sdComm}
    Suppose $\phi = \phi^{a_1\ldots a_k}{}_{b_1 \ldots b_l}{}^{A_1\ldots A_K}{}_{B_1\ldots B_L}$ is a sufficiently regular tensor with $S^2_{\varphi}$ and $S^2_{\theta}-$tangent indices, then for $s\leqslant s_1-1$
    \[
    \abs{[\slashed{D}_\lambda, \lambda \slashed{D}] \phi}^N_{s, p_0}(\Lambda) \lesssim \epsilon^{\frac{1}{2}} \abs{\phi}^N_{s, p_0}(\Lambda).
    \]
\end{corollary}
\begin{proof}
    From \eqref{eq:sDcomm2} it suffices to estimate $\lambda \times $\eqref{eq:sDcomm1}. Doing so with the results of \ref{lem:CK.733} we see that with the bootstrap assumptions we are able to bound all of the terms. For example
    \begin{align*}
    \abs{(\slashed{w}^A_a \slashed{v}^B - \slashed{w}_a^B \slashed{v}^A) [\nab_A, \nab_B]\phi}^N_{s}(\lambda, \varphi) &\lesssim |\slashed{w}|^N_s \times \frac{1}{\lambda}|\slashed{v}|^N_s \times |r^2 K|_{s, 0, 0} \times \left(\frac{\lambda}{r}\right)^2 \times \abs{\phi}^N_{s}(\lambda, \varphi) \\
    & \lesssim \epsilon^{\frac{1}{2}} \abs{\phi}^N_{s}(\lambda, \varphi).
    \end{align*}
    We can control the other terms in \eqref{eq:sDcomm1} similarly, whence the result follows.
\end{proof}
With this result we can now establish
\begin{lemma}\label{lem:lamint}
    Suppose $\phi = \phi^{a_1\ldots a_k}{}_{b_1 \ldots b_l}{}^{A_1\ldots A_K}{}_{B_1\ldots B_L}$ is a sufficiently regular tensor with $S^2_{\varphi}$ and $S^2_{\theta}-$tangent indices, and that $|\phi|_{0}^N(\lambda, \varphi)\to0$ as $\lambda \to 0$. Then for $s\leqslant s_1$ and $p_0>-1$ we have
    \[
    \abs{\phi}^N_{s, p_0}(\Lambda) \lesssim \abs{\slashed{D}_\lambda \phi}^N_{s, p_0-1}(\Lambda).
    \]
\end{lemma}
\begin{proof}
    Let
    \begin{align*}
    \langle \phi, \psi\rangle_{\mathring{\gamma}, \gamma}(\lambda, \varphi) &= \phi^{a_1\ldots a_k}{}_{b_1 \ldots b_l}{}^{A_1\ldots A_K}{}_{B_1\ldots B_L}\times \psi^{a'_1\ldots a'_k}{}_{b'_1 \ldots b'_l}{}^{A'_1\ldots A'_K}{}_{B'_1\ldots B'_L} \times \mathring{\gamma}_{a_1a'_1} \times\cdots \times \mathring{\gamma}_{a_ka'_k}\times\\&\qquad
 \times \mathring{\gamma}^{b_1b'_1} \times\cdots \times \mathring{\gamma}^{b_lb'_l}\times {\gamma}_{A_1A'_1} \times\cdots \times {\gamma}_{A_KA'_K}\times {\gamma}^{B_1B'_1} \times\cdots \times {\gamma}^{B_LB'_L},
    \end{align*}
    i.e., the inner product with $S^2_{\varphi}-$tangent indices contracted with $\mathring{\gamma}_{ab}(\lambda)$ and $S^2_{\theta}-$tangent indices contracted with $\gamma_{AB}\circ j$. Then it follows that
    \[
    \frac{\partial}{\partial \lambda}  \langle \phi, \psi\rangle_{\mathring{\gamma}, \gamma} = \langle \slashed{D}_\lambda \phi, \psi\rangle_{\mathring{\gamma}, \gamma}+\langle \phi, \slashed{D}_\lambda  \psi\rangle_{\mathring{\gamma}, \gamma}
    \]
    and
    \[
    \langle \phi, \phi\rangle_{\mathring{\gamma}, \gamma}^{\frac{1}{2}} \sim |\phi|_0^N
    \]
    We compute
    \[
    \frac{\partial}{\partial \lambda} \left( \langle \phi, \phi\rangle_{\mathring{\gamma}, \gamma} +\varepsilon\right)^\frac{1}{2} = \frac{\langle \phi, \slashed{D}_\lambda\phi\rangle_{\mathring{\gamma}, \gamma}}{\left( \langle \phi, \phi\rangle_{\mathring{\gamma}, \gamma} +\varepsilon\right)^\frac{1}{2}}\lesssim \lambda^{p_0}\abs{\slashed{D}_\lambda\phi}_{0, p_0}^N(\Lambda)
    \]
    Integrating from $0$ to $\lambda'<\Lambda$ and sending $\epsilon \to 0$ we deduce
    \[
    |\phi|_{0, p_0}^N(\Lambda) \lesssim \abs{\slashed{D}_\lambda\phi}_{0, p_0}^N(\Lambda)
    \]
    whence the result with $s=0$ follows. Now assume for induction that the result has been established for some $s<s_1$. Clearly
    \[
    |\lambda \slashed{D}_{\lambda} \phi|_{s, p_0}^N(\Lambda) \lesssim \abs{\slashed{D}_\lambda(\lambda \slashed{D}_{\lambda} \phi)}_{s, p_0-1}^N(\Lambda)\lesssim \abs{\slashed{D}_\lambda \phi}_{s+1, p_0-1}^N(\Lambda).
    \]
    Furthermore, 
    \begin{align*}
        |\lambda \slashed{D} \phi|_{s, p_0}^N(\Lambda) &\lesssim  \abs{\slashed{D}_\lambda(\lambda \slashed{D} \phi)}_{s, p_0-1}^N(\Lambda) \lesssim  \abs{\slashed{D}_\lambda \phi}_{s+1, p_0-1}^N(\Lambda) + \abs{[\slashed{D}_\lambda, \lambda \slashed{D}] \phi)}_{s, p_0-1}^N(\Lambda) \\
        & \lesssim \abs{\slashed{D}_\lambda \phi}_{s+1, p_0-1}^N(\Lambda) + \varepsilon^{\frac{1}{2}}   \abs{ \phi}_{s, p_0-1}^N(\Lambda)  \lesssim  \abs{\slashed{D}_\lambda \phi}_{s+1, p_0-1}^N(\Lambda)
    \end{align*}
    where we have used Corollary \ref{cor:sdComm} and the induction hypothesis in the last line. This recovers the induction assumption for $s+1$.
\end{proof}
With these results in hand we can now establish 
\begin{proof}[Proof of Theorem \ref{thm:geomcont}]
    From \eqref{eq:geodsvA}-\eqref{eq:geodsv4}, the boostrap assumptions and the lemmas above, we can immediately estimate
    \begin{align*}
        |\slashed{v}|^N_{s_1, 1^+}(\Lambda) &\lesssim \epsilon & \abs{v^3 +\frac{1}{2}}^N_{s_1, 1^+}(\Lambda) &\lesssim \epsilon & \abs{v^4 -\frac{1}{2}}^N_{s_1, 1^+}(\Lambda) &\lesssim \epsilon
    \end{align*}
    Rewriting \eqref{eq:Nembedd1} as
    \[
    \frac{\partial \theta^A}{\partial \lambda} = \slashed{v}^A + 2b^A v^3, \qquad \frac{\partial u}{\partial \lambda} = 2 v^3, \qquad \frac{\partial r}{\partial \lambda} = v^4-f v^3
    \]
    we can then estimate
    \begin{align*}
        \abs{\frac{r}{\lambda}-1}^N_{s_1, 1^+}(\Lambda)&\lesssim \epsilon & \abs{\frac{u-u_0}{\lambda}+1}^N_{s_1, 1^+}(\Lambda)&\lesssim \epsilon  & \sup_{(0,\Lambda)\times S^2} \lambda^{-2-\delta} d_{S^2, \mathring{\gamma}}(\theta, \varphi) &\lesssim \epsilon
    \end{align*}
    where for the final estimate we note that we may think of $\lambda \mapsto \theta(\lambda, \varphi)$ as describing a curve on the round \emph{unit} sphere with speed bounded by $\epsilon \lambda^\delta$, hence the intrinsic distance between initial and final points is bounded by $\epsilon \lambda^{1+\delta}$ which gives the result after rescaling back to the sphere of radius $\lambda$. To recover the bounds on $w$, we first note
    \[
    \frac{\slashed{D}}{d\lambda} \delta_a{}^A = \chib^A{}_a v^3 + \chi^A{}_a v^4 - 2 \frac{\partial b^A}{\partial \theta^a} v^3 + \Gamma_B{}^A{}_a \slashed{v}^B
    \]
    which implies
    \[
    \abs{\slashed{D}_{\lambda} \left(\frac{\lambda}{r} \delta_a{}^A\right)}^N_{s_1, 0^+} (\Lambda)\lesssim \epsilon.
    \]
    Combining this with \eqref{eq:vwcom} we can estimate
    \begin{align*}
        \abs{\slashed{D}_\lambda\left(\frac{\lambda}{r} [\slashed{w}-\delta]\right)}^N_{s_1-1, 0^+}(\Lambda) &\lesssim \epsilon&\abs{\slashed{D}_\lambda  (\lambda w^3)}^N_{s_1-1, 1^+}(\Lambda)&\lesssim \epsilon&\abs{\slashed{D}_\lambda  (\lambda w^4)}^N_{s_1-1, 1^+}(\Lambda) &\lesssim \epsilon
    \end{align*}
 Applying Lemma \ref{lem:lamint} we deduce 
\begin{align*}
    |\slashed{w} - \delta|^N_{s_1-1, 1^+}(\lambda_0) &\lesssim \epsilon & \abs{w^3}^N_{s_1-1, 1^+}(\lambda_0) &\lesssim \epsilon &\abs{w^4}^N_{s_1-1, 1^+}(\lambda_0) &\lesssim \epsilon.
\end{align*}
    
    The implicit constants in all the estimates above depend on $s_1, T$ but not $\epsilon$. It follows that for $\epsilon$ sufficiently small we contradict the maximality of $\Lambda$ unless $\Lambda = \lambda_0$ and we are done.
\end{proof}

\begin{theorem}\label{thm:indmetest}
    Let $h = j^*g$ be the induced metric on $N$. Write
    \[
    h = \mathfrak{f} \ud \lambda^2+2 \mathfrak{b}_a \ud \varphi^a \ud \lambda + \mathfrak{g}_{ab}\ud \varphi^a \ud \varphi^b.
    \]
    Then we can estimate
    \[
    \abs{\mathfrak{f}-1}^N_{s_1-1, 1^+}+\abs{\mathfrak{b}}^N_{s_1-1, 1^+}+\abs{\mathfrak{g} - \mathring{\gamma}(\lambda)}^N_{s_1-1, 1^+} \lesssim \epsilon^{\frac{1}{2}}.
    \]
\end{theorem}
\begin{proof}
    We have
    \begin{align*}
    \mathfrak{f} = g(j_*\partial_\lambda, j_*\partial_\lambda) &= \slashed{v}^A\slashed{v}^B \gamma_{AB} - 4 v^3 v^4 \\
    &=1 +\slashed{v}^A\slashed{v}^B \gamma_{AB} - (2v^3+1)+(2v^4-1) - (2v^3+1)(2v^4-1)
    \end{align*}
    which we can estimate directly. Next
    \begin{align*}
    \mathfrak{g}_{ab} = g(j_*\partial_{\varphi^a}, j_*\partial_{\varphi^b}) &= \slashed{w}^A_a\slashed{w}^B_b \gamma_{AB} - 2 w_a^3 w_b^4- 2 w_b^3 w_a^4 \\
    &=\mathring{\gamma}_{ab}(\lambda) +(\slashed{w}_a^A-\delta_a^A)(\slashed{w}_b^B-\delta_b^B)\gamma_{AB} +\delta_a^A(\slashed{w}_b^B-\delta_b^B)\gamma_{AB}+\delta_b^B(\slashed{w}_a^A-\delta_a^A)\gamma_{AB}\\
    &\qquad + \delta_a^A\delta_b^B(\gamma_{AB} - \mathring{\gamma}_{AB}(r)) + \left(\frac{r^2}{\lambda^2}-1 \right)\mathring{\gamma}_{ab}(\lambda)- 2 w_a^3 w_b^4- 2 w_b^3 w_a^4
    \end{align*}
    which can again be directly estimated. Finally,
    \begin{align*}
    \mathfrak{b}_{a} = g(j_*\partial_{\varphi^a}, j_*\partial_{\lambda}) = w_a^Av^B\gamma_{AB} - 2 w_a^3v^4-2w_a^4v^3
    \end{align*}
    which again can be estimated directly.
\end{proof}
\subsubsection{Second fundamental form} In order to bound the second fundamental form of $N$, we first need to determine the normal to the surface. This will be given by
\[
n = -\kappa [\Omega(j_*\partial_{\varphi^1}, j_*\partial_{\varphi^2},j_*\partial_\lambda, \cdot)]^\sharp
\]
where $\Omega$ is the spacetime volume-form which we can take to be given by
\[
\Omega = 2 \in_{AB} e^A \wedge e^B \wedge e^3 \wedge e^4
\]
and $\kappa$ is a normalisation constant. Observing that we have
\begin{align*}
j_*\partial_{\varphi^a} &= \delta_a{}^A e_A + (\slashed{w}_a^A -\delta_a{}^A)e_A + w_a^3e_3 + w_a^4e_4  \\
j_*\partial_\lambda &= \frac{1}{2}(e_4-e_3) + \slashed{v}^Ae_A + \left(v^3+\frac{1}{2} \right)e_3 + \left(v^4-\frac{1}{2} \right)e_4
\end{align*}
we deduce (after dividing through by a near-unity factor) that
\[
n = \frac{1}{2}(e_3+e_4) + \slashed{\delta n}^Ae_A + \delta n^3 e_3 + \delta n^4 e_4
\]
where we can estimate
\[
|\slashed{\delta n}|^N_{s_1-1, 1^+}+ |{\delta n}^3|^N_{s_1-1, 1^+}+ |{\delta n}^4|^N_{s_1-1, 1^+} \lesssim \epsilon^\frac{1}{2}.
\]

The second fundamental form of $N$ is the symmetric two-tensor given by
\[
k(V, W) = g(n, D_V W)
\]
where $V, W$ are vectors tangent to the surface.

\begin{theorem}\label{thm:sffest}
    Writing
    \[
    k = F \ud \lambda^2 + 2 G_a \ud \lambda \ud \varphi^a + H_{ab} \ud \varphi^a \ud \varphi^b
    \]
    we can bound
    \[
    |F|^N_{s_1-2, 0^+} + |G|^N_{s_1-2, 0^+} + |H|^N_{s_1-2, 0^+} \lesssim \epsilon^{\frac{1}{2}}. 
    \]
\end{theorem}
\begin{proof}
    We write $V = \slashed{V}^Ae_a + V^3e_3 +V^4e_4$ and similarly for $W$ and let $\underline{\nab}_V = \slashed{V}^A\nab_A + V^3 \nab_3 + V^4 \nab_4$. Then
    \begin{align*}
        D_VW &= \left[\underline{\nab}_V \slashed{W} + \slashed{V}^A\chib_A{}^BW^3 + V^A \chi_A{}^BW^4 + 2 \xib^BV^3W^3 \right]e_B \\
        &\qquad + \left[\underline{\nab}_V W^3 +\frac{1}{2} \chi_{AB}\slashed{V}^A\slashed{W}^A + (V^3 \slashed{W}^A + V^A \slashed{W}^3)\eta_A - 2\omegab V^3 W^3 \right]e_3\\
        &\qquad+  \left[\underline{\nab}_V W^4 +\frac{1}{2} \chib_{AB}\slashed{V}^A\slashed{W}^A - (V^4 \slashed{W}^A + V^A \slashed{W}^4)\eta_A + 2\omegab V^3 W^4 \right]e_4
    \end{align*}
    Setting $V=W=j_*\partial_\lambda$ this expression vanishes by the geodesic equation, so in fact $F=0$. To find $G_a$ we set $V= j_*\partial_\lambda$, $W_a= j_*\partial_{\varphi^a}$ and observe that the $\abs{\cdot}^N_{s_1-2, 0^+}$ norm of all of the components of $D_VW_a$ can be bounded by a multiple of $\epsilon^{1/2}$. Contracting with $n$ we have the required bound on $G_a$.
    
    To determine $H_{ab}$ we consider $w_a = j_*\partial_{\varphi^a}$, $W_b = j_*\partial_{\varphi^b}$ and observe that the $\abs{\cdot}^N_{s_1-2, 0^+}$ norm of all of the components of
    \[
    D_{W_a}W_b - \frac{\mathring{\gamma}_{ab}}{2\lambda} (e_3 + e_4)
    \]
    can again be bounded by a multiple of $\epsilon^{1/2}$. Contracting $D_{W_a}W_b$ with $n$ we see a cancellation at top order due to the fact that $(e_3-e_4)$ and $(e_3+e_4)$ are orthogonal, and the subleading terms can be estimated to establish the required bound on $H_{ab}$.
\end{proof}

\subsubsection{Embedding in $H^s$}

Let $\mathbb{R}^3$ be equipped with the standard Euclidean metric in Cartesian coordinates. For $T$ a smooth tensor field and $U$ an open subset of $\mathbb{R}^3$ we define a weighted $C^k$ norm by
\[
|T|_{C^{k,p}(U)} =  \sup_{x\in U}\sum_{i\leqslant k} |x|^{i-p} \abs{\nabla^i T(x)}
\]
We define the $H^{s}(U)$ (fractional Sobolev) norm by extension:
\begin{equation*}
	\nrm{T}_{H^{s}(U)} := \inf_{\tilde{T} \in H^{s}(\mathbb{R}^{3}) : \tilde{T} |_{U} = T} \nrm{\tilde{T}}_{H^{s}(\mathbb{R}^{3})},
\end{equation*}
where we use the standard Cartesian basis to express $\tilde{T}$ and $T$ on both sides.

\begin{lemma}
    Suppose $f\in C^\infty_c(\mathbb{R}^{3})$, that $0\leqslant s<k$, and $p>-\frac{3}{2}+s$. Then we can estimate
    \[
    \left\|{f}\right\|_{H^s(B)} \lesssim \abs{f}_{C^{k,p}(B)}
    \]
    where the implicit constant depends on $n$, $p$, $s$, and $B = \{|x|< 1\}$.
\end{lemma}
\begin{proof}

We start with a simple but key observation. On the annulus $A_{1} := \{\frac{1}{2} < \abs{x} < 2 \}$, we clearly have
\begin{equation*}
	\nrm{f}_{H^{s}(A_{1})} \aleq \nrm{f}_{H^{k}(A_{1})} \aleq \abs{f}_{C^{k, p}(A_{1})},
\end{equation*}
By rescaling, on any smaller annuli $A_{R} := \{\frac{R}{2} < \abs{x} < 2 R \}$ ($0 < R \leqslant 1$), we have
\begin{equation*}
	\nrm{f}_{H^{s}(A_{R})} \aleq R^{p-s+\frac{3}{2}} \abs{f}_{C^{k, p}(A_{R})},
\end{equation*}
where we used the scaling invariance of $R^{p} \abs{f}_{C^{k, p}(A_{R})}$ (for any $R > 0$) and the inequality
$\nrm{f}_{H^{s}(A_{R})} \aleq R^{-s+\frac{3}{2}} \nrm{f(R(\cdot))}_{H^{s}(A_{1})}$ (true only for $0 < R \leqslant1$, as it is the inhomogeneous Sobolev norm). 

Now we start the proof in earnest. We claim that it suffices to prove the bound on the whole space, and for $f$ that is supported in $2B = \{ \abs{x} < 2 \}$. Indeed, by a standard extension procedure applied to $f$ restricted to $B$, we may find another function $\tilde{f}$ supported in $2B$ that (i)~agrees with $f$ on $B$ (so that $\nrm{f}_{H^{s}(B)} \leqslant \nrm{\tilde{f}}_{H^{s}(\mathbb{R}^{3})}$), and (ii)~satisfies $|\tilde{f}|_{C^{k, p}(\mathbb{R}^{3})} \leqslant2 \abs{f}_{C^{k, p}(B)}$. Clearly, the bound on the whole space for $\tilde{f}$ would imply the desired bound for $f$. In what follows, we abuse the notation and write $f$ for $\tilde{f}$.

Next, we perform a dyadic decomposition $f = \sum_{j = 0}^{\infty} f_{j}$, where each $f_{j}$ is $f$ multiplied by a smooth cutoff supported in $A_{2^{-j}}$ for each $j \geqslant 0$. Then we may bound
\begin{equation*}
	\nrm{f}_{H^{s}} \leqslant \sum_{j = 0}^{\infty} \nrm{f_{j}}_{H^{s}(A_{2^{-j}})} \aleq \sum_{j = 0}^{\infty} 2^{-(p - s +\frac{3}{2}) j} \abs{f_{j}}_{C^{k, p}} \aleq \abs{f}_{C^{k, p}},
\end{equation*}
where we used $p - s +\frac{3}{2} > 0$ and the obvious bound $\abs{f_{j}}_{C^{k, p}} \aleq \abs{f}_{C^{k, p}}$.
This completes the proof. \qedhere

\end{proof}
This result immediately generalises to tensors expressed with respect to the standard Cartesian basis. Now we wish to restate this result in terms of polar coordinates. For this, we write the Euclidean metric in polar coordinates:
\[
\mathring{h} = d\lambda^2 + \mathring{\gamma}_{ab}(\lambda) d\varphi^a d\varphi^b.
\]
Then if $X = X^\lambda \partial_\lambda + \slashed{X}^a \partial_{\varphi^a}$ we can express the components of $\nabla X$ in terms of $\slashed{D}_\lambda, \slashed{D}_a$ by:
\begin{align*}
\nabla_\lambda X^\lambda &= \slashed{D}_\lambda X^\lambda, &\nabla_a X^\lambda &= \slashed{D}_a X^\lambda - \frac{1}{\lambda} \gamma_{ab} \slashed{X}^b, \\
\nabla_\lambda \slashed{X}^a, &= \slashed{D}_\lambda \slashed{X}^a & \nabla_a \slashed{X}^b &=\slashed{D}_a \slashed{X}^b +\frac{\delta_a{}^b}{\lambda} X^\lambda.
\end{align*}
We observe that
\[
|x||\nabla X| + |X| \sim |\lambda \slashed{D}_\lambda X^\lambda| + |\lambda \slashed{D} X^\lambda|_{\mathring{\gamma}}+ |X^\lambda|+|\lambda \slashed{D}_\lambda \slashed{X}|_{\mathring{\gamma}} + |\lambda \slashed{D} \slashed{X}|_{\mathring{\gamma}}+\abs{\slashed{X}}_{\mathring{\gamma}}.
\]
Recalling that (after raising all indices with the metric) we can write any tensor as a sum of products of vectors we deduce:
    \begin{theorem}
    Suppose $\phi$ is a sufficiently regular $S^2_{\varphi}-$tangent tensor on $N$, that $0\leqslant s<k$, and $p>-\frac{3}{2}+s$. Then
    \[
    \left\|\phi\right\|_{H^s} \lesssim \abs{\phi}^N_{k, p}.
    \]
\end{theorem}

Together with Theorems \ref{thm:indmetest}, \ref{thm:sffest} this completes the proof of Theorem \ref{thm:geomcont1}, provided $s_1\geqslant 4$.

\subsection{Future completeness}

\begin{lemma}\label{lem:complete}
    Suppose $(M, g)$ is a smooth solution of Einstein's equations as above, such that after constructing Newman--Unti coordinates the weak bootstrap assumptions hold in $S$. Then $(M, g)$ is future geodesically complete.
\end{lemma}
\begin{proof}
    We need to show that if $\gamma:[0, \lambda_0) \to M$ is a future-directed causal geodesic with $\lambda_0<\infty$, then it can be extended. In $S$ we write $\gamma(\lambda) = (u(\lambda), r(\lambda), \theta^A(\lambda))$ and
    \[
    \dot{\gamma}(\lambda) = \slashed{v}^Ae_a + v^3e_3+v^4e_4,
    \] 
    where $u, r, \theta^A$ are piecewise smooth, with isolated discontinuities when the geodesic passes through the axis\footnote{the case $\gamma=\Gamma$ is trivial.}, at which points $u, \dot{u}, r$ are continuous and $\dot{r}, \theta^A$ have finite discontinuities. The geodesic equations are \eqref{eq:Nembedd1}, \eqref{eq:geodsvA}, \eqref{eq:geodsv3}, \eqref{eq:geodsv4}. The condition that $\gamma$ is future directed and causal implies $v^3, v^4 \geqslant0$.
    In view of the fact that $g$ is smooth, the only way that $\gamma$ can fail to be extendable is if either $r(\lambda)$ or $u(\lambda)$ tends to infinity as $\lambda \nearrow \lambda_0$, so it suffices to show that $\dot{u}, \dot{r}$ are bounded. 
    
    Let $E= \frac{1}{4}|\slashed{v}|^2+\frac{1}{2}(v^3)^2 + \frac{1}{2}(v^4)^2$.  Since at an axis crossing $|\slashed{v}|=0$ and $(v^3,  v^4) \mapsto (v^4, v^3)$, $E$ is continuous. A brief computation shows that away from axis crossings
    \[
    \frac{\ud E}{\ud \lambda} = -\frac{1}{2}\slashed{v}\cdot(\chi + \chib)\cdot\slashed{v}(v^3+v^4) + 2\slashed{v}\cdot\eta ((v^4)^2 - (v^3)^2)+ \left\{(2\omegab [(v^3)^2 - (v^4)^2] - \slashed{v}\cdot\xib(v^4+v^3)\right\}v^3.
    \]
    Observing that $\chi + \chib = \chih + \chibh +\f 12 [(\trch -2r^{-1})+ (\trchb +2r^{-1})]\gamma$, and also that the terms involving $\omegab, \xib$ which do not decay in $r$ always come with a factor of $v^3$, we see that $\frac{\ud E}{\ud \lambda}$ is continuous and we can estimate using the bootstrap assumptions
    \[
    \frac{\ud E}{\ud \lambda} \lesssim \epsilon \left( \frac{v^4}{(\jb{u}+r)^{1^+}} + \frac{v^3}{\jb{u}^{1^+}} \right) E.
    \]
    Making use of \eqref{eq:Nembedd1}, and using $|f-1|\lesssim \epsilon r/\jb{u}^{1^+}$ we deduce
    \[
    \frac{\ud E}{\ud \lambda} \lesssim \epsilon \left( \frac{1}{\jb{r}^{1^+}}\frac{\ud  r}{\ud \lambda} + \frac{1}{\jb{u}^{1^+}}\frac{\ud  u}{\ud \lambda} \right) E,
    \]
    and thus by Gr\"onwall's inequality
    \[
    E(\lambda) \leqslant E(0) \exp\left[ C\epsilon\left(1+\frac{1}{\jb{r(\lambda)}^{0^+}}+ \frac{1}{\jb{u(\lambda)}^{0^+}}\right) \right] \lesssim E(0).
    \]
    for some constant $C>0$ independent of $\epsilon$. Next we note that \eqref{eq:Nembedd1} implies
    \[
    \abs{\frac{\ud r}{\ud \lambda}} \lesssim (1+r) E \quad \implies \quad r(\lambda) \lesssim \exp(C' E(0) \lambda)
    \]
    so that $r$ remains bounded and hence, together with the boundedness of $E$, we deduce that all components of $\dot{\gamma}$ remain bounded so that $\gamma(\lambda)$ can be extended smoothly past $\lambda_0$.
\end{proof}

\section{Proof of the main theorem}\label{sec:proof}
We are now ready to complete our proof.
\begin{proof}[Proof of Theorem \ref{thm:main}]
Fix $\epsilon$ with $0<\epsilon<\epsilon_0$. Suppose we have $(\Sigma, h, k)$ such that after constructing coordinates the weak initial conditions \eqref{eq:data.alp.weak}, \eqref{eq:data.weak} hold. Let 
\[
T_{max} = \sup\{T>0: \Psi^*g\text{ exists and is smooth on }S_T\text{ and }\eqref{eq:BA.weak.pointwise}, \eqref{eq:BA.weak.integrated}\text{ hold}\}.
\]
By Corollary \ref{cor:loc.exist.bounds} and continuity $T_{max}$ is well defined and strictly positive. The result will follow if $T_{max} = \infty$. Suppose for contradiction that $T_{max}<\infty$. Set $u_0 = T_{max}-\epsilon^{\frac{1}{2}}$ and let $N(u_0)$ be the surface constructed in Theorem \ref{thm:geomcont1}, which is the image of a disc of radius $1/2$ under the exponential map. The first and second fundamental form $h', k'$ induced on $N(u_0)$ are necessarily smooth and by Theorem \ref{thm:geomcont1} we control $\|(h', k')\|_{H^{\frac{5}{2}^+}\times H^{\frac{3}{2}^+}} \lesssim \epsilon$. It follows from standard local existence theory for the Einstein equations (see for example \cite{Hughes1977Well-posedRelativity} Theorem V, together with the observation below Theorem III on propagation of regularity) that our solution can be extended smoothly to the future of $N(u_0)$ with metric components with respect to harmonic coordinates bounded in $C^0_tH_x^{\frac{5}{2}^+}\cap C^1_tH_x^{\frac{3}{2}^+}$ by a multiple of $\epsilon$. Noting that this implies a uniform bound on Christoffel symbols in harmonic coordinates, hence control of existence time for geodesics, we can further assume that within this extension $\Gamma$ can be continued from $N(u_0)$ to the future for a proper time at least $1/10$, provided $\epsilon_0$ is sufficiently small. We deduce then that our spacetime can be extended smoothly to a neighbourhood of $\Psi^{-1}(\Sigma_{T_{max}})\cap \Gamma$. 

Away from $\Gamma$ our coordinates are regular and we control at least three derivatives of the metric pointwise so it follows that the spacetime can be extended to a neighbourhood of $\Psi^{-1}(\Sigma_{T_{max}})$. By continuity the coordinate construction of \S\ref{sec:coords} can be extended to a (possibly smaller) neighbourhood of $\Psi^{-1}(\Sigma_{T_{max}})$. We conclude that there exists $T'>T_{max}$ such that $\Psi^*g\text{ exists and is smooth on }S_{T'}$.

Now, by Corollary \ref{cor:alpha.final} and Propositions \ref{prop:fixed.sphere}, \ref{prop:integrated} there exists a constant $C$ depending only on $s_0, \delta$ such that \eqref{eq:BA.weak.pointwise}, \eqref{eq:BA.weak.integrated} hold in $S_{T_{max}}$ with $\epsilon$ replaced by $C\epsilon^2$. Decreasing $\epsilon_0$ if needed so that $\epsilon_0 \leqslant 1/(2C)$, by continuity \eqref{eq:BA.weak.pointwise}, \eqref{eq:BA.weak.integrated} will hold in $S_{T''}$ for some $T_{max}<T''<T'$, which contradicts the maximality of $T_{max}$, hence $T_{max}=\infty$. Lemma \ref{lem:complete} establishes that the resulting spacetime is future complete.

If the metric additionally satisfies the strong initial data assumptions \eqref{eq:data.alp.strong}, \eqref{eq:data.strong} then we can repeat this argument (for a possibly smaller $\epsilon_0$) using the strong bootstrap assumptions \eqref{eq:BA.strong.pointwise}, \eqref{eq:BA.strong.integrated} in addition to \eqref{eq:BA.weak.pointwise}, \eqref{eq:BA.weak.integrated}.
\end{proof} 

\appendix

\section{Localising the main stability theorem}\label{sec:appendix}

In this appendix, we formulate the theorems regarding exterior stability of Minkowski and the stability for mixed Cauchy-characteristic data. Both results involve a simple localisation of our main theorem. As mentioned in the introduction, these results generalise those in \cite{Klainerman2003TheNicolo.} and \cite{Graf2020GlobalData}. We begin with the exterior stability result.


\begin{corollary}[Exterior nonlinear stability of Minkowski spacetime]
    Let $(\Sigma, h,k)$ be a smooth set of Cauchy data for the Einstein vacuum equations such that for a compact set $\mathcal K \subset \Sigma$, $\Sigma \setminus \mathcal K$ is diffeomorphic to $\mathbb R^3 \setminus B_R(0)$ for some $R>0$ and thus equipped with coordinates $x^1,x^2,x^3$. Suppose that for some $s \geqslant 30$ and $\nu>0$, $(\Sigma\setminus \mathcal K,h,k)$ satisfies
    \begin{equation}\label{eq:assumption.exterior}
        (h_{ij}, k_{ij}) \in (\delta_{ij} + H^{s, \nu - \frac{3}{2}}, H^{s-1, (\nu + 1) - \frac{3}{2}}).
    \end{equation}
    Then, for every $\de \in (0, \min\{\nu,\f 1{20} \})$, there exists $\widetilde{R} > R$ such that the following conclusions hold:
    \begin{enumerate}[i)]
        \item Let $\widetilde{\mathcal K} \supseteq \mathcal K$ be the compact set such that $\Sigma \setminus \widetilde{\mathcal K}$ corresponds to $ \mathbb R^3 \setminus B_{\widetilde{R}}(0)$ under the diffeomorphism above. Then the maximal future Cauchy development $(M,g)$ of $\Sigma\setminus \widetilde{\mathcal K}$ may be covered by a single Newman--Unti coordinate chart, and the outgoing null hypersurface emanating from $\rd \widetilde{K}$ (given as the constant $\{u=1-\widetilde{R}\}$ hypersurface in the Newman--Unti coordinates) is a future null boundary of the spacetime. 
        \item There exists $\ep>0$ such that the estimates \eqref{eq:BA.weak.pointwise}, \eqref{eq:BA.weak.integrated}  hold in $(M,g)$.
        \item If, in addition, \eqref{eq:assumption.exterior} holds with $\nu = 3-\de$, then after choosing $\widetilde{R}$ larger, the estimates \eqref{eq:BA.strong.pointwise}, \eqref{eq:BA.strong.integrated} hold for some $\ep>0$.
    \end{enumerate}
\end{corollary}
\begin{proof}
    Let $\de \in (0, \min\{\nu,\f 1{20} \})$. In particular, since $\nu - \de>0$, the assumption \eqref{eq:assumption.exterior} implies that for every $\mathfrak e>0$, there exists $\widetilde{R} > R$ sufficiently large such that
    $$\mathfrak n_e:=\| h_{ij} - \delta_{ij} \|_{H^{s, \delta-\frac{3}{2}}(\mathbb R^3 \setminus B_{\widetilde{R}}(0))} + \| k_{ij} \|_{H^{s-1, \delta-\frac{1}{2}}(\mathbb R^3 \setminus B_{\widetilde{R}}(0))} < \mathfrak e.$$
    Thus, the estimates in Theorem~\ref{thm:main-1} can be applied to construct a solution in the Newman--Unti gauge for $u \leqslant 1-\widetilde{R}$. A similar argument applies to show the stronger estimates when \eqref{eq:assumption.exterior} holds with $\nu = 3-\de$. 
    
   It then suffices to show that the region constructed is the maximal future Cauchy development. If not, then there is a smooth globally hyperbolic extension $(\widetilde{M},\widetilde{g})$ and a timelike curve $\gamma: (-\eta,\eta)\to \widetilde{M}$ such that $\gamma|_{(-\eta,0)} \subset M$ and $\lim_{s\to 0^-} u(\gamma(s)) = 1-\widetilde{R}$. However, it follows that $J^+(\Sigma) \cap J^-(\gamma(0))$ is non-compact, contradicting global hyperbolicity (see \cite[Theorem~8.3.12]{rmW1984}).
\end{proof}

We next formulate the result for stability of Minkowski spacetime for the spacelike-characteristic initial value problem. The data are given on a spacelike hypersurface $\Sigma= B_1(0)$ and a characteristic hypersurface $H_0$ which share a boundary diffeomorphic to $\mathbb S^2$. On $\Sigma$ we prescribe the usual Cauchy data, but observe that we only need to prescribe characteristic data for $\gamma$ on $H_0$. The other components, i.e., $b$ and $f$, as well as the other geometric quantities depending on the transversal derivatives of the metric components (e.g., $\chib$ and some curvature components), are determined by the corresponding transport equation and the initial data computed from $(\Sigma,h,k)$.

\begin{corollary}[Stability for the spacelike-characteristic initial value problem]\label{cor:graf}
    Let $s \in \mathbb Z_{\geqslant 30}$. Consider the spacelike-characteristic initial value problem satisfying the following conditions:
    \begin{itemize}
        \item Suppose smooth (up to the boundary) Cauchy data $(\Sigma= B_1(0),h,k)$ are given. Denote
        $$\mathfrak n_\Sigma:=\|h_{ij}- \de_{ij} \|_{H^s} + \|k_{ij}\|_{H^s}.$$
        \item Suppose smooth characteristic initial data for $\gamma$ are given on $H_0 = [1,\infty) \times \mathbb S^2$ associated with a geodesic foliation. Define $\chi$ and $\alp$ by (the third equation in) \eqref{gauge.con} and \eqref{eq:4chih}, respectively, and assume that the constraint equation \eqref{eq:4trchi} holds. For $s_0\in \mathbb Z_{\geq 20}$, $\nu>0$, denote by $\mathfrak n_H$ the following norm\footnote{Observe that the middle index corresponding to $p_0$ in the norms here are irrelevant since the data are posed in a region where $r \geq 1$.} on $\alp$:
        \begin{equation}\label{eq:epH.def}
            \mathfrak n_H:= \| r \nab_4 \alp + 5 \alp \|_{s_0, -1, 2+\nu}(H_0) + \|\alp\|_{s_0+1,-1,\f 32 +\nu}(H_0).
        \end{equation}
        \item Identify the boundary $\rd\Sigma = \rd B_1(0)$ with the sphere $\{1\}\times \mathbb S^2\subset H_0$. On $\Sigma$, define $\gamma$ to be the restriction of $h$ to the spheres $\rd B_r(0)$ for $r \in (0,1)$. We then require that $\gamma$ agrees at the intersection $\rd B_1(0) = \{1\}\times \mathbb S^2$ up to all orders. 
    \end{itemize}
    There exists $\mathfrak e_C = \mathfrak e_C(s,s_0,\de)>0$ such that if $\mathfrak n_\Sigma + \mathfrak n_H < \mathfrak e_C$, then the following conclusions hold for some $C = C(s,s_0,\de)>0$.
    \begin{enumerate}[i)]
    \item The maximal future Cauchy development $(M,g)$ is globally smooth and future complete and may be covered by a single Newman--Unti chart, which is regular away from the centre geodesic.
    \item The estimates \eqref{eq:BA.weak.pointwise}, \eqref{eq:BA.weak.integrated}  hold in $(M,g)$ with $\ep$ replaced by $C(\mathfrak n_\Sigma + \mathfrak n_H)$.
        \item There exists $\mathfrak e'_C = \mathfrak e'_C(s,s_0,\de)>0$ such that if $\nu = 3-\de$ in \eqref{eq:epH.def} and $\mathfrak n_\Sigma + \mathfrak n_H < \mathfrak e'_C$, then the estimates \eqref{eq:BA.strong.pointwise}, \eqref{eq:BA.strong.integrated} hold with $\ep$ replaced by $C(\mathfrak n_\Sigma + \mathfrak n_H)$.
    \end{enumerate}
    
\end{corollary}
\begin{proof}
    We will explain how to construct a local solution under the given assumptions. The estimates we need are then identical to those in the proof of Theorem~\ref{thm:main}. Indeed, for the wave estimates for $\alp$, we argue exactly as in \S\ref{sec:teukolsky} and \S\ref{TeukEsts}, except that we now have initial data on $H_0$ instead of $\tilde{\Sigma}_1$; observe that the data on $H_0$ exactly have to be controlled in the norm \eqref{eq:epH.def}. The bounds for the other geometric quantities can be derived in the same manner as in \S\ref{sec:other-geom}.

    To obtain a local solution, we first solve for the maximal globally hyperbolic future development of the data $(\Sigma,h,k)$ using \cite{yCBrG1969}. By standard domain of dependence argument (for instance in wave coordinates), locally the solution is smoothly extendible up to an incoming null hypersurface $\Hb$ emanating from $\rd B_1(0)$ in the initial data. The result of \cite{adR1990} can then be used to solve the characteristic initial value problem with data on $\Hb$ and $H_0$, constructing a local solution to the mixed spacelike-characteristic initial value problem. \qedhere

\end{proof}

\begin{remark} \label{rem:strong-peeling-spnull}
    In \eqref{eq:epH.def}, if we additionally require $\nu = 3+\f{\de}{2}$, then a very similar proof shows that the smallness of this norm propagates; see Remarks~\ref{rem:strong-peeling-rp} and \ref{rem:strong-peeling-geom} for more details. This will in particular allow us to show using the fundamental theorem of calculus and Cauchy--Schwarz that $\lim_{r\to \infty} (r^5 \alp)(u,r)$ exists and that Bondi--Sachs peeling holds (see \eqref{eq:Bondi.Sachs.def}).
    
    Observe that there are a couple of key differences between the case of Cauchy data in Theorem~\ref{thm:main-1} and the case of the spacelike-characteristic initial value problem in Corollary~\ref{cor:graf}. First, in Theorem~\ref{thm:main-1}, if one imposes a smallness condition for $\nu > 3$, then this forces $\lim_{r\to \infty} (r^5 \alp)(1-r,r) =0$ (see the norms on the right-hand side of Corollary~\ref{cor:BasicEst}). On the other hand, in the setting of Corollary~\ref{cor:graf}, imposing a smallness condition for $\nu > 3$ is consistent with $\lim_{r\to \infty} (r^5 \alp)(u,r)$ being well-defined but non-trivial (see the norms in \eqref{eq:epH.def}). Second, in Theorem~\ref{thm:main-1}, the decay rates of various initial data quantities are constrained by the slowly decaying terms in Definition~\ref{def:AF.data} (see also Remark~\ref{rem:AF.data.cancel}), which restrict the decay rates that can be propagated via the $e_{4}$-transport equations. In contrast, in the setting of Corollary~\ref{cor:graf}, the initialization of these transport equations occurs entirely on the compact region $\Sigma = B_{1}(0)$, and thus such constraints do not exist (see also Remark~\ref{rem:strong-peeling-geom}).
\end{remark}

\section{Proof of theorem \ref{thm:data.lemma}} \label{sec:appendix2}

In this section we show that bounds on the Cauchy data for the Einstein equations in suitable (weighted) Sobolev spaces imply our initial data assumptions at the level of geometric quantities tied to the Newman--Unti coordinates. We first need to establish some comparison results for these quantities for solutions with initial data which are close.

Suppose $\mathfrak{S}_j=(h_j, k_j), j=1, 2$ are two smooth initial data sets defined on $\mathbb{R}^3$, and set $\mathfrak{S}_0=(\delta, 0)$. For $0\leqslant r_0 <r_1$ we let
\[
\Delta^s_{r_0, r_1}(\mathfrak{S}_1, \mathfrak{S}_2) =  \|h_1-h_2\|_{H^{s+1}(\{r_0<|x|<r_1\})}+\|k_1-k_2\|_{H^{s}(\{r_0<|x|<r_1\})}.
\]
For each data set $\mathfrak{S}_j$ we construct the local extension in Newman--Unti coordinates, as in \S\ref{sec:coords}, recalling that the construction differs inside and out of a $\delta r$ neighbourhood of the axis. We let $(f_j, \gamma_j, b_j)$ be the respective metric coefficients. 

\begin{lemma}\label{thm:data.lemma1}
    Suppose $3\leqslant s \in \mathbb N$ and
    \[
    \Delta^{s}_{r_0, r_1}(\mathfrak{S}_1, \mathfrak{S}_0)+\Delta^{s+3}_{r_0, r_1}(\mathfrak{S}_2, \mathfrak{S}_0)=P.
    \]
    Then for any $\delta r<r_0<r_1$, $P<P_0$, $p_0, p_\infty$ we have
    \begin{align}
    &\|f_1-f_2\|_{s+1, p_0, p_\infty}(A_{r_0,r_1}) +\|\gamma_1-\gamma_2\|_{s+1, p_0, p_\infty}(A_{r_0,r_1}) +\|b_1-b_2\|_{s+1, p_0, p_\infty}(A_{r_0,r_1})\nonumber \\
     &+|f_1-f_2|_{s, p_0, p_\infty}(A_{r_0,r_1})+|\gamma_1-\gamma_2|_{s, p_0, p_\infty}(A_{r_0,r_1}) +|b_1-b_2|_{s, p_0, p_\infty}(A_{r_0,r_1})\label{eq:data.lemma1.est}\\&\qquad \leqslant c \Delta^{s}_{r_0, r_1}(\mathfrak{S}_1, \mathfrak{S}_2)\nonumber
    \end{align}
    where $c$ depends on $s, r_0, r_1, p_0, p_\infty$ and $P_0$ and $A_{r_0,r_1} = \tilde{\Sigma}_1\cap \{r_0<r<r_1\}$.
\end{lemma}
\begin{proof}
    Standard local existence theory for Einstein's equations in harmonic coordinates $(x^\mu)$ guarantees a solution $g_j$ in a neighbourhood of $\{x^0=1, r_0<|x|<r_1\}$ corresponding to initial data $(h_j, k_j)$, and moreover that we can bound 
    \begin{align*}
     &\sup_{-t_0<t-1<t_0}\left[\sum_{i=0}^{s+1 } \|\partial_{0}^{i}(g_1-g_2)\|_{H^{s+1-i}(\{x^0=t\}\cap U)} \right] \leqslant c \Delta^{s}_{r_0, r_1}(\mathfrak{S}_1, \mathfrak{S}_2) =:c\Delta
    \end{align*}
    for some $t_0>0$ and $U$ an open neighbourhood of $\{x^0=1, r_0<|x|<r_1\} \simeq \tilde{\Sigma}_1\cap \{r_0<r<r_1\}$. See for example \cite{Ringstrom2009TheRelativity}, using (9.33) for the continuity estimate. Here and below $c$ is as in the statement of the lemma, but may increase from line to line. 
    
    Next, for each metric we construct the vector field $W_j = W_j(r, u, \theta)$, (c.f.\ \S \ref{sec:coords}), order by order on $\tilde{\Sigma}_1$ using the geodesic equation starting from the initial condition \eqref{eq:e4initdef}. We deduce from the bound above that (writing $W_j^\mu$ for the components of $W_j$ in the $(x^\mu)$ coordinates):
    \begin{align*}
    &\sum_{i=0}^{s+1} \|\partial_r^{i}(W_1^\mu-W_2^\mu)\|_{H^{s+1-i}(A_{r_0,r_1})} \leqslant c \Delta.
    \end{align*}
    Note that this implies Sobolev bounds on all derivatives of $W_1^\mu - W_2^\mu$ up to order $s+1$ on the intial surface. We deduce from this and the fact that $X_1=X_2$ on $\tilde{\Sigma}_1$ that the embedding maps $X_j^\mu(r, u, \theta)$ satisfy the estimate
    
    \begin{align*}
    &\sum_{i=0}^{s+2} \|\partial_r^{i}(X_1^\mu-X_2^\mu)\|_{H^{s+2-i}(A_{r_0,r_1})} \leqslant c \Delta.
    \end{align*}
    Noting that
    \[
    f = \partial_u X^\mu \partial_u X^\nu g_{\mu \nu}, \quad \gamma_{AB} = \partial_{\theta^A} X^\mu \partial_{\theta^B} X^\nu g_{\mu \nu}, \qquad b_A = \partial_{\theta^A} X^\mu \partial_{u} X^\nu g_{\mu \nu}
    \]
    control of the integrated norms in \eqref{eq:data.lemma1.est} follows. Control of the remaining pointwise norms follows by a Sobolev embedding in the radial direction.
\end{proof}

Next we consider a neighbourhood of the centre

\begin{lemma}
    Suppose $h_1(0)=h_2(0)=\delta$, $5\leqslant s' \in \mathbb N$ and
    \[
    \Delta^{s'}_{0, \delta r}(\mathfrak{S}_1, \mathfrak{S}_0)+\Delta^{s'+3}_{0, \delta r}(\mathfrak{S}_2, \mathfrak{S}_0)=P.
    \]
    Then for any $p_\infty$, $P<P_0$ and integer $s<s'-4$ we have
    \begin{align}
    &\|f_1-f_2\|_{s+1, \frac{3}{2}^+, p_\infty}(A_{\delta r}) +\|\gamma_1-\gamma_2\|_{s+1, \frac{3}{2}^+, p_\infty}(A_{\delta r}) +\|b_1-b_2\|_{s+1, \frac{3}{2}^+, p_\infty}(A_{\delta r})\nonumber \\
     &+|f_1-f_2|_{s, 1^+, p_\infty}(A_{\delta r})+|\gamma_1-\gamma_2|_{s, 1^+, p_\infty}(A_{\delta r}) +|b_1-b_2|_{s, 1^+, p_\infty}(A_{\delta r})\label{eq:data.lemma2.est}\\&\qquad \leqslant c \Delta^{s}_{0,{\delta r}}(\mathfrak{S}_1, \mathfrak{S}_2)\nonumber
    \end{align}
     where $c$ depends on $s, s', p_\infty$ and $P_0$, and $A_{\delta r} = \tilde{\Sigma}_1\cap \{r<\delta r\}$.
\end{lemma}
\begin{proof}
    As in the proof of Lemma \ref{thm:data.lemma1} we can bound
    \begin{align*}
     &\sup_{-t_0<t-1<t_0}\left[\sum_{i=0}^{s+6} \|\partial_{0}^{i}(g_1-g_2)\|_{H^{s+4-i}(\{x^0=t\}\cap U)} \right] \leqslant c \Delta^{s'}_{0, \delta r}(\mathfrak{S}_1, \mathfrak{S}_2) =:c\Delta
    \end{align*}
    for some $t_0>0$ and $U$ an open neighbourhood of $\{x^0=1, |x|<\delta r\} \simeq \tilde{\Sigma}_1\cap \{r<\delta r\}$.
    Examining the construction of coordinates near the centre in \S\ref{sec:coordscentre}, we see that when constructing the centre geodesics $\Gamma_j$ it is necessary to perform a Sobolev embedding in order to restrict the Christoffel symbols to a curve. As a result, writing the coordinates of $\Gamma_j$ as $\Gamma_j^\mu(u)$ in the coordinate chart coming from the local existence result we have
    \[
    \left| \Gamma_1^\mu -\Gamma_2^\mu \right|_{C^{s+4}} \leqslant c \Delta.
    \]
    Constructing the vector fields $V_j(u, \theta)$, we observe that the equation for parallel propagation involves only functions along $\Gamma_j$ -- that is to say $\theta$ appears only in the initial conditions, and moreover in a smooth fashion. As a result, we can estimate on $\Gamma$
    \[
    \left|(\partial_u)^p (\partial_{\theta^1})^q (\partial_{\theta^2})^l (V_1^\mu -  V_2^\mu) \right|\leqslant c \Delta
    \]
    for $p\leqslant s+3$ and \emph{any} $q,l$. We next extend away from $\Gamma$ by solving the geodesic equation to construct the embedding map $X_j^\mu(r, u, \theta)$. In view of the fact that by continuity we can assume $\{u+r=\textrm{const.}\}$ remains spacelike near $\tilde{\Sigma}_1$ we can estimate for $\tilde{U}$ an open neighbourhood of $\tilde{\Sigma}_1\cap \{r<\delta r\}$
    \[
    \sup_t\left\|r^{-d}(\partial_u)^p (\partial_{\theta^1})^q (\partial_{\theta^2})^l(\partial_r)^m (X_1^\mu - X_2^\mu )\right\|_{L^2(\tilde{\Sigma}_t\cap \tilde{U})}\leqslant c \Delta
    \]
    where $p+q+l+m\leqslant s+5$, $p\leqslant s+3$ and $d=0$ if $p=q=0$ and $d=1$ otherwise. Constructing the metric functions as in the previous lemma, and recalling from \S\ref{sec:coords.form} that $0=\partial_r^a (f_1-f_2)|_{r=0} = \partial_r^b (b_1-b_2)_{A}|_{r=0} = \partial_r^c (\gamma_1-\gamma_2)_{AB}|_{r=0}$ for $0\leqslant a \leqslant 1$, $0\leqslant b \leqslant 2$, $0\leqslant c \leqslant 3$ we can iteratively apply Hardy's inequality to deduce
    \begin{align*}
    \left\|r^{(-3)^-}(\partial_u)^p (\partial_{\theta^1})^q (\partial_{\theta^2})^l(\partial_r)^m (f_1 - f_2 )\right\|_{L^2(\tilde{\Sigma}_1\cap \tilde{U})}&\leqslant c \Delta\\
    \left\|r^{(-4)^-}(\partial_u)^p (\partial_{\theta^1})^q (\partial_{\theta^2})^l(\partial_r)^m (b_1 - b_2 )_A\right\|_{L^2(\tilde{\Sigma}_1\cap \tilde{U})}&\leqslant c \Delta\\
    \left\|r^{(-5)^-}(\partial_u)^p (\partial_{\theta^1})^q (\partial_{\theta^2})^l(\partial_r)^m (\gamma_1 - \gamma_2 )_{AB}\right\|_{L^2(\tilde{\Sigma}_1\cap \tilde{U})}&\leqslant c \Delta
    \end{align*}
    for $p+q+l+m\leqslant s+1$. Re-writing in terms of our $S^2$-adapted norms gives the first result. The second part of the result again follows after a radial Sobolev embedding.
\end{proof}

\begin{corollary}\label{cor:idest}
    Suppose $h_1(0)=h_2(0)=\delta$,  $0\leqslant r_0<r_1$, $5\leqslant s'\in \mathbb N$  and
    \[
    \Delta^{s'}_{r_0,  r_1}(\mathfrak{S}_1, \mathfrak{S}_0)+\Delta^{s'+3}_{r_0,  r_1}(\mathfrak{S}_2, \mathfrak{S}_0)=P.
    \]
    Then for any $P<P_0$, and integer $s<s'-4$ 
    \begin{align}
    &\|\mathfrak{m}_1-\mathfrak{m}_2\|_{s+1, \frac{3}{2}^+, p_\infty}(A_{r_0,r_1}) +\|\mathfrak{s}_1-\mathfrak{s}_2\|_{s, \frac{1}{2}^+, p_\infty}(A_{r_0,r_1}) +\|b_1-b_2\|_{s-1, (-\frac{1}{2})^+, p_\infty}(A_{r_0,r_1})\nonumber \\
     &+|\mathfrak{m}_1-\mathfrak{m}_2|_{s, 1^+, p_\infty}(A_{r_0,r_1})+|\mathfrak{s}_1-\mathfrak{s}_2|_{s-1, 0^+, p_\infty}(A_{r_0,r_1}) +|\mathfrak{w}_1-\mathfrak{w}_2|_{s-2, (-1)^+, p_\infty}(A_{r_0,r_1})\label{eqn:leproof2}\\&\qquad \leqslant c \Delta^{s}_{r_1, r_2}(\mathfrak{S}_1, \mathfrak{S}_2)\nonumber
    \end{align}
    where $\mathfrak{m}$ is any metric component, $\mathfrak{s}$ is any Ricci coefficient, $\mathfrak{w}$ is any curvature component and $c$ depends on  $s, r_1, p_\infty, P_0$.
\end{corollary}

Note that the loss of derivatives in the above results  may not be sharp --- as elsewhere in the paper we have not attempted to optimise for regularity. We are now ready to give the proof of Theorem \ref{thm:data.lemma}

\begin{proof}[Proof of Theorem \ref{thm:data.lemma}]
    We first establish the statement under the weaker decay assumptions. Restricted to $|x|<1$, the result already follows from Corollary \ref{cor:idest} with $\mathfrak{S}_2 = \mathfrak{S}_0$. Consider for fixed $R>1$ the coordinate transformation $\varphi_R:(u, r, \theta^A)\mapsto (Ru+(1-R), Rr, \theta^A)$. This map preserves $\tilde{\Sigma}_1$ and restricts as a homothety, rescaling $x\mapsto Rx$. We see that if $\tilde{f} = f\circ\varphi_R$, $\tilde{\gamma}_{AB} = R^{-2} {\gamma}_{AB}\circ\varphi_R$ and $\tilde{b}^A = R b^A\circ\varphi_R$, then $\varphi_{R}^* g = R^2 \tilde{g}$ with obvious notation. Clearly $\tilde{g}$ will be Einstein if $g$ is Einstein. 
    
    Let $A_R = \{x\in \mathbb{R}^3 : R<|x|<2R\}$ and consider establishing the result restricted to $A_R$ for $R\geqslant 1$. Denote $\mathfrak{e}_R=\|h_{ij} - \delta_{ij}\|_{H^{s'+1, \nu-\frac{3}{2}}(A_R)} + \|k_{ij}\|_{H^{s', \nu-\frac{1}{2}}(A_R)}$. The Cauchy data, $\tilde{\mathfrak{S}}=(\tilde{h}, \tilde{k})$, induced by $\tilde{g}$ on $A_1$ is related to that induced on $A_R$ by $g$ by
    \[
    \tilde{h}_{ij}(y) = h_{ij}(Ry), \qquad \tilde{k}_{ij}(y) = R k_{ij}(Ry)
    \]
    and we can verify 
    \[
    \Delta^{s'}_{1, 2}(\tilde{\mathfrak{S}}, \mathfrak{S}_0) \lesssim R^{-\nu}\mathfrak{e}_R.
    \]
    Applying the corollary we have already established with $\mathfrak{S}_2=\mathfrak{S}_0$, we control all quantities associated to $\tilde{g}$ by a constant multiple of $R^{-\nu}\mathfrak{e}_R$. 
    
    We can verify that $\|f-1\|_{s+1, p_0, p_\infty}({\tilde{\Sigma}_1\cap A_R})\sim R^{p_\infty}\|\tilde{f}-1\|_{s+1, p_0, p_\infty}({\tilde{\Sigma}_1\cap A_1}) $ and similarly for $\gamma-\mathring{\gamma}$ and $b$ and also for the pointwise norms. From Proposition \ref{prop:gauge.con} we see that $\|\omegab\|_{s, p_0, p_\infty}({\tilde{\Sigma}_1\cap A_R}) \sim R^{p_\infty-1}\|\tilde{\omegab}\|_{s, p_0, p_\infty}({\tilde{\Sigma}_1\cap A_1})$ and similarly for the other connection coefficients. Finally, by considering the transformation properties of the Weyl tensor under conformal transformations we have $\|{\alpha}\|_{s, p_0, p_\infty}({\tilde{\Sigma}_1\cap A_R}) \sim R^{p_\infty-2}\|\tilde\alpha\|_{s, p_0, p_\infty}({\tilde{\Sigma}_1\cap A_1})$ and similarly for other Weyl components. Since $\nu>\delta$ we see that on setting $R = 2^k$ and summing on $k$ we control all quantities required on the initial data at the regularity needed.

    Let us turn now to the strong decay assumption. We can deal with the region $|x|<R_0$ by the previous result. To handle the large $|x|$ region, consider the Kerr solution to the vacuum Einstein equations in Boyer-Lindquist coordinates $(t, r, \theta, \phi)$ with mass $M$ and spin $a$ (see for example \cite{mtw} (33.2) with $Q=0$ for the explicit form). Introducing Cartesian coordinates for $r$ sufficiently large by
    \begin{align*}
    \left(1+\frac{M}{2|x|} \right)^2 x^1 &= \sqrt{r^2+a^2} \sin \theta \cos \phi,\\
    \left(1+\frac{M}{2|x|} \right)^2 x^2 &= \sqrt{r^2+a^2} \sin \theta \sin \phi, \\
    \left(1+\frac{M}{2|x|} \right)^2 x^3 &= r \cos \theta,
    \end{align*}
    and expanding the metric functions, we find that in these coordinates the Kerr metric has the asymptotic form
    \begin{align*}
        g_{M,a} &= -\left( 1-\frac{2M}{|x|} + \frac{2M^2}{|x|^2}\right) \ud t^2 + \left( 1+\frac{2M}{|x|} + \frac{3M^2}{2|x|^2}\right) \delta_{ij}\ud x^i \ud x^j - \frac{4M a}{|x|^3} \ud t(x^1 \ud x^2 - x^2 \ud x^1) + O \left(\frac{1}{|x|^3} \right)
    \end{align*}
    We note that this agrees with (19.13) in \cite{mtw}, derived as the leading order far-field metric due to an arbitrary source. We can compute the induced metric and second fundamental form of the surface $\{t=0\}$ and we find that they satisfy
    \begin{align*}
    h^{a, M}_{ij} &= (1 + \tfrac{2M}{|x|} + \tfrac{3 M^{2}}{2 |x|^{2}}) \delta_{ij}
        + O(|x|^{-3}),  \\
    k^{a,M}_{ij} &= {3} (\in_{i k \ell} J^{\ell} \tfrac{x_{j} x^{k}}{|x|^{5}} + \in_{j k \ell} J^{\ell} \tfrac{x_{i} x^{k}}{|x|^{5}}) + O(|x|^{-4})
    \end{align*}
    where $J = (0, 0, Ma)$, and we can act on the error terms with $|x|\partial_i$ an arbitrary number of times without changing the decay rate. Constructing Newman--Unti coordinates starting from $\{t=0\}$ in the Kerr geometry we can compute the Ricci components and curvature components and we find that \eqref{eq:data.alp.strong}, \eqref{eq:data.strong} restricted to $\tilde{\Sigma}_1\cap\{|x|>R_0\}$ hold for the Kerr initial data in these coordinates, provided $M, a$ are sufficiently small. We can now repeat the dyadic argument above. On $A_R$, after pulling back to $A_1$ we now compare $(\tilde{h}, \tilde{k})$ to $(h^R_2, k^R_2) = (R^{-2} \varphi_R^*h^{a, M}, R^{-1} \varphi_R^*k^{a, M})$. We can verify that for any $s$
    \[
    \|h^R_2-\delta\|_{H^{s+1}(\{1<|x|<2\})}+\|k^R_2\|_{H^{s}(\{1<|x|<2\})} 
    \]
    is uniformly bounded for $R>R_0$ and so we can repeat our previous argument.
\end{proof}

\bibliographystyle{amsalpha}
\bibliography{refs}

\newcommand{\etalchar}[1]{$^{#1}$}
\providecommand{\bysame}{\leavevmode\hbox to3em{\hrulefill}\thinspace}
\providecommand{\MR}{\relax\ifhmode\unskip\space\fi MR }
\providecommand{\MRhref}[2]{%
  \href{http://www.ams.org/mathscinet-getitem?mr=#1}{#2}
}
\providecommand{\href}[2]{#2}
\begin{thebibliography}{BvdBM62}

\bibitem[ABWY23]{andersson2023global}
Lars Andersson, Pieter Blue, Zoe Wyatt, and Shing-Tung Yau, \emph{Global
  stability of spacetimes with supersymmetric compactifications}, Analysis \&
  PDE \textbf{16} (2023), no.~9, 2079--2107.

\bibitem[AC22]{sAntC2022}
Spyros Alexakis and Nathan~Thomas Carruth, \emph{{Squeezing a fixed amount of
  gravitational energy to arbitrarily small scales, in U(1) symmetry}}, arXiv
  \textbf{2205.05526} (2022).

\bibitem[BFJ{\etalchar{+}}21]{Bigorgne2021AsymptoticMatter}
Léo Bigorgne, David Fajman, Jérémie Joudioux, Jacques Smulevici, and
  Maximilian Thaller, \emph{{Asymptotic Stability of Minkowski Space-Time with
  Non-compactly Supported Massless Vlasov Matter}}, Archive for Rational
  Mechanics and Analysis \textbf{242} (2021), no.~1, 1--147.

\bibitem[Bie10]{Bieri2010AnRelativity}
Lydia Bieri, \emph{{An Extension of the Stability Theorem of the Minkowski
  Space in General Relativity}}, J. Diff. Geom. \textbf{86} (2010), no.~1,
  17--70.

\bibitem[Bon60]{hB1960}
H.~Bondi, \emph{{Gravitational {W}aves in {G}eneral {R}elativity}}, Nature
  \textbf{4724} (1960), no.~184, 535.

\bibitem[BvdBM62]{hBmgjvdBawkM1962}
H.~Bondi, M.~G.~J. van~der Burg, and A.~W.~K. Metzner, \emph{Gravitational
  waves in general relativity. {VII}. {W}aves from axi-symmetric isolated
  systems}, Proc. Roy. Soc. London Ser. A \textbf{269} (1962), 21--52.

\bibitem[CBG69]{yCBrG1969}
Yvonne Choquet-Bruhat and Robert Geroch, \emph{Global aspects of the {C}auchy
  problem in general relativity}, Comm. Math. Phys. \textbf{14} (1969),
  329--335.

\bibitem[Chr86]{dC1986}
Demetrios Christodoulou, \emph{The problem of a self-gravitating scalar field},
  Comm. Math. Phys. \textbf{105} (1986), no.~3, 337--361.

\bibitem[Chr93]{dC1993}
\bysame, \emph{Bounded variation solutions of the spherically symmetric
  {E}instein-scalar field equations}, Comm. Pure Appl. Math. \textbf{46}
  (1993), no.~8, 1131--1220.

\bibitem[Chr02]{dC2002}
Demetrios. Christodoulou, \emph{{The Global Initial Value Problem in General
  Relativity}}, The Ninth Marcel Grossmann Meeting, World Scientific Publishing
  Company, December 2002, pp.~44--54.

\bibitem[CK90]{ChrKl90}
D.~Christodoulou and S.~Klainerman, \emph{Asymptotic properties of linear field
  equations in minkowski space}, Communications on Pure and Applied Mathematics
  \textbf{43} (1990), no.~2, 137--199.

\bibitem[CK93]{CK94}
Demetrios Christodoulou and Sergiu Klainerman, \emph{{The Global Nonlinear
  Stability of the Minkowski Space (PMS-41)}}, Princeton University Press,
  1993.

\bibitem[CK25]{xtCsK2025}
Xuantao Chen and Sergiu Klainerman, \emph{{Solving the constraint equation for
  general free data}}, arXiv \textbf{2512.22704} (2025).

\bibitem[CK26]{xtCsK2026}
\bysame, \emph{{Formation of Trapped Surfaces in Geodesic Foliation}},
  Communications in Mathematical Physics \textbf{407} (2026), no.~98.

\bibitem[Daf06]{mD2006}
Mihalis Dafermos, \emph{A note on the collapse of small data self-gravitating
  massless collisionless matter}, J. Hyperbolic Differ. Equ. \textbf{3} (2006),
  no.~4, 589--598.

\bibitem[DHR19]{DHR}
Mihalis Dafermos, Gustav Holzegel, and Igor Rodnianski, \emph{The linear
  stability of the {S}chwarzschild solution to gravitational perturbations},
  Acta Math. \textbf{222} (2019), no.~1, 1--214. \MR{3941803}

\bibitem[DHRT21]{DHRT}
Mihalis Dafermos, Gustav Holzegel, Igor Rodnianski, and Martin Taylor,
  \emph{The non-linear stability of the {S}chwarzschild family of black holes},
  arXiv \textbf{2104.08222} (2021).

\bibitem[DR09]{Dafermos2009ASpacetimes}
Mihalis Dafermos and Igor Rodnianski, \emph{A new physical-space approach to
  decay for the wave equation with applications to black hole spacetimes},
  XVIth International Congress on Mathematical Physics (2009).

\bibitem[FJS21]{Fajman2021TheSystem}
David Fajman, Jérémie Joudioux, and Jacques Smulevici, \emph{{The stability
  of the Minkowski space for the Einstein–Vlasov system}}, Analysis {\&} PDE
  \textbf{14} (2021), 425--531.

\bibitem[Fri86]{Friedrich1986OnStructure}
Helmut Friedrich, \emph{{On the existence of n-geodesically complete or future
  complete solutions of Einstein's field equations with smooth asymptotic
  structure}}, Communications in Mathematical Physics \textbf{107} (1986),
  no.~4, 587 – 609.

\bibitem[Fri18]{hF2018}
\bysame, \emph{Peeling or not peeling---is that the question?}, Classical
  Quantum Gravity \textbf{35} (2018), no.~8, 083001, 18.

\bibitem[FST25]{ajFjSaT2025}
Allen~Juntao Fang, J\'er\'emie Szeftel, and Arthur Touati, \emph{Spacelike
  initial data for black hole stability}, Comm. Math. Phys. \textbf{406}
  (2025), no.~10, Paper No. 235, 40. \MR{4951473}

\bibitem[Gau26]{oG2026}
Onyx Gautam, \emph{{Late-time tails for linear waves on radially symmetric
  stationary spacetimes of two space dimensions}}, arXiv \textbf{2605.03220}
  (2026).

\bibitem[GKS24]{GKS}
Elena Giorgi, Sergiu Klainerman, and J\'er\'emie Szeftel, \emph{Wave equations
  estimates and the nonlinear stability of slowly rotating {K}err black holes},
  Pure Appl. Math. Q. \textbf{20} (2024), no.~7, 2865--3849.

\bibitem[Gra20]{Graf2020GlobalData}
Olivier Graf, \emph{{Global nonlinear stability of Minkowski space for
  spacelike-characteristic initial data}}, M{\'{e}}moires de la
  Soci{\'{e}}t{\'{e}} math{\'{e}}matique de France (2020).

\bibitem[Hin23]{Hintz.exterior}
Peter Hintz, \emph{Exterior stability of {M}inkowski space in generalized
  harmonic gauge}, Arch. Ration. Mech. Anal. \textbf{247} (2023), no.~5, Paper
  No. 99, 45.

\bibitem[HKM77]{Hughes1977Well-posedRelativity}
Thomas J~R Hughes, Tosio Kato, and Jerrold~E Marsden, \emph{{Well-posed
  quasi-linear second-order hyperbolic systems with applications to nonlinear
  elastodynamics and general relativity}}, Archive for Rational Mechanics and
  Analysis \textbf{63} (1977), no.~3, 273--294.

\bibitem[HSW23]{Huneau2023TheSpacetime}
Cécile Huneau, Annalaura Stingo, and Zoe Wyatt, \emph{{The global stability of
  the Kaluza-Klein spacetime}}, arXiv \textbf{2307.15267} (2023).

\bibitem[Hun18]{huneau2018stability}
C{\'e}cile Huneau, \emph{Stability of minkowski space-time with a translation
  space-like killing field}, Annals of PDE \textbf{4} (2018), no.~1, 12.

\bibitem[HV20]{Hintz2017StabilityMetric}
Peter Hintz and András Vasy, \emph{Stability of {M}inkowski space and
  polyhomogeneity of the metric}, Annals of PDE \textbf{6} (2020).

\bibitem[IP22]{Ionescu2022TheAMS-213}
Alexandru Ionescu and Benoit Pausader, \emph{{The Einstein-Klein-Gordon Coupled
  System: Global Stability of the Minkowski Solution: (AMS-213)}}, 4 2022.

\bibitem[Keh22]{lmaK2022}
Leonhard M.~A. Kehrberger, \emph{The case against smooth null infinity {I}:
  {H}euristics and counter-examples}, Ann. Henri Poincar\'e \textbf{23} (2022),
  no.~3, 829--921.

\bibitem[Kei18]{Keir}
Joseph Keir, \emph{The weak null condition and global existence using the
  $p$-weighted energy method}, arXiv \textbf{1808.09982} (2018).

\bibitem[KK25]{KaKe2025}
Istvan Kadar and Lionor Kehrberger, \emph{{Scattering, Polyhomogeneity and
  Asymptotics for Quasilinear Wave Equations From Past to Future Null
  Infinity}}, arXiv \textbf{2501.09814} (2025).

\bibitem[KL23]{Kauffman2023GlobalCoordinates}
Christopher Kauffman and Hans Lindblad, \emph{{Global Stability of Minkowski
  Space for the Einstein--Maxwell--Klein--Gordon System in Generalized Wave
  Coordinates}}, Annales Henri Poincare \textbf{24} (2023), no.~11, 3837--3919.

\bibitem[KN03a]{Klainerman2003TheNicolo.}
Sergiu Klainerman and Francesco Nicol\`o, \emph{The evolution problem in
  general relativity}, Progress in mathematical physics v. 25,
  Birkh{\"{a}}user, Boston, 2003 (eng).

\bibitem[KN03b]{Klainerman2003PeelingEquations}
Sergiu Klainerman and Francesco Nicol{\`{o}}, \emph{{Peeling properties of
  asymptotically flat solutions to the Einstein vacuum equations}}, Classical
  and Quantum Gravity \textbf{20} (2003), no.~14, 3215.

\bibitem[KS23]{KS.Kerr}
Sergiu Klainerman and J\'er\'emie Szeftel, \emph{Kerr stability for small
  angular momentum}, Pure Appl. Math. Q. \textbf{19} (2023), no.~3, 791--1678.

\bibitem[Lin17]{hL2017}
Hans Lindblad, \emph{On the asymptotic behavior of solutions to the {E}instein
  vacuum equations in wave coordinates}, Comm. Math. Phys. \textbf{353} (2017),
  no.~1, 135--184.

\bibitem[LM16]{LeFloch2016TheFields}
Philippe~G LeFloch and Yue Ma, \emph{{The Global Nonlinear Stability of
  Minkowski Space for Self-gravitating Massive Fields}}, Communications in
  Mathematical Physics \textbf{346} (2016), no.~2, 603--665.

\bibitem[LO22]{LO.dispersive}
Jonathan Luk and Sung-Jin Oh, \emph{Global nonlinear stability of large
  dispersive solutions to the {E}instein equations}, Ann. Henri Poincar\'e
  \textbf{23} (2022), no.~7, 2391--2521.

\bibitem[LO24]{LO}
\bysame, \emph{Late time tail of waves on dynamic asymptotically flat
  spacetimes of odd space dimensions}, arXiv \textbf{2404.02220} (2024).

\bibitem[Loi09]{Loizelet2009}
Julien Loizelet, \emph{Solutions globales des équations
  d’{E}instein-{M}axwell}, Annales de la faculté des sciences de Toulouse
  Mathématiques \textbf{18} (2009), no.~3, 495--540 (fre).

\bibitem[LOY18]{LOY1}
Jonathan Luk, Sung-Jin Oh, and Shiwu Yang, \emph{Solutions to the
  {E}instein-scalar-field system in spherical symmetry with large bounded
  variation norms}, Ann. PDE \textbf{4} (2018), no.~1, Paper No. 3, 59.

\bibitem[LR05]{Lindblad2005GlobalCoordinates}
Hans Lindblad and Igor Rodnianski, \emph{{Global existence for the Einstein
  vacuum equations in wave coordinates}}, Commun. Math. Phys. \textbf{256}
  (2005), 43--110.

\bibitem[LT20]{Lindblad2020GlobalGauge}
Hans Lindblad and Martin Taylor, \emph{{Global Stability of Minkowski Space for
  the Einstein–Vlasov System in the Harmonic Gauge}}, Archive for Rational
  Mechanics and Analysis \textbf{235} (2020), no.~1, 517--633.

\bibitem[Mao26]{ycM2026}
Yuchen Mao, \emph{{L}ate-time {T}ails of the {E}instein {V}acuum equation near
  {M}inkowski {S}pacetime in {W}ave {G}auges}, in preparation (2026).

\bibitem[Mos16]{gM2016}
Georgios Moschidis, \emph{The {$r^p$}-weighted energy method of {D}afermos and
  {R}odnianski in general asymptotically flat spacetimes and applications},
  Ann. PDE \textbf{2} (2016), no.~1, Art. 6, 194.

\bibitem[MOT26]{MOT2026}
Yuchen Mao, Sung-Jin Oh, and Zhongkai Tao, \emph{Flexibility of asymptotically
  flat general relativistic initial data sets via recovery on curves}, in
  preparation (2026).

\bibitem[MTW73]{mtw}
Ch.W. Misner, K.S. Thorne, and J.A. Wheeler, \emph{{Gravitation}}, W. Freeman,
  1973.

\bibitem[NU62]{Newman1962BehaviorSpaces}
Ezra~T Newman and Theodore W~J Unti, \emph{{Behavior of Asymptotically Flat
  Empty Spaces}}, J. Math. Phys. \textbf{3} (1962), no.~5, 891.

\bibitem[Nü25]{Nutzi}
Andrea Nützi, \emph{Perturbations of minkowski spacetime with regular
  conformal compactification}, arXiv \textbf{2510.01964} (2025).

\bibitem[Ren90]{adR1990}
A.~D. Rendall, \emph{Reduction of the characteristic initial value problem to
  the {C}auchy problem and its applications to the {E}instein equations}, Proc.
  Roy. Soc. London Ser. A \textbf{427} (1990), no.~1872, 221--239.

\bibitem[Rin09]{Ringstrom2009TheRelativity}
Hans Ringstr{\"{o}}m, \emph{{The Cauchy problem in general relativity}}, EMS
  Press, 2009.

\bibitem[RR92]{gRadR1992}
G.~Rein and A.~D. Rendall, \emph{Global existence of solutions of the
  spherically symmetric {V}lasov-{E}instein system with small initial data},
  Comm. Math. Phys. \textbf{150} (1992), no.~3, 561--583.

\bibitem[Sac62]{rkS1962}
R.~K. Sachs, \emph{Gravitational waves in general relativity. {VIII}. {W}aves
  in asymptotically flat space-time}, Proc. Roy. Soc. London Ser. A
  \textbf{270} (1962), 103--126.

\bibitem[Sch13]{vS2013}
Volker Schlue, \emph{Decay of linear waves on higher-dimensional
  {S}chwarzschild black holes}, Anal. PDE \textbf{6} (2013), no.~3, 515--600.

\bibitem[She23]{Shen2023GlobalDecay}
Dawei Shen, \emph{{Global stability of Minkowski spacetime with minimal
  decay}}, arXiv \textbf{2310.07483} (2023).

\bibitem[She24a]{Shen2024StabilityRegions}
\bysame, \emph{{Stability of Minkowski spacetime in exterior regions}}, Pure
  Appl. Math. Quart. \textbf{20} (2024), no.~2, 757--868.

\bibitem[She24b]{dwShen2024}
\bysame, \emph{Stability of {M}inkowski spacetime in exterior regions}, Pure
  Appl. Math. Q. \textbf{20} (2024), no.~2, 757--868.

\bibitem[She26]{dwShen2026}
\bysame, \emph{{The Stability of Minkowski Spacetime}}, arXiv
  \textbf{2605.26506} (2026).

\bibitem[Smu25]{jS2025}
Jacques Smulevici, \emph{The stability of {Minkowski} space and its influence
  on the mathematical analysis of {General} {Relativity}}, Comptes Rendus.
  M\'ecanique \textbf{353} (2025), 519--542 (en).

\bibitem[Spe10]{Speck2010TheCoordinates}
Jared Speck, \emph{{The Global Stability of the Minkowski Spacetime Solution to
  the Einstein-Nonlinear Electromagnetic System in Wave Coordinates}}, Analysis
  {\&} PDE \textbf{7} (2010).

\bibitem[SY81]{Schoen1981ProofII}
Richard Schoen and Shing-Tung Yau, \emph{{Proof of the positive mass theorem.
  II}}, Communications in Mathematical Physics \textbf{79} (1981), no.~2,
  231--260.

\bibitem[Tay17]{Taylor2017TheSystem}
Martin Taylor, \emph{{The Global Nonlinear Stability of Minkowski Space for the
  Massless Einstein–Vlasov System}}, Annals of PDE \textbf{3} (2017), no.~1,
  9.

\bibitem[Wal84]{rmW1984}
Robert~M. Wald, \emph{General relativity}, University of Chicago Press,
  Chicago, IL, 1984.

\bibitem[Wan20]{Wang2020AnEquations}
Qian Wang, \emph{{An intrinsic hyperboloid approach for Einstein Klein–Gordon
  equations}}, Journal of Differential Geometry \textbf{115} (2020), no.~1, 27
  – 109.

\bibitem[Wan22]{Wang2022GlobalSystem}
Xuecheng Wang, \emph{{Global stability of the Minkowski spacetime for the
  Einstein-Vlasov system}}, arXiv \textbf{2210.00512} (2022).

\bibitem[Wit81]{Witten1981ATheorem}
Edward Witten, \emph{{A new proof of the positive energy theorem}},
  Communications in Mathematical Physics \textbf{80} (1981), no.~3, 381--402.

\bibitem[ZY00]{Zipser2000TheEquations}
Nina Zipser and S~T Yau, \emph{{The global nonlinear stability of the trivial
  solution of the Einstein--Maxwell equations}}, Ph.D. thesis, United States --
  Massachusetts, 2000, p.~198.

\end{thebibliography}

\end{document}